\documentclass[journal,onecolumn,draftclsnofoot,]{IEEEtran}

\usepackage[utf8]{inputenc}
\usepackage[english]{babel}
\usepackage{epsfig}
\usepackage{amsfonts}
\usepackage{ifpdf}
\usepackage{multirow}
\usepackage{amssymb}
\usepackage{amsthm}
\usepackage{mathtools}
\usepackage{mathrsfs}
\usepackage{textgreek}
\usepackage{grffile}
\usepackage{xcolor}
\usepackage{bm}
\usepackage{cite}
\usepackage{setspace}
\usepackage{nccmath}

\usepackage[usestackEOL]{stackengine}

  \stackMath

\ifCLASSINFOpdf
\else
\fi
\usepackage{amsmath}
\usepackage{algorithmic}
\usepackage{algorithm}
\usepackage{graphicx}

\newtheorem{definition}{Definition}

\newtheorem{remark}{Remark}
\newtheorem{lemma}{Lemma}
\newtheorem{theorem}{Theorem}
\newtheorem{corollary}{Corollary}

\newcommand{\norm}[1]{\left\lVert#1\right\rVert}

\ifCLASSOPTIONcompsoc
  \usepackage[caption=false,font=normalsize,labelfont=sf,textfont=sf]{subfig}
\else
  \usepackage[caption=false,font=footnotesize]{subfig}
\fi

\usepackage{stfloats}
\begin{document}
%


\title{Fundamentals of LEO Based Localization}
\author{Don-Roberts~Emenonye,
        Harpreet~S.~Dhillon,
        and~R.~Michael~Buehrer
\thanks{D.-R. Emenonye, H. S. Dhillon, and R. M.  Buehrer are with Wireless@VT,  Bradley Department of Electrical and Computer Engineering, Virginia Tech,  Blacksburg,
VA, 24061, USA. Email: \{donroberts, hdhillon, rbuehrer\}@vt.edu. The support of the US National Science Foundation (Grants ECCS-2030215, CNS-1923807, and CNS-2107276) is gratefully acknowledged. Part of this work has been accepted to IEEE VTC 2024, Washington, DC, USA \cite{emenonye2023_VTC_conf_Minimal,emenonye2023_VTC_conf_unsyn}.
}
}

\maketitle
\IEEEpeerreviewmaketitle
\begin{abstract}
In this paper, we derive the fundamental limits of low earth orbit (LEO) enabled localization by analyzing the available information in signals from multiple LEOs during different transmission time slots received on a multiple antennas and evaluate the utility of these signals for $9$D localization ($3$D position, $3$D orientation, and $3$D velocity estimation). We start by deriving the Fisher Information Matrix (FIM) for the channel parameters that are present in the signals received from LEOs in the same or multiple constellations during multiple transmission time slots.  To accomplish this, we define a system model that captures i)  time offset between LEOs caused by having relatively cheap clocks, ii) frequency offset between LEOs, iii) the unknown Doppler rate caused by high mobility LEOs, and iv) multiple transmission time slots from a particular LEO. We transform the FIM for the channel parameters to the FIM for the location parameters and determine the required conditions for localization. To do this, we start with the $3$D localization cases: i) $3$D positioning with known velocity and orientation, ii) $3$D orientation estimation with known position and velocity, and iii)  $3$D velocity estimation with known position and orientation. Subsequently, we derive the FIM for the full $9$D localization case ($3$D position, $3$D orientation, and $3$D velocity estimation) in terms of the FIM for the $3$D localization. Using these results, we determine the number of LEOs, the operating frequency, the number of transmission time slots, and the number of receive antennas that allow for different levels of location estimation. We then provide insights into the interaction between the number of LEOs, the operating frequency, the number of transmission time slots, and the number of receive antennas. One key result is that in the presence of time and frequency offsets and Doppler rate, it is possible to perform $9$D localization ($3$D position, $3$D velocity, and $3$D orientation estimation) of a receiver by utilizing the signals from three LEO satellites observed during three transmission time slots received through multiple receive antennas.

\end{abstract}

\section{Introduction}
There has been renewed interest in the use of low earth orbit satellites as evidenced by the launch of several new satellites into existing LEO constellations, such as Orbcomm, Iridium, and Globalstar, as well as the creation of new constellations such as Boeing, SpaceMobile, Oneweb, Telesat, Kuiper, and Starlink. Because this cluster of mega-constellations will be closer to the earth than the current satellites in global navigation satellite systems (GNSS), they will have shorter propagation delays and encounter lower losses, thereby providing greater potential accuracy in specific localization scenarios. Moreover, LEO satellites could be used when the GNSS signals are unavailable (such as in deep
urban canyons, under dense foliage, during unintentional interference,
and intentional jamming) or untrustworthy (e.g.,
under malicious spoofing attacks). Due to these reasons, utilizing these LEO satellites for localization is a research direction of increasing interest.

Research on LEO-based localization is ongoing and ranges from opportunistic use to dedicated localization systems where the difference between these approaches lies in the amount of information available about the properties of the LEO signals and the synchronization accuracy. Although existing research has varied from opportunistic systems that may utilize multiple constellations while having full knowledge of the signal properties of specific constellations and no or partial knowledge of the signal properties of the other constellations to dedicated systems that always have full knowledge of all the signal properties of all constellations. The fundamental limits of the information available for LEO-based $9$D localization ($3$D position, $3$D orientation, and $3$D velocity estimation) from the same or different constellations have not been thoroughly derived. Hence, in this paper, we use the Fisher Information Matrix (FIM) to characterize the available information for $9$D localization in the signals at a single/multiple antenna receiver that arrives from LEOs in the same or different constellations during single/multiple transmission time slots, which leads to critical insights into the number of LEO satellites, the operating frequency, the number of transmission time slots, and the number of receive antennas required for different levels of $9$D location estimation.

\subsection{Related Works}
The following three existing research directions are related to the work in this paper: i) LEO-based localization, ii) Large antenna array-enabled localization, and iii) Opportunistic localization in $5$G systems. The pertinent works from these directions are reviewed in this section.

\subsubsection{LEO-based localization}
Localization using LEOs ranges from dedicated \cite{Fundamental_Performance_Bounds_for_Carrier_Phase_Positioning_LEO_PNT,Broadband_LEO_Constellations_for_Navigation,Economical_Fused_LEO_GNSS,Empowering_the_Tracking_Performance_of_LEOBased_Positioning_by_Means_of_Meta_Signals,Performance_Analysis_of_a_Multi_Slope_Chirp_Spread_Spectrum,Integrated_Communications_and_Localization_for_Massive_MIMO_LEO_Satellite} to semi-opportunistic to opportunistic techniques \cite{Psiaki2020NavigationUC,Kassas2019NewAgeSN,Navigation_With_Differential_Carrier_Phase_Measurements_From_Megaconstellation_LEO_Satellites,A_Hybrid_Analytical_Machine_Learning_Approach_for_LEO_Satellite_Orbit_Prediction,Doppler_effect_Downlink_Receivers_OFDM_Low_earth_orbit_satellites_Bandwidth_Synchronization_Doppler_positioning_low_Earth_orbit,Ad_Astra_STAN_With_Megaconstellation_LEO_Satellites,A_Hybrid_Analytical_Machine_Learning_Approach_for_LEO_Satellite_Orbit_Prediction_1,Assessing_Machine_Learning_for_LEO_Satellite_Orbit_Determination_in_Simultaneous_Tracking_and_Navigation,Cognitive_Navigation_With_Unknown_OFDM_signals_With_Application_Terrestrial_5G_Starlink,Observability_Analysis_of_Receiver_Localization_Pseudorange,Positioning_with_Starlink_LEO_Satellites_A_Blind_Doppler_Spectral_Approach,Receiver_Design_for_Doppler_Positioning_with_Leo_Satellites,Unveiling_Starlink_LEO_Satellite_OFDM_Like_Signal_Structure_Enabling_Precise_Positioning}. This spectrum is defined by the knowledge or lack of knowledge of the structure of the reference signal, which is determined by the length of the reference signal, its values, and its periodicity. At the dedicated end of the spectrum, the signal structure is completely known, and at the other end, the signal structure is completely unknown. In between the two extremes exist semi-opportunistic techniques where only partial knowledge of the signal structure of some satellites is known.

The authors in \cite{Broadband_LEO_Constellations_for_Navigation} propose a dedicated framework for utilizing broadband LEO constellations for navigation. The proposed framework utilizes delay measurements and is evaluated by considering the position errors as a function of the product of the geometric dilution of precision (GDOP) and the ranging errors. The ranging errors incorporate clock offset, and the orbit is assumed to be known, with the justification that the orbit can be accurately predicted by fifteen observing ground stations (just like the framework used in Galileo). In \cite{Economical_Fused_LEO_GNSS}, a vision of a fused GNSS and dedicated LEO architecture in which existing clocks, modems, antennas, and spectrum of broadband satellite mega-constellations are dual-purposed for delay-based positioning is provided. In \cite{Empowering_the_Tracking_Performance_of_LEOBased_Positioning_by_Means_of_Meta_Signals}, the authors provide a dedicated framework for utilizing Doppler measurements from Amazon Kuiper Satellites for receiver positioning. Reference signal design for a dedicated framework using delay measurements is proposed in \cite{Performance_Analysis_of_a_Multi_Slope_Chirp_Spread_Spectrum}. A dedicated framework is provided for massive multiple-input multiple-output (MIMO)-based integrated localization and communication, in which the delay and Doppler are used to position a receiver to improve the transmission rate in \cite{Integrated_Communications_and_Localization_for_Massive_MIMO_LEO_Satellite}. {\em In summary, current research on the dedicated end has not provided a rigorous explanation of the available information in LEO satellites from the same or different constellations that are received by a multiple antenna receiver during multiple transmission time slots. Hence, in this paper, with the assumption that the LEO orbits are known\footnote{The satellites’ positions and velocities
can be obtained through two-line element (TLE) files. }, we use information theory to rigorously characterize this available information and its utility for the $9$D localization of a receiver ($3$D position, $3$D orientation, and $3$D velocity estimation).}

The authors in \cite{Psiaki2020NavigationUC} provide an opportunistic experimental framework to use at least eight Doppler shift measurements to give estimates for the $3$D position, $3$D velocity, time offset, and time offset rate. In \cite{Kassas2019NewAgeSN}, an opportunistic framework that combines inertial measurement units with the signals from the LEO satellites is developed to experimentally estimate the position of LEO satellites, the time offset, and the LEO orbit. Authors in \cite{Navigation_With_Differential_Carrier_Phase_Measurements_From_Megaconstellation_LEO_Satellites} demonstrate experimentally that the signals from two Orbcomm satellites can be used to opportunistically track an unmanned aerial vehicle (UAV) for two minutes with a position error of $15 \text{ m}.$ In \cite{A_Hybrid_Analytical_Machine_Learning_Approach_for_LEO_Satellite_Orbit_Prediction}, an opportunistic framework using Doppler measurement taken from Orbocomm satellites over multiple time intervals is used to estimate the orbital parameters and receiver position. Authors in \cite{Doppler_effect_Downlink_Receivers_OFDM_Low_earth_orbit_satellites_Bandwidth_Synchronization_Doppler_positioning_low_Earth_orbit} develop an opportunistic framework to detect reference signals from six Starlink satellites, and subsequently use this framework to achieve a receiver positioning error of $20 \text{ m}.$ The authors in \cite{Ad_Astra_STAN_With_Megaconstellation_LEO_Satellites} propose an opportunistic framework to combine IMU measurements with range and Doppler measurements. This framework utilizes two Orbcomm, one Iridium and three Starlink satellites to achieve position errors of 27.1 m and 18.4 m, respectively. A ground receiver localizes itself while estimating the noise covariance matrix of a single Orbcomm satellite with Doppler observed over time in \cite{A_Hybrid_Analytical_Machine_Learning_Approach_for_LEO_Satellite_Orbit_Prediction_1}. The authors in \cite{Assessing_Machine_Learning_for_LEO_Satellite_Orbit_Determination_in_Simultaneous_Tracking_and_Navigation} provide an opportunistic machine learning framework that uses Doppler measurements to give an estimate of the LEO orbit and receiver position. A Doppler-based opportunistic experimental framework is developed to utilize $5$G and LEO signals to localize a stationary receiver in \cite{Cognitive_Navigation_With_Unknown_OFDM_signals_With_Application_Terrestrial_5G_Starlink}. The authors in \cite{Observability_Analysis_of_Receiver_Localization_Pseudorange} develop an opportunistic framework that uses range measurements from a single Orbocmm satellite with a known orbit to position a receiver. In \cite{Positioning_with_Starlink_LEO_Satellites_A_Blind_Doppler_Spectral_Approach}, the received signal frequency spectrum is mathematically characterized, a characterization that accounts for the high dynamic nature of the channel, which results from the speed of LEO satellites. A Doppler discriminator is proposed to differentiate between satellites, and a Doppler tracking algorithm is proposed. Finally, this opportunistic framework uses Doppler measurements to obtain a $2$D position error of 4.3 m. The authors in \cite{Receiver_Design_for_Doppler_Positioning_with_Leo_Satellites} assume that the satellite's position and velocity are known, then a receiver is proposed. Subsequently, Doppler measurements from two Orbcomm satellites are used to position a receiver to an accuracy of 11 m opportunistically. In \cite{Unveiling_Starlink_LEO_Satellite_OFDM_Like_Signal_Structure_Enabling_Precise_Positioning}, the spectrum of six LEO satellites is investigated. It was noticed that three of the satellites used tones while the other three used OFDM. The satellites were used to achieve a horizontal positioning error of 6.5 m opportunistically. {\em Current research on the opportunistic end has primarily focused on the experimental use of Doppler measurements to find the $3$D position of a receiver. Our contribution to this end of the research involves providing a favorable lower bound on the achievable accuracy when Doppler measurements are used to aid the $9$D localization of a receiver ($3$D position, $3$D orientation, and $3$D velocity estimation).}

\subsubsection{Large antenna array enabled localization}
The use of higher frequency bands has resulted in the possibility of having a large number of antennas on transceivers, which is termed massive MIMO. The degrees of freedom offered by massive MIMO has enabled single-anchor localization under the parameterization of the received signal with the angle of departure (AOD), angle of arrival (AOA), and time of arrival (TOA) \cite{garcia2017direct,8240645,8515231,8356190,guerra2018single,emenonye2023limits,fascista2021downlink,8755880,li2019massive,9082200}. Authors in \cite{garcia2017direct} use distributed anchors to estimate multiple time of arrival (TOA) and angle of arrival (AOA) measurements, and the TOAs are used to restrict the source location to a convex set. While the AOAs are used to provide an exact estimate of the source location. The FIM for $2$D positioning is derived in \cite{8240645}, and expectation-maximization algorithms that attain these bounds are presented. These bounds are used in \cite{8515231} to show that for sufficiently high temporal and spatial resolution, non-line of sight parameters always provide position and orientation information, increasing position and orientation estimation accuracy. These bounds are extended to the $3$D positioning, and $2$D orientation case in \cite{8356190}, and the structure of these bounds is rigorously presented. Additional limits are presented for the uplink of a single anchor system in \cite{guerra2018single}. In that paper, it is shown that the Cramer Rao Lower Bound (CRLB) is unique in the limit of the number of receive antennas because each possible transmit position leads to distinct observations at the base station (BS). The authors in \cite{emenonye2023limits} show that while the transmitter or receiver’s 3D orientation can be
jointly estimated with the transmitter or receiver’s 3D position in
the near-field propagation regime, only the transmitter or receiver’s
2D orientation can be jointly estimated with the transmitter or
receiver’s 2D position in the far-field propagation regime. The task of single antenna receiver position is tackled with multiple observations \cite{8755880} and with a single observation \cite{fascista2021downlink}. In \cite{li2019massive}, the sparsity offered by the massive MIMO setup is used for simultaneous localization and mapping, and localization under hardware impairments is investigated in \cite{9082200}. An alternative to the parameterizations presented in \cite{garcia2017direct,8240645,8515231,8356190,guerra2018single,emenonye2023limits,fascista2021downlink,8755880,li2019massive,9082200} is presented in \cite{5571889,9606768,7364259}. In \cite{5571900}, {\em a priori} information about the channel parameters and {\em a priori} information about the transmitter or receiver location are included in a detailed investigation of the fundamental limits of wideband localization. This work also shows that NLOS is only helpful with {\em a priori} information about the channel parameters, and it is extended to cooperative localization in \cite{5571889}. The cooperative localization framework is then developed for collaborative positioning in massive networks \cite{9606768}. Finally, the use of reconfigurable intelligent surfaces (RISs) to aid localization is presented in \cite{8264743,9729782,9781656
,9508872,9625826,9500663,9782100,9528041,emenonye2022fundamentals,emenonye2023_ICC_conf_workshop,emenonye2022ris,emenonye2023_ICC_conf, RIS_Aided_Kinematic,OTFS_Enabled_RIS,9774917}. These prior works on RIS-aided localization can be grouped mainly into i) continuous RIS\cite{8264743,9729782,9781656} and discrete RIS \cite{emenonye2022fundamentals,emenonye2022ris,emenonye2023_ICC_conf,emenonye2023_ICC_conf_workshop, RIS_Aided_Kinematic, OTFS_Enabled_RIS,9508872,9625826,9500663,9782100,9528041,9774917}, and ii) near-field\cite{emenonye2022ris,emenonye2023_ICC_conf, RIS_Aided_Kinematic, OTFS_Enabled_RIS,8264743,9729782,9781656,9508872,9500663,9625826,9774917} and far-field propagation \cite{emenonye2022fundamentals,emenonye2023_ICC_conf_workshop,9782100,9528041}.
{\em Although extensive work has been done in this area, our contribution to existing work is as follows: i) we derive the available information about channel parameters in the received signals from fast-moving anchors with known trajectories that transmit during multiple transmission time
slots, ii) we present the available information for $9$D localization utilizing signals from multiple fast-moving anchors which are not synchronized in time and frequency that are received by a multiple antenna receiver, and iii) we present the minimal infrastructure needed for different levels of $9$D localization considering time and frequency offset among the anchors, in addition to a high Doppler rate.}

\subsubsection{Opportunistic localization in $5$G systems}
Several signals with good correlation properties are transmitted as always-on signals, such as the primary and secondary synchronization signals and the physical broadcast channel block, or on-demand signals, such as demodulation reference signals, phase tracking reference signals, and sounding reference signals. In \cite{9573365}, a cognitive opportunistic receiver is designed to detect the active $5$G base stations and estimate the number of base stations and their associated reference signals. The reference signals can be either always-on or on-demand signals. In that paper, a sequential generalized likelihood ratio detector
is used to detect the presence of multiple base stations on the
same channel and provide an estimate of the number of active BSs. This detector uses Doppler frequencies of the base stations to define the signal subspace, after which the reference signals are estimated. Subsequently, a UAV is tracked over a $416 \text{ m}$ trajectory with an error of $4.15 \text{ m}$. The authors in \cite{9369049} present a framework to utilize always-on reference signals for localization; after separating the effect of the clock bias and drift from the estimated range measurements, the ranging error standard deviation is shown to be $1.19 \text{ m}$. Finally, in \cite{kassas2021carpe}, an opportunistic utilization of $5$G always-on signals is studied, and a software-defined radio is used to extract navigation observables from $5$G signals. Finally, a localization framework using an extended Kalman filter is utilized to estimate the receiver’s position.

\subsection{Contribution}
This paper focuses on the localization of a receiver with multiple antennas using reference signals received from LEOs, which could be from the same or various constellations, during multiple transmission time slots. With this setup, our main contributions are:
\subsubsection{determining the available information about channel parameters in the received signal from LEOs during multiple transmission time slots} We derive the FIM for the channel parameters that are present in the signals received from LEOs in the same or multiple constellations during multiple transmission time slots.  To enable this, we define a system model that captures i) the possibility of a time offset between LEOs caused by having cheap synchronization clocks, ii) the possibility of a frequency offset between LEOs, iii) the unknown Doppler rate caused by the short coherence time in high mobility LEO based satellite systems, and iv) multiple transmission time slots from a particular LEO. One key result is that while the channel gain affects the information available about the frequency offset and/or Doppler rate, knowledge, or lack thereof of the delays, the channel gain and the time offset do not affect the available information about the frequency offset and/or Doppler rate. Hence, the accuracy achievable while estimating the frequency offset and/or Doppler rate is independent of the knowledge of the delays, channel gain, and the time offset.

\subsubsection{determining the available information for $9$D localization}

We transform the FIM for the channel parameters to the FIM for the location parameters and show the possible localization conditions. To do this, we start with the $3$D localization cases: i) $3$D positioning with known $3$D velocity and known $3$D orientation, ii) $3$D orientation estimation with known $3$D position and known $3$D velocity, and iii)  $3$D velocity estimation with known $3$D position and known $3$D orientation. Subsequently, we derive the FIM for the $6$D localization cases in terms of the FIM for the $3$D localization cases:  i) $3$D positioning and $3$D orientation estimation with known $3$D velocity, ii) $3$D velocity and orientation estimation with known $3$D position, and iii)  $3$D position and velocity estimation with known $3$D orientation. Finally, we derive the FIM for the $9$D localization ($3$D position, $3$D orientation, and $3$D velocity estimation) in terms of the FIM for the $3$D localization. With these derivations, we present the number of LEOs, the operating frequency, the number of transmission time slots, and the number of receive antennas that allow for different levels of location estimation. 

\subsubsection{determining the minimal infrastructure for $9$D localization}

We present the minimal infrastructure needed for various levels of localization. These key insights are direct results of the interaction between the number of LEOs, the operating frequency, the number of transmission time slots, and the number of receive antennas. They are as follows: i) the $3$D position estimation of a receiver in the presence of time and frequency offsets and Doppler rate when the $3$D orientation and  $3$D velocity are known is possible by utilizing the signal from a single LEO satellite observed during four time slots on a single receive antenna, ii) the $3$D estimation of a receiver's orientation in the presence of time and frequency offsets and Doppler rate when the $3$D position and  $3$D velocity are known is possible by utilizing the signal from two LEO satellites observed during one transmission time slot on multiple receive antennas or by utilizing the signal from a single LEO satellite observed during two transmission time slots on multiple receive antennas, iii) the $3$D estimation of a receiver's velocity in the presence of time and frequency offsets and Doppler rate when the $3$D position and  $3$D orientation are known is possible by utilizing the signal from two LEO satellites observed during two transmission time slots on a single receive antenna or by utilizing the signal from a single LEO satellite observed during three transmission time slots on a single receive antenna, and iv) even in the presence of time and frequency offsets and Doppler rate, it is possible to perform $9$D localization ($3$D position, $3$D velocity, and $3$D orientation estimation) of a receiver by utilizing the signals from  three LEO satellites observed during three transmission time slots received through multiple receive antennas.

It is important to note that while we have focused on LEO-based systems, this work applies to general $9$D localization using mobile anchors with known trajectories and velocities. For example, our result states that even in the presence of time and frequency offsets and Doppler rate, the $3$D orientation estimation of a multi-antenna receiver, when the $3$D position and $3$D velocity are known, can be estimated using signals from at least two satellites during a single transmission time slot can be generalized to the case when the anchors are not satellites. 

Another example of a result that generalizes to the case where the anchors are not satellites is that even in the presence of time and frequency offsets and Doppler rate, it is possible to perform $9$D localization ($3$D position, $3$D velocity, and $3$D orientation estimation) of a receiver by utilizing the signals from three LEO satellites observed during three transmission time slots received through multiple receive antennas.

\textit{Notation:}
 The function $
 \bm{F}_{\bm{v}}(\bm{w} ; \bm{x}, \bm{y}) \triangleq \mathbb{E}_{\bm{v}}\left\{\left[\nabla_{ \bm{x}} \ln f(\bm{w})\right]\left[\nabla_{\bm{y}} \ln f(\bm{w})\right]^{\mathrm{T}}\right\}
 $ and $ \bm{G}_{\bm{v}}(\bm{w} ; \bm{x}, \bm{y})$
 describes the loss of information in the FIM defined by $
 \bm{F}_{\bm{v}}(\bm{w} ; \bm{x}, \bm{y})$ due to uncertainty in the nuisance parameters. The inner product of $\bm{x}$ is $\norm{\bm{x}}^2$ and the outer product of $\bm{x}$ is $\norm{\bm{x}^{\mathrm{T}}}^2$. $\nabla_{x} y$ is the first derivative of $y$ with respect to $x$. 
\section{System Model}
We consider $N_B$ single antenna LEO satellites, each communicating with a receiver with $N_U$ antennas, through transmissions in $N_{K}$ different time slots. The transmission slots are spaced by $\Delta_t$. At the $k^{\text{th}}$ transmission time slot, the $N_B$ LEO satellites are located at $\bm{p}_{b,k}, \; \; b \in \{1,2,\cdots, N_B\}  \text{ and } k \in \{1,2,\cdots, N_K\}$. The points, $\bm{p}_{b,k}$, are described with respect to a global origin. During the $k^{\text{th}}$ time slot, the receiver has an arbitrary but known geometry with its centroid located at $\bm{p}_{U,k}.$ During the $k^{\text{th}}$ time slot, the point, $\bm{s}_{u,k}$, describes the $u^{\text{th}}$ receive antenna with respect to the centroid while the point, $\bm{p}_{u,k}$, describes the position of this element with respect to the global origin as $\bm{p}_{u,k}= \bm{p}_{U,k} + \bm{s}_{u,k}$. The point, $\bm{s}_{u}$, can be written as $\bm{s}_{u} = \bm{Q}_{U} \Tilde{\bm{s}}_{u}$ where $\Tilde{\bm{s}}_{u}$ aligns with the global reference axis  and $\bm{Q}_{U} = \bm{Q}\left(\alpha_{U}, \psi_{U}, \varphi_{U}\right)$ defines a $3$D rotation matrix \cite{lavalle2006planning}. The orientation angles of the receiver are vectorized as $\bm{\Phi}_{U} = \left[\alpha_{U}, \psi_{U}, \varphi_{U}\right]^{\mathrm{T}}$.  The centroid of the receiver at point, $\bm{p}_{U,k}$ with respect to the $b^{\text{th}}$ LEO can be written as $\bm{p}_{U,k} = \bm{p}_{b,k} + d_{bU,k}\bm{\Delta}_{bU,k}$ where  $d_{bU,k}$ is the distance from point $\bm{p}_{b,k}$ to point $\bm{p}_{U,k}$ and $\bm{\Delta}_{bU,k}$ is the corresponding unit direction vector $\bm{\Delta}_{bU,k} = [\cos \phi_{bU,k} \sin \theta_{bU,k}, \sin \phi_{bU,k} \sin \theta_{bU,k}, \cos \theta_{bU,k}]^{\mathrm{T}}$. During the $k^{\text{th}}$ transmission time slot, the angles $\phi_{bU,k}$  and $\theta_{bU,k}$ represent the angle in the azimuth and elevation from the $b^{\text{th}}$ LEO satellite to the receiver.

\begin{figure}
\centering
    \fbox{\includegraphics[clip, trim=9.5cm 9.2cm 10.5cm 3cm,width=0.7
    \textwidth]{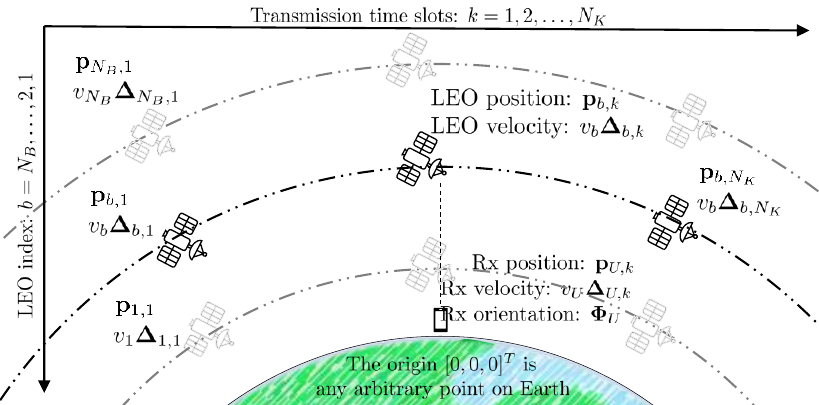}}
    \caption{LEO-based localization systems with $N_B$ LEOs transmitting during $N_K$ transmission time slots to a receiver with $N_U$ antennas.}
    \label{System_model_1}
\end{figure}

\subsection{Transmit and Receive Processing}
The $N_B$ LEO satellites transmit in $N_K$ consecutive time slots that are equally spaced in time. At time $t$, during the $k^{\text{th}}$ time slot, the $b^{\text{th}}$ LEO satellite uses quadrature modulation and transmits the following signal to the receiver
\begin{equation}
    \begin{aligned}
        x_{b,k}[t] = s_{b,k}[t] \operatorname{exp}_{}{(j2 \pi f_c t )},
    \end{aligned}
\end{equation}
where $s_{b,k}[t]$ is the complex signal envelope of the signal transmitted by the $b^{\text{th}}$ LEO satellite during the $k^{\text{th}}$ time slot, and $f_c = c / \lambda$ is the operating frequency of LEO satellites. The speed of light is $c$, and $\lambda$ is the operating wavelength. The channel model from the LEO satellites to the receiver consists only of the LOS paths. With this channel model and the transmit signal, the signal at the $u^{\text{th}}$ receive antenna during the $k^{th}$ time slot is
\begin{equation}
\begin{aligned}
\label{equ:receive_signal}
y_{u,k}[t] &= \sum_{b}^{N_B} y_{bu,k}[t], = \sum_{b=1}^{N_b} \beta_{bu,k}  \sqrt{2} \Re\left\{s_{b,k}[t_{obu,k}]  \operatorname{exp}(j( 2 \pi f_{ob,k} t_{obu,k}))\right\} + {n}_{u,k}[t], \\
&= {\mu}^{}_{u,k}[t] + {n}_{u,k}[t],
\end{aligned}
\end{equation}
where  ${\mu}^{}_{u,k}[t]$ and ${n}_{u,k}[t] \sim \mathcal{C}\mathcal{N}(0,N_0)$ are the noise-free part (useful part) of the signal and the thermal noise local to the receiver's antenna array, respectively. Also, $\beta_{bu,k}$ is the channel gain from the $b^{\text{th}}$ LEO satellite observed at the $u^{\text{th}}$ receive antenna during the $k^{\text{th}}$ time slot, $f_{ob,k} = f_c(1-\nu_{b,k}) +\epsilon_{b}$ is the observed frequency at receiver with respect to the $b^{\text{th}}$ LEO satellite, and $t_{obu,k}= t-\tau_{bu,k} +\delta_{b}$ is the effective time duration. In the observed frequency, $\nu_{b,k}$ is the Doppler with respect to the $b^{\text{th}}$ LEO satellite, and $\epsilon_{b}$ is the frequency offset measured with respect to the $b^{\text{th}}$ LEO satellite.

In the effective time duration, $\delta_{b}$ is the time offset at the receiver measured with respect to the $b^{\text{th}}$ LEO satellite, and the delay from the $u^{\text{th}}$ receive antenna to the $b^{\text{th}}$ LEO satellite during the $k^{\text{th}}$ time slot is
$$
\tau_{bu,k} \triangleq  \frac{\left\|\mathbf{p}_{u,k}-\mathbf{p}_{b,k}\right\|}{c}.
$$

\begin{remark}
The unknown Doppler rate due to the speed of the LEOs is captured by the fact that the Doppler measurements, $\nu_{b,k}$, change every transmission time slot. In the subsequent sections, we will show that these changing Doppler measurements provide more degrees of freedom for localization. Also, the offset, $\delta_{b}$, captures the unknown time offset as well as the unknown ionospheric and tropospheric delay concerning the $b^{\text{th}}$ LEO satellite.
\end{remark}

Here, the position of the $b^{\text{th}}$ LEO satellite and the $u^{\text{th}}$ receive antenna during the $k^{\text{th}}$ time slot is
$$
\begin{aligned}
    \mathbf{p}_{b,k} &= \mathbf{p}_{b,o} + \Tilde{\mathbf{p}}_{b,k}, \\
    \mathbf{p}_{u,k} &= \mathbf{p}_{u,o} + \Tilde{\mathbf{p}}_{U,k},
\end{aligned}
$$
where $\mathbf{p}_{b,o}$ and $\mathbf{p}_{u,o}$ are the reference points of the $b^{\text{th}}$ LEO satellite and the $u^{\text{th}}$ receive antenna, respectively. The distance travelled by the $b^{\text{th}}$ LEO satellite and the $u^{\text{th}}$ receive antenna are  $\Tilde{\mathbf{p}}_{b,k}$ and $\Tilde{\mathbf{p}}_{u,k}$, respectively. These traveled distances can be described as
$$
\begin{aligned}
\Tilde{\mathbf{p}}_{b,k} &= (k - 1) \Delta_{t} v_{b} \mathbf{\Delta}_{b,k}, \\
    \Tilde{\mathbf{p}}_{U,k} &= (k - 1) \Delta_{t} v_{U} \bm{\Delta}_{U,k}. \\
    \end{aligned}
$$
Here, $v_{b}$ and $v_{U}$ are speeds of the $b^{\text{th}}$ LEO satellite and receiver, respectively. The associated directions are defined as $\bm{\Delta}_{b,k} = [\cos \phi_{b,k} \sin \theta_{b,k}, \sin \phi_{b,k} \sin \theta_{b,k}, \cos \theta_{b,k}]^{\mathrm{T}}$ and $\bm{\Delta}_{U,k} = [\cos \phi_{U,k} \sin \theta_{U,k}, \sin \phi_{U,k} \sin \theta_{U,k}, \cos \theta_{U,k}]^{\mathrm{T}}$, respectively. Now, the velocity of the $b^{\text{th}}$ LEO satellite and the velocity of the receiver are $\bm{v}_{b} = v_{b} \bm{\Delta}_{b,k}$ and $\bm{v}_{U} = v_{U} \bm{\Delta}_{U,k}$, respectively. Hence, the Doppler observed by the receiver from the $b^{\text{th}}$ LEO satellite is
$$
\nu_{b,k} = \bm{\Delta}_{bU,k}^{\mathrm{T}} \frac{(\bm{v}_{b,k} - \bm{v}_{U,k})}{c}.
$$

\subsection{Properties of the Received Signal}
In this section, we discuss the properties that are observable in the signal at the receiver across all receive antennas and during all the transmission slots. To accomplish this, we consider the Fourier transform of the baseband signal that is transmitted by the $b^{\text{th}}$ LEO satellite at time $t$ during the $k^{\text{th}}$ time slot
$$
S_{b,k}[f] \triangleq \frac{1}{\sqrt{2 \pi}} \int_{-\infty}^{\infty} s_{b,k}[t] \operatorname{exp}_{}{(-j2 \pi f t )} \; \;    \text{d} t.
$$
This Fourier transform is called the spectral density. In the above equation, due to the finite bandwidth constraint (bandlimited signal), the baseband signal $s_{b,k}[t]$ must be infinite in time and takes the form $s_{b,k}[t] = \sum_{i = -\infty}^{\infty} s_{ib,k} f[t-i T]$. Here, $s_{ib,k}$ is the appropriate communication symbol from a set of equally likely communication symbols, $f[t-i T]$ is the appropriate pulse, and $T$ is the pulse duration.
In any communication system, a fundamental property is the bandwidth occupied by the non-negative real-valued spectral density, $S_{b,k}[f]$. The bandwidth is a measure of the width of a non-negative real-valued spectral density defined for all frequencies. We next use the following mathematical formulations to define the span of frequencies occupied by $S_{b,k}[f]$ \cite{kay1993fundamentals}.

\subsubsection{Effective Baseband Bandwidth}
This can be viewed as the average of the squared of all frequencies normalized by the area occupied by the spectral density, $S_{b,k}$. Mathematically, the effective baseband bandwidth is
$$
\alpha_{1b,k} \triangleq\left(\frac{\int_{-\infty}^{\infty} f^2\left|S_{b,k}[f]\right|^2 d f}{\int_{-\infty}^{\infty}\left|S_{b,k}[f]\right|^2 d f}\right)^{\frac{1}{2}}.
$$
Clearly, if we view $\left|S_{b,k}[f]\right|^2$ as a probability density function (PDF), if normalized to unit area, and assume that the $\left|S_{b,k}[f]\right|^2$ is centered at $f = 0$, then the effective baseband bandwidth is simply the variance of the frequencies occupied by the spectral density.

\subsubsection{Baseband-Carrier Correlation (BCC)}
Mathematically, the BCC is
$$
\alpha_{2b,k} \triangleq\frac{\int_{-\infty}^{\infty} f\left|S_{b,k}[f]\right|^2 d f}{\left(\int_{-\infty}^{\infty} f^2\left|S_{b,k}[f]\right|^2 d f \right)^{\frac{1}{2}} \left(\int_{-\infty}^{\infty}\left|S_{b,k}[f]\right|^2 d f\right)^{\frac{1}{2}}}.
$$
In later sections, the term $\alpha_{2b,k}$ will help provide a compact representation of the mathematical description of the available information in the received signals. In line with viewing $\left|S_{b,k}[f]\right|^2$ as a PDF, the numerator of the BCC is the expected value of the frequency occupied by the spectral density. Hence, the BCC is the ratio of the mean of the frequencies occupied by the spectral density to the variance of the frequencies occupied by the spectral density.
\subsubsection{Received Signal-to-Noise Ratio}
The SNR is the ratio of the power of the signal across its occupied frequencies to the noise spectral density. Mathematically, the SNR is
$$
\underset{bu,k}{\operatorname{SNR}} \triangleq \frac{8 \pi^2 \left|\beta_{bu,k}\right|^2}{N_0} \int_{-\infty}^{\infty}\left|S_{b,k}[f]\right|^2 d f.
$$
If there is no beam split, the channel gain is constant across all receive antennas and we have
$$
\underset{b,k}{\operatorname{SNR}} \triangleq \frac{8 \pi^2 \left|\beta_{b,k}\right|^2}{N_0} \int_{-\infty}^{\infty}\left|S_{b,k}[f]\right|^2 d f.
$$
If the same signal is transmitted across all $N_{K}$ time slots, and the channel gain is constant across all receive antennas and time slots, we have
$$
\underset{b}{\operatorname{SNR}} \triangleq \frac{8 \pi^2 \left|\beta_{b}\right|^2}{N_0} \int_{-\infty}^{\infty}\left|S_{b}[f]\right|^2 d f.
$$
The mathematical description of the available information useful for localization is written in terms of these received signal properties.
\section{Available Information in the Received Signal}
In this section, we define the parameters that need to be estimated, both geometric channel parameters and nuisance parameters. The definition of these parameters serves as an intermediate step to investigating the available geometric information provided by LEOs, which subsequently helps the investigation of the feasibility of LEO-based localization under different types of LEO constellations, number of LEOs, beam split, and number of receive antennas.

\subsection{Obsevrable Parameters}
The analysis in this section is based on the received signal given by (\ref{equ:receive_signal}), which is obtained from $N_B$ LEO satellites on $N_U$ receive antennas during $N_{K}$ distinct time slots of $T$ duration each. The parameters observable in the signal received by a receiver from the $b^{\text{th}}$ LEO satellite on its $N_U$ receive antenna during the $N_{K}$ different time slots are subsequently presented. The delays observed across the $N_U$ receive antennas during the $k^{\text{th}}$ time slots are presented in vector form
$$
\bm{\tau}_{b,k}
\triangleq\left[{\tau}_{b1,k}, {\tau}_{b2,k}, \cdots,
{\tau}_{bN_U,k}\right]^{\mathrm{T}}, 
$$
then the delays across the $N_U$ receive antennas during all $N_{K}$ time slots are also vectorized as follows
$$
\bm{\tau}_{b}
\triangleq\left[\bm{\tau}_{b,1}^{\mathrm{T}}, \bm{\tau}_{b,2}^{\mathrm{T}}, \cdots, \bm{\tau}_{b,N_K}^{\mathrm{T}}\right]^{\mathrm{T}}.
$$
The Doppler observed with respect to the $b^{\text{th}}$ LEO satellite across all the $N_{K}$ transmission time slots is
$$
\bm{\nu}_{b}
\triangleq\left[{\nu}_{b,1}, {\nu}_{b,2}, \cdots, \nu_{b,N_K}\right]^{\mathrm{T}}.
$$
Next, 
the channel gain across the $N_U$ receive antennas during the $k^{\text{th}}$ time slots are presented in vector form
$$
\bm{\beta}_{b,k}
\triangleq\left[{\beta}_{b1,k}, {\beta}_{b2,k}, \cdots, {\beta}_{bN_U,k}\right]^{\mathrm{T}}, 
$$
then the delays across the $N_U$ receive antennas during all $N_{K}$ time slots are also vectorized as follows
$$
\bm{\beta}_{b}
\triangleq\left[\bm{\beta}_{b,1}^{\mathrm{T}}, \bm{\beta}_{b,2}^{\mathrm{T}}, \cdots, \bm{\beta}_{b,N_K}^{\mathrm{T}}\right]^{\mathrm{T}}.
$$
Note that if there is no beam split, the channel gain remains constant across all antennas and is simply
$$
\bm{\beta}_{b}
\triangleq\left[{\beta}_{b,1}^{\mathrm{T}}, {\beta}_{b,2}^{\mathrm{T}}, \cdots, {\beta}_{b,N_K}^{\mathrm{T}}\right]^{\mathrm{T}}.
$$
Moreover, if the channel gain is constant across all time slots, we can further represent the $b^{\text{th}}$ LEO transmission by the scalar, $\beta_b$. Finally, with these vectorized forms, the total parameters observable in the signals received at a receiver from the $b^{\text{th}}$ LEO satellite on its $N_U$ receive antenna during the $N_{K}$ different time slots are vectorized as follows
$$
\bm{\eta}_{b} \triangleq\left[\bm{\tau}_{b}^{\mathrm{T}}, \bm{\nu}_{b}^{\mathrm{T}}, \bm{\beta}_{b}^{\mathrm{T}}, \delta_{b}, \epsilon_{b}\right]^{\mathrm{T}}.
$$
All signals observable from all $N_B$ LEO satellites across $N_U$ receive antennas during the $N_{K}$ different time slots are vectorized as
$$
\bm{\eta}_{} \triangleq\left[\bm{\eta}_{1}^{\mathrm{T}}, \bm{\eta}_{2}^{\mathrm{T}}, \cdots, \bm{\eta}_{N_B}^{\mathrm{T}}\right]^{\mathrm{T}}.
$$
After specifying the parameters that are present in the signals received from the LEO satellites, considering the time slots and receive antennas, we present the mathematical preliminaries needed for further discussions.

\subsection{Mathematical Preliminaries}
Although we have specified the parameters in the signals received in a LEO-based localization system, we still have to investigate the estimation accuracy achievable when estimating these parameters. Moreover, it is unclear whether all the parameters presented are separately observable and can contribute to a localization framework. One way of answering these two questions is by using the FIM. To introduce the FIM, we assume that for the received signal expressions, there exists an unbiased estimator of the parameter vector and note that the error covariance matrix of the unbiased estimator, $\hat{\bm{\eta}}$, satisfies the following information inequality
$
\mathbb{E}_{\bm{y}; \boldsymbol{\eta}}\left\{(\hat{\boldsymbol{\eta}}-\boldsymbol{\eta})(\hat{\boldsymbol{\eta}}-\boldsymbol{\eta})^{\mathrm{T}}\right\} \succeq \mathbf{J}_{ \bm{\bm{y}}; \bm{\eta}}^{-1},
$
where $\mathbf{J}_{ \bm{\bm{y}}; \bm{\eta}}$ is the general FIM for the parameter vector $\boldsymbol{\eta}.$
\begin{definition}
\label{definition_FIM_1}
The general FIM for a parameter vector, $\bm{\eta}$, defined as $\mathbf{J}_{ \bm{\bm{y}}; \bm{\eta}} =  \bm{F}_{\bm{\bm{y}}; \bm{\eta} }(\bm{y}; \bm{\eta} ; \bm{\eta}, \bm{\eta})$ is the summation of the FIM obtained from the likelihood due to the observations defined as  $\mathbf{J}_{\bm{y}|\bm{\eta}} =  \bm{F}_{{\bm{y} }}(\bm{y}| \bm{\eta} ;\bm{\eta},\bm{\eta})$ and the FIM from {\em a priori} information about the parameter vector defined as $\mathbf{J}_{ \bm{\eta}} =  \bm{F}_{{\bm{\eta} }}( \bm{\eta} ;\bm{\eta},\bm{\eta})$. In mathematical terms, we have
\begin{equation}
\label{definition_equ:definition_FIM_1}
\begin{aligned}
\mathbf{J}_{ \bm{\bm{y}}; \bm{\eta}} &\triangleq 
-\mathbb{E}_{\bm{y};\bm{\eta}_{}}\left[\frac{\partial^{2} \ln \chi(\bm{y}_{};  \bm{\eta}_{} )}{\partial \bm{\eta}_{} \partial \bm{\eta}_{}^{\mathrm{T}}}\right] \\
&= -\mathbb{E}_{\bm{y}  } \left[\frac{\partial^{2} \ln \chi(\bm{y}_{}|  \bm{\eta}_{} )}{\partial \bm{\eta}_{} \partial \bm{\eta}_{}^{\mathrm{T}}}\right] -\mathbb{E}_{ \bm{\eta}_{}}\left[\frac{\partial^{2} \ln \chi(  \bm{\eta}_{} )}{\partial \bm{\eta}_{} \partial \bm{\eta}_{}^{\mathrm{T}}}\right] \\ &= \mathbf{J}_{\bm{y}|\bm{\eta}} + \mathbf{J}_{ \bm{\eta}},
\end{aligned}
\end{equation}
where  $\chi(\bm{y}_{};  \bm{\eta}_{} )$ denotes the joint probability density function (PDF) of $\bm{y}$ and $\bm{\eta}$.
\end{definition}
While the FIM is a useful statistical tool, valuable in measuring the achievable accuracy, quantifying the available information, and determining the observable parameters, it scales quadratically with the size of the parameter vector. Hence, it is favorable to partition the FIM to focus on the parameter of interest, and the Schur's complement\cite{horn2012matrix} provides a method of achieving this partition. The resulting partition is the Equivalent FIM (EFIM).

\begin{definition}
\label{definition_EFIM}
Given a parameter vector, $ \bm{\eta}_{} \triangleq\left[\bm{\eta}_{1}^{\mathrm{T}}, \bm{\eta}_{2}^{\mathrm{T}}\right]^{\mathrm{T}}$, where $\bm{\eta}_{1}$ is the parameter of interest, the resultant FIM has the structure 
$$
\mathbf{J}_{ \bm{\bm{y}}; \bm{\eta}}=\left[\begin{array}{cc}
\mathbf{J}_{ \bm{\bm{y}}; \bm{\eta}_1}^{}  & \mathbf{J}_{ \bm{\bm{y}}; \bm{\eta}_1, \bm{\eta}_2}^{} \\
 \mathbf{J}_{ \bm{\bm{y}}; \bm{\eta}_1, \bm{\eta}_2}^{\mathrm{T}} &\mathbf{J}_{ \bm{\bm{y}}; \bm{\eta}_2}^{}
\end{array}\right],
$$
where $\bm{\eta} \in \mathbb{R}^{N}, \bm{\eta}_{1} \in \mathbb{R}^{n}, \mathbf{J}_{ \bm{\bm{y}}; \bm{\eta}_1}^{} \in \mathbb{R}^{n \times n},  \mathbf{J}_{ \bm{\bm{y}}; \bm{\eta}_1, \bm{\eta}_2}\in \mathbb{R}^{n \times(N-n)}$, and $\mathbf{J}_{ \bm{\bm{y}}; \bm{\eta}_2}^{}\in$ $\mathbb{R}^{(N-n) \times(N-n)}$ with $n<N$, 
and the EFIM \cite{5571900} of  parameter ${\bm{\eta}_{1}}$ is given by 
$\mathbf{J}_{ \bm{\bm{y}}; \bm{\eta}_1}^{\mathrm{e}} =\mathbf{J}_{ \bm{\bm{y}}; \bm{\eta}_1}^{} - \mathbf{J}_{ \bm{\bm{y}}; \bm{\eta}_1}^{nu} =\mathbf{J}_{ \bm{\bm{y}}; \bm{\eta}_1}^{}-
\mathbf{J}_{ \bm{\bm{y}}; \bm{\eta}_1, \bm{\eta}_2}^{} \mathbf{J}_{ \bm{\bm{y}}; \bm{\eta}_2}^{-1} \mathbf{J}_{ \bm{\bm{y}}; \bm{\eta}_1, \bm{\eta}_2}^{\mathrm{T}}.$

Note that the term $\mathbf{J}_{ \bm{\bm{y}}; \bm{\eta}_1}^{nu}  = \mathbf{J}_{ \bm{\bm{y}}; \bm{\eta}_1, \bm{\eta}_2}^{} \mathbf{J}_{ \bm{\bm{y}}; \bm{\eta}_2}^{-1} \mathbf{J}_{ \bm{\bm{y}}; \bm{\eta}_1, \bm{\eta}_2}^{\mathrm{T}}$ describes the loss of information about ${\bm{\eta}_{1}}$  due to uncertainty in the nuisance parameters ${\bm{\eta}_{2}}$. This EFIM captures all the required information about the parameters of interest present in the FIM; as observed from the relation $(\mathbf{J}_{ \bm{\bm{y}}; \bm{\eta}_1}^{\mathrm{e}})^{-1} = [\mathbf{J}_{ \bm{\bm{y}}; \bm{\eta}}^{-1}]_{[1:n,1:n]}$.
\end{definition}

\subsection{Fisher Information Matrix for Channel Parameters}
In the definitions of the FIM and EFIM given in the previous section, the expression of the likelihood of the received signal conditioned on
the parameter vector is required. This likelihood for the received signal conditioned on
the parameter vector is defined considering the $N_B$ LEO satellites, $N_U$ receive antennas, and the $N_{K}$ time slots, and can be written as

\begin{equation}
\label{equ:likelihood}
\begin{aligned}
    &\chi(\bm{y}_{}[t]|  \bm{\eta}_{}) \propto \prod_{b = 1}^{N_B}\prod_{u = 1}^{N_U}\prod_{k = 1}^{N_K}      \operatorname{exp} \left\{\frac{2}{N_{0}} \int_0^{T_{}} \Re\left\{{\mu}_{bu,k}^{\mathrm{H}}[t] {y}_{bu,k}[t]\right\} \text{d} t-\frac{1}{N_{0}} \int_0^{T_{}}|{\mu}_{bu,k}[t]|^{2} \; \text{d} t \right\}.
    \end{aligned}
\end{equation}

Subsequently, this FIM due to the observations from the $N_B$ LEO satellite, received across the $N_U$ antennas, and during the $N_{K}$ distinct time slots can be computed with the likelihood function (\ref{equ:likelihood}) and Definition 
\ref{definition_FIM_1}, and it results in the diagonal matrix\footnote{With the assumption that the parameters from different LEO satellites are independent.}.
\begin{equation}
\begin{aligned}
\mathbf{J}_{\bm{y}|\bm{\eta}} =  \bm{F}_{{\bm{y} }}(\bm{y}| \bm{\eta} ;\bm{\eta},\bm{\eta}) = \operatorname{diag}\left\{\bm{F}_{{\bm{y} }}(\bm{y}| \bm{\eta} ;\bm{\eta}_{1},\bm{\eta}_{1}), \ldots, \bm{F}_{{\bm{y} }}(\bm{y}| \bm{\eta} ;\bm{\eta}_{N_B},\bm{\eta}_{N_B})\right\}.
    \end{aligned}
\end{equation}
The entries in FIM due to the observations of the received signals from $b^{\text{th}}$ LEO satellite can be obtained through the simplified expression:
$$
\begin{aligned}
     \bm{F}_{{\bm{y} }}(\bm{y}| \bm{\eta} ;\bm{\eta}_{b},\bm{\eta}_{b}) =   \frac{1}{N_{0}} \sum_{u,k}^{N_U N_K}\Re\left\{ \int \nabla_{\bm{\eta}_{b}}{\mu}_{bu,k} [t]\nabla_{\bm{\eta}_{b}}{\mu}_{bu,k}^{\mathrm{H}}[t] \; \; \text{d} t \right\}.
\end{aligned}
$$
Considering the $b^{\text{th}}$ LEO satellite, the FIM focusing on the delays at the $u^{\text{th}}$ receive antenna during the $k^{\text{th}}$ time slot is
$$
\begin{aligned}
\bm{F}_{{\bm{y} }}(\bm{y}| \bm{\eta} ;{\tau}_{bu,k},{\tau}_{bu,k}) &= -\bm{F}_{{\bm{y} }}(\bm{y}| \bm{\eta} ;{\tau}_{bu,k},{\delta}_{b}) = \underset{bu,k}{\operatorname{SNR}} \omega_{b,k},
\end{aligned}
$$
where $\omega_{b,k} = \Bigg[ \alpha_{1b,k}^2 +  2f_{ob,k} \alpha_{1b,k} \alpha_{2b,k} +  f_{ob,k}^2   \Bigg].$
The FIM relating the delay and the channel gain due to the observations of the received signals from $b^{\text{th}}$ LEO satellite to the $u^{\text{th}}$ receive antenna during the $k^{\text{th}}$ time slot is
$$
\begin{aligned}
\bm{F}_{{\bm{y} }}(\bm{y}| \bm{\eta} ;{\tau}_{bu,k},{\beta}_{bu,k}) &= \frac{1}{2 \pi \left|\beta_{bu,k}\right|^2}{\beta}^{\mathrm{H}}_{bu,k}\underset{bu,k}{\operatorname{SNR}} \; \Re\left\{j (\alpha_{2b,k} \alpha_{1b,k} + 1)
\right\}.
\end{aligned}
$$
The FIM relating the delay and the Doppler due to the observations of the received signals from $b^{\text{th}}$ LEO satellite to the $u^{\text{th}}$ receive antenna during the $k^{\text{th}}$ time slot is
$$
\begin{aligned}
\bm{F}_{{\bm{y} }}(\bm{y}| \bm{\eta} ;{\tau}_{bu,k},{\nu}_{b,k}) =& \frac{2\left|\beta_{bu,k}\right|^2}{N_0}  \Re\left\{-j 2 \pi f_{c} \int_{}^{} t_{obu,k} \nabla_{{\tau}_{bu,k}}s^{\mathrm{H}}(t_{obu,k}) s^{}(t_{obu,k}) \; dt_{obu,k}
\right\} \\
+&  \frac{2\left|\beta_{bu,k}\right|^2}{N_0} \Re\left\{-j^2 4 \pi^2 f_{ob,k}^{} f_{c} \int_{}^{} t_{obu,k} \left|s(t_{obu,k})\right|^2  \; dt_{obu,k}
\right\}.
\end{aligned}
$$
The FIM of the delay and the frequency offset due to the observations of the received signals from $b^{\text{th}}$ LEO satellite to the $u^{\text{th}}$ receive antenna during the $k^{\text{th}}$ time slot is
$$
\begin{aligned}
\bm{F}_{{\bm{y} }}(\bm{y}| \bm{\eta} ;{\tau}_{bu,k},{\epsilon}_{b}) =& \frac{2\left|\beta_{bu,k}\right|^2}{N_0}  \Re\left\{j 2 \pi  \int_{}^{} t_{obu,k} \nabla_{{\tau}_{bu,k}}s^{\mathrm{H}}(t_{obu,k}) s^{}(t_{obu,k}) \; dt_{obu,k}
\right\} \\
+&  \frac{2\left|\beta_{bu,k}\right|^2}{N_0} \Re\left\{j^2 4 \pi^2 f_{ob,k}^{}  \int_{}^{} t_{obu,k} \left|s(t_{obu,k})\right|^2  \; dt_{obu,k}
\right\}.
\end{aligned}
$$
\begin{remark}
    While the channel gain affects the information available about the delay, knowledge, or lack thereof of the Dopplers, the channel gain and the frequency offset do not affect the available information about the delay. Hence, the accuracy achievable while estimating the delays is independent of the knowledge of the Dopplers, channel gain, and the frequency offset.
    \end{remark}
\begin{proof}
    The proof follows by noticing that the FIM of the delay is dependent on the channel gain through the SNR, but the FIM relating the delay to the Dopplers, the channel gain, and the frequency offset is zero. 
\end{proof}
The FIM focusing on the Doppler observed with respect to the $b^{\text{th}}$ LEO satellite at the receiver during the $k^{\text{th}}$ time slot is presented next
$$
\begin{aligned}
\bm{F}_{{\bm{y} }}(\bm{y}| \bm{\eta} ;{\nu}_{b,k}, {\tau}_{bu,k}) = \bm{F}_{{\bm{y} }}(\bm{y}| \bm{\eta} ; {\tau}_{bu,k}, {\nu}_{b,k})^{\mathrm{H}}.
\end{aligned}
$$
The FIM of the Doppler observed with respect to the $b^{\text{th}}$ LEO satellite at the receiver during the $k^{\text{th}}$ time slot is
$$
\begin{aligned}
\bm{F}_{{\bm{y} }}(\bm{y}| \bm{\eta} ;{\nu}_{b,k}, {\nu}_{b,k}) = 0.5  \underset{bu,k}{\operatorname{SNR}} f_{c}^2 t_{obu,k}^2.
\end{aligned}
$$
The FIM of the Doppler observed with respect to the $b^{\text{th}}$ LEO satellite and the corresponding  channel gain at the receiver during the $k^{\text{th}}$ time slot is
$$
\begin{aligned}
\bm{F}_{{\bm{y} }}(\bm{y}| \bm{\eta} ;{\nu}_{b,k}, {\beta}_{bu,k}) &= \frac{2 \beta_{bu,k}}{N_0}   \Re\left\{ j 2 \pi f_{c} \int_{}^{} t_{obu,k} \left|s(t_{obu,k})\right|^2  \; dt_{obu,k} \right\}.
\end{aligned}
$$
The FIM of the Doppler observed with respect to the $b^{\text{th}}$ LEO satellite and the corresponding time offset at the receiver during the $k^{\text{th}}$ time slot is 
$$
\begin{aligned}
\bm{F}_{{\bm{y} }}(\bm{y}| \bm{\eta} ;{\nu}_{b,k}, {\delta}_{b}) =&\frac{2\left|\beta_{bu,k}\right|^2}{N_0} \Re\left\{j 2 \pi f_{c} \int_{}^{} t_{obu,k} s^{\mathrm{H}}(t_{obu,k}) \nabla_{{\delta}_{b}}s^{}(t_{obu,k}) \; dt_{obu,k}
\right\} \\
+&  \frac{2\left|\beta_{bu,k}\right|^2}{N_0} \Re\left\{j^2 4 \pi^2 f_{ob,k}^{} f_{c} \int_{}^{} t_{obu,k} \left|s(t_{obu,k})\right|^2  \; dt_{obu,k}
\right\}.
\end{aligned}
$$
The FIM of the Doppler observed with respect to the $b^{\text{th}}$ LEO satellite and the corresponding frequency offset during the $k^{\text{th}}$ time slot is
$$
\begin{aligned}
\bm{F}_{{\bm{y} }}(\bm{y}| \bm{\eta} ;{\nu}_{b,k}, {\epsilon}_{b}) = - 0.5 * \underset{bu,k}{\operatorname{SNR}} f_{c} t_{obu,k}^2.
\end{aligned}
$$
\begin{remark}
    From the above, the channel gain influences the information we have about Doppler. Also, the knowledge of delays, channel gain, and the time offset does not affect the information available about the Dopplers. Therefore, the accuracy in estimating Doppler remains unaffected by the knowledge of delays, channel gain, or time offset.
    \end{remark}
\begin{proof}
    The proof follows by noticing that the FIM of the Doppler is dependent on the channel gain through the SNR, but the FIM relating the Dopplers to the delays, the channel gain, and the time offset is zero. 
\end{proof}
The FIM focusing on the channel gain at the $u^{\text{th}}$ receive antenna during the $k^{\text{th}}$ time slot is presented next
$$
\begin{aligned}
\bm{F}_{{\bm{y} }}(\bm{y}| \bm{\eta} ;{\beta}_{bu,k},{\tau}_{bu,k}) &= \bm{F}_{{\bm{y} }}(\bm{y}| \bm{\eta} ;{\tau}_{bu,k},{\beta}_{bu,k})^{\mathrm{H}}.
\end{aligned}
$$

The FIM between the channel gain and the Doppler in the FIM due to the observations of the received signals from $b^{\text{th}}$ LEO satellite to the $u^{\text{th}}$ receive antenna during the $k^{\text{th}}$ time slot is
$$
\begin{aligned}
\bm{F}_{{\bm{y} }}(\bm{y}| \bm{\eta} ; {\beta}_{bu,k},{\nu}_{b,k})  = \bm{F}_{{\bm{y} }}(\bm{y}| \bm{\eta} ;{\nu}_{b,k}, {\beta}_{bu,k})^{\mathrm{H}}.
\end{aligned}
$$

The FIM of the channel gain in the FIM due to the observations of the received signals from $b^{\text{th}}$ LEO satellite to the $u^{\text{th}}$ receive antenna during the $k^{\text{th}}$ time slot is
$$
\begin{aligned}
\bm{F}_{{\bm{y} }}(\bm{y}| \bm{\eta} ; {\beta}_{bu,k}, {\beta}_{bu,k}) = \frac{1}{4 \pi^2 \left|\beta_{bu,k}\right|^2}\underset{bu,k}{\operatorname{SNR}}.
\end{aligned}
$$
The FIM between the channel gain and the time offset in the FIM due to the observations of the received signals from $b^{\text{th}}$ LEO satellite to the $u^{\text{th}}$ receive antenna during the $k^{\text{th}}$ time slot is
$$
\bm{F}_{{\bm{y} }}(\bm{y}| \bm{\eta} ;{\beta}_{bu,k},{\delta}_{b}) = - \bm{F}_{{\bm{y} }}(\bm{y}| \bm{\eta} ;{\beta}_{bu,k},{\tau}_{bu,k}).
$$
The FIM between the channel gain and the frequency offset in the FIM due to the observations of the received signals from $b^{\text{th}}$ LEO satellite to the $u^{\text{th}}$ receive antenna during the $k^{\text{th}}$ time slot is
$$
\begin{aligned}
\bm{F}_{{\bm{y} }}(\bm{y}| \bm{\eta} ;{\beta}_{bu,k},{\epsilon}_{b}) &= \frac{4 \pi \beta_{bu,k}}{N_0}    \Re\left\{ j   \int_{}^{} t_{obu,k} \left|s(t_{obu,k})\right|^2  \; dt_{obu,k} \right\}.
\end{aligned}
$$
\begin{remark}
The information available for the estimation of the channel gain is independent of the other parameters - delay, Doppler, time, and frequency offset. Hence, the accuracy achievable in the estimation of the channel gain does not depend on the knowledge of the other parameters.
\end{remark}
\begin{proof}
    The proof follows by noticing that the FIM relating the channel gain to the delays, Dopplers, time, and frequency offset is zero. 
\end{proof}

 The FIM focusing on the time offset at the $u^{\text{th}}$ receive antenna during the $k^{\text{th}}$ time slot with respect to the $b^{\text{th}}$ LEO satellite is presented next. The FIM between the time offset and the delay in the FIM due to the observations of the received signals from $b^{\text{th}}$ LEO satellite to the $u^{\text{th}}$ receive antenna during the $k^{\text{th}}$ time slot is
$$
\begin{aligned}
\bm{F}_{{\bm{y} }}(\bm{y}| \bm{\eta} ;{\delta}_{b}, {\tau}_{bu,k}) = \bm{F}_{{\bm{y} }}(\bm{y}| \bm{\eta} ; {\tau}_{bu,k}, {\delta}_{b}).
\end{aligned}
$$
The FIM between the time offset and the Doppler in the FIM due to the observations of the received signals from $b^{\text{th}}$ LEO satellite to the $u^{\text{th}}$ receive antenna during the $k^{\text{th}}$ time slot is
$$
\begin{aligned}
\bm{F}_{{\bm{y} }}(\bm{y}| \bm{\eta} ;{\delta}_{b}, {\nu}_{bu,k}) = \bm{F}_{{\bm{y} }}(\bm{y}| \bm{\eta} ;{\nu}_{bu,k}, {\delta}_{b})^{\mathrm{H}}.
\end{aligned}
$$
The FIM between the time offset  and the channel gain in the FIM due to the observations of the received signals from $b^{\text{th}}$ LEO satellite to the $u^{\text{th}}$ receive antenna during the $k^{\text{th}}$ time slot is
$$
\begin{aligned}
\bm{F}_{{\bm{y} }}(\bm{y}| \bm{\eta} ;{\delta}_{b}, {\beta}_{bu,k}) = \bm{F}_{{\bm{y} }}(\bm{y}| \bm{\eta} ; {\beta}_{bu,k},{\delta}_{b})^{\mathrm{H}}.
\end{aligned}
$$
The FIM of the time offset in the FIM due to the observations of the received signals from $b^{\text{th}}$ LEO satellite to the $u^{\text{th}}$ receive antenna during the $k^{\text{th}}$ time slot is
$$
\begin{aligned}
\bm{F}_{{\bm{y} }}(\bm{y}| \bm{\eta} ;{\delta}_{b}, {\delta}_{b}) = \bm{F}_{{\bm{y} }}(\bm{y}| \bm{\eta} ; {\tau}_{bu,k},{\tau}_{bu,k}) = -\bm{F}_{{\bm{y} }}(\bm{y}| \bm{\eta} ;{\delta}_{b}, {\tau}_{bu,k}).
\end{aligned}
$$
The FIM between the time offset  and the frequency offset in the FIM due to the observations of the received signals from $b^{\text{th}}$ LEO satellite to the $u^{\text{th}}$ receive antenna during the $k^{\text{th}}$ time slot is
$$
\begin{aligned}
\bm{F}_{{\bm{y} }}(\bm{y}| \bm{\eta} ;{\delta}_{b},{\epsilon}_{b}) = &\frac{2\left|\beta_{bu,k}\right|^2}{N_0} \Re\left\{j 2 \pi  \int_{}^{} t_{obu,k} \nabla_{{\delta}_{b}}s^{\mathrm{H}}(t_{obu,k}) s^{}(t_{obu,k}) \; dt_{obu,k}
\right\} \\
+&  \frac{2\left|\beta_{bu,k}\right|^2}{N_0} \Re\left\{-j^2 4 \pi^2 f_{ob,k}^{}  \int_{}^{} t_{obu,k} \left|s(t_{obu,k})\right|^2  \; dt_{obu,k}
\right\}.
\end{aligned}
$$
\begin{remark}
Neither the channel gain nor the frequency offset influences the information available about the time offset itself. Therefore, the accuracy in estimating the time offset is unaffected by knowledge of the Doppler, channel gain, or frequency offset.    
\end{remark}
\begin{proof}
    The proof follows by noticing that the FIM of the time offset is dependent on the channel gain through the SNR, but the FIM relating the time offset to the Dopplers, the channel gain, and the frequency offset is zero. 
\end{proof}
 The FIM focusing on the frequency offset at the $u^{\text{th}}$ receive antenna during the $k^{\text{th}}$ time slot with respect to the $b^{\text{th}}$ LEO satellite is presented next. The FIM between the frequency offset and the delay in the FIM due to the observations of the received signals from $b^{\text{th}}$ LEO satellite to the $u^{\text{th}}$ receive antenna during the $k^{\text{th}}$ time slot is
$$
\begin{aligned}
\bm{F}_{{\bm{y} }}(\bm{y}| \bm{\eta} ;{\epsilon}_{b}, {\tau}_{bu,k}) &= \bm{F}_{{\bm{y} }}(\bm{y}| \bm{\eta} ; {\tau}_{bu,k}, {\epsilon}_{b})^{\mathrm{H}}.
\end{aligned}
$$
The FIM of the frequency offset and the corresponding Doppler observed with respect to the $b^{\text{th}}$ LEO satellite  during the $k^{\text{th}}$ time slot is
$$
\begin{aligned}
\bm{F}_{{\bm{y} }}(\bm{y}| \bm{\eta} ;{\epsilon}_{b}, {\nu}_{b,k}) = - 0.5 * \underset{bu,k}{\operatorname{SNR}} f_{c}t_{obu,k}^2.
\end{aligned}
$$
The FIM between the frequency offset and the channel gain in the FIM due to the observations of the received signals from $b^{\text{th}}$ LEO satellite to the $u^{\text{th}}$ receive antenna during the $k^{\text{th}}$ time slot is
$$
\begin{aligned}
\bm{F}_{{\bm{y} }}(\bm{y}| \bm{\eta} ;{\epsilon}_{b}, {\beta}_{bu,k}) = \bm{F}_{{\bm{y} }}(\bm{y}| \bm{\eta} ;{\beta}_{bu,k},{\epsilon}_{b})^{\mathrm{H}}.
\end{aligned}
$$
The FIM between the frequency offset and the time offset in the FIM due to the observations of the received signals from $b^{\text{th}}$ LEO satellite to the $u^{\text{th}}$ receive antenna during the $k^{\text{th}}$ time slot is
$$
\begin{aligned}
\bm{F}_{{\bm{y} }}(\bm{y}| \bm{\eta} ;{\epsilon}_{b}, {\delta}_{b}) = \bm{F}_{{\bm{y} }}(\bm{y}| \bm{\eta} ;{\delta}_{b},{\epsilon}_{b})^{\mathrm{H}}.
\end{aligned}
$$

The FIM of the frequency offset in the FIM due to the observations of the received signals from $b^{\text{th}}$ LEO satellite to the $u^{\text{th}}$ receive antenna during the $k^{\text{th}}$ time slot is
$$
\begin{aligned}
\bm{F}_{{\bm{y} }}(\bm{y}| \bm{\eta} ;{\epsilon}_{b}, {\epsilon}_{b}) = 0.5 * \underset{bu,k}{\operatorname{SNR}}  t_{obu,k}^2.
\end{aligned}
$$

\begin{remark}

The channel gain and the time offset do not impact the information we have about the frequency offset itself. Thus, the accuracy in estimating the frequency offset is unaffected by the knowledge of delays, channel gain, and time offset. 
\end{remark}
\begin{proof}
    The proof follows by noticing that the FIM of the frequency offset is dependent on the channel gain through the SNR, but the FIM relating the frequency offset to the delays, the channel gain, and the time offset is zero. 
\end{proof}

The FIM of the channel parameters, based on the observations of the received signals, is used to derive the FIM of the receiver's location in the next section.

\section{Fisher Information Matrix for Location Parameters}
In the previous section, we highlighted the useful and nuisance parameters present in the signals received from the $N_B$ LEO satellites across the $N_U$ receive antennas during $N_{K}$ different time slots. Subsequently, we derived the information about these parameters present in the received signals and presented the structure of these parameters. In this section, we use the FIM for channel parameters to derive the FIM for the location parameters and highlight the FIM structure. This FIM for the location parameters will help us determine how feasible it is to localize a receiver with the signals received from LEO satellites. 

To proceed, we define $\bm{p}_{U} = \bm{p}_{U,0}$ and $\bm{v}_{U} = \bm{v}_{U,k}$, and the location parameters 
$$\bm{\kappa} = [\bm{p}_{U}, \bm{\Phi}_{U}, \bm{v}_{U}, \bm{\zeta}_{1}, \bm{\zeta}_{2},\cdots, \bm{\zeta}_{N_B}  ],$$
$$\text{where}$$ 
$$
\bm{\zeta}_{b} = \left[\bm{\beta}_{b}^{\mathrm{T}}, \delta_{b}, \epsilon_{b}\right]^{\mathrm{T}},
$$
and our goal is to derive the FIM of the entire location parameter vector, or different combinations of parameters, under different levels of uncertainty about the channel parameters. The FIM for the location parameters,  $\mathbf{J}_{\bm{y}|\bm{\kappa}}
$ can be obtained from the FIM for the channel parameters, $\mathbf{J}_{ \bm{\bm{y}}| \bm{\eta}}$, using the bijective transformation  $\mathbf{J}_{\bm{y}|\bm{\kappa}} \triangleq \mathbf{\Upsilon}_{\bm{\kappa}} \mathbf{J}_{ \bm{\bm{y}}| \bm{\eta}} \mathbf{\Upsilon}_{\bm{\kappa}}^{\mathrm{T}}$, where $\mathbf{\Upsilon}_{\bm{\kappa}}$ represents derivatives of the non-linear relationship between the geometric channel parameters, $ \bm{\eta}$, and the location parameters\cite{kay1993fundamentals}. The elements in the bijective transformation matrix $\mathbf{\Upsilon}_{\bm{\kappa}}$ are given in Appendix \ref{Appendix_Entries_in_transformation_matrix}. With no {\em a priori} information about the location parameters $\bm{\kappa}$, $\mathbf{J}_{\bm{y}; \bm{\kappa}} = \mathbf{J}_{\bm{y}|\bm{\kappa}}$. The EFIM taking  $\bm{\kappa}_{1} = [\bm{p}_{U}, \bm{\Phi}_{U}, \bm{v}_{U}]$ as the parameter of interest and $\bm{\kappa}_{2} = [\bm{\zeta}_{1}, \bm{\zeta}_{2},\cdots, \bm{\zeta}_{N_B}  ]$ as the nuisance parameters is now derived. 
\subsection{Elements in $\mathbf{J}_{ \bm{\bm{y}}; \bm{\kappa}_1}^{}$}
\label{subsection:FIM_channel_parameters}
The elements in $\mathbf{J}_{ \bm{\bm{y}}; \bm{\kappa}_1}^{}$ are presented through the following Lemmas. This FIM corresponds to the available information of the location parameters $\bm{\kappa}_{1}$ when the nuisance parameters are known. 
\begin{lemma}
\label{lemma:FIM_3D_position}
The FIM of the $3$D position of the receiver is
\begin{equation}
\label{equ_lemma:FIM_3D_position}
\begin{aligned}
{\bm{F}_{{{y} }}(\bm{y}_{}| \bm{\eta} ;\bm{p}_{U},\bm{p}_{U}) = } {\sum_{b,k^{},u^{}} \underset{bu,k}{\operatorname{SNR}}  \Bigg[\frac{\omega_{b,k} }{c^2}   \bm{\Delta}_{bu,k} 
 \bm{\Delta}_{bu^{},k^{}}^{\mathrm{T}}}  + {\frac{f_{c}^2 t_{obu,k}^2\nabla_{\bm{p}_{U}} \nu_{b,k} \nabla_{\bm{p}_{U}}^{\mathrm{T}}\nu_{b,k}}{2}  \Bigg]}.
\end{aligned}
\end{equation}

\end{lemma}
\begin{proof}
See Appendix \ref{Appendix_lemma_FIM_3D_position}.
\end{proof} 

\begin{lemma}
\label{lemma:FIM_3D_position_3D_orientation}
The FIM relating the $3$D position and $3$D orientation of the receiver is
\begin{equation}
\label{equ_lemma:FIM_3D_position_3D_orientation}
\begin{aligned}
&{\bm{F}_{{{y} }}(\bm{y}_{}| \bm{\eta} ;\bm{p}_{U},\bm{\Phi}_{U}) = }  {\sum_{b,k^{},u^{}} \underset{bu,k}{\operatorname{SNR}}  \Bigg[ \frac{\omega_{b,k}}{c}   \bm{\Delta}_{bu,k} 
 \nabla_{\bm{\Phi}_{U}}^{\mathrm{T}} \tau_{bu^{},k^{}}}    \Bigg].
\end{aligned}
\end{equation}

\end{lemma}
\begin{proof}
See Appendix \ref{Appendix_lemma_FIM_3D_position_3D_orientation}.
\end{proof} 

\begin{lemma}
\label{lemma:FIM_3D_position_3D_velocity}
The FIM relating the $3$D position and $3$D velocity of the receiver is
\begin{equation}
\label{equ_lemma:FIM_3D_position_3D_velocity}
\begin{aligned}
{\bm{F}_{{{y} }}(\bm{y}_{}| \bm{\eta} ;\bm{p}_{U},\bm{v}_{U}) = } 
\medmath{{\sum_{b,k^{},u^{}} \underset{bu,k}{\operatorname{SNR}}  \Bigg[ \frac{(k - 1) \omega_{b,k}\Delta_{t}}{c^2}    \bm{\Delta}_{bu,k} 
\bm{\Delta}_{bu,k}^{\mathrm{T}} }    - {\frac{f_{c}^2 t_{obu,k}^2\nabla_{\bm{p}_{U}} \nu_{b,k} \bm{\Delta}_{bU,k}^{\mathrm{T}}}{2 c}  \Bigg]}}.
\end{aligned}
\end{equation}

\end{lemma}
\begin{proof}
See Appendix \ref{Appendix_lemma_FIM_3D_position_3D_velocity}.
\end{proof}

\begin{lemma}
\label{lemma:FIM_3D_orientation_3D_orientation}
The FIM of the $3$D orientation of the receiver is
\begin{equation}
\label{equ_lemma:FIM_3D_orientation_3D_orientation}
\begin{aligned}
&{\bm{F}_{{{y} }}(\bm{y}_{}| \bm{\eta} ;\bm{\Phi}_{U},\bm{\Phi}_{U}) = } {\sum_{b,k^{},u^{}} \underset{bu,k}{\operatorname{SNR}}  \Bigg[ \omega_{b,k}     \nabla_{\bm{\Phi}_{U}} \tau_{bu^{},k^{}}
 \nabla_{\bm{\Phi}_{U}}^{\mathrm{T}} \tau_{bu^{},k^{}}}    \Bigg].
\end{aligned}
\end{equation}
\end{lemma}
\begin{proof}
See Appendix \ref{Appendix_lemma_FIM_3D_orientation}.
\end{proof}

\begin{lemma}
\label{lemma:FIM_3D_orientation_3D_velocity}
The FIM relating the $3$D orientation and $3$D velocity of the receiver is
\begin{equation}
\label{equ_lemma:FIM_3D_orientation_3D_velocity}
\begin{aligned}
&{\bm{F}_{{{y} }}(\bm{y}_{}| \bm{\eta} ;\bm{\Phi}_{U},\bm{v}_{U})  } = {\sum_{b,k^{},u^{}} \underset{bu,k}{\operatorname{SNR}}  \Bigg[ \frac{(k - 1) \Delta_{t}^{}\omega_{b,k}}{c}     \nabla_{\bm{\Phi}_{U}} \tau_{bu^{},k^{}}
 \bm{\Delta}_{bu,k}^{\mathrm{T}} }    \Bigg].
\end{aligned}
\end{equation}

\end{lemma}
\begin{proof}
See Appendix \ref{Appendix_lemma_FIM_3D_orientation_3D_velocity}.
\end{proof}

\begin{lemma}
\label{lemma:FIM_3D_velocity}
The FIM of the $3$D velocity of the receiver is
\begin{equation}
\label{equ_lemma:FIM_3D_velocity}
\begin{aligned}{\bm{F}_{{{y} }}(\bm{y}_{}| \bm{\eta} ;\bm{v}_{U},\bm{v}_{U}) = } \medmath{\sum_{b,k^{},u^{}} \underset{bu,k}{\operatorname{SNR}}  \Bigg[ \frac{(k - 1)^2 \Delta_{t}^{2}\omega_{b,k}}{c^2}     \bm{\Delta}_{bu,k} 
\bm{\Delta}_{bu,k}^{\mathrm{T}} }  + \medmath{\frac{f_{c}^2 t_{obu,k}^2\bm{\Delta}_{bU,k}\bm{\Delta}_{bU,k}^{\mathrm{T}}}{2 c^2}  \Bigg]}.
\end{aligned}
\end{equation}

\end{lemma}
\begin{proof}
See Appendix \ref{Appendix_lemma_FIM_3D_velocity}.
\end{proof}
\begin{remark}
    When all other location parameters are known, the information available at the $k^{\text{th}}$ time slot for the estimation of $3$D velocity through the delay is a factor $ \frac{(k - 1)^2 \Delta_{t}^{2}}{c^2}$ more than the information available for the estimation of $3$D position.
\end{remark}
\begin{remark}
It is impossible to estimate the $3$D velocity of the receiver using only the delays observed at one-time slot.
\end{remark}
\subsection{Elements in $\mathbf{J}_{ \bm{\bm{y}}; \bm{\kappa}_1}^{nu}$}
\label{subsection:information_loss_FIM_channel_parameters}
The elements in $\mathbf{J}_{ \bm{\bm{y}}; \bm{\kappa}_1}^{nu}$ are presented in this section. These elements represent  the loss of information about ${\bm{\kappa}_{1}}$  due to uncertainty in the nuisance parameters ${\bm{\kappa}_{2}}$.
\begin{lemma}
\label{lemma:information_loss_FIM_3D_position}
The loss of information about $3$D position of the receiver due to uncertainty in the nuisance parameters ${\bm{\kappa}_{2}}$ is
\begin{equation}
\label{equ_lemma:information_loss_FIM_3D_position}
\begin{aligned}
{\bm{G}_{{{y} }}(\bm{y}_{}| \bm{\eta} ;\bm{p}_{U},\bm{p}_{U}) = } \medmath{\sum_{b}   \norm{\sum_{k^{},u^{}} \underset{bu^{},k^{}}{\operatorname{SNR}}\bm{\Delta}_{bu^{},k^{}}^{\mathrm{T}} \frac{ \omega_{b,k}}{c} }}^{2}    \medmath{   \left(\sum_{u,k} \underset{bu,k}{\operatorname{SNR}} \omega_{b,k}\right)^{\mathrm{-1}}  }  + \medmath{\sum_{b}\norm{{\sum_{k^{},u^{} } \underset{bu^{},k^{}}{\operatorname{SNR}} \; \; \nabla_{\bm{p}_{U}}^{\mathrm{T}} \nu_{b,k^{}}   }  \frac{(f_{c}^{}) (t_{obu,k^{}}^{2})}{2} }^{2}  \left(\sum_{u,k} \frac{\underset{bu,k}{\operatorname{SNR}}  t_{obu,k}^{2}}{2}\right)^{-1}}
\end{aligned}
\end{equation}

\end{lemma}
\begin{proof}
See Appendix \ref{Appendix_subsection:information_loss_3D_position_3D_position}.
\end{proof}

\begin{lemma}
\label{lemma:information_loss_FIM_3D_position_3D_orientation}
The loss of information about the FIM of the $3$D position and $3$D orientation of the receiver due to uncertainty in the nuisance parameters ${\bm{\kappa}_{2}}$ is
\begin{equation}
\label{equ_lemma:information_loss_FIM_3D_position_3D_orientation}
\begin{aligned}
{\bm{G}_{{{y} }}(\bm{y}_{}| \bm{\eta} ;\bm{p}_{U},\bm{\Phi}_{U}) = }  \medmath{ \frac{1}{c} \sum_{b}  \sum_{k^{},u^{}k^{'},u^{'}} \underset{bu^{},k^{}}{\operatorname{SNR}} \underset{bu^{'},k^{'}}{\operatorname{SNR}} \bm{\Delta}_{bu^{},k^{}}  \nabla_{\bm{\Phi}_{U}}^{\mathrm{T}} \tau_{bu^{'},k^{'}} \omega_{b,k} \omega_{b,k^{'}}}    \medmath{   \left(\sum_{u,k} \underset{bu,k}{\operatorname{SNR}} \omega_{b,k}\right)^{\mathrm{-1}}  } 
\end{aligned}
\end{equation}

\end{lemma}
\begin{proof}
See Appendix \ref{Appendix_subsection:information_loss_3D_position_3D_orientation}.
\end{proof}

\begin{lemma}
\label{lemma:information_loss_FIM_3D_position_3D_veloctiy}
The loss of information about the FIM of the $3$D position and $3$D velocity of the receiver due to uncertainty in the nuisance parameters ${\bm{\kappa}_{2}}$ is
\begin{equation}
\label{equ_lemma:information_loss_FIM_3D_position_3D_veloctiy}
\begin{aligned}
\bm{G}_{{{y} }}(\bm{y}_{}| \bm{\eta} ;\bm{p}_{U},\bm{v}_{U}) &= \frac{\Delta_{t}}{c^2} \sum_{b,k^{},u^{}k^{'},u^{'}} \underset{bu^{},k^{}}{\operatorname{SNR}} \underset{bu^{'},k^{'}}{\operatorname{SNR}}  \medmath{   \bm{\Delta}_{bu^{},k^{}}  (k^{'} - 1) {\bm{\Delta}_{bu^{'},k^{'}}^{\mathrm{T}}} }    \medmath{   \left(\sum_{u,k} \underset{bu,k}{\operatorname{SNR}} \omega_{b,k}\right)^{\mathrm{-1}}}\\ &   - \frac{1}{c}\sum_{b,k^{},u^{}k^{'},u^{'}} \underset{bu^{},k^{}}{\operatorname{SNR}} \underset{bu^{'},k^{'}}{\operatorname{SNR}}    {{ \nabla_{\bm{p}_{U}} \nu_{b,k^{}} \bm{\Delta}_{bU,k^{'}}^{\mathrm{T}}   }  \frac{(f_{c}^{2}) (t_{obu,k^{}}^{2}t_{obu^{'},k^{'}}^{2})}{4}   \left(\sum_{u,k} \frac{\underset{bu,k}{\operatorname{SNR}}  t_{obu,k}^{2}}{2}\right)^{-1}}
\end{aligned}
\end{equation}
\end{lemma}
\begin{proof}
See Appendix \ref{Appendix_subsection:information_loss_3D_position_3D_velocity}.
\end{proof}

\begin{lemma}
\label{lemma:information_loss_FIM_3D_orientation_3D_orientation}
The loss of information about $3$D orientation of the receiver due to uncertainty in the nuisance parameters ${\bm{\kappa}_{2}}$ is
\begin{equation}
\label{equ_lemma:information_loss_FIM_3D_orientation_3D_orientation}
\begin{aligned}
{\bm{G}_{{{y} }}(\bm{y}_{}| \bm{\eta} ;\bm{\Phi}_{U},\bm{\Phi}_{U}) = } \medmath{\sum_{b}  \sum_{k^{},u^{}k^{'},u^{'}} \underset{bu^{},k^{}}{\operatorname{SNR}} \underset{bu^{'},k^{'}}{\operatorname{SNR}} \nabla_{\bm{\Phi}_{U}} \tau_{bu^{},k^{}}\nabla_{\bm{\Phi}_{U}}^{\mathrm{T}} \tau_{bu^{'},k^{'}} \omega_{b,k} \omega_{b,k^{'}}}    \medmath{   \left(\sum_{u,k} \underset{bu,k}{\operatorname{SNR}} \omega_{b,k}\right)^{\mathrm{-1}}  } 
\end{aligned}
\end{equation}

\end{lemma}
\begin{proof}
See Appendix \ref{Appendix_subsection:information_loss_3D_orientation_3D_orientation}.
\end{proof}

\begin{lemma}
\label{lemma:information_loss_FIM_3D_orientation_3D_veloctiy}
The loss of information about the FIM of the $3$D orientation and $3$D velocity of the receiver due to uncertainty in the nuisance parameters ${\bm{\kappa}_{2}}$ is
\begin{equation}
\label{equ_lemma:information_loss_FIM_3D_orientation_3D_veloctiy}
\begin{aligned}
{\bm{G}_{{{y} }}(\bm{y}_{}| \bm{\eta} ;\bm{\Phi}_{U},\bm{v}_{U}) = \frac{\Delta_{t}}{c} \sum_{b}  } \medmath{ \sum_{k^{},u^{}k^{'},u^{'}} \underset{bu^{},k^{}}{\operatorname{SNR}} \underset{bu^{'},k^{'}}{\operatorname{SNR}} \nabla_{\bm{\Phi}_{U}} \tau_{bu^{},k^{}}(k^{'} - 1) {\bm{\Delta}_{bu^{'},k^{'}}^{\mathrm{T}}} \omega_{b,k} \omega_{b,k^{'}}}    \medmath{   \left(\sum_{u,k} \underset{bu,k}{\operatorname{SNR}} \omega_{b,k}\right)^{\mathrm{-1}}  } 
\end{aligned}
\end{equation}
\end{lemma}
\begin{proof}
See Appendix \ref{Appendix_subsection:information_loss_3D_orientation_3D_velocity}.
\end{proof}

\begin{lemma}
\label{lemma:information_loss_FIM_3D_velocity_3D_velocity}
The loss of information about $3$D velocity of the receiver due to uncertainty in the nuisance parameters ${\bm{\kappa}_{2}}$ is
\begin{equation}
\label{equ_lemma:information_loss_FIM_3D_velocity_3D_velocity}
\begin{aligned}
\medmath{\bm{G}_{{{y} }}(\bm{y}_{}| \bm{\eta} ;\bm{v}_{U},\bm{v}_{U})} = \medmath{\frac{\Delta_{t}^{2}}{c^2} \sum_{b}} \medmath{  \norm{\sum_{k^{},u^{}} \underset{bu^{},k^{}}{\operatorname{SNR}}  (k - 1)  \bm{\Delta}_{bu,k}^{\mathrm{T}}  \omega_{b,k} }^2
}    \medmath{   \left(\sum_{u,k} \underset{bu,k}{\operatorname{SNR}} \omega_{b,k}\right)^{\mathrm{-1}}  } 
+ \frac{1}{c^2}\medmath{\sum_{b}\norm{{\sum_{k^{},u^{} } \underset{bu^{},k^{}}{\operatorname{SNR}} \; \; \bm{\Delta}_{bU,k}^{\mathrm{T}}   }  \frac{(f_{c}^{}) (t_{obu,k^{}}^{2})}{2} }^{2}  \left(\sum_{u,k} \frac{\underset{bu,k}{\operatorname{SNR}}  t_{obu,k}^{2}}{2}\right)^{-1}}
\end{aligned}
\end{equation}
\end{lemma}
\begin{proof}
See Appendix \ref{Appendix_subsection:information_loss_3D_velocity_3D_velocity}.
\end{proof}
The elements in the EFIM for the location parameters, $\mathbf{J}_{ \bm{\bm{y}}; \bm{\kappa}_1}^{\mathrm{e}}$ are obtained by appropriately combining the Lemmas in Section \ref{subsection:FIM_channel_parameters} and the Lemmas in Section \ref{subsection:information_loss_FIM_channel_parameters}. The EFIM for the location parameters is 
$$\mathbf{J}_{ \bm{\bm{y}}; \bm{\kappa}_1}^{\mathrm{e}} =\mathbf{J}_{ \bm{\bm{y}}; \bm{\kappa}_1}^{} - \mathbf{J}_{ \bm{\bm{y}}; \bm{\kappa}_1}^{nu} =\mathbf{J}_{ \bm{\bm{y}}; \bm{\kappa}_1}^{}-
\mathbf{J}_{ \bm{\bm{y}}; \bm{\kappa}_1, \bm{\kappa}_2}^{} \mathbf{J}_{ \bm{\bm{y}}; \bm{\kappa}_2}^{-1} \mathbf{J}_{ \bm{\bm{y}}; \bm{\kappa}_1, \bm{\kappa}_2}^{\mathrm{T}}.$$

\clearpage
\subsection{FIM for $3$D Localization}
In this section, we consider available information for estimating one of the location parameters when the other two location parameters are known. We start with the FIM for the $3$D position estimation when both the $3$D orientation and $3$D velocity are known, and then we proceed to the FIM for the $3$D orientation estimation when both the $3$D position and $3$D velocity are known. Finally, we present the FIM for the $3$D velocity estimation when the $3$D position and $3$D orientation are known.

\begin{theorem}
\label{theorem:FIM_3D_position}
First, when both the $3$D orientation and $3$D velocity are known, the EFIM of the $3$D position of the receiver is given by (\ref{equ_theorem:FIM_3D_position}). Second, when both the $3$D position and $3$D velocity are known, the EFIM of the $3$D orientation of the receiver is given by (\ref{equ_theorem:FIM_3D_orientation}). Finally, when both the $3$D position and $3$D orientation are known, the EFIM of the $3$D velocity of the receiver is given by (\ref{equ_theorem:FIM_3D_velocity}).

\begin{figure*}
\begin{align}
\begin{split}
\label{equ_theorem:FIM_3D_position}
\mathbf{J}_{ \bm{\bm{y}}; \bm{p}_{U}}^{\mathrm{e}} &= [\mathbf{J}_{ \bm{\bm{y}}; \bm{\kappa}_{1}}^{\mathrm{e}}]_{[1:3,1:3]} 
= \medmath{\bm{F}_{{{y} }}(\bm{y}_{}| \bm{\eta} ;\bm{p}_{U},\bm{p}_{U})  } - \medmath{ \bm{G}_{{{y} }}(\bm{y}_{}| \bm{\eta} ;\bm{p}_{U},\bm{p}_{U}) } \\ &= {\sum_{b,k^{},u^{}} \underset{bu,k}{\operatorname{SNR}}  \Bigg[ \frac{\omega_{b,k}}{c^2}   \bm{\Delta}_{bu,k} 
 \bm{\Delta}_{bu^{},k^{}}^{\mathrm{T}}}  + {\frac{f_{c}^2 t_{obu,k}^2\nabla_{\bm{p}_{U}} \nu_{b,k} \nabla_{\bm{p}_{U}}^{\mathrm{T}}\nu_{b,k}}{2}  \Bigg]} \\ &- \medmath{\Bigg[\sum_{b}   \frac{1}{c^2}\norm{\sum_{k^{},u^{}} \underset{bu^{},k^{}}{\operatorname{SNR}}\bm{\Delta}_{bu^{},k^{}}^{\mathrm{T}}  \omega_{b,k} }}^{2}    \medmath{   \left(\sum_{u,k} \underset{bu,k}{\operatorname{SNR}} \omega_{b,k}\right)^{\mathrm{-1}}  }  + \medmath{\sum_{b}\norm{{\sum_{k^{},u^{} } \underset{bu^{},k^{}}{\operatorname{SNR}} \; \; \nabla_{\bm{p}_{U}}^{\mathrm{T}} \nu_{b,k^{}}   }  \frac{(f_{c}^{}) (t_{obu,k^{}}^{2})}{2} }^{2}  \left(\sum_{u,k} \frac{\underset{bu,k}{\operatorname{SNR}}  t_{obu,k}^{2}}{2}\right)^{-1}\Bigg]}.
\end{split}
\end{align}
\end{figure*}

\begin{figure*}
\begin{align}
\begin{split}
\label{equ_theorem:FIM_3D_orientation}
\mathbf{J}_{ \bm{\bm{y}}; \bm{\Phi}_{U}}^{\mathrm{e}} &= [\mathbf{J}_{ \bm{\bm{y}}; \bm{\kappa}_{1}}^{\mathrm{e}}]_{[4:6,4:6]} 
= \medmath{\bm{F}_{{{y} }}(\bm{y}_{}| \bm{\eta} ;\bm{\Phi}_{U},\bm{\Phi}_{U})  } - \medmath{ \bm{G}_{{{y} }}(\bm{y}_{}| \bm{\eta} ;\bm{\Phi}_{U},\bm{\Phi}_{U}) } \\ &= {\sum_{b,k^{},u^{}} \underset{bu,k}{\operatorname{SNR}}  \Bigg[ \omega_{b,k}     \nabla_{\bm{\Phi}_{U}} \tau_{bu^{},k^{}}
 \nabla_{\bm{\Phi}_{U}}^{\mathrm{T}} \tau_{bu^{},k^{}}}    \Bigg]  - \medmath{\sum_{b}  \sum_{k^{},u^{}k^{'},u^{'}} \underset{bu^{},k^{}}{\operatorname{SNR}} \underset{bu^{'},k^{'}}{\operatorname{SNR}} \nabla_{\bm{\Phi}_{U}} \tau_{bu^{},k^{}}\nabla_{\bm{\Phi}_{U}}^{\mathrm{T}} \tau_{bu^{'},k^{'}} \omega_{b,k} \omega_{b,k^{'}}}    \medmath{   \left(\sum_{u,k} \underset{bu,k}{\operatorname{SNR}} \omega_{b,k}\right)^{\mathrm{-1}}}.
\end{split}
\end{align}
\end{figure*}

\begin{figure*}
\begin{align}
\begin{split}
\label{equ_theorem:FIM_3D_velocity}
\mathbf{J}_{ \bm{\bm{y}}; \bm{v}_{U}}^{\mathrm{e}} &= [\mathbf{J}_{ \bm{\bm{y}}; \bm{\kappa}_{1}}^{\mathrm{e}}]_{[7:9,7:9]} 
= \medmath{\bm{F}_{{{y} }}(\bm{y}_{}| \bm{\eta} ;\bm{v}_{U},\bm{v}_{U})  } - \medmath{ \bm{G}_{{{y} }}(\bm{y}_{}| \bm{\eta} ;\bm{v}_{U},\bm{v}_{U}) } \\  \\ &=\medmath{\sum_{b,k^{},u^{}} \underset{bu,k}{\operatorname{SNR}}  \Bigg[ \frac{(k - 1)^2 \Delta_{t}^{2}\omega_{b,k}}{c^2}     \bm{\Delta}_{bu,k} 
\bm{\Delta}_{bu,k}^{\mathrm{T}} }  + \medmath{\frac{f_{c}^2 t_{obu,k}^2\bm{\Delta}_{bU,k}\bm{\Delta}_{bU,k}^{\mathrm{T}}}{2}  \Bigg]} \\ &- \Bigg[\medmath{\frac{\Delta_{t}^{2}}{c^2} \sum_{b}} \medmath{  \norm{\sum_{k^{},u^{}} \underset{bu^{},k^{}}{\operatorname{SNR}}  (k - 1)  \bm{\Delta}_{bu,k}^{\mathrm{T}}  \omega_{b,k} }^2
}    \medmath{   \left(\sum_{u,k} \underset{bu,k}{\operatorname{SNR}} \omega_{b,k}\right)^{\mathrm{-1}}  } 
+ \medmath{ \frac{1}{c^2}\sum_{b}\norm{{\sum_{k^{},u^{} } \underset{bu^{},k^{}}{\operatorname{SNR}} \; \; \bm{\Delta}_{bU,k}^{\mathrm{T}}   }  \frac{(f_{c}^{}) (t_{obu,k^{}}^{2})}{2} }^{2}  \left(\sum_{u,k} \frac{\underset{bu,k}{\operatorname{SNR}}  t_{obu,k}^{2}}{2}\right)^{-1}}\Bigg].
\end{split}
\end{align}
\end{figure*}

\end{theorem}
\begin{proof}
The proof follows by subtracting the appropriate Lemmas in Section \ref{subsection:information_loss_FIM_channel_parameters} from the Lemmas in Section \ref{subsection:FIM_channel_parameters}, following the EFIM definition and then selecting the appropriate diagonal.
\end{proof}

\clearpage
\subsection{FIM for $6$D Localization}
Here, we focus on the estimation of any two of the location parameters. First, we derive the FIM for the $3$D position and $3$D orientation when the $3$D velocity is known. In this scenario, the resulting EFIM matrix can be written as
\begin{equation}
\label{equ:FIM_6D_3D_position_3D_orientation}
\begin{aligned}
\left[\begin{array}{cc}
\mathbf{J}_{ \bm{\bm{y}}; \bm{p}_{U}}^{\mathrm{e}} & \mathbf{J}_{ \bm{\bm{y}};[ \bm{p}_{U}, \bm{\Phi}_{U}]}^{\mathrm{e}} \\
(\mathbf{J}_{ \bm{\bm{y}};[ \bm{p}_{U}, \bm{\Phi}_{U}]}^{\mathrm{e}})^{\mathrm{T}} & \mathbf{J}_{ \bm{\bm{y}}; \bm{\Phi}_{U}}^{\mathrm{e}} \\
\end{array}\right].
\end{aligned}
\end{equation}
Now, to isolate and obtain the EFIM for $3$D position for this case of $6$D estimation with known $3$D velocity, we apply the EFIM to the matrix in (\ref{equ:FIM_6D_3D_position_3D_orientation}) and obtain
\begin{equation}
\label{equ:FIM_6D_3D_position_3D_position_3D_orientation}
\mathbf{J}_{ \bm{\bm{y}}; \bm{p}_{U}}^{\mathrm{ee}} = \mathbf{J}_{ \bm{\bm{y}}; \bm{p}_{U}}^{\mathrm{e}} - \mathbf{J}_{ \bm{\bm{y}}; [ \bm{p}_{U}, \bm{\Phi}_{U}]}^{\mathrm{e}}  [\mathbf{J}_{ \bm{\bm{y}}; \bm{\Phi}_{U}}^{\mathrm{e}}]^{-1}  [\mathbf{J}_{ \bm{\bm{y}}; [ \bm{p}_{U}, \bm{\Phi}_{U}]}^{\mathrm{e}}]^{\mathrm{T}}.
\end{equation}

\begin{theorem}
\label{theorem:FIM_6D_3D_position_3D_position_3D_orientation}
The following are necessary conditions to estimate the $3$D position when the $3$D orientation is unknown but the $3$D velocity is known: i) there has to be enough information to estimate the $3$D position when both the $3$D velocity and $3$D orientation are known, and ii) there has to be enough information to estimate the $3$D orientation when both the $3$D position and $3$D velocity are known.
\end{theorem}
\begin{proof}
    The proof follows by examining (\ref{equ:FIM_6D_3D_position_3D_position_3D_orientation}) and noticing that  $\mathbf{J}_{ \bm{\bm{y}}; \bm{\Phi}_{U}}^{\mathrm{e}}$ has to be invertible for $\mathbf{J}_{ \bm{\bm{y}}; \bm{p}_{U}}^{\mathrm{ee}}$ to exist and $\mathbf{J}_{ \bm{\bm{y}}; \bm{p}_{U}}^{\mathrm{e}}$ has to be positive definite for $\mathbf{J}_{ \bm{\bm{y}}; \bm{p}_{U}}^{\mathrm{ee}}$ to be positive definite.
\end{proof}

\begin{corollary}
\label{corollary:FIM_6D_3D_position_3D_position_3D_orientation_1}

The minimal infrastructure to estimate the $3$D orientation when the $3$D position and $3$D velocity are known is a necessary condition to estimate the $3$D position when the $3$D orientation is unknown but the $3$D velocity is known. 
\end{corollary}

\begin{corollary}
\label{corollary:FIM_6D_3D_position_3D_position_3D_orientation_2}

The minimal infrastructure to estimate the $3$D position when the $3$D velocity and $3$D orientation are known is a necessary condition to estimate the $3$D position when the $3$D orientation is unknown but the $3$D velocity is known. 
\end{corollary}

Similarly, to isolate and obtain the EFIM for $3$D orientation for this case of $6$D estimation with known $3$D velocity, we apply the EFIM to the matrix in (\ref{equ:FIM_6D_3D_position_3D_orientation}) and obtain
\begin{equation}
\label{equ:FIM_6D_3D_orientation_3D_position_3D_orientation}
\mathbf{J}_{ \bm{\bm{y}}; \bm{\Phi}_{U}}^{\mathrm{ee}} = \mathbf{J}_{ \bm{\bm{y}}; \bm{\Phi}_{U}}^{\mathrm{e}} -  [\mathbf{J}_{ \bm{\bm{y}}; [ \bm{p}_{U}, \bm{\Phi}_{U}]}^{\mathrm{e}}]^{\mathrm{T}} [\mathbf{J}_{ \bm{\bm{y}}; \bm{p}_{U}}^{\mathrm{e}}]^{-1} \mathbf{J}_{ \bm{\bm{y}}; [ \bm{p}_{U}, \bm{\Phi}_{U}]}^{\mathrm{e}}  .
\end{equation}

\begin{theorem}
\label{theorem:FIM_6D_3D_orientation_3D_position_3D_velocity}
The following are necessary conditions to estimate the $3$D orientation when the $3$D position is unknown but the $3$D velocity is known: i) there has to be enough information to estimate the $3$D orientation when both the $3$D position and $3$D velocity are known, and ii) there has to be enough information to estimate the $3$D position when both the $3$D velocity and $3$D orientation are known.
\end{theorem}
\begin{proof}
    The proof follows by examining (\ref{equ:FIM_6D_3D_orientation_3D_position_3D_orientation}) and noticing that  $\mathbf{J}_{ \bm{\bm{y}}; \bm{p}_{U}}^{\mathrm{e}}$ has to be invertible for $\mathbf{J}_{ \bm{\bm{y}}; \bm{\Phi}_{U}}^{\mathrm{ee}}$ to exist and $\mathbf{J}_{ \bm{\bm{y}}; \bm{\Phi}_{U}}^{\mathrm{e}}$ has to be positive definite for $\mathbf{J}_{ \bm{\bm{y}}; \bm{\Phi}_{U}}^{\mathrm{ee}}$ to be positive definite.
\end{proof}

\begin{corollary}
\label{corollary:FIM_6D_3D_orientation_3D_position_3D_orientation_1}

The minimal infrastructure to estimate the $3$D position when the $3$D velocity and $3$D orientation are known is a necessary condition to estimate the $3$D orientation when the $3$D position is unknown but the $3$D velocity is known. 
\end{corollary}

\begin{corollary}
\label{corollary:FIM_6D_3D_orientation_3D_position_3D_orientation_2}

The minimal infrastructure to estimate the $3$D orientation when the $3$D position and $3$D velocity are known is a necessary condition to estimate the $3$D orientation when the $3$D position is unknown but the $3$D velocity is known. 
\end{corollary}

The following theorem handles the case of joint estimation of the $3$D position and $3$D orientation.
\begin{theorem}
\label{theorem:FIM_6D_3D_joint_3D_position_3D_orientation}
The $3$D position and $3$D orientation  can be jointly estimated if the following conditions are met: i) there is enough information to estimate the $3$D position when the $3$D velocity and $3$D orientation are known, ii) there is enough information to estimate the $3$D orientation when the $3$D position and $3$D velocity are known, and iii) either (\ref{equ:FIM_6D_3D_position_3D_position_3D_orientation}) or  (\ref{equ:FIM_6D_3D_orientation_3D_position_3D_orientation}) is invertible.
\end{theorem}
\begin{proof}
    The proof follows from applying the Schur's complement and the matrix inversion lemma to (\ref{equ:FIM_6D_3D_position_3D_orientation}).
\end{proof}

Now to investigate the $6$D case  when the $3$D position is known, we start by presenting the appropriate EFIM matrix, which can be written as
\begin{equation}
\label{equ:FIM_6D_3D_orientation_3D_velocity}
\begin{aligned}
\left[\begin{array}{cc}
\mathbf{J}_{ \bm{\bm{y}}; \bm{\Phi}_{U}}^{\mathrm{e}} & \mathbf{J}_{ \bm{\bm{y}};[  \bm{\Phi}_{U}, \bm{v}_{U}]}^{\mathrm{e}} \\
(\mathbf{J}_{ \bm{\bm{y}};[  \bm{\Phi}_{U}, \bm{v}_{U}]}^{\mathrm{e}})^{\mathrm{T}} & \mathbf{J}_{ \bm{\bm{y}}; \bm{v}_{U}}^{\mathrm{e}} \\
\end{array}\right].
\end{aligned}
\end{equation}
Now, to isolate and obtain the EFIM for $3$D orientation for this case of $6$D estimation with known $3$D position, we apply the EFIM to the matrix in (\ref{equ:FIM_6D_3D_orientation_3D_velocity}) and obtain
\begin{equation}
\label{equ:FIM_6D_3D_orientation_3D_orientation_3D_velocity}
\mathbf{J}_{ \bm{\bm{y}}; \bm{\Phi}_{U}}^{\mathrm{ee}} = \mathbf{J}_{ \bm{\bm{y}}; \bm{\Phi}_{U}}^{\mathrm{e}} - \mathbf{J}_{ \bm{\bm{y}}; [ \bm{\Phi}_{U}, \bm{v}_{U}]}^{\mathrm{e}}  [\mathbf{J}_{ \bm{\bm{y}}; \bm{v}_{U}}^{\mathrm{e}}]^{-1}  [\mathbf{J}_{ \bm{\bm{y}}; [ \bm{\Phi}_{U}, \bm{v}_{U}]}^{\mathrm{e}}]^{\mathrm{T}}.
\end{equation}

\begin{theorem}
\label{theorem:FIM_6D_3D_orientation_3D_velocity_3D_orientation}
To estimate the 3D orientation when the 3D velocity is unknown but the 3D position is known, the following conditions must be satisfied: i) there must be sufficient information to estimate the 3D orientation when both the 3D velocity and 3D position are known, and ii) there must be sufficient information to estimate the 3D velocity when both the 3D position and 3D orientation are known.
\end{theorem}
\begin{proof}
    The proof follows by examining (\ref{equ:FIM_6D_3D_orientation_3D_orientation_3D_velocity}) and noticing that  $\mathbf{J}_{ \bm{\bm{y}}; \bm{v}_{U}}^{\mathrm{e}}$ has to be invertible for $\mathbf{J}_{ \bm{\bm{y}}; \bm{\Phi}_{U}}^{\mathrm{ee}}$ to exist and $\mathbf{J}_{ \bm{\bm{y}}; \bm{\Phi}_{U}}^{\mathrm{e}}$ has to be positive definite for $\mathbf{J}_{ \bm{\bm{y}}; \bm{\Phi}_{U}}^{\mathrm{ee}}$ to be positive definite.
\end{proof}

\begin{corollary}
The minimal infrastructure to estimate the $3$D orientation when the $3$D position and $3$D velocity are known is a necessary condition to estimate the $3$D orientation when the $3$D velocity is unknown but the $3$D position is known. 
\end{corollary}

\begin{corollary}
The minimal infrastructure to estimate the $3$D velocity when the $3$D position and $3$D orientation are known is a necessary condition to estimate the $3$D orientation when the $3$D velocity is unknown but the $3$D position is known. 
\end{corollary}

Similarly, to isolate and obtain the EFIM for $3$D velocity for this case of $6$D estimation with known $3$D position, we apply the EFIM to the matrix in (\ref{equ:FIM_6D_3D_orientation_3D_velocity}) and obtain
\begin{equation}
\label{equ:FIM_6D_3D_velocity_3D_orientation_3D_velocity}
\mathbf{J}_{ \bm{\bm{y}}; \bm{v}_{U}}^{\mathrm{ee}} = \mathbf{J}_{ \bm{\bm{y}}; \bm{v}_{U}}^{\mathrm{e}} -  [\mathbf{J}_{ \bm{\bm{y}}; [ \bm{\Phi}_{U},\bm{v}_{U}]}^{\mathrm{e}}]^{\mathrm{T}} [\mathbf{J}_{ \bm{\bm{y}}; \bm{\Phi}_{U}}^{\mathrm{e}}]^{-1} \mathbf{J}_{ \bm{\bm{y}}; [  \bm{\Phi}_{U},\bm{v}_{U}]}^{\mathrm{e}} .
\end{equation}

\begin{theorem}
\label{theorem:FIM_6D_3D_velocity_3D_velocity_3D_orientation}
The following are necessary conditions to estimate the $3$D velocity when the $3$D orientation is unknown but the $3$D position is known: i) there has to be enough information to estimate the $3$D orientation when both the $3$D velocity and $3$D position are known, and ii) there has to be enough information to estimate the $3$D velocity when both the $3$D position and $3$D orientation are known.
\end{theorem}
\begin{proof}
    The proof follows by examining (\ref{equ:FIM_6D_3D_velocity_3D_orientation_3D_velocity}) and noticing that  $\mathbf{J}_{ \bm{\bm{y}}; \bm{\Phi}_{U}}^{\mathrm{e}}$ has to be invertible for $\mathbf{J}_{ \bm{\bm{y}}; \bm{v}_{U}}^{\mathrm{ee}}$ to exist and $\mathbf{J}_{ \bm{\bm{y}}; \bm{v}_{U}}^{\mathrm{e}}$ has to be positive definite for $\mathbf{J}_{ \bm{\bm{y}}; \bm{v}_{U}}^{\mathrm{ee}}$ to be positive definite.
\end{proof}

\begin{corollary}
\label{corollary:FIM_6D_3D_velocity_3D_velocity_3D_orientation_1}
The minimal infrastructure to estimate the $3$D velocity when the $3$D position and $3$D orientation are known is a necessary condition to estimate the $3$D velocity when the $3$D orientation is unknown but the $3$D position is known. 
\end{corollary}

\begin{corollary}
\label{corollary:FIM_6D_3D_velocity_3D_velocity_3D_orientation_2}

The minimal infrastructure to estimate the $3$D orientation when the $3$D position and $3$D velocity are known is a necessary condition to estimate the $3$D velocity when the $3$D orientation is unknown but the $3$D position is known. 
\end{corollary}

The following theorem handles the case of joint estimation of the $3$D velocity and $3$D orientation.
\begin{theorem}
\label{theorem:FIM_6D_3D_joint_3D_velocity_3D_orientation}
The $3$D velocity and $3$D orientation  can be jointly estimated if the following conditions are met: i) there is enough information to estimate the $3$D velocity when the $3$D position and $3$D orientation are known, ii) there is enough information to estimate the $3$D orientation when the $3$D position and $3$D velocity are known, and iii) either (\ref{equ:FIM_6D_3D_orientation_3D_orientation_3D_velocity}) or  (\ref{equ:FIM_6D_3D_velocity_3D_orientation_3D_velocity}) is invertible.
\end{theorem}
\begin{proof}
    The proof follows from applying the Schur's complement and the matrix inversion lemma to (\ref{equ:FIM_6D_3D_position_3D_orientation}).
\end{proof}

Finally, we present the scenario for the $3$D position and $3$D velocity estimation when the $3$D orientation is known. We present the appropriate EFIM matrix below
\begin{equation}
\label{equ:FIM_6D_3D_position_3D_velocity}
\begin{aligned}
\left[\begin{array}{cc}
\mathbf{J}_{ \bm{\bm{y}}; \bm{p}_{U}}^{\mathrm{e}} & \mathbf{J}_{ \bm{\bm{y}};[ \bm{p}_{U}, \bm{v}_{U}]}^{\mathrm{e}} \\
(\mathbf{J}_{ \bm{\bm{y}};[ \bm{p}_{U}, \bm{v}_{U}]}^{\mathrm{e}})^{\mathrm{T}} & \mathbf{J}_{ \bm{\bm{y}}; \bm{v}_{U}}^{\mathrm{e}} \\
\end{array}\right].
\end{aligned}
\end{equation}
Now, to isolate and obtain the EFIM for $3$D position for this case of $6$D estimation with known $3$D orientation, we apply the EFIM to the matrix in (\ref{equ:FIM_6D_3D_position_3D_velocity}) and obtain
\begin{equation}
\label{equ:FIM_6D_3D_position_3D_position_3D_velociy}
\mathbf{J}_{ \bm{\bm{y}}; \bm{p}_{U}}^{\mathrm{ee}} = \mathbf{J}_{ \bm{\bm{y}}; \bm{p}_{U}}^{\mathrm{e}} - \mathbf{J}_{ \bm{\bm{y}}; [ \bm{p}_{U}, \bm{v}_{U}]}^{\mathrm{e}}  [\mathbf{J}_{ \bm{\bm{y}}; \bm{v}_{U}}^{\mathrm{e}}]^{-1}  [\mathbf{J}_{ \bm{\bm{y}}; [ \bm{p}_{U}, \bm{v}_{U}]}^{\mathrm{e}}]^{\mathrm{T}}.
\end{equation}
\begin{theorem}
\label{theorem:FIM_6D_3D_position_3D_position_3D_velocity}
The following are necessary conditions to estimate the $3$D position when the $3$D velocity is unknown but the $3$D orientation is known: i) there has to be enough information to estimate the $3$D position when both the $3$D velocity and $3$D orientation are known, and ii) there has to be enough information to estimate the $3$D velocity when both the $3$D position and $3$D orientation are known.
\end{theorem}
\begin{proof}
    The proof follows by examining (\ref{equ:FIM_6D_3D_position_3D_position_3D_velociy}) and noticing that  $\mathbf{J}_{ \bm{\bm{y}}; \bm{v}_{U}}^{\mathrm{e}}$ has to be invertible for $\mathbf{J}_{ \bm{\bm{y}}; \bm{p}_{U}}^{\mathrm{ee}}$ to exist and $\mathbf{J}_{ \bm{\bm{y}}; \bm{p}_{U}}^{\mathrm{ee}}$ has to be positive definite for $\mathbf{J}_{ \bm{\bm{y}}; \bm{p}_{U}}^{\mathrm{ee}}$ to be positive definite.
\end{proof}

\begin{corollary}
\label{corollary:FIM_6D_3D_position_3D_position_3D_velocity_1}
The minimal infrastructure to estimate the $3$D velocity when the $3$D position and $3$D orientation are known is a necessary condition to estimate the $3$D position when the $3$D velocity is unknown but the $3$D orientation is known. 
\end{corollary}

\begin{corollary}
\label{corollary:FIM_6D_3D_position_3D_position_3D_velocity_2}
The minimal infrastructure to estimate the $3$D position when the $3$D velocity and $3$D orientation are known is a necessary condition to estimate the $3$D position when the $3$D velocity is unknown but the $3$D orientation is known. 
\end{corollary}

Similarly, to isolate and obtain the EFIM for $3$D velocity for this case of $6$D estimation with known $3$D orientation, we apply the EFIM to the matrix in (\ref{equ:FIM_6D_3D_position_3D_velocity}) and obtain
\begin{equation}
\label{equ:FIM_6D_3D_velocity_3D_position_3D_velocity}
\mathbf{J}_{ \bm{\bm{y}}; \bm{v}_{U}}^{\mathrm{ee}} = \mathbf{J}_{ \bm{\bm{y}}; \bm{v}_{U}}^{\mathrm{e}} -  [\mathbf{J}_{ \bm{\bm{y}}; [ \bm{p}_{U}, \bm{v}_{U}]}^{\mathrm{e}}]^{\mathrm{T}} [\mathbf{J}_{ \bm{\bm{y}}; \bm{p}_{U}}^{\mathrm{e}}]^{-1}  \mathbf{J}_{ \bm{\bm{y}}; [ \bm{p}_{U}, \bm{v}_{U}]}^{\mathrm{e}}.
\end{equation}

\begin{theorem}
\label{theorem:FIM_6D_3D_velocity_3D_position_3D_velocity}
The following are necessary conditions to estimate the $3$D velocity when the $3$D position is unknown but the $3$D orientation is known: i) there has to be enough information to estimate the $3$D velocity when both the $3$D position and $3$D orientation are known, and ii) there has to be enough information to estimate the $3$D position when both the $3$D velocity and $3$D orientation are known.\end{theorem}
\begin{proof}
    The proof follows by examining (\ref{equ:FIM_6D_3D_velocity_3D_position_3D_velocity}) and noticing that  $\mathbf{J}_{ \bm{\bm{y}}; \bm{p}_{U}}^{\mathrm{e}}$ has to be invertible for $\mathbf{J}_{ \bm{\bm{y}}; \bm{v}_{U}}^{\mathrm{ee}}$ to exist and $\mathbf{J}_{ \bm{\bm{y}}; \bm{v}_{U}}^{\mathrm{ee}}$ has to be positive definite for $\mathbf{J}_{ \bm{\bm{y}}; \bm{v}_{U}}^{\mathrm{ee}}$ to be positive definite.
\end{proof}

\begin{corollary}
\label{corollary:FIM_6D_3D_velocity_3D_position_3D_velocity_1}
The minimal infrastructure to estimate the $3$D position when the $3$D velocity and $3$D orientation are known is a necessary condition to estimate the $3$D velocity when the $3$D position is unknown but the $3$D orientation is known. 
\end{corollary}

\begin{corollary}
\label{corollary:FIM_6D_3D_velocity_3D_position_3D_velocity_2}
The minimal infrastructure to estimate the $3$D velocity when the $3$D position and $3$D orientation are known is a necessary condition to estimate the $3$D velocity when the $3$D position is unknown but the $3$D orientation is known. 
\end{corollary}

The following theorem handles the case of joint estimation of the $3$D position and $3$D velocity.
\begin{theorem}
\label{theorem:FIM_6D_3D_joint_3D_position_3D_velocity}
The $3$D position and $3$D velocity  can be jointly estimated if the following conditions are met: i) there is enough information to estimate the $3$D position when the $3$D velocity and $3$D orientation are known, ii) there is enough information to estimate the $3$D velocity when the $3$D position and $3$D orientation are known, and iii) either (\ref{equ:FIM_6D_3D_position_3D_position_3D_velociy}) or  (\ref{equ:FIM_6D_3D_velocity_3D_position_3D_velocity}) is invertible.
\end{theorem}
\begin{proof}
    The proof follows from applying the Schur's complement and the matrix inversion lemma to (\ref{equ:FIM_6D_3D_position_3D_velocity}).
\end{proof}

\clearpage
\subsection{FIM for $9$D Localization}
In this section, we analyze the FIM when all location parameters are unknown. More specifically, we present the available information when all the location parameters - $3$D position, $3$D orientation and $3$D velocity are parameters to be estimated.

\begin{theorem}
\label{theorem:FIM_9D_position}
If $\mathbf{J}_{ \bm{\bm{y}}; \bm{\Phi}_{U}}^{\mathrm{e}}$ is invertible then the loss in information about $\bm{p}_{U}$ due to the unknown $\bm{\Phi}_{U}$ and $\bm{v}_{U}$ which is specified by  $\mathbf{J}_{ \bm{\bm{y}}; \bm{p}_{U}}^{nu}$ exists if and only if $ \bm{S} = \mathbf{J}_{ \bm{\bm{y}}; \bm{v}_{U}}^{\mathrm{e}} - \mathbf{J}_{ \bm{\bm{y}}; [ \bm{v}_{U}, \bm{\Phi}_{U}]}^{\mathrm{e}}  [\mathbf{J}_{ \bm{\bm{y}}; \bm{\Phi}_{U}}^{\mathrm{e}}]^{-1}  \mathbf{J}_{ \bm{\bm{y}}; [ \bm{\Phi}_{U}, \bm{v}_{U}]}^{\mathrm{e}}$ is invertible. Subsequently, $\mathbf{J}_{ \bm{\bm{y}}; \bm{p}_{U}}^{nu}$ is given by (\ref{equ_theorem:FIM_3D_nuisance_position}), and the EFIM for the $3$D position in this $9$D localization scenario is 
\begin{equation}
\label{equ:FIM_9D_3D_position_3D_position_3D_orientation_3D_velocity}
\mathbf{J}_{ \bm{\bm{y}}; \bm{p}_{U}}^{\mathrm{eee}} = \mathbf{J}_{ \bm{\bm{y}}; \bm{p}_{U}}^{\mathrm{e}} -  \mathbf{J}_{ \bm{\bm{y}}; \bm{p}_{U}}^{nu}.
\end{equation}

\begin{figure*}
\begin{align}
\begin{split}
\label{equ_theorem:FIM_3D_nuisance_position}
\mathbf{J}_{ \bm{\bm{y}}; \bm{p}_{U}}^{nu}  &= \mathbf{J}_{ \bm{\bm{y}}; [\bm{p}_{U}, \bm{\Phi}_{U}]}^{e} [\mathbf{J}_{ \bm{\bm{y}};  \bm{\Phi}_{U}}^{e} ]^{-1} \mathbf{J}_{ \bm{\bm{y}}; [ \bm{\Phi}_{U}, \bm{p}_{U}]}^{e} + \mathbf{J}_{ \bm{\bm{y}}; [\bm{p}_{U}, \bm{\Phi}_{U}]}^{e} [\mathbf{J}_{ \bm{\bm{y}};  \bm{\Phi}_{U}}^{e} ]^{-1} \mathbf{J}_{ \bm{\bm{y}}; [\bm{\Phi}_{U}, \bm{v}_{U}]}^{e} \bm{S}^{-1}  \mathbf{J}_{ \bm{\bm{y}}; [ \bm{v}_{U}, \bm{\Phi}_{U}]}^{e} [\mathbf{J}_{ \bm{\bm{y}};  \bm{\Phi}_{U}}^{e} ]^{-1} \mathbf{J}_{ \bm{\bm{y}}; [ \bm{\Phi}_{U}, \bm{p}_{U}]}^{e} \\ 
& -  \mathbf{J}_{ \bm{\bm{y}}; [\bm{p}_{U}, \bm{v}_{U}]}^{e} \bm{S}^{-1}  \mathbf{J}_{ \bm{\bm{y}}; [ \bm{v}_{U}, \bm{\Phi}_{U}]}^{e} [\mathbf{J}_{ \bm{\bm{y}};  \bm{\Phi}_{U}}^{e} ]^{-1} \mathbf{J}_{ \bm{\bm{y}}; [ \bm{\Phi}_{U}, \bm{p}_{U}]}^{e} - \mathbf{J}_{ \bm{\bm{y}}; [\bm{p}_{U}, \bm{\Phi}_{U}]}^{e} [\mathbf{J}_{ \bm{\bm{y}};  \bm{\Phi}_{U}}^{e} ]^{-1} \mathbf{J}_{ \bm{\bm{y}}; [\bm{\Phi}_{U}, \bm{v}_{U}]}^{e} \bm{S}^{-1}  \mathbf{J}_{ \bm{\bm{y}}; [ \bm{v}_{U}, \bm{p}_{U}]}^{e} \\
& + \mathbf{J}_{ \bm{\bm{y}}; [\bm{p}_{U}, \bm{v}_{U}]}^{e} \bm{S}^{-1}  \mathbf{J}_{ \bm{\bm{y}}; [ \bm{v}_{U}, \bm{p}_{U}]}^{e}.
\end{split}
\end{align}
\end{figure*}
\end{theorem}
\begin{proof}
See Appendix \ref{Appendix_subsection:FIM_9D_position}.
\end{proof}

Next, we focus on the available information for the estimation of $\bm{v}_{U}.$

\begin{theorem}
\label{theorem:FIM_9D_velocity}
If $\mathbf{J}_{ \bm{\bm{y}}; \bm{p}_{U}}^{\mathrm{e}}$ is invertible then the loss in information about $\bm{v}_{U}$ due to the unknown $\bm{p}_{U}$ and $\bm{\Phi}_{U}$ which is specified by  $\mathbf{J}_{ \bm{\bm{y}}; \bm{v}_{U}}^{nu}$ exists if and only if $ \bm{S} = \mathbf{J}_{ \bm{\bm{y}}; \bm{\Phi}_{U}}^{\mathrm{e}} - \mathbf{J}_{ \bm{\bm{y}}; [  \bm{\Phi}_{U}, \bm{p}_{U}]}^{\mathrm{e}}  [\mathbf{J}_{ \bm{\bm{y}}; \bm{p}_{U}}^{\mathrm{e}}]^{-1}  [\mathbf{J}_{ \bm{\bm{y}}; [ \bm{p}_{U}, \bm{\Phi}_{U}]}^{\mathrm{e}}]$ is invertible. Subsequently, $\mathbf{J}_{ \bm{\bm{y}}; \bm{v}_{U}}^{nu}$ is given by (\ref{equ_theorem:FIM_3D_nuisance_velocity}), and the EFIM for the $3$D velocity in this $9$D localization scenario is 
\begin{equation}
\label{equ:FIM_9D_3D_velocity_3D_position_3D_orientation_3D_velocity}
\mathbf{J}_{ \bm{\bm{y}}; \bm{v}_{U}}^{\mathrm{eee}} = \mathbf{J}_{ \bm{\bm{y}}; \bm{v}_{U}}^{\mathrm{e}} -  \mathbf{J}_{ \bm{\bm{y}}; \bm{v}_{U}}^{nu}.
\end{equation}
\begin{figure*}
\begin{align}
\begin{split}
\label{equ_theorem:FIM_3D_nuisance_velocity}
\mathbf{J}_{ \bm{\bm{y}}; \bm{v}_{U}}^{nu}  &= \mathbf{J}_{ \bm{\bm{y}}; [\bm{v}_{U},\bm{p}_{U}]}^{e} [\mathbf{J}_{ \bm{\bm{y}};  \bm{p}_{U}}^{e} ]^{-1} \mathbf{J}_{ \bm{\bm{y}}; [  \bm{p}_{U}, \bm{v}_{U}]}^{e} + \mathbf{J}_{ \bm{\bm{y}}; [ \bm{v}_{U}, \bm{p}_{U}]}^{e} [\mathbf{J}_{ \bm{\bm{y}};  \bm{p}_{U}}^{e} ]^{-1} \mathbf{J}_{ \bm{\bm{y}}; [\bm{p}_{U}, \bm{\Phi}_{U}]}^{e} \bm{S}^{-1}  \mathbf{J}_{ \bm{\bm{y}}; [ \bm{\Phi}_{U}, \bm{p}_{U}]}^{e} [\mathbf{J}_{ \bm{\bm{y}};  \bm{p}_{U}}^{e} ]^{-1} \mathbf{J}_{ \bm{\bm{y}}; [  \bm{p}_{U}, \bm{v}_{U}]}^{e} \\ 
& -  \mathbf{J}_{ \bm{\bm{y}}; [\bm{v}_{U}, \bm{\Phi}_{U}]}^{e} \bm{S}^{-1}  \mathbf{J}_{ \bm{\bm{y}}; [ \bm{\Phi}_{U}, \bm{p}_{U}]}^{e} [\mathbf{J}_{ \bm{\bm{y}};  \bm{p}_{U}}^{e} ]^{-1} \mathbf{J}_{ \bm{\bm{y}}; [ \bm{p}_{U}, \bm{v}_{U}]}^{e} - \mathbf{J}_{ \bm{\bm{y}}; [\bm{v}_{U}, \bm{p}_{U}]}^{e} [\mathbf{J}_{ \bm{\bm{y}};  \bm{p}_{U}}^{e} ]^{-1} \mathbf{J}_{ \bm{\bm{y}}; [\bm{p}_{U}, \bm{\Phi}_{U}]}^{e} \bm{S}^{-1}  \mathbf{J}_{ \bm{\bm{y}}; [ \bm{\Phi}_{U}, \bm{v}_{U}]}^{e} \\
& + \mathbf{J}_{ \bm{\bm{y}}; [ \bm{v}_{U}, \bm{\Phi}_{U}]}^{e} \bm{S}^{-1}  \mathbf{J}_{ \bm{\bm{y}}; [ \bm{\Phi}_{U}, \bm{v}_{U}]}^{e}.
\end{split}
\end{align}
\end{figure*}
\end{theorem}
\begin{proof}
See Appendix \ref{Appendix_subsection:FIM_9D_velocity}.
\end{proof}

Finally, we focus on the available information for the estimation of $\bm{\Phi}_{U}.$

\begin{theorem}
\label{theorem:FIM_9D_orientation}
If $\mathbf{J}_{ \bm{\bm{y}}; \bm{p}_{U}}^{\mathrm{e}}$ is invertible then the loss in information about $\bm{\Phi}_{U}$ due to the unknown $\bm{p}_{U}$ and $\bm{v}_{U}$ which is specified by  $\mathbf{J}_{ \bm{\bm{y}}; \bm{\Phi}_{U}}^{nu}$ exists if and only if $ \bm{S} = \mathbf{J}_{ \bm{\bm{y}}; \bm{v}_{U}}^{\mathrm{e}} - \mathbf{J}_{ \bm{\bm{y}}; [ \bm{v}_{U}, \bm{p}_{U}]}^{\mathrm{e}}  [\mathbf{J}_{ \bm{\bm{y}}; \bm{p}_{U}}^{\mathrm{e}}]^{-1}  [\mathbf{J}_{ \bm{\bm{y}}; [ \bm{p}_{U}, \bm{v}_{U}]}^{\mathrm{e}}]$ is invertible. Subsequently, $\mathbf{J}_{ \bm{\bm{y}}; \bm{\Phi}_{U}}^{nu}$ is given by (\ref{equ_theorem:FIM_3D_nuisance_orientation}), and the EFIM for the $3$D orientation in this $9$D localization scenario is 
\begin{equation}
\label{equ:FIM_9D_3D_orientation_3D_position_3D_orientation_3D_velocity}
\mathbf{J}_{ \bm{\bm{y}}; \bm{\Phi}_{U}}^{\mathrm{eee}} = \mathbf{J}_{ \bm{\bm{y}}; \bm{\Phi}_{U}}^{\mathrm{e}} -  \mathbf{J}_{ \bm{\bm{y}}; \bm{\Phi}_{U}}^{nu}.
\end{equation}

\begin{figure*}
\begin{align}
\begin{split}
\label{equ_theorem:FIM_3D_nuisance_orientation}
\mathbf{J}_{ \bm{\bm{y}}; \bm{\Phi}_{U}}^{nu}  &= \mathbf{J}_{ \bm{\bm{y}}; [\bm{\Phi}_{U},\bm{p}_{U}]}^{e} [\mathbf{J}_{ \bm{\bm{y}};  \bm{p}_{U}}^{e} ]^{-1} \mathbf{J}_{ \bm{\bm{y}}; [  \bm{p}_{U}, \bm{\Phi}_{U}]}^{e} + \mathbf{J}_{ \bm{\bm{y}}; [ \bm{\Phi}_{U}, \bm{p}_{U}]}^{e} [\mathbf{J}_{ \bm{\bm{y}};  \bm{p}_{U}}^{e} ]^{-1} \mathbf{J}_{ \bm{\bm{y}}; [\bm{p}_{U}, \bm{v}_{U}]}^{e} \bm{S}^{-1}  \mathbf{J}_{ \bm{\bm{y}}; [ \bm{v}_{U}, \bm{p}_{U}]}^{e} [\mathbf{J}_{ \bm{\bm{y}};  \bm{p}_{U}}^{e} ]^{-1} \mathbf{J}_{ \bm{\bm{y}}; [  \bm{p}_{U}, \bm{\Phi}_{U}]}^{e} \\ 
& -  \mathbf{J}_{ \bm{\bm{y}}; [\bm{\Phi}_{U}, \bm{v}_{U}]}^{e} \bm{S}^{-1}  \mathbf{J}_{ \bm{\bm{y}}; [ \bm{v}_{U}, \bm{p}_{U}]}^{e} [\mathbf{J}_{ \bm{\bm{y}};  \bm{p}_{U}}^{e} ]^{-1} \mathbf{J}_{ \bm{\bm{y}}; [ \bm{p}_{U}, \bm{\Phi}_{U}]}^{e} - \mathbf{J}_{ \bm{\bm{y}}; [\bm{\Phi}_{U}, \bm{p}_{U}]}^{e} [\mathbf{J}_{ \bm{\bm{y}};  \bm{p}_{U}}^{e} ]^{-1} \mathbf{J}_{ \bm{\bm{y}}; [\bm{p}_{U}, \bm{v}_{U}]}^{e} \bm{S}^{-1}  \mathbf{J}_{ \bm{\bm{y}}; [ \bm{v}_{U}, \bm{\Phi}_{U}]}^{e} \\
& + \mathbf{J}_{ \bm{\bm{y}}; [\bm{\Phi}_{U}, \bm{v}_{U}]}^{e} \bm{S}^{-1}  \mathbf{J}_{ \bm{\bm{y}}; [ \bm{v}_{U}, \bm{\Phi}_{U}]}^{e}.
\end{split}
\end{align}
\end{figure*}
\end{theorem}
\begin{proof}
See Appendix \ref{Appendix_subsection:FIM_9D_orientation}.
\end{proof}



\section{Numerical Results}
This section presents simulation results that describe the available information in signals received from LEO satellites during multiple transmission time slots on receivers with multiple antennas. We start by showing the minimum infrastructure needed to estimate different location parameters. More specifically, we present the minimum number of LEO satellites, time slots, and receive antennas that contribute to i) $3$D position estimation when the $3$D orientation and $3$D velocity are known, ii) $3$D orientation estimation when the $3$D position and $3$D velocity are known,  iii) $3$D velocity estimation when the $3$D position and $3$D orientation are known, iv) $3$D position and $3$D orientation estimation when the $3$D velocity is known, v) $3$D position and $3$D velocity estimation when the $3$D orientation is known,  vi) $3$D orientation and $3$D velocity estimation when the $3$D position is known, and vii)  $3$D position, $3$D velocity, and $3$D orientation estimation. We also present the CRLB for $3$D position, $3$D orientation, and $3$D velocity estimation in the $9$D localization case. We present the CRLB for the $3$D position as a function of the receive antenna, the CRLB for the $3$D position as a function of the spacing between transmission time slots, and the CRLB for the $3$D position as a function of the operating frequency.  We present the CRLB for the $3$D velocity as a function of the receive antenna, the CRLB for the $3$D velocity as a function of the spacing between transmission time slots, and the CRLB for the $3$D velocity as a function of the operating frequency. Finally, we present the CRLB for the $3$D orientation as a function of the receive antenna, the CRLB for the $3$D orientation as a function of the spacing between transmission time slots, and the CRLB for the $3$D orientation as a function of the operating frequency. 

We use the following simulation parameters. The SNR is assumed constant across the transmission time slots and receive antennas, and the following set of SNR values is considered: $\{40 \text{ dB}, 20 \text{ dB}, 0 \text{ dB}, -20 \text{ dB}\}$. The $x,y, \text{ and } z$ components of the position of the LEO satellites are randomly chosen, but LEO satellites are approximately $2000 \text{ km}$ from the receiver. The $x,y, \text{ and } z$ components of the velocity of the LEO satellites are randomly chosen and change every transmission time slot to depict acceleration, but the LEO satellites have a speed of  $8000 \text{ m/s}$. The receiver's position's $x,y, \text{ and } z$ components are randomly chosen, but the receiver is approximately $30 \text{ m}$ from the origin. The $x,y, \text{ and } z$ components of the receiver's velocity are randomly chosen and remain constant to depict constant velocity, but the receiver has a speed of  $25 \text{ m/s}$. The effective baseband bandwidth, $\alpha_{1b,k}$, is $100 \text{ MHz}$ and the BCC, $\alpha_{2b,k}$, is $0 \text{ MHz}$. 

\begin{table*}[!t]
\renewcommand{\arraystretch}{1.3}
\caption{Legend for Estimation Possibilities for $\bm{p}_{U}$ in the $3$D Localization Setup under Different Combinations of $N_B$, $N_K$, and $N_U$}
\label{Legend_Estimation_Possibilities_for_PU_in_the_3D_Localization}
\centering
\begin{tabular}{|c|c|}
\hline
a & A single TOA from a LEO satellite to the centroid of the receiver with the unit vector from the LEO satellite to\\  &  the centroid of the receiver.\\
\hline
b & Two TOAs from two LEO satellites and a single Doppler measurement from either of the two LEO satellites.\\
\hline
c &  Two Dopplers from two LEO satellites and a single TOA measurement from either of the two LEO satellites.\\
\hline
d & Two unit vectors from two distinct LEOs obtainable due to the presence of multiple receive antennas, \\ & which capture multiple TOAs from a single LEO satellite.\\
\hline
e &   A single TOA measurement from a distinct LEO and a unit vector from the other LEO satellites.\\
\hline
f &  A single Doppler measurement from a distinct LEO and a unit vector from the other LEO satellites.\\
\hline
g &   Three TOA measurements from three distinct LEO satellites.\\
\hline
h &  Three Doppler measurements from three distinct LEO satellites.\\
\hline
i &   Two TOAs from the same LEO satellite obtained during two distinct time slots in combination \\ & with the Doppler  measurements obtained during either of the time slots with respect to the LEO satellite.\\
\hline
j &  Two Doppler measurements obtained during the two distinct time slots with respect to the LEO satellite in \\ & combination  with the  TOA obtained during either of the time slots.\\
\hline
k &  Two unit vectors obtained during the two distinct time slots from the same LEO satellite \\ & obtainable due to the presence of multiple receive antennas.\\
\hline
l &   A single TOA measurement from the LEO during either of the time slots \\ & and a unit vector from the LEO during either of the time slots.\\
\hline
m &  A single Doppler measurement from the LEO during either of the time slots \\ & and a unit vector from the  LEO during either of the time slots.\\
\hline
n &   Three TOA measurements obtained during three different time slots from a single LEO satellite.\\
\hline
o &  Three Doppler measurements obtained during three different time slots from a single LEO satellite.\\
\hline
p &   Four TOA measurements obtained during four different time slots from a single LEO satellite \\ & with one of the TOA measurements serving as a reference measurement for time differencing.\\
\hline
q &  Four Doppler measurements obtained during four different time slots from a single LEO satellite \\ &  with one of the Doppler measurements serving as a reference measurement for frequency differencing.\\
\hline
\end{tabular}
\end{table*}

\subsection{$3$D Position Estimation}
Here, we investigate the minimal number of time slots, LEO satellites, and receive antennas that produce a positive definite FIM for the $3$D position of the receiver, which is defined by (\ref{equ_theorem:FIM_3D_position}).   We start with the cases where we only have measurements taken during a single time slot from various LEO satellites considering single or multiple antennas.

\subsubsection{$N_K = 1$, $N_B = 1$, \textit{and} $N_U > 1$}
Under this condition, the information is insufficient to find the $3$D position of the receiver when there is a time offset. However, without a time offset, the information is enough to find the $3$D position of the receiver. The $3$D position of the receiver can be obtained by combining the available TOA measurements from the single LEO satellite with the multiple receive antennas.

\subsubsection{$N_K = 1$, $N_B = 2$, \textit{and} $N_U = 1$}
Under this condition, the information is insufficient to find the $3$D position of the receiver when there is a time offset, a frequency offset, or both. However, without a time offset or frequency offset, the information is enough to find the $3$D position of the receiver. These can be done with two TOAs from two LEO satellites and a single Doppler measurement from either of the two LEO satellites or with two Dopplers from two LEO satellites and a single TOA measurement from either of the two LEO satellites.

\subsubsection{$N_K = 1$, $N_B = 2$, \textit{and} $N_U > 1$}
Under this condition, there is always information to find the $3$D position of the receiver, even with the presence of a time offset, a frequency offset, or both. Without a time offset or frequency offset, the $3$D position can be found using: i) two TOA measurements from two distinct LEOs and a Doppler measurement from either of the LEO, ii)  two Doppler measurements from two distinct LEOs and a TOA measurement from either of the LEO satellites, iii) two unit vectors from two distinct LEOs obtainable due to the presence of multiple receive antennas, which capture multiple TOAs from a single LEO satellite, iv) one TOA measurement from a distinct LEO and a unit vector from the other LEO satellites, and v) one Doppler measurement from a distinct LEO and a unit vector from the other LEO. With only a time offset, iii) and v) can be used to find a $3$D position.  With only a frequency offset, iii) and iv) can be used to find the $3$D position. When both a time and frequency offset are present, iii) can be used to find the $3$D position.

\begin{figure}[htb!]
\centering
\subfloat[]{\includegraphics[ width= 3.2in]{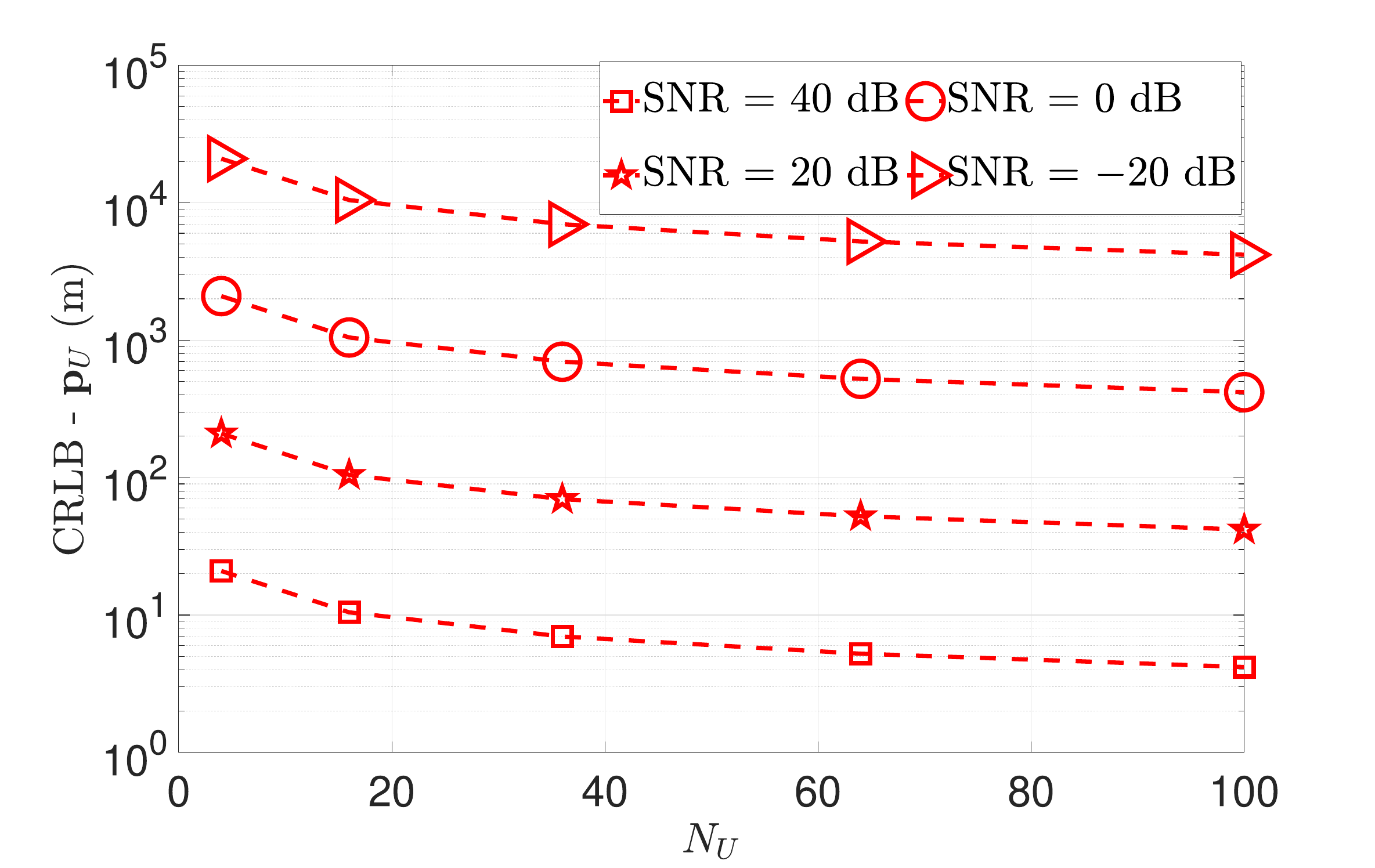}
\label{fig:Results/9D_PU_NU/_NB_3_N_K_3_N_U___delta_t_index_5_fcIndex_1_SNRIndex_1}}
\hfil
\subfloat[]{\includegraphics[ width= 3.2in]{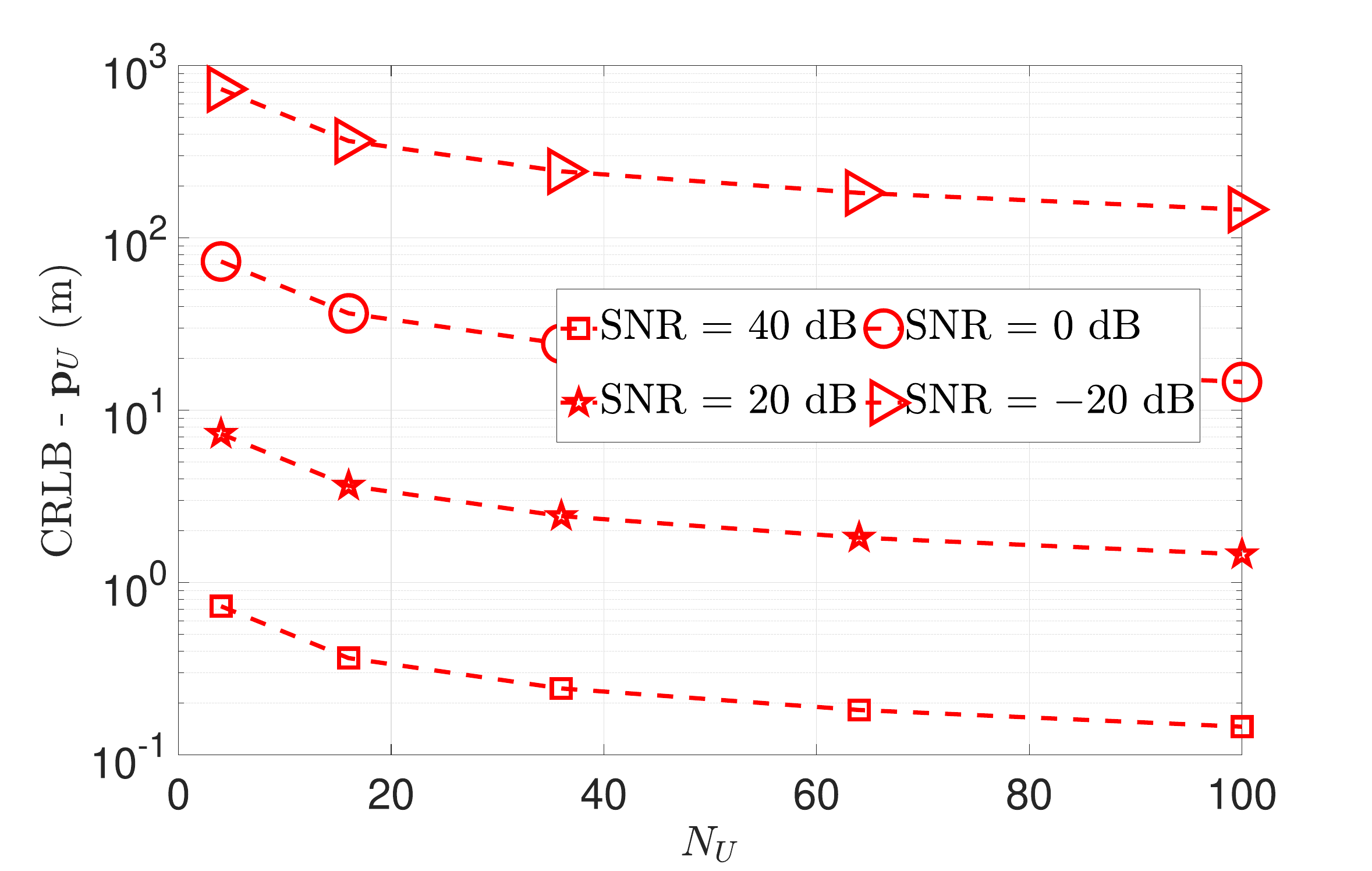}
\label{fig:Results/9D_PU_NU/_NB_3_N_K_3_N_U___delta_t_index_8_fcIndex_1_SNRIndex_1}}
\caption{CRLB for $\bm{p}_{U}$ in the $9$D localization scenario with $f_c = 1 \text{ GHz}$: (a) $\Delta_t = 10 \text{ s}$ and (b) $\Delta_t = 80 \text{ s}$.}
\label{Results:9D_PU_NB_3_N_K_3_N_U___delta_t_index_5_8_fcIndex_1_SNRIndex_1_NU}
\end{figure}

\begin{figure}[htb!]
\centering
\subfloat[]{\includegraphics[ width= 3.2in]{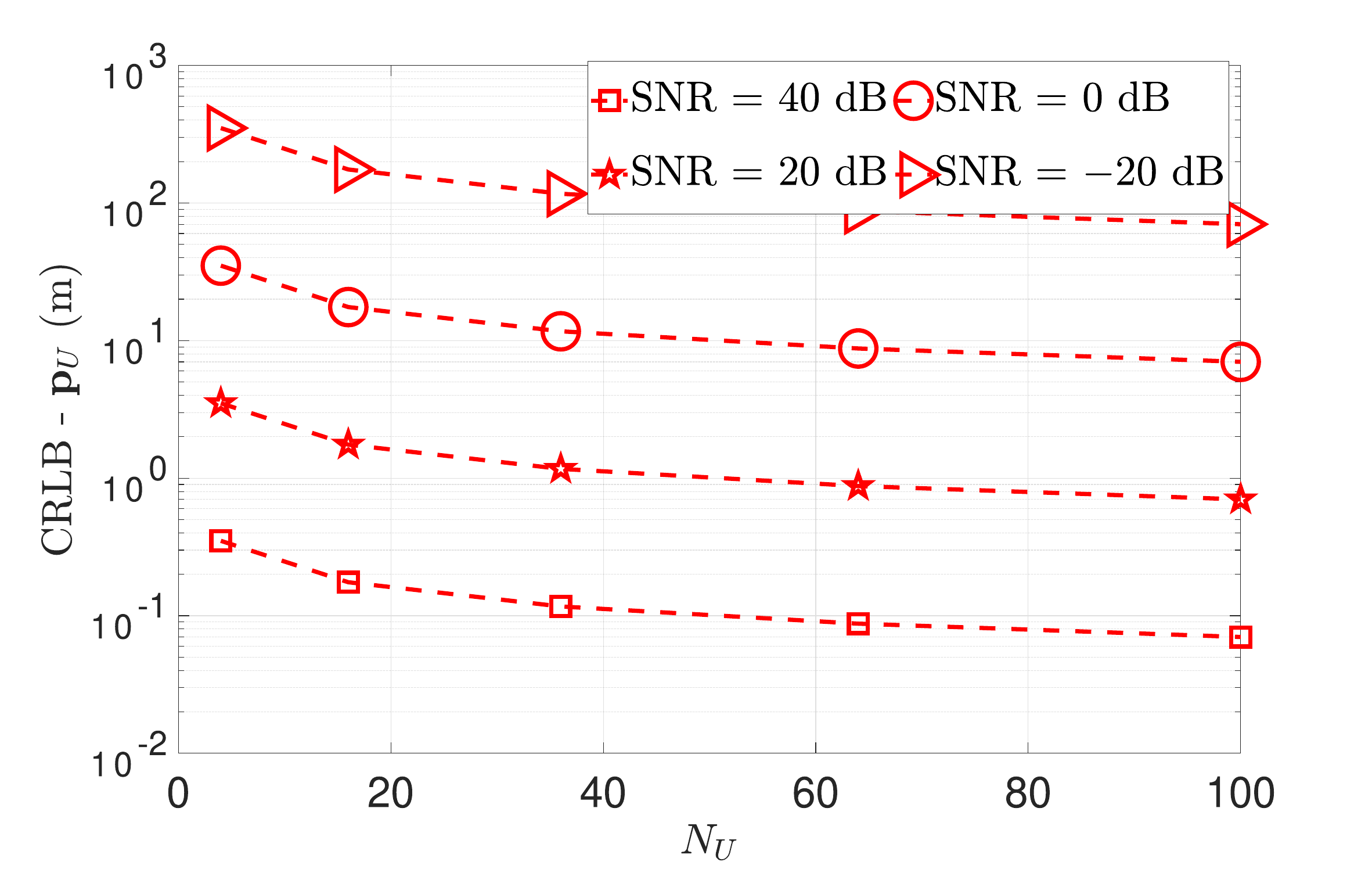}
\label{fig:Results/9D_PU_NU/_NB_3_N_K_3_N_U___delta_t_index_5_fcIndex_4_SNRIndex_1}}
\hfil
\subfloat[]{\includegraphics[ width= 3.2in]{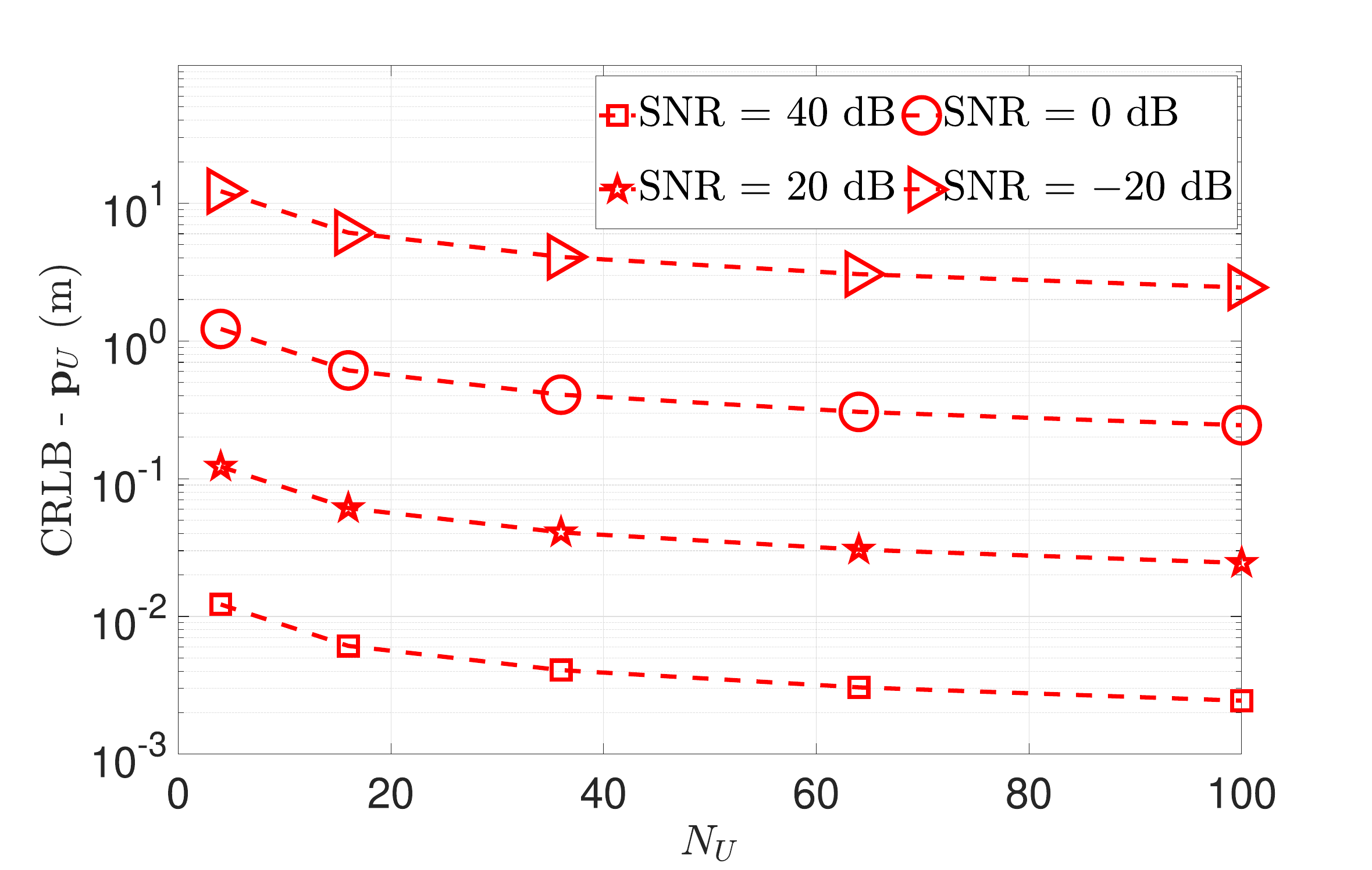}
\label{fig:Results/9D_PU_NU/_NB_3_N_K_3_N_U___delta_t_index_8_fcIndex_4_SNRIndex_1}}
\caption{CRLB for $\bm{p}_{U}$ in the $9$D localization scenario with $f_c = 60 \text{ GHz}$: (a) $\Delta_t = 10 \text{ s}$ and (b) $\Delta_t = 80 \text{ s}$.}
\label{Results:9D_PU_NB_3_N_K_3_N_U___delta_t_index_5_8_fcIndex_4_SNRIndex_1_NU}
\end{figure}

\begin{figure}[htb!]
\centering
\subfloat[]{\includegraphics[ width= 3.2in]{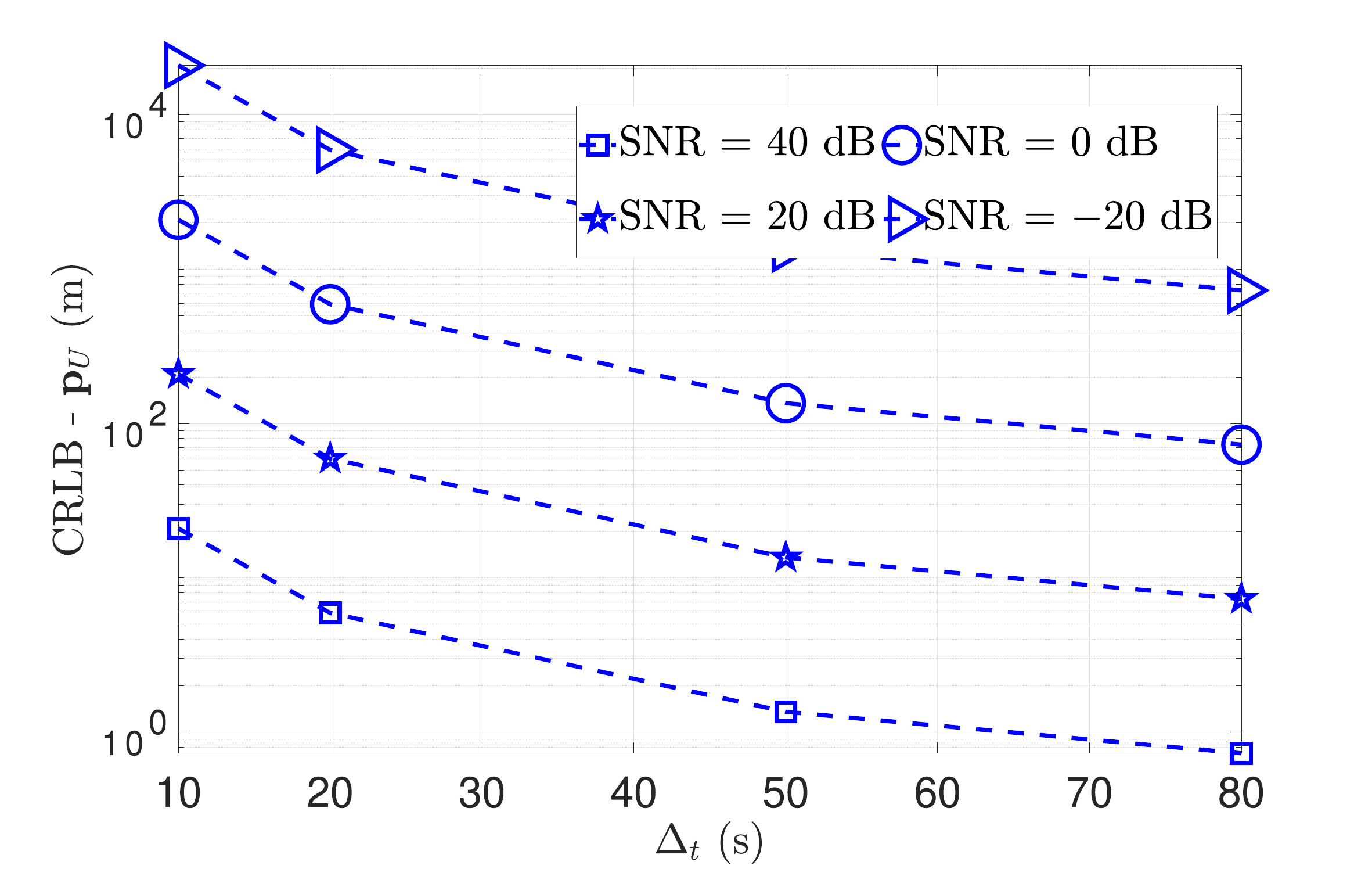}
\label{fig:Results/9D_PU_Time/_NB_3_N_K_3_N_U__delta_t_index___fcIndex_1_SNRIndex__}}
\hfil
\subfloat[]{\includegraphics[ width= 3.2in]{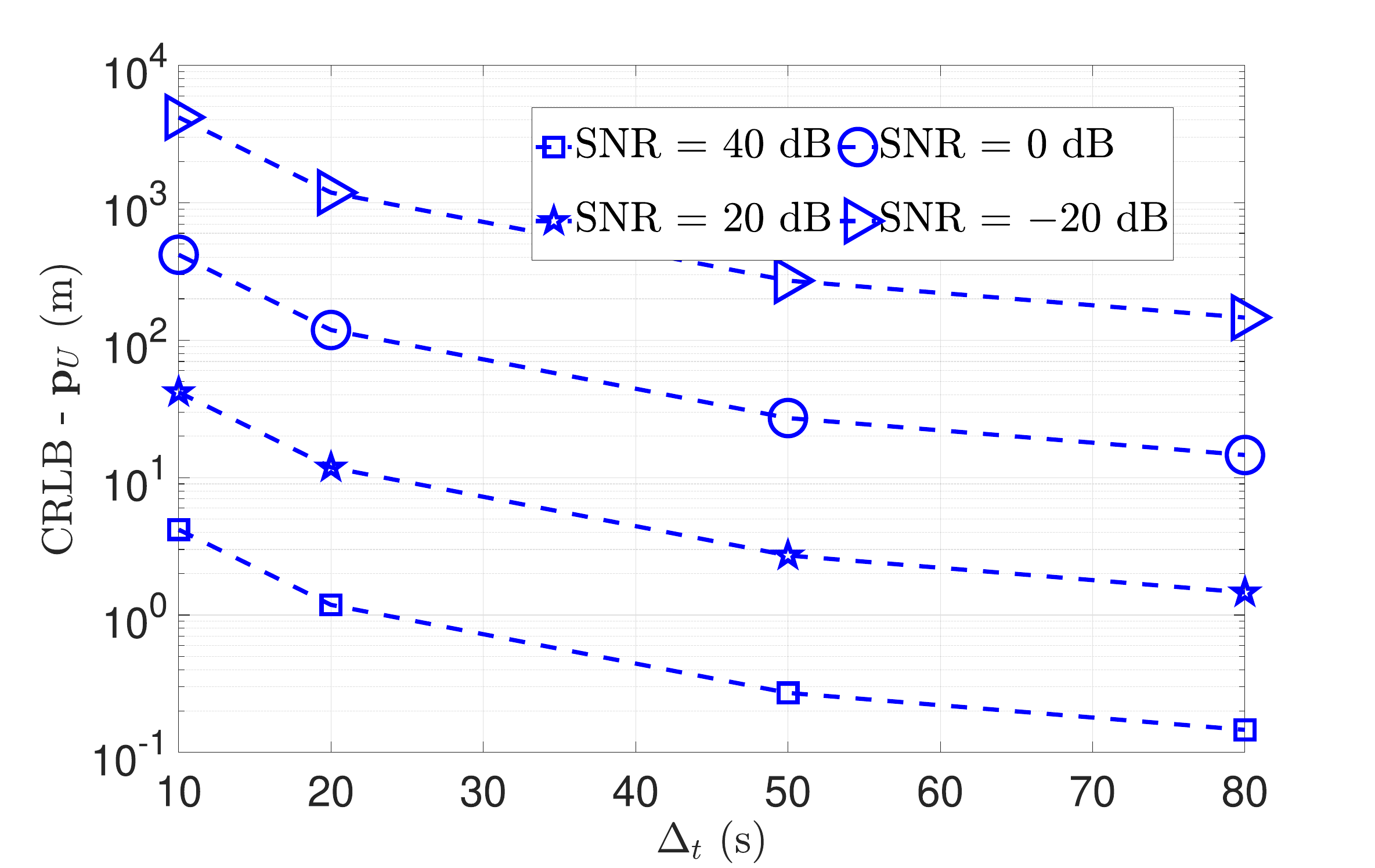}
\label{fig:Results/9D_PU_Time/_NB_3_N_K_3_N_U_10_delta_t_index___fcIndex_1_SNRIndex__}}
\caption{CRLB for $\bm{p}_{U}$ in the $9$D localization scenario with $f_c = 1 \text{ GHz}$: (a) $N_U = 4$  and (b) $N_U = 100$.}
\label{Results:_NB_3_N_K_3_N_U_2_10_delta_t_index___fcIndex_1_SNRIndex_1deltat}
\end{figure}

\begin{figure}[htb!]
\centering
\subfloat[]{\includegraphics[ width= 3.2in]{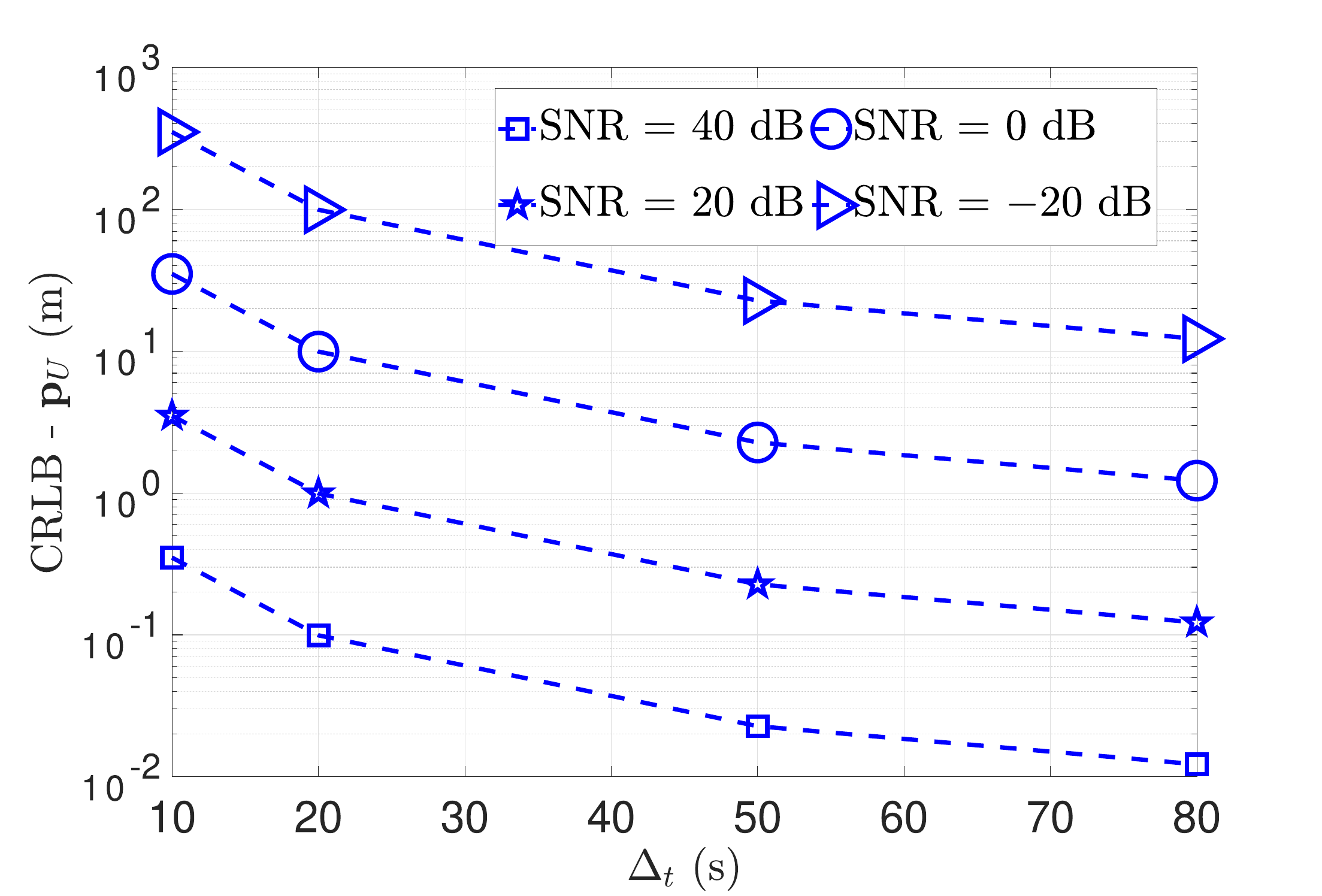}
\label{fig:Results/9D_PU_Time/_NB_3_N_K_3_N_U_2_delta_t_index___fcIndex_4_SNRIndex__}}
\hfil
\subfloat[]{\includegraphics[ width= 3.2in]{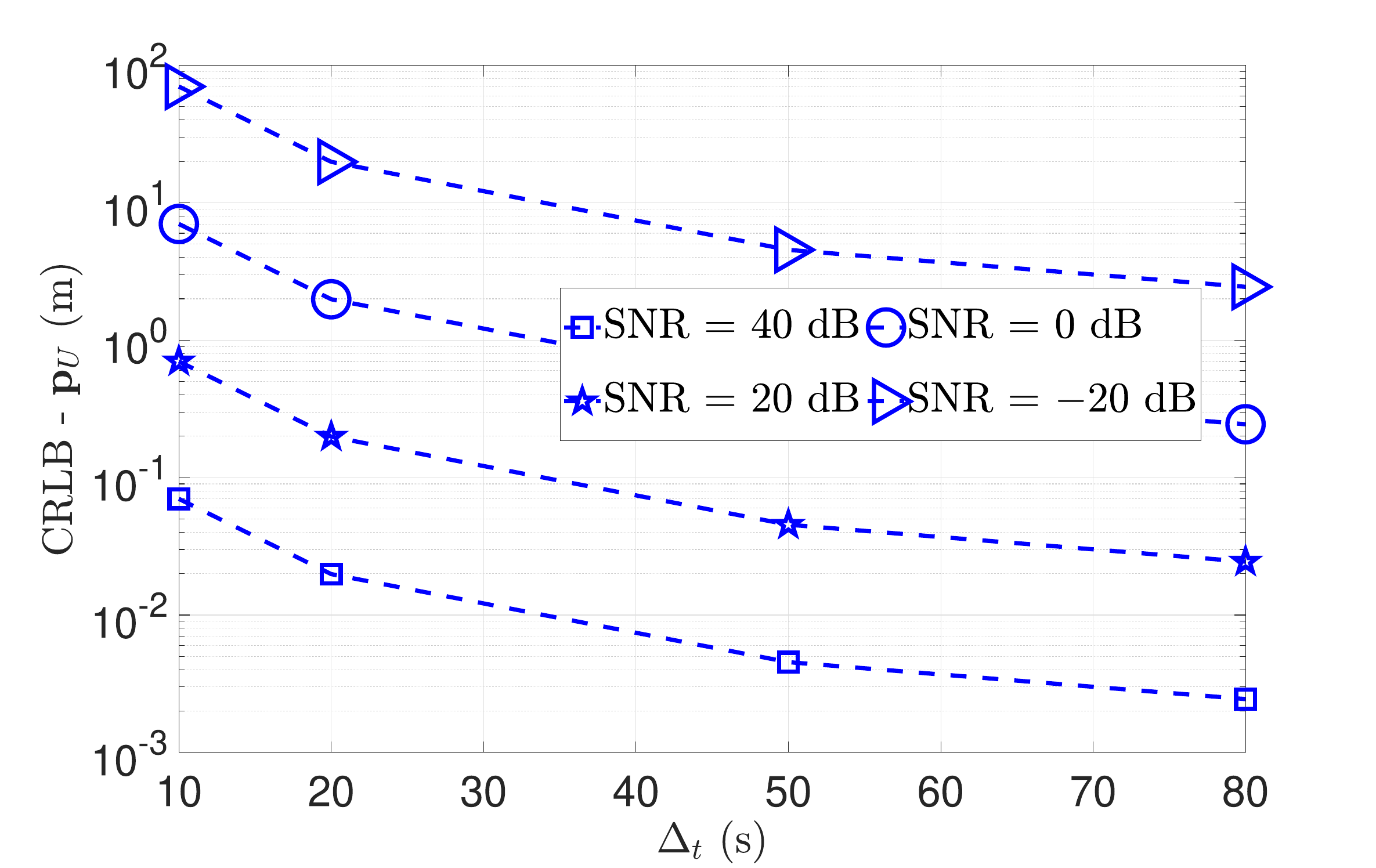}
\label{fig:Results/9D_PU_Time/_NB_3_N_K_3_N_U_10_delta_t_index___fcIndex_5_SNRIndex__}}
\caption{CRLB for $\bm{p}_{U}$ in the $9$D localization scenario with $f_c = 60 \text{ GHz}$: (a) $N_U = 4$  and (b) $N_U = 100$.}

\label{Results:_NB_3_N_K_3_N_U_2_10_delta_t_index___fcIndex_4_SNRIndex_1deltat}
\end{figure}

\begin{figure}[htb!]
\centering
\subfloat[]{\includegraphics[ width= 3.2in]{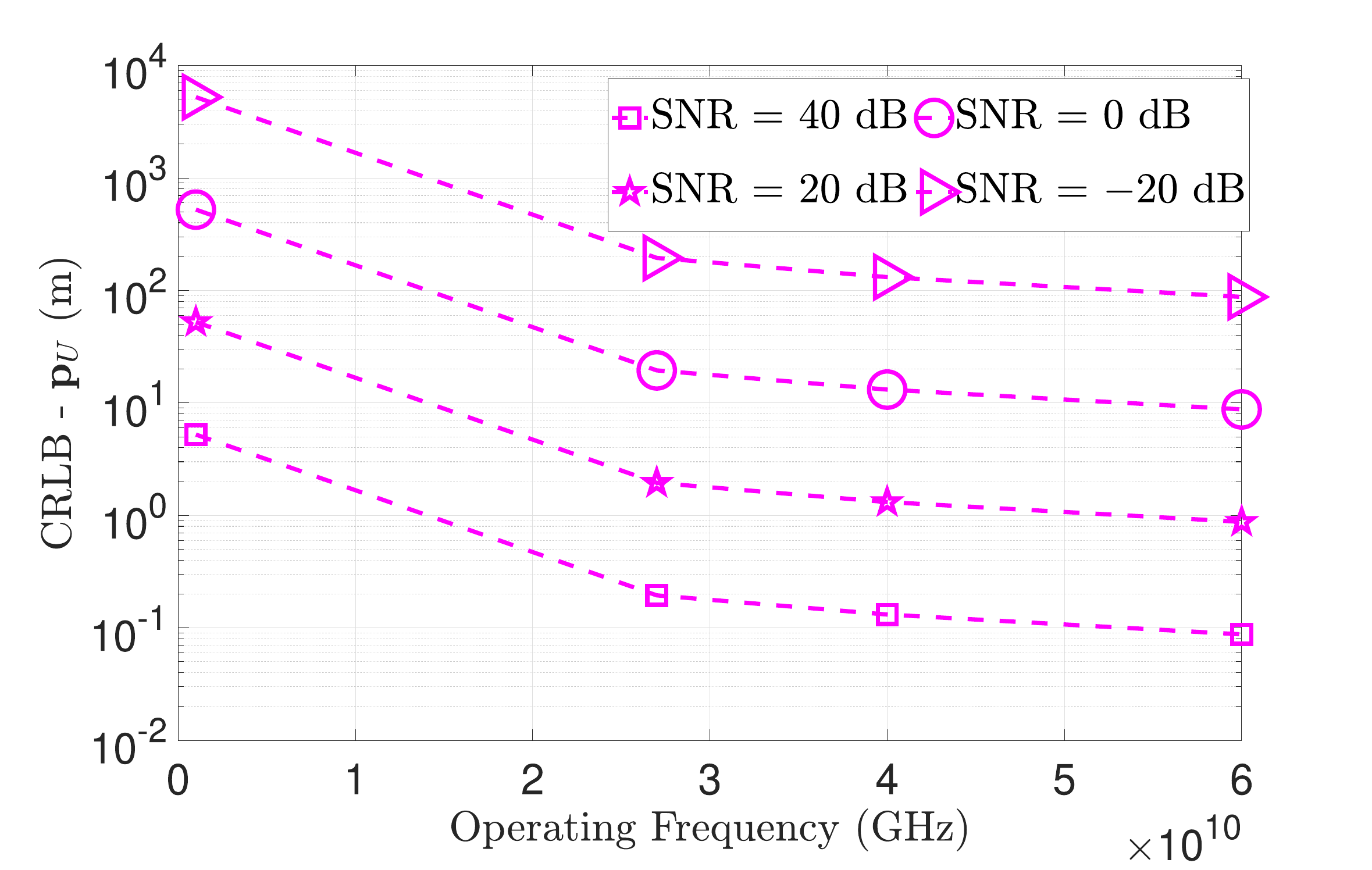}
\label{fig:Results/9D_PU_fc/_NB_3_N_K_3_N_U_2_delta_t_index_1_fcIndex_1_SNRIndex_1}}
\hfil
\subfloat[]{\includegraphics[ width= 3.2in]{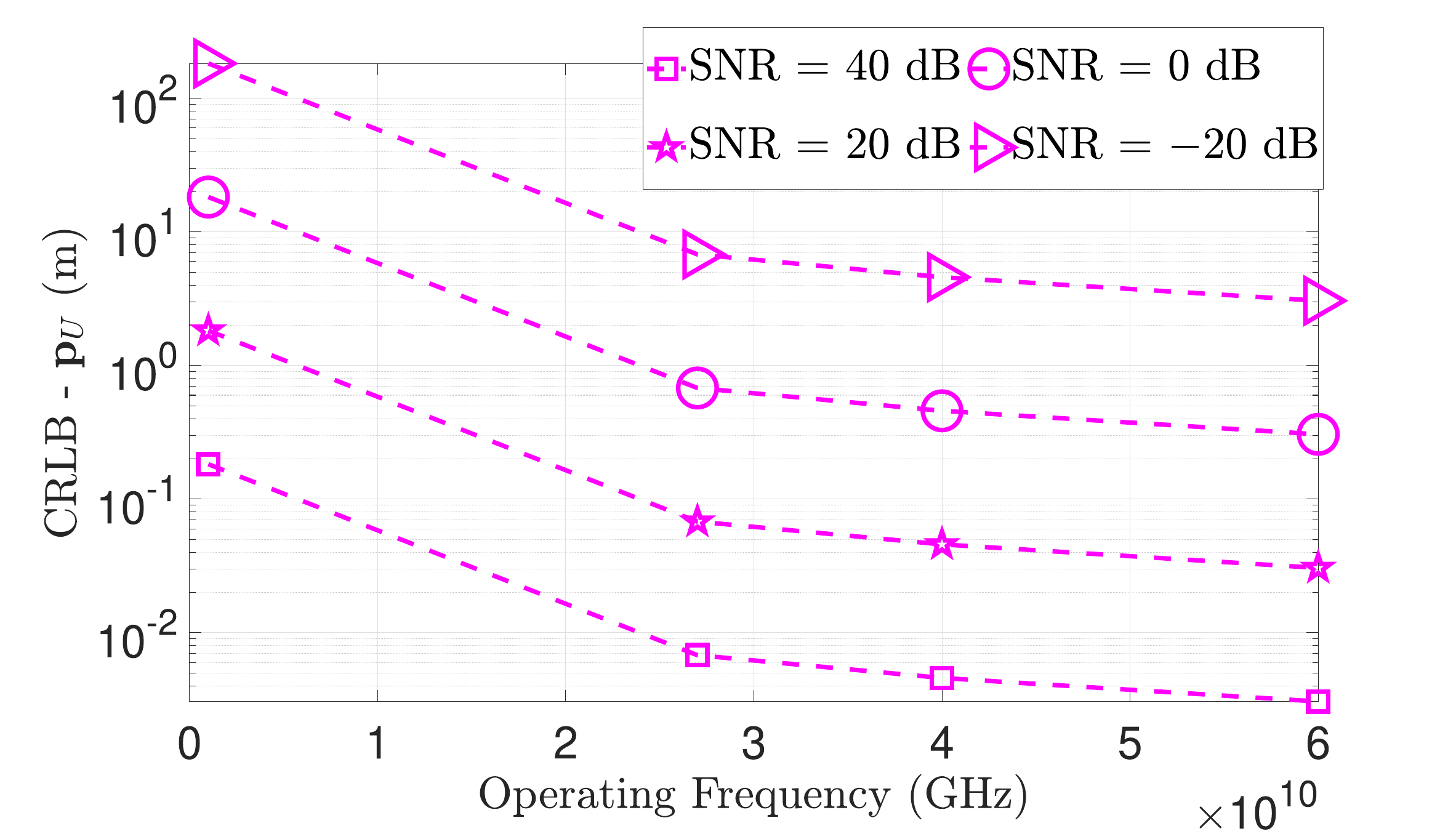}
\label{fig:Results/9D_PU_fc/_NB_3_N_K_3_N_U_2_delta_t_index_3_fcIndex_1_SNRIndex_1}}
\caption{CRLB for $\bm{p}_{U}$ in the $9$D localization scenario with $N_U = 64$: (a) $\Delta_t = 10 \text{ s}$ and (b)  $\Delta_t = 80 \text{ s}$.}
\label{Results:9D_PU_fc_NB_3_N_K_3_N_U_2_delta_t_index_1_3_fcIndex_1_SNRIndex_1_NU}
\end{figure}

\subsubsection{$N_K = 1$, $N_B = 3$, \textit{and} $N_U = 1$}
Under this condition, the information is insufficient to find the $3$D position of the receiver when both a time offset and frequency offset are present. Without a time offset and frequency offset, the $3$D position can be found using: i) three TOA measurements from three distinct LEO satellites, ii) three Doppler measurements from three distinct LEO satellites, iii) two TOA measurements from two distinct LEOs and Doppler measurement from either of the LEO, and iv)  two Doppler measurements from two distinct LEOs and a TOA measurement from either of the LEO satellites.
With only a time offset, ii) can be used to find the $3$D position, while with only a frequency offset, i) can be utilized to find the $3$D position.

\subsubsection{$N_K = 1$, $N_B = 3$, \textit{and} $N_U > 1$}
Under this condition, the information is enough to find the $3$D position of the receiver even when there is a time offset, a frequency offset, or both. 
Without a time offset or frequency offset, the $3$D position can be found using: i) three TOA measurements from three distinct LEO satellites, ii) three Doppler measurements from three distinct LEO satellites, iii) two TOA measurements from two distinct LEOs and a Doppler measurement from either of the LEO, iv)  two Doppler measurements from two distinct LEOs and a TOA measurement from either of the LEO satellites, v) two unit vectors from two distinct LEO satellites obtainable due to the presence of multiple receive antennas, which capture multiple TOAs from a single LEO satellite, vi) one TOA measurement from a distinct LEO and a unit vector from the other LEO satellites, and vii) one Doppler measurement from a distinct LEO and a unit vector from the other LEO. With only a time offset, ii), v), and vii) can be used to find $3$D position, while with only a frequency offset, i), v), and vi) can be used to find $3$D position. When both a time offset and frequency offset are present, v) can be used to find the $3$D position.

Now, we focus on the cases when we only have measurements taken from a single LEO satellite considering single and multiple receive antennas, and different number of time slots. 

\subsubsection{$N_K = 2$, $N_B = 1$, \textit{and} $N_U = 1$}
Under this condition, the presence of either a time or frequency offset means there is insufficient information to find the $3$D position of the receiver. Without a time and frequency offset, there is enough information to find the $3$D position of the receiver using: i) two TOAs from the same LEO satellite obtained during two distinct time slots in combination with the Doppler measurements obtained during either of the time slots with respect to the LEO satellite, and ii) two Doppler measurements obtained during the two distinct time slots with respect to the LEO satellite in combination with the TOA obtained during either of the time slots.

\subsubsection{$N_K = 2$, $N_B = 1$, \textit{and} $N_U > 1$}
Without a time and frequency offset, there is enough information to find the $3$D position of the receiver using: i) two TOAs from the same LEO satellite obtained during two distinct time slots in combination with the Doppler measurements obtained during either of the time slots with respect to the LEO satellite, and ii) two Doppler measurements obtained during the two distinct time slots with respect to the LEO satellite in combination with the TOA obtained during either of the time slots.  iii) two unit vectors obtained during the two distinct time slots from the same LEO satellite obtainable due to the presence of multiple receive antennas, iv) one TOA measurement from the LEO during either of the time slots and a unit vector from the LEO during either of the time slots and v) one Doppler measurement from the LEO during either of the time slots and a unit vector from the other LEO during either of the time slots. With only a time offset, iii) and v) can be used to find the $3$D position.  With only a frequency offset, iii) and iv) can be used to find the $3$D position. When both a time offset and frequency offset
are present, iii) can be used to find the $3$D position.

\subsubsection{$N_K = 3$, $N_B = 1$, \textit{and} $N_U = 1$}
Under this condition, the presence of either a time or frequency offset means there is insufficient information to find the $3$D position of the receiver. Without a time and frequency offset, there is enough information to find the $3$D position of the receiver using: i) three TOA measurements obtained during three different time slots from a single LEO satellite, ii) three Doppler measurements obtained during three different time slots from a single LEO satellite, iii) two TOAs from the same LEO satellite obtained during two distinct time slots in combination with the Doppler measurements obtained during either of the time slots with respect to the LEO satellite, and iv) two Doppler measurements obtained during the two distinct time slots with respect to the LEO satellite in combination with the TOA obtained during either of the time slots. While with a time offset, ii) can be used to provide the $3$D position of the receiver, with a frequency offset, i) can be used to provide the $3$D position of the receiver.

\subsubsection{$N_K = 3$, $N_B = 1$, \textit{and} $N_U > 1$}
Without a time and frequency offset, there is enough information to find the $3$D position of the receiver using: i) three TOA measurements obtained during three different time slots from a single LEO satellite, ii) three Doppler measurements obtained during three different time slots from a single LEO satellite, iii) two TOAs from the same LEO satellite obtained during two distinct time slots in combination with the Doppler measurements obtained during either of the time slots with respect to the LEO satellite, and iv) two Doppler measurements obtained during the two distinct time slots with respect to the LEO satellite in combination with the TOA obtained during either of the time slots.  v) two unit vectors obtained during the two distinct time slots from the same LEO satellite obtainable due to the presence of multiple receive antennas, vi) one TOA measurement from the LEO during either of the time slots and a unit vector from the LEO during either of the time slots, and vii) one Doppler measurement from the LEO during either of the time slots and a unit vector from the other LEO during either of the time slots. With only a time offset, ii), v), and vii) can be used to find the $3$D position.  With only a frequency offset, i), v), and vi)  can be used to find the $3$D position. When both a time offset and frequency offset are present, v) can be used to find the $3$D position.

We next present cases for $N_K > 3$. These cases are unique and special because they allow time and frequency difference techniques to handle scenarios with both time and frequency offsets and only a single receive antenna is available.

\subsubsection{$N_K = 4$, $N_B = 1$, \textit{and} $N_U = 1$}
Without a time and frequency offset, there is enough information to find the $3$D position of the receiver using: i) four TOA measurements obtained during four different time slots from a single LEO satellite with one of the TOA measurements serving as a reference measurement for time differencing, ii) four Doppler measurements obtained during four different time slots from a single LEO satellite with one of the Doppler measurements serving as a reference measurement for frequency differencing, iii) three TOA measurements obtained during three different time slots from a single LEO satellite, iv) three Doppler measurements obtained during three different time slots from a single LEO satellite, v) two TOAs from the same LEO satellite obtained during two distinct time slots in combination with the Doppler measurements obtained during either of the time slots with respect to the LEO satellite, and vi) two Doppler measurements obtained during the two distinct time slots with respect to the LEO satellite in combination with the TOA obtained during either of the time slots. While with a time offset, i), ii), and iv) can be used to provide the $3$D position of the receiver, with a frequency offset, i), ii), and iii) can be used to provide the $3$D position of the receiver. With both a time and frequency offset, i) and ii) can be used to provide the $3$D position of the receiver.

\subsubsection{$N_K = 4$, $N_B = 1$, \textit{and} $N_U > 1$}
Without a time and frequency offset, there is enough information to find the $3$D position of the receiver using: i) four TOA measurements obtained during four different time slots from a single LEO satellite with one of the TOA measurements serving as a reference measurement for time differencing, ii) four Doppler measurements obtained during four different time slots from a single LEO satellite with one of the Doppler measurements serving as a reference measurement for frequency differencing, iii) three TOA measurements obtained during three different time slots from a single LEO satellite, iv) three Doppler measurements obtained during three different time slots from a single LEO satellite, v) two TOAs from the same LEO satellite obtained during two distinct time slots in combination with the Doppler measurements obtained during either of the time slots with respect to the LEO satellite, and vi) two Doppler measurements obtained during the two distinct time slots with respect to the LEO satellite in combination with the TOA obtained during either of the time slots.  vii) two unit vectors obtained during the two distinct time slots from the same LEO satellite obtainable due to the presence of multiple receive antennas, viii) one TOA measurement from the LEO during either of the time slots and a unit vector from the LEO during either of the time slots and ix) one Doppler measurement from the LEO during either of the time slots and a unit vector from the LEO during either of the time slots. 
With only a time offset, i) ii), iv), and v) can be used to find the $3$D position.  With only a frequency offset, i), ii), (iii), v), (vii), and viii)  can be used to find the $3$D position. When both a time offset and frequency offset are present, i), ii), and v) can be used to find the $3$D position.

\begin{table*}[!t]
\renewcommand{\arraystretch}{1.3}
\caption{Legend for Estimation Possibilities for $\bm{\Phi}_{U}$ in the $3$D Localization Setup under Different Combinations of $N_B$, $N_K$, and $N_U$}
\label{Legend_Estimation_Possibilities_for_PHIU_in_the_3D_Localization}
\centering
\begin{tabular}{|c|c|}
\hline
a & Multiple TOA measurements received across the receive antennas from two LEO satellites.\\
\hline
b & Multiple TOA measurements received across the receive antennas from a single LEO satellite during both time slots.\\
\hline
\end{tabular}
\end{table*}

\begin{table*}[!t]
\renewcommand{\arraystretch}{1.3}
\caption{Legend for Estimation Possibilities for $\bm{v}_{U}$ in the $3$D Localization Setup under Different Combinations of $N_B$, $N_K$, and $N_U$}
\label{Legend_Estimation_Possibilities_for_vU_in_the_3D_Localization}
\centering
\begin{tabular}{|c|c|}
\hline
a & Three Doppler measurements from three LEO satellites.\\
\hline
b & A combination of the frequency differencing of two Doppler measurements from two distinct \\ & time slots from the first LEO satellite, frequency differencing of the other two Doppler measurements from\\ & two distinct time slots from the other LEO satellite, and time differencing of any \\ &  two delay measurements from any two distinct time slots from any of the LEO satellites.\\
\hline
c & A combination of the time differencing of two delay measurements from two distinct \\ & time slots from the first LEO satellite,  time differencing of two delay measurements from \\ & two distinct time slots from the second LEO satellite, and frequency differencing of two \\ & Doppler measurements from two distinct time slots from any of the two satellites.\\
\hline
d &  A combination of the frequency differencing of two Doppler measurements from two distinct \\ & time slots from the first LEO satellite,  frequency differencing of two Doppler measurements from \\ & two distinct time slots from the second LEO satellite, and two \\ & delay measurements from two distinct time slots from either of the two satellites. \\
\hline
e &  A combination of the time differencing of two delay measurements from two distinct \\ & time slots from the first  LEO satellite,   time differencing of two delay measurements from \\ & two distinct time slots from the second LEO satellite,  and one of the  Doppler measurement from \\ &either of the distinct time slots from either of the LEO satellites. \\
\hline
f & A combination of two delay  measurements from the two distinct  time slots from either of the \\ & first satellite,  two delay  measurements from  the two distinct time slots from \\ & the second satellite,  and one of the four Doppler measurements from \\ & either of the distinct time slots from either of the LEO satellites. \\
\hline
g & Three of the four Doppler measurements from both of the distinct time slots from both of the LEO satellites.\\
\hline
\end{tabular}
\end{table*}

\begin{table*}[!t]
\renewcommand{\arraystretch}{1.3}
\caption{Legend for Estimation Possibilities for $\bm{v}_{U}$ in the $3$D Localization Setup under Different Combinations of $N_B$, $N_K$, and $N_U$}
\label{Legend_Estimation_Possibilities_for_vU_in_the_3D_Localization_1}
\centering
\begin{tabular}{|c|c|}
\hline
h & Four TOA measurements obtained during four different time slots from a single LEO satellite \\ & with one of the TOA measurements serving as a reference measurement for time differencing. \\
\hline
i & Four Doppler measurements obtained during four different time slots from a single LEO satellite \\ & with one of the Doppler measurements serving as a reference measurement for frequency differencing. \\
\hline
j & A combination of the frequency differencing of two Doppler measurements from two distinct time slots from \\ & the first LEO satellite,  frequency differencing of \\ &  the other two Doppler measurements from two distinct time slots from the LEO satellite, \\ & and time differencing of any two delay measurements from any two distinct time slots from the LEO satellite. \\
\hline
k & A combination of the time differencing of two delay measurements  from two distinct time slots \\ & from the LEO satellite,  time differencing of  the other two delay measurements \\ & from two distinct time slots from the LEO satellite,  and frequency differencing \\ & of any of the two Doppler measurements from two distinct time slots from the LEO satellite. \\
\hline
l & A combination of the frequency differencing of two Doppler measurements  \\ & from two distinct time slots from the LEO satellite, \\ &  frequency differencing of the other two Doppler measurements from two distinct time slots from the LEO satellite,  \\ & and two delay measurements from two distinct time slots from the LEO satellite. \\ 
\hline
m & A combination of the time differencing of two delay measurements from two distinct time slots from the first \\ & LEO satellite,  time differencing of the other two delay measurements from two distinct time slots from the \\ &  LEO satellite, and one of the Doppler measurement from one of the distinct time slots from the LEO satellite. \\ 
\hline
n & A combination of two delay measurements from the two distinct time slots from the LEO satellite, \\ & two delay measurements from the two distinct time slots from the LEO satellite,  \\ & and one of the three Doppler measurements from the distinct time slots from the LEO satellite. \\ 
\hline 
o & Three Doppler measurements from the distinct time slots from the LEO satellite. \\
\hline
\end{tabular}
\end{table*}

\begin{figure}[htb!]
\centering
\subfloat[]{\includegraphics[ width= 3.2in]{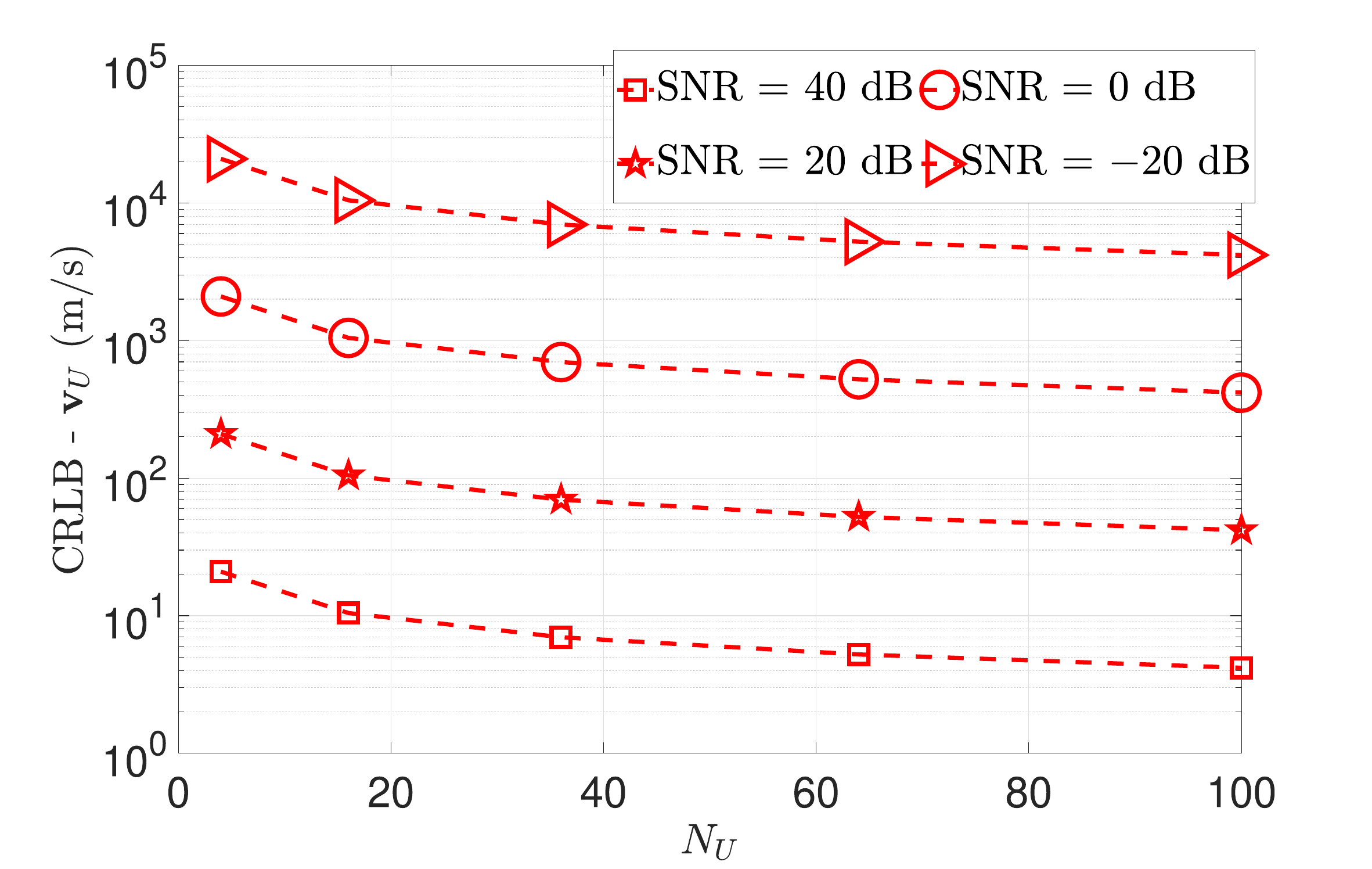}
\label{fig:Results/9D_VU_NU/_NB_3_N_K_3_N_U___delta_t_index_5_fcIndex_1_SNRIndex__}}
\hfil
\subfloat[]{\includegraphics[ width= 3.2in]{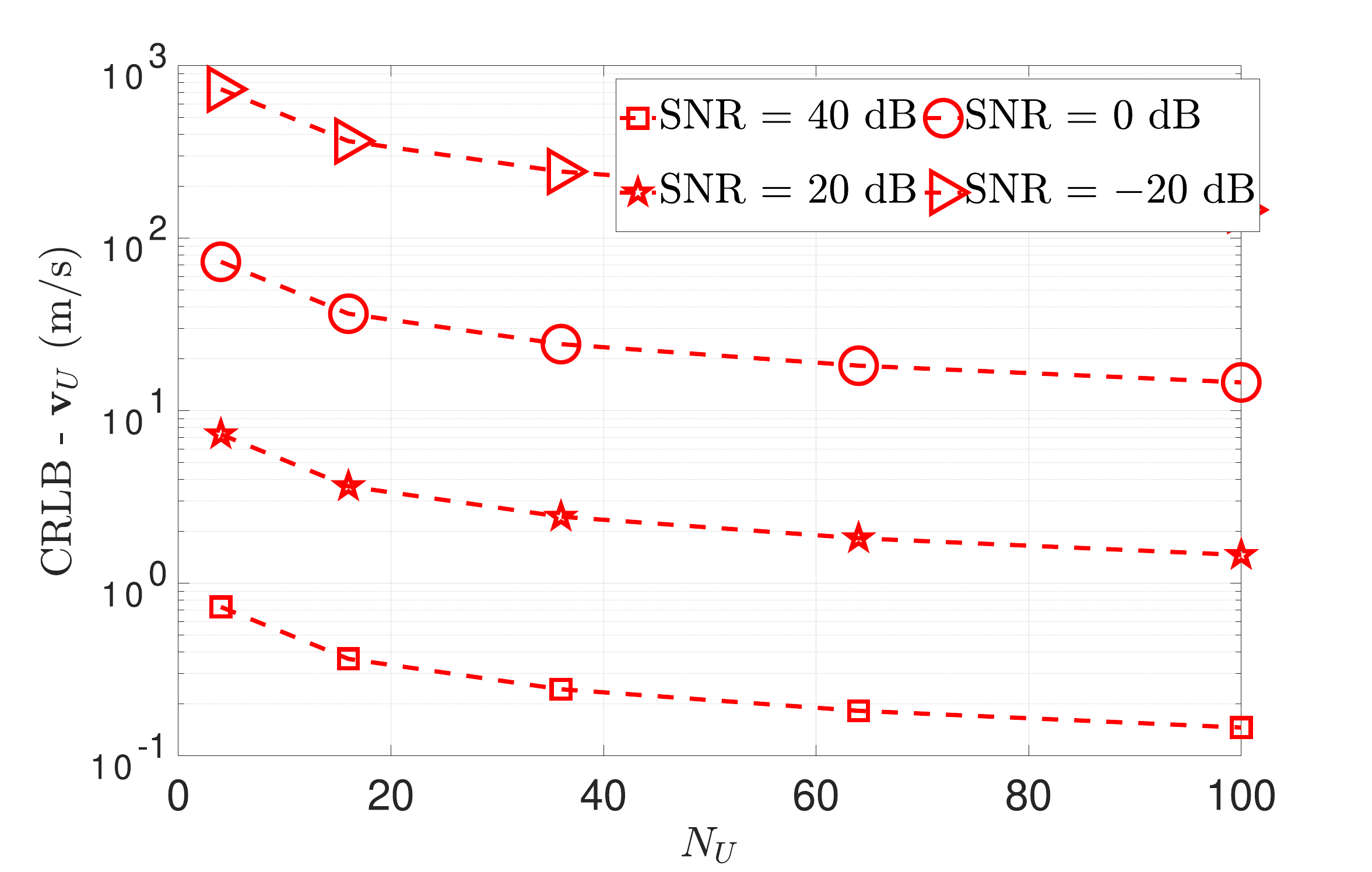}
\label{fig:Results/9D_VU_NU/_NB_3_N_K_3_N_U___delta_t_index_8_fcIndex_1_SNRIndex__}}
\caption{CRLB for $\bm{v}_{U}$ in the $9$D localization scenario with $f_c = 1 \text{ GHz}$: (a) $\Delta_t = 10 \text{ s}$ and (b) $\Delta_t = 80 \text{ s}$.}
\label{fig:Results/9D_VU_NU/_NB_3_N_K_3_N_U___delta_t_index_5_8_fcIndex_1_SNRIndex__}
\end{figure}

\begin{figure}[htb!]
\centering
\subfloat[]{\includegraphics[ width= 3.2in]{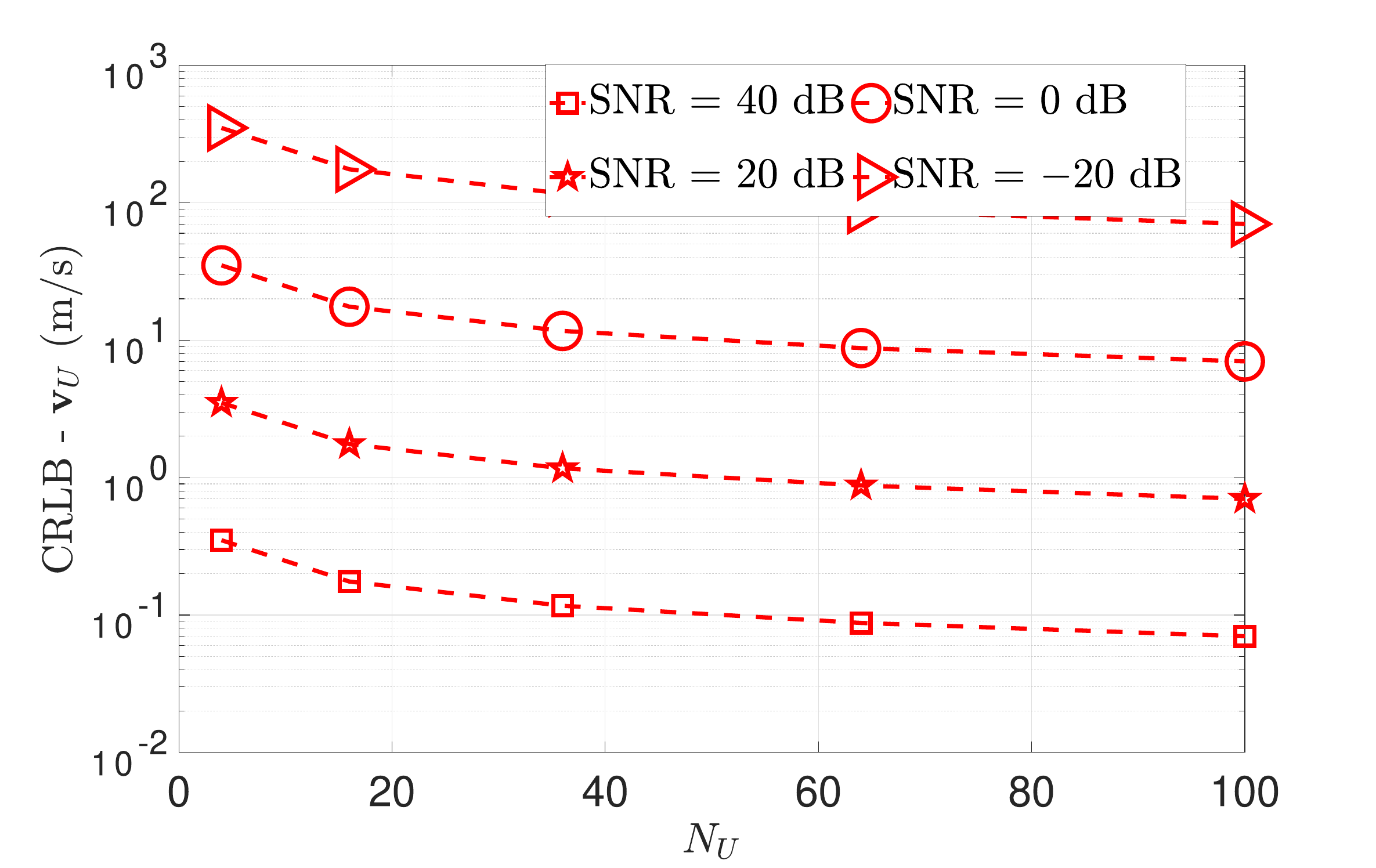}
\label{fig:Results/9D_VU_NU/_NB_3_N_K_3_N_U___delta_t_index_5_fcIndex_4_SNRIndex__}}
\hfil
\subfloat[]{\includegraphics[ width= 3.2in]{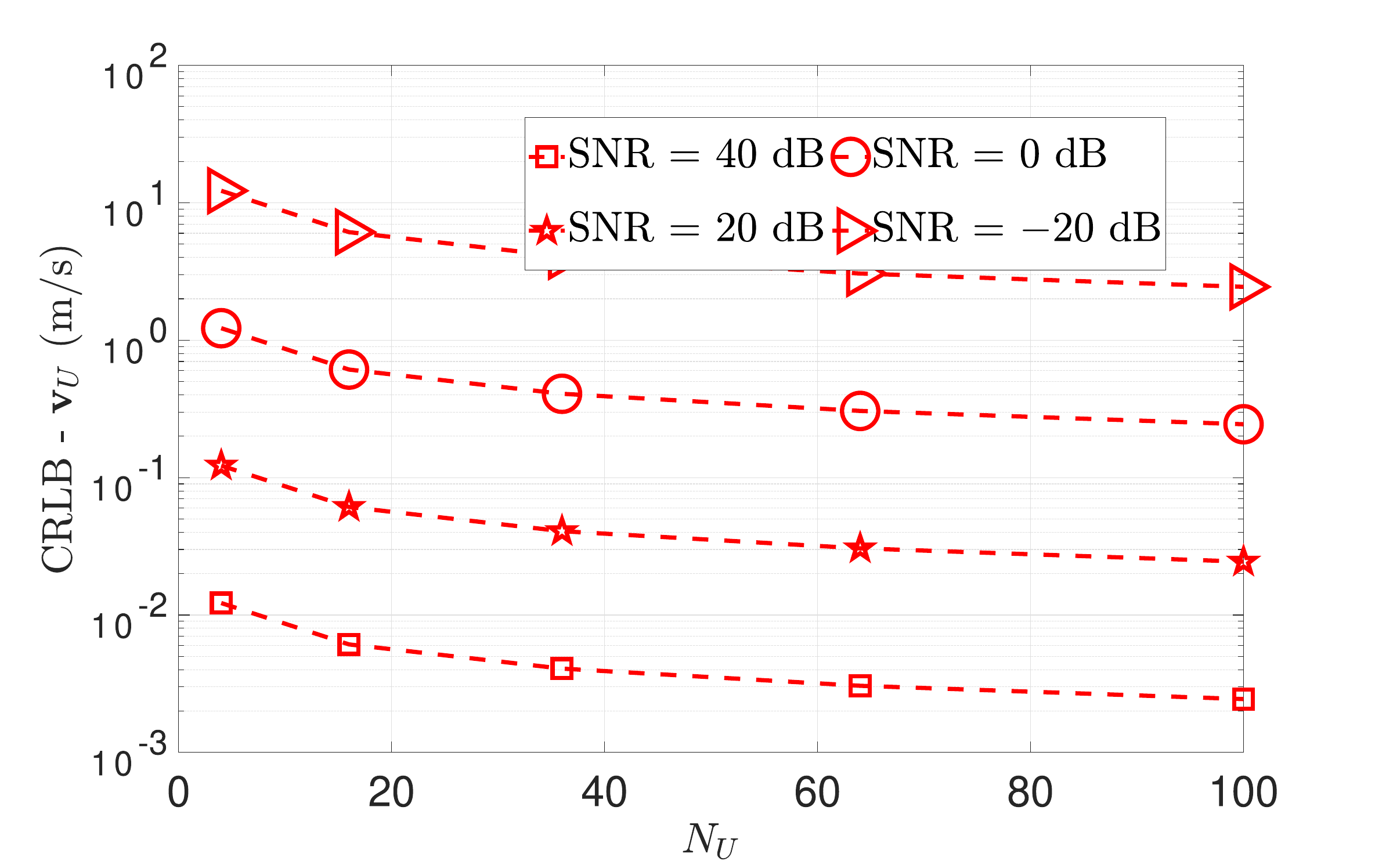}
\label{fig:Results/9D_VU_NU/_NB_3_N_K_3_N_U___delta_t_index_8_fcIndex_4_SNRIndex__}}
\caption{CRLB for $\bm{v}_{U}$ in the $9$D localization scenario with $f_c = 60 \text{ GHz}$: (a) $\Delta_t = 10 \text{ s}$ and (b) $\Delta_t = 80 \text{ s}$.}
\label{fig:Results/9D_VU_NU/_NB_3_N_K_3_N_U___delta_t_index_5_8_fcIndex_4_SNRIndex__}
\end{figure}

\begin{figure}[htb!]
\centering
\subfloat[]{\includegraphics[ width= 3.2in]{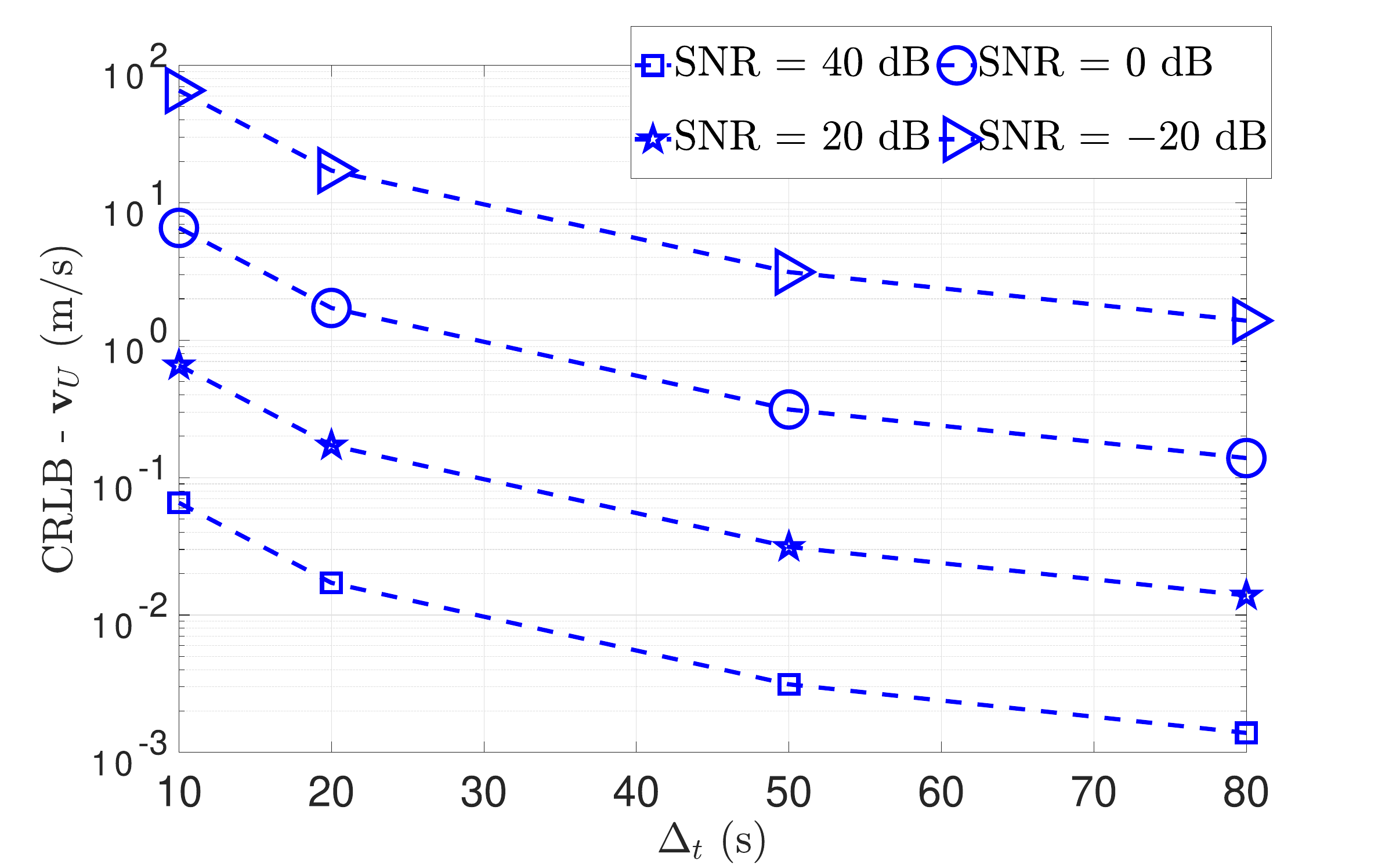}
\label{fig:Results/9D_VU_Time/_NB_3_N_K_3_N_U_2_delta_t_index___fcIndex_1_SNRIndex__}}
\hfil
\subfloat[]{\includegraphics[ width= 3.2in]{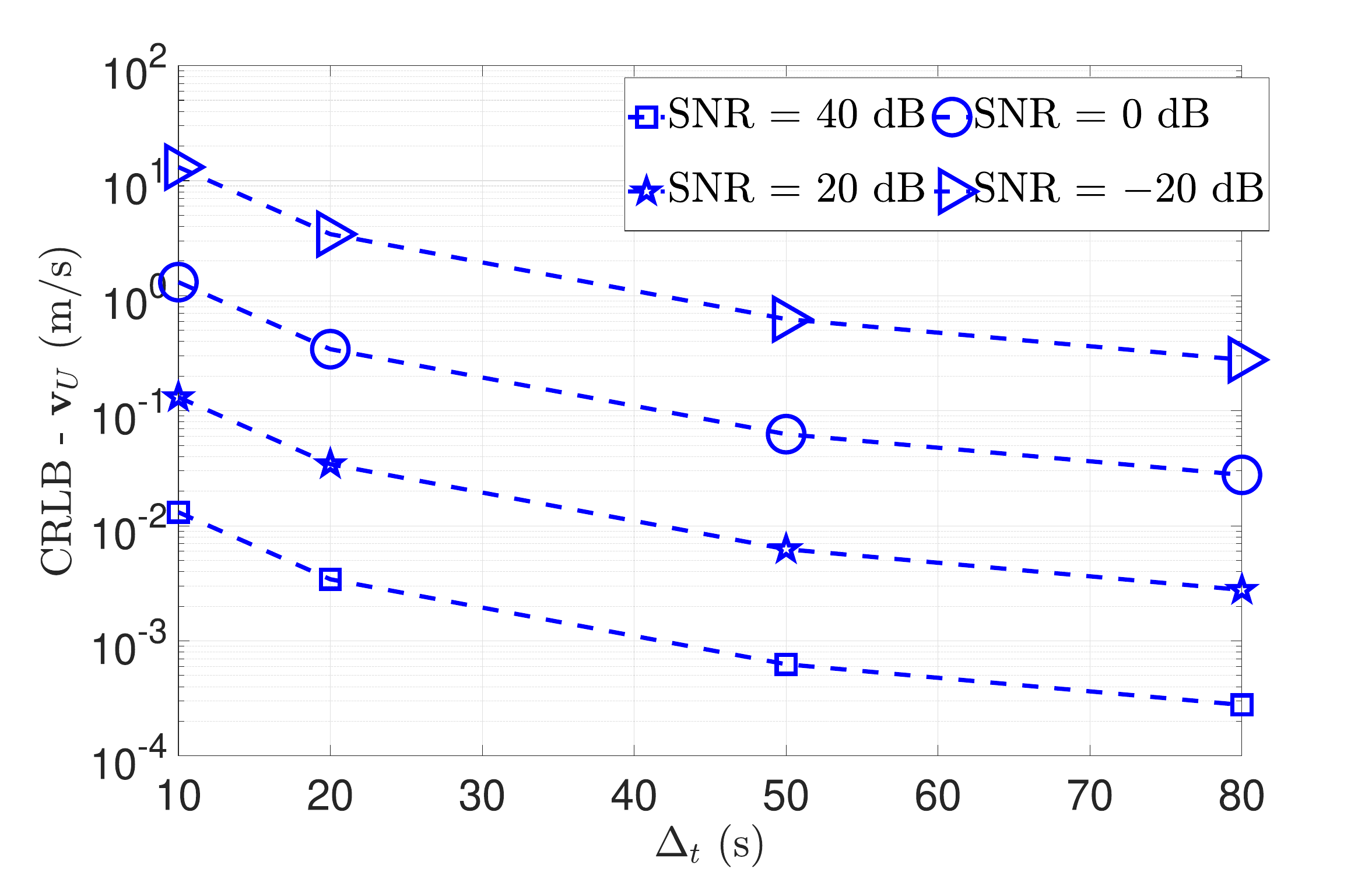}
\label{fig:Results/9D_VU_Time/_NB_3_N_K_3_N_U_10_delta_t_index___fcIndex_1_SNRIndex__}}
\caption{CRLB for $\bm{v}_{U}$ in the $9$D localization scenario with $f_c = 1 \text{ GHz}$: (a) $N_U = 4$  and (b) $N_U = 100$.}
\label{Results/9D_VU_Time/_NB_3_N_K_3_N_U_2_10_delta_t_index___fcIndex_1_SNRIndex__}
\end{figure}

\begin{figure}[htb!]
\centering
\subfloat[]{\includegraphics[ width= 3.2in]{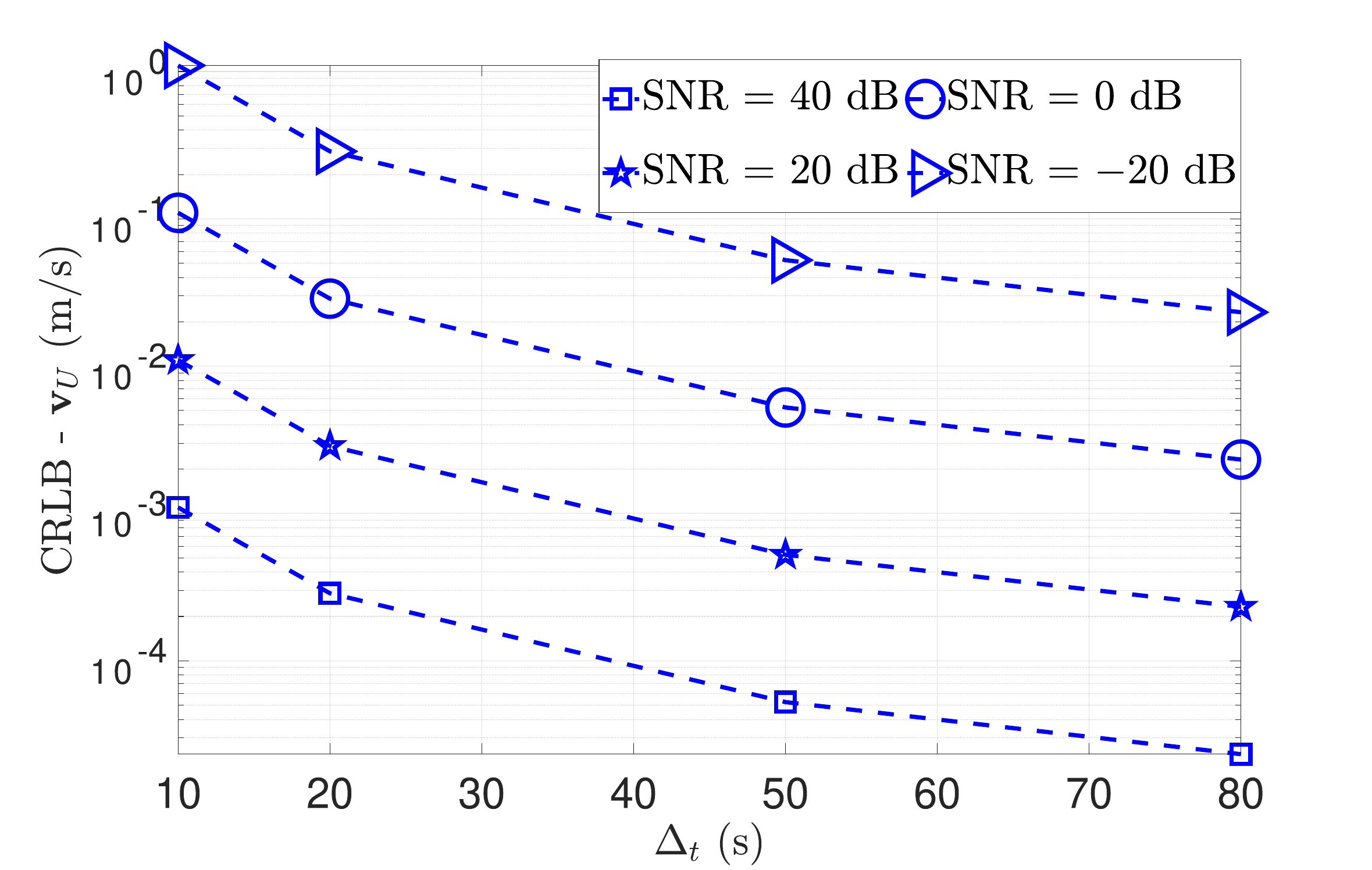}
\label{fig:Results/9D_VU_Time/_NB_3_N_K_3_N_U_2_delta_t_index___fcIndex_4_SNRIndex__}}
\hfil
\subfloat[]{\includegraphics[ width= 3.2in]{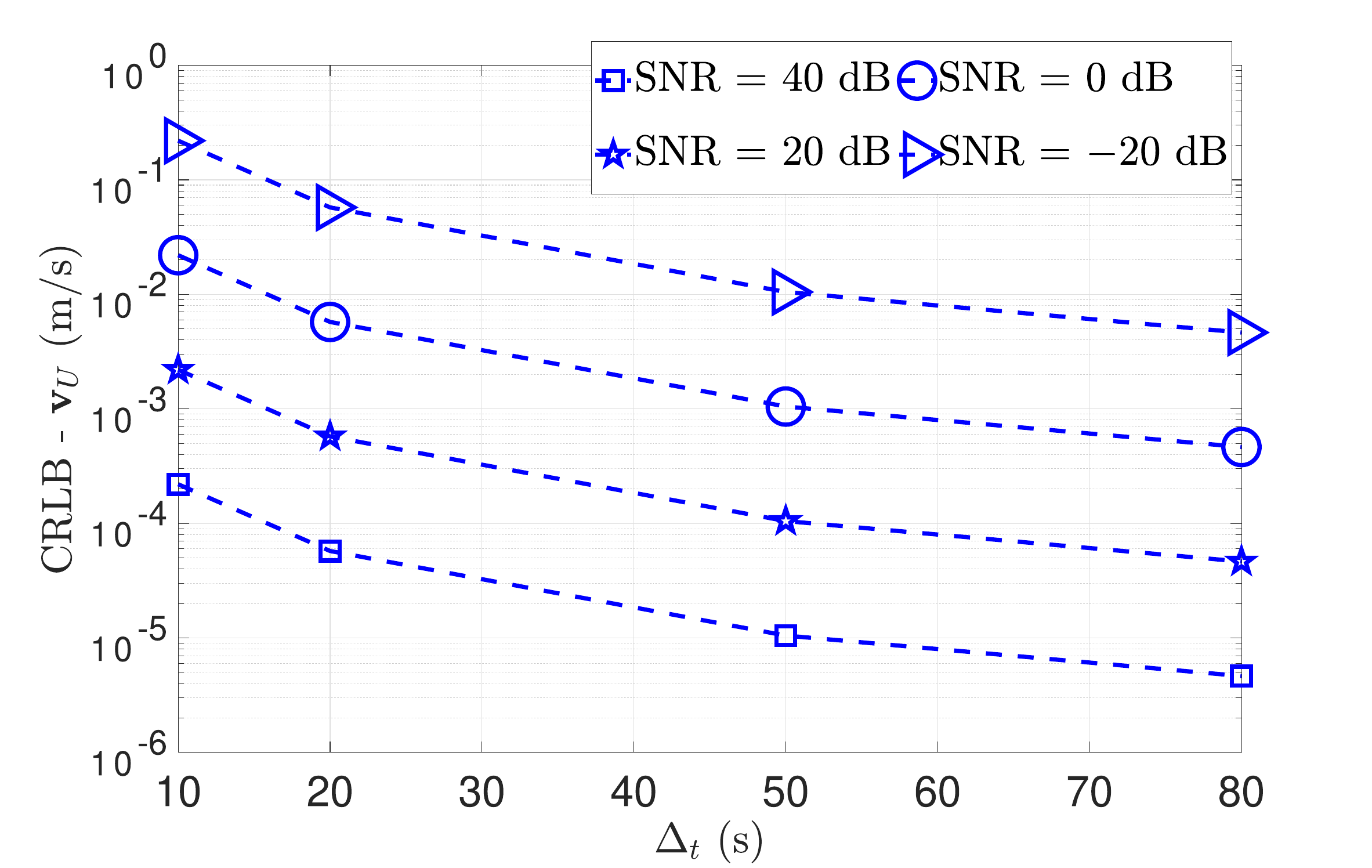}
\label{fig:Results/9D_VU_Time/_NB_3_N_K_3_N_U_10_delta_t_index___fcIndex_4_SNRIndex__}}
\caption{CRLB for $\bm{v}_{U}$ in the $9$D localization scenario with $f_c = 60 \text{ GHz}$: (a) $N_U = 4$  and (b) $N_U = 100$.}
\label{Results/9D_VU_Time/_NB_3_N_K_3_N_U_2_10_delta_t_index___fcIndex_4_SNRIndex__}
\end{figure}

\begin{figure}[htb!]
\centering
\subfloat[]{\includegraphics[ width= 3.2in]{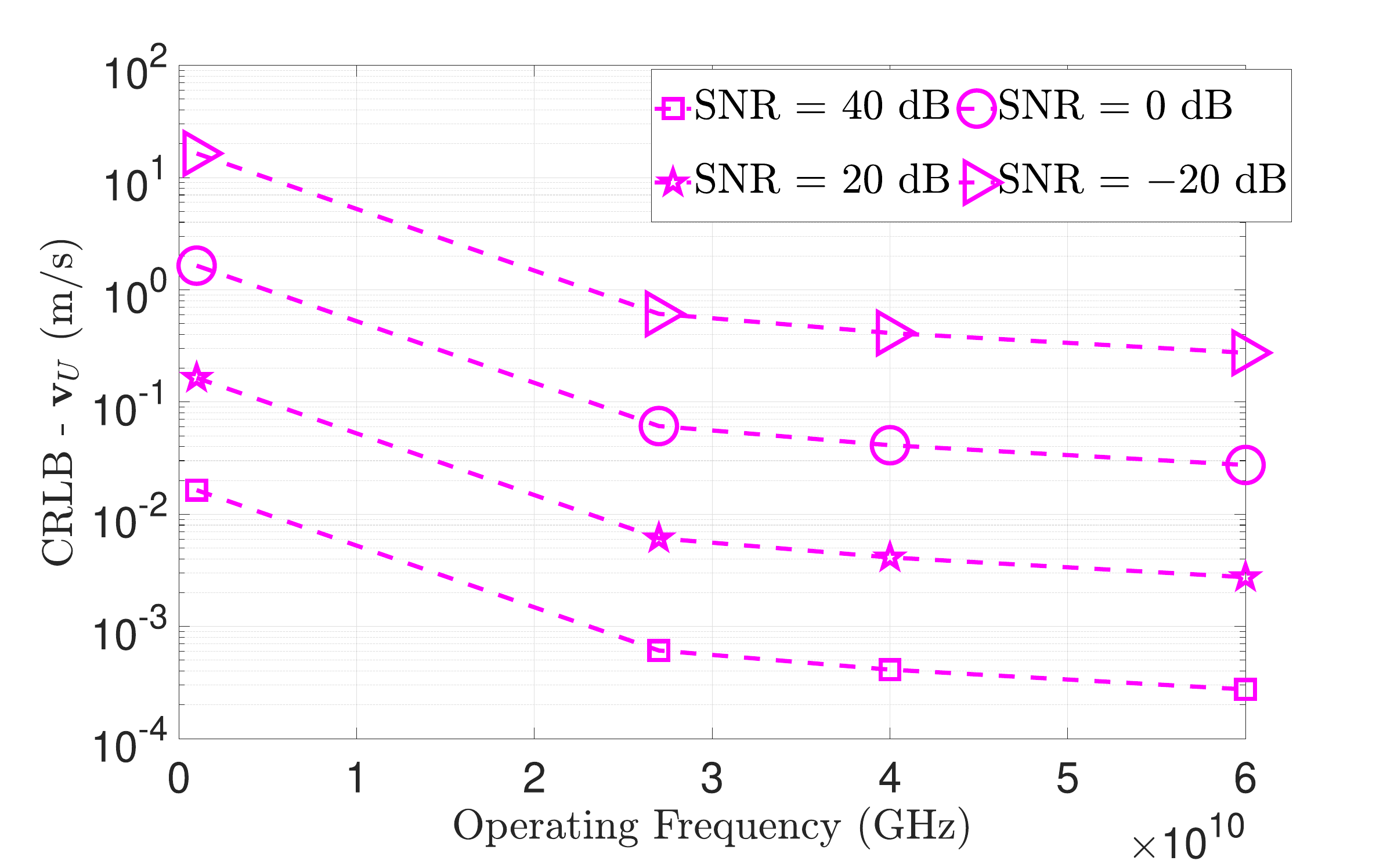}
\label{fig:Results/9D_VU_fc/_NB_3_N_K_3_N_U_8_delta_t_index_5_fcIndex___SNRIndex__}}
\hfil
\subfloat[]{\includegraphics[ width= 3.2in]{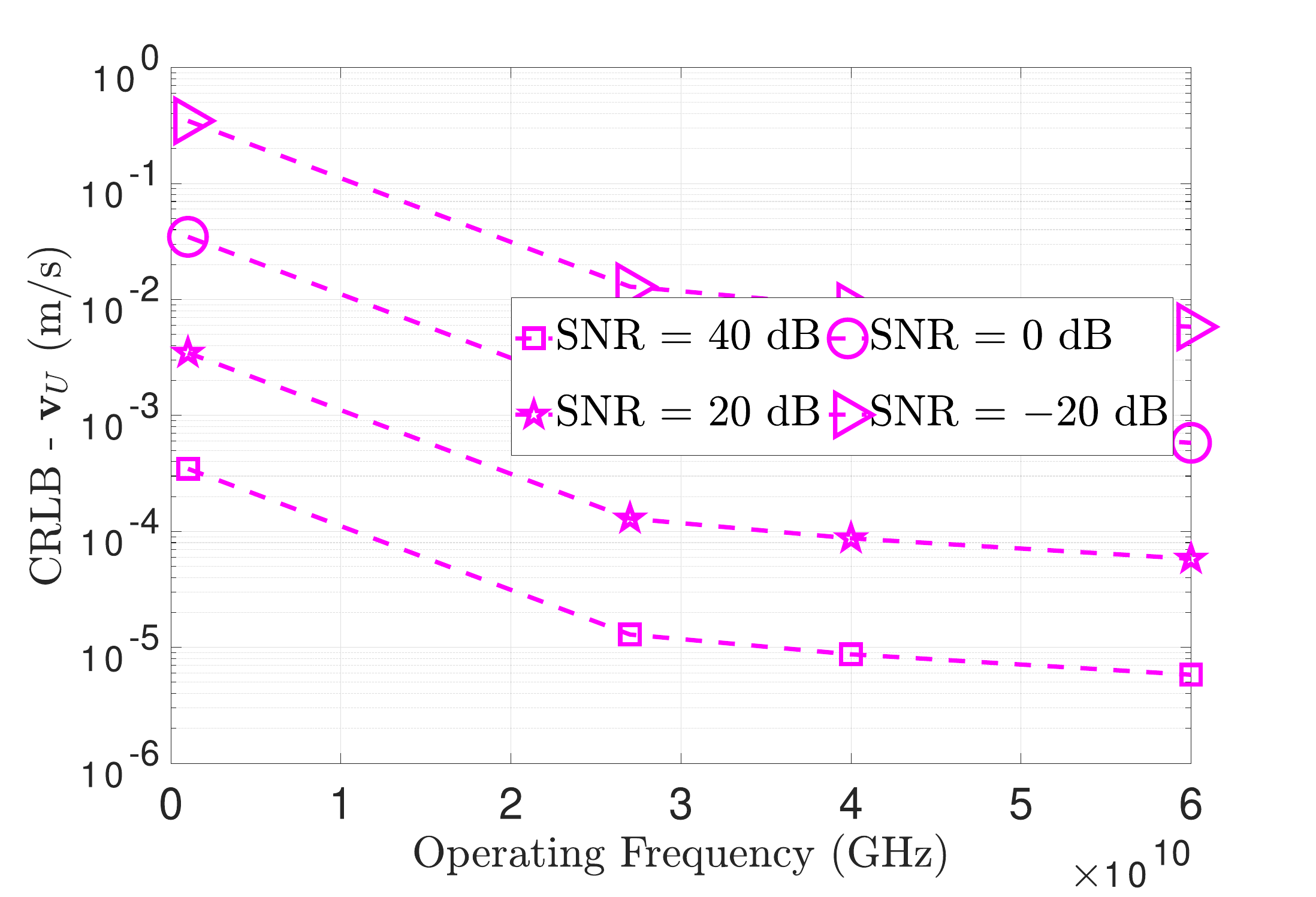}
\label{fig:Results/9D_VU_fc/_NB_3_N_K_3_N_U_8_delta_t_index_8_fcIndex___SNRIndex__}}
\caption{CRLB for $\bm{v}_{U}$ in the $9$D localization scenario with $N_U = 64$: (a) $\Delta_t = 10 \text{ s}$ and (b)  $\Delta_t = 80 \text{ s}$.}
\label{fig:Results/9D_VU_fc/_NB_3_N_K_3_N_U_8_delta_t_index_5_8_fcIndex___SNRIndex__}
\end{figure}

\subsection{Information Available to Find the $3$D Orientation of the receiver}
Here, we investigate the minimal number of time slots, LEO satellites, and receive antennas that produce a positive definite FIM for the $3$D orientation of the receiver, which is defined by (\ref{equ_theorem:FIM_3D_orientation}).

\subsubsection{$N_K = 1$, $N_B = 2$, \textit{and} $N_U > 1$}
Irrespective of the presence or absence of both time and frequency offsets, the $3$D orientation of the receiver can be estimated through the multiple TOA measurements received across the receive antennas from both LEO satellites.

\subsubsection{$N_K = 2$, $N_B = 1$, \textit{and} $N_U > 1$}
Irrespective of the presence or absence of both time and frequency offsets, the $3$D orientation of the receiver can be estimated through the multiple TOA measurements received across the receive antennas during both time slots from the LEO satellite.

\begin{table*}[!t]
\renewcommand{\arraystretch}{1.3}
\caption{Scenarios that allow for $\bm{p}_{U}$ estimation with ideal synchronization (No offset) and a lack of time or frequency synchronization (see Table \ref{Legend_Estimation_Possibilities_for_PU_in_the_3D_Localization} for legend)}
\label{Estimation_Possibilities_for_PU_in_the_3D_Localization}
\centering
\begin{tabular}{|c|c|c|c|c|}
\hline
 & No offset & Time offset & Frequency offset & Both \\
\hline
$N_K = 1$, $N_B = 1$, \textit{and} $N_U > 1$ & a & - & a & - \\
\hline
 $N_K = 1$, $N_B = 2$, \textit{and} $N_U = 1$& b, c & - & - & - \\
\hline
 $N_K = 1$, $N_B = 2$, \textit{and} $N_U > 1$& a, b, c, d, e, f & d, f &  d, e & d \\
\hline
 $N_K = 1$, $N_B = 3$, \textit{and} $N_U = 1$& b, c, g, h & h & g & - \\
\hline
 $N_K = 1$, $N_B = 3$, \textit{and} $N_U > 1$& b, c, d, e, f, g, h & d, f, h & d, e, g & d \\
\hline
 $N_K = 2$, $N_B = 1$, \textit{and} $N_U = 1$& i, j & - & - & - \\
\hline
 $N_K = 2$, $N_B = 1$, \textit{and} $N_U > 1$& i, j, k, l, m & k, m & k, l & k \\
\hline
 $N_K = 3$, $N_B = 1$, \textit{and} $N_U = 1$& i, j, n, o & o & n & - \\
\hline
 $N_K = 3$, $N_B = 1$, \textit{and} $N_U > 1$& i, j, k, l, m, n, o & k, m, o & k, l, n & k \\
\hline
 $N_K = 4$, $N_B = 1$, \textit{and} $N_U = 1$& i, j, n, o, p, q & o, p, q &  n, p, q & p, q \\
\hline
 $N_K = 4$, $N_B = 1$, \textit{and} $N_U > 1$& i, j, k, l, m, n, o, p, q & k, m, o, p, q & k, l, n,  p, q & k, p, q \\
\hline
\end{tabular}
\end{table*}

\begin{table*}[!t]
\renewcommand{\arraystretch}{1.3}
\caption{Scenarios that allow for $\bm{\Phi}_{U}$ estimation with ideal synchronization (No offset) and a lack of time or frequency synchronization (see Table \ref{Legend_Estimation_Possibilities_for_PHIU_in_the_3D_Localization} for legend)}
\label{Estimation_Possibilities_for_PHIU_in_the_3D_Localization}
\centering
\begin{tabular}{|c|c|c|c|c|}
\hline
 & No offset & Time offset & Frequency offset & Both \\
\hline
$N_K = 1$, $N_B = 2$, \textit{and} $N_U > 1$ & a & a & a & a \\
\hline
 $N_K = 2$, $N_B = 1$, \textit{and} $N_U > 1$& b & b & b & b \\
\hline
\end{tabular}
\end{table*}

\begin{table*}[!t]
\renewcommand{\arraystretch}{1.3}
\caption{Scenarios that allow for $\bm{v}_{U}$ estimation with ideal synchronization (No offset) and a lack of time or frequency synchronization (see Tables \ref{Legend_Estimation_Possibilities_for_vU_in_the_3D_Localization} and \ref{Legend_Estimation_Possibilities_for_vU_in_the_3D_Localization_1} for legend)}
\label{Estimation_Possibilities_for_vU_in_the_3D_Localization}
\centering
\begin{tabular}{|c|c|c|c|c|}
\hline
 & No offset & Time offset & Frequency offset & Both \\
\hline
$N_K = 1$, $N_B = 3$, \textit{and} $N_U = 1$ & a & a & - & - \\
\hline
$N_K = 2$, $N_B = 2$, \textit{and} $N_U = 1$ & b, c, d, e, f, g & b, c, d, e, f, g & b, c, d & b, c, d \\
\hline
$N_K = 4$, $N_B = 1$, \textit{and} $N_U = 1$ &  h, i, j, k, l, m, n, o & h, i, j, k, l, m, n, o & h, i, j, k, l,  & h, i, j, k, l \\
\hline
\end{tabular}
\end{table*}

\subsection{Information available to find the $3$D velocity of the receiver}
Here, we investigate the minimal number of time slots, LEO satellites, and receive antennas that produce a positive definite FIM for the $3$D velocity of the receiver, which is defined by (\ref{equ_theorem:FIM_3D_velocity}). 

\subsubsection{$N_K = 1$, $N_B = 3$, \textit{and} $N_U = 1$}
Here, the presence of frequency offsets causes the information to be insufficient to find the $3$D velocity of the receiver. Without frequency offsets, there is enough information to find the $3$D velocity of the receiver, and the presence of time offsets does not affect the feasibility of finding the $3$D velocity of the receiver. The delay measurements are not useful in this setup, whether in the presence or absence of time offsets. We can find the $3$D velocity of the receiver using three Doppler measurements from the three LEO satellites.

\subsubsection{$N_K = 2$, $N_B = 2$, \textit{and} $N_U = 1$}
 Without time and frequency offsets, we find the $3$D velocity of the receiver using: i) a combination of the frequency differencing of two Doppler measurements from two distinct time slots from the first LEO satellite,  frequency differencing of two Doppler measurements from two distinct time slots from the second LEO satellite, and time differencing of two delay measurements from two distinct time slots from any of the two satellites, ii) a combination of the time differencing of two delay measurements from two distinct time slots from the first LEO satellite,  time differencing of two delay measurements from two distinct time slots from the second LEO satellite, and frequency differencing of two Doppler measurements from two distinct time slots from any of the two satellites, iii) a combination of the frequency differencing of two Doppler measurements from two distinct time slots from the first LEO satellite,  frequency differencing of two Doppler measurements from two distinct time slots from the second LEO satellite, and two delay  measurements from two distinct time slots from either of the two satellites, iv) a combination of the time differencing of two delay measurements from two distinct time slots from the first LEO satellite,  time differencing of two delay measurements from two distinct time slots from the second LEO satellite, and one of the Doppler measurement from either of the distinct time slots from either of the LEO satellites, v)  a combination of  two delay  measurements from the two distinct time slots from either of the first satellite, two delay  measurements from the two distinct time slots from the second satellite, and one of the four Doppler measurements from either of the distinct time slots from either of the LEO satellites, vi) three of the four Doppler measurements from both of the distinct time slots from both of the LEO satellites. With a time offset, we find the $3$D velocity of the receiver using: 
 i), ii), iii), iv), v), vi) while with frequency offsets,  we find the $3$D velocity of the receiver using i), ii), and iii). With both offsets, we can use i), ii), and iii) to find the $3$D velocity of the receiver.

 \subsubsection{$N_K = 3$, $N_B = 1$, \textit{and} $N_U = 1$}
Without time and frequency offsets, we find the $3$D velocity of the receiver using: i) a combination of the frequency differencing of two Doppler measurements from two distinct time slots from the first LEO satellite, frequency differencing of the other two Doppler measurements from two distinct time slots from the LEO satellite, and time differencing of any two delay measurements from any two distinct time slots from the LEO satellite, ii) a combination of the time differencing of two delay measurements from two distinct time slots from the LEO satellite,  time differencing of  the other two delay measurements from two distinct time slots from the LEO satellite, and frequency differencing of any of the two Doppler measurements from two distinct time slots from the LEO satellite, iii) a combination of the frequency differencing of two Doppler measurements from two distinct time slots from the LEO satellite,  frequency differencing of the other two Doppler measurements from two distinct time slots from the LEO satellite, and two delay  measurements from two distinct time slots from the LEO satellite, iv) a combination of the time differencing of two delay measurements from two distinct time slots from the first LEO satellite,  time differencing of the other two delay measurements from two distinct time slots from the LEO satellite, and one of the Doppler measurement from one of the distinct time slots from the LEO satellite, v)  a combination of  two delay  measurements from the two distinct time slots from the LEO satellite, two delay  measurements from the two distinct time slots from the LEO satellite, and one of the three Doppler measurements from the distinct time slots from the LEO satellite, and vi) three Doppler measurements from the distinct time slots from the LEO satellite.  With a time offset, we find the $3$D velocity of the receiver using: 
 i), ii), iii), iv), v), vi) while with frequency offsets,  we find the $3$D velocity of the receiver using i), ii), and iii). With both offsets, we can use i), ii), and iii) to find the $3$D velocity of the receiver.

\subsubsection{$N_K = 4$, $N_B = 1$, \textit{and} $N_U = 1$}
Without time and frequency offsets, we find the $3$D velocity of the receiver using: i) four TOA measurements obtained during four different time slots from a single LEO satellite with one of the TOA measurements serving as a reference measurement for time differencing, ii) four Doppler measurements obtained during four different time slots from a single LEO satellite with one of the Doppler measurements serving as a reference measurement for frequency differencing, iii)
a combination of the frequency differencing of two Doppler measurements from two distinct time slots from the first LEO satellite,  frequency differencing of the other two Doppler measurements from two distinct time slots from the LEO satellite, and time differencing of any two delay measurements from any two distinct time slots from the LEO satellite, iv) a combination of the time differencing of two delay measurements from two distinct time slots from the LEO satellite,  time differencing of  the other two delay measurements from two distinct time slots from the LEO satellite, and frequency differencing of any of the two Doppler measurements from two distinct time slots from the LEO satellite, v) a combination of the frequency differencing of two Doppler measurements from two distinct time slots from the LEO satellite,  frequency differencing of the other two Doppler measurements from two distinct time slots from the LEO satellite, and two delay  measurements from two distinct time slots from the LEO satellite, vi) a combination of the time differencing of two delay measurements from two distinct time slots from the first LEO satellite,  time differencing of the other two delay measurements from two distinct time slots from the LEO satellite, and one of the Doppler measurement from one of the distinct time slots from the LEO satellite, vii)  a combination of  two delay  measurements from the two distinct time slots from the LEO satellite, two delay  measurements from the two distinct time slots from the LEO satellite, and one of the three Doppler measurements from the distinct time slots from the LEO satellite, and viii) three Doppler measurements from the distinct time slots from the LEO satellite. With time offset, we find the $3$D velocity of the receiver using: 
 i), ii), iii), iv), v), vi), vii), and viii) while with frequency offsets,  we find the $3$D velocity of the receiver using i), ii), iii), iv), and v). With both offsets, we can use i), ii), iii), iv), and v) to find the $3$D velocity of the receiver.

\subsection{Information available to find the $3$D position  and $3$D velocity of the receiver when the $3$D orientation is known }
 \subsubsection{$N_K = 1$, $N_B = 3$, \textit{and} $N_U = 1$}
The necessary conditions in Corollary (\ref{corollary:FIM_6D_3D_position_3D_position_3D_velocity_1}) and Corollary (\ref{corollary:FIM_6D_3D_position_3D_position_3D_velocity_2}) are met. Hence, the possibility of estimating only the $3$D position exists. Simulation results indicate that the FIM in this case of $3$D position estimation, $\mathbf{J}_{ \bm{\bm{y}}; \bm{p}_{U}}^{\mathrm{ee}}$ is positive definite when there is no time or frequency offset. Hence, the $3$D position can be estimated when there is no time or frequency offset. With either a time or a frequency offset, we can not estimate the $3$D position. From Theorem \ref{theorem:FIM_6D_3D_joint_3D_position_3D_velocity}, $\mathbf{J}_{ \bm{\bm{y}}; \bm{v}_{U}}^{\mathrm{ee}}$  also has to be positive definite and $\bm{v}_{U}$ can be estimated when there is no time or frequency offset. Also, $\bm{p}_{U}$ and $\bm{v}_{U}$ can be jointly estimated in this case with no time or frequency offset. While the estimation of $\bm{p}_{U}$ or $\bm{v}_{U}$ or both is feasible with no offsets, simulation results indicate that the errors resulting from estimating $\bm{v}_{U}$ is too large. Hence, only $\bm{p}_{U}$ should be estimated.

\begin{table*}[!t]
\renewcommand{\arraystretch}{1.3}
\caption{Information available for $6$D localization when the $3$D orientation is known - $\tau$ indicates that delay can be used to estimate the corresponding location parameters, $\nu$ indicates that Doppler can be used to estimate the corresponding location parameters, b indicates that combinations of the delay and Doppler can be used to estimate the corresponding location parameters.}
\label{Information_available_for_6_D_localization_when_the_3_D_orientation_is_known}
\centering
\begin{tabular}{|c|c|c|c|c|}
\hline
 & No offset & Time offset & Frequency offset & Both \\
\hline
$N_K = 1$, $N_B = 3$, \textit{and} $N_U = 1$ & $\bm{p}_{U}$ - B  & - & - & - \\
\hline
$N_K = 1$, $N_B = 6$, \textit{and} $N_U = 1$ & $\bm{p}_{U}$, $\bm{v}_{U}$ - $\nu$, B   & $\bm{p}_{U}$, $\bm{v}_{U}$ - $\nu$, B  & - & - \\
\hline
$N_K = 4$, $N_B = 1$, \textit{and} $N_U = 1$ & $\bm{p}_{U}$, $\bm{v}_{U}$ - B & $\bm{p}_{U}$, $\bm{v}_{U}$ - B & $\bm{p}_{U}$, $\bm{v}_{U}$ - B & $\bm{p}_{U}$, $\bm{v}_{U}$ - B \\
\hline
$N_K = 3$, $N_B = 2$, \textit{and} $N_U = 1$ & $\bm{p}_{U}$, $\bm{v}_{U}$ - $\tau, \nu$, B & $\bm{p}_{U}$, $\bm{v}_{U}$ - $\nu$, B & $\bm{p}_{U}$, $\bm{v}_{U}$ - $\tau$, B & $\bm{p}_{U}$, $\bm{v}_{U}$ - B \\
\hline
$N_K = 2$, $N_B = 3$, \textit{and} $N_U = 1$ & $\bm{p}_{U}$, $\bm{v}_{U}$ - $\tau, \nu$, B & $\bm{p}_{U}$, $\bm{v}_{U}$ - $\nu$, B & $\bm{p}_{U}$, $\bm{v}_{U}$ - $\tau$, B & $\bm{p}_{U}$, $\bm{v}_{U}$ - B \\
\hline
$N_K = 3$, $N_B = 3$, \textit{and} $N_U = 1$ & $\bm{p}_{U}$, $\bm{v}_{U}$ - $\tau, \nu$, B & $\bm{p}_{U}$, $\bm{v}_{U}$ - $\tau, \nu$, B  & $\bm{p}_{U}$, $\bm{v}_{U}$ - $\tau, \nu$, B & $\bm{p}_{U}$ , $\bm{v}_{U}$ - $\tau, \nu$, B \\
\hline
\end{tabular}
\end{table*}

 \subsubsection{$N_K = 1$, $N_B = 6$, \textit{and} $N_U = 1$}
The necessary conditions in Corollary (\ref{corollary:FIM_6D_3D_position_3D_position_3D_velocity_1}) and Corollary (\ref{corollary:FIM_6D_3D_position_3D_position_3D_velocity_2}) are met. Hence, the possibility of estimating only the $3$D position exists. Simulation results indicate that the FIM in this case of $3$D position estimation, $\mathbf{J}_{ \bm{\bm{y}}; \bm{p}_{U}}^{\mathrm{ee}}$ is positive definite except when there is a frequency offset. From Theorem \ref{theorem:FIM_6D_3D_joint_3D_position_3D_velocity}, $\mathbf{J}_{ \bm{\bm{y}}; \bm{v}_{U}}^{\mathrm{ee}}$  also has to be positive definite and $\bm{v}_{U}$ can be estimated when there is no frequency offset. FInally, $\bm{p}_{U}$ and $\bm{v}_{U}$ can be jointly estimated in this case with no frequency offset. Simulation results show that it is possible to use only the Doppler for estimating $\bm{p}_{U}$, $\bm{v}_{U}$, or both when there is no frequency offset.

\begin{figure}[htb!]
\centering
\subfloat[]{\includegraphics[ width= 3.2in]{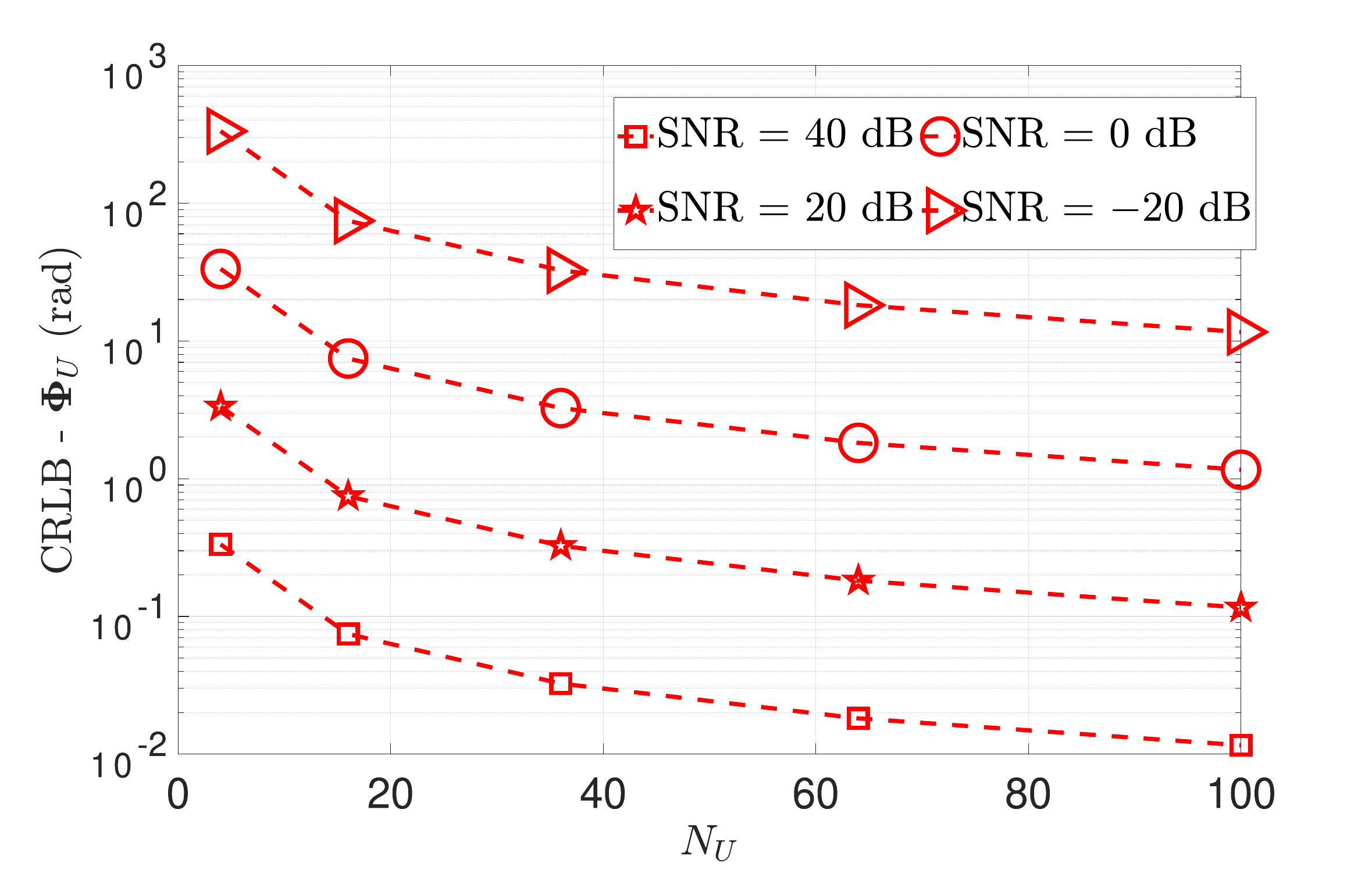}
\label{Results/9D_PHIU_NU/_NB_3_N_K_3_N_U___delta_t_index_5_fcIndex_1_SNRIndex__}}
\hfil
\subfloat[]{\includegraphics[ width= 3.2in]{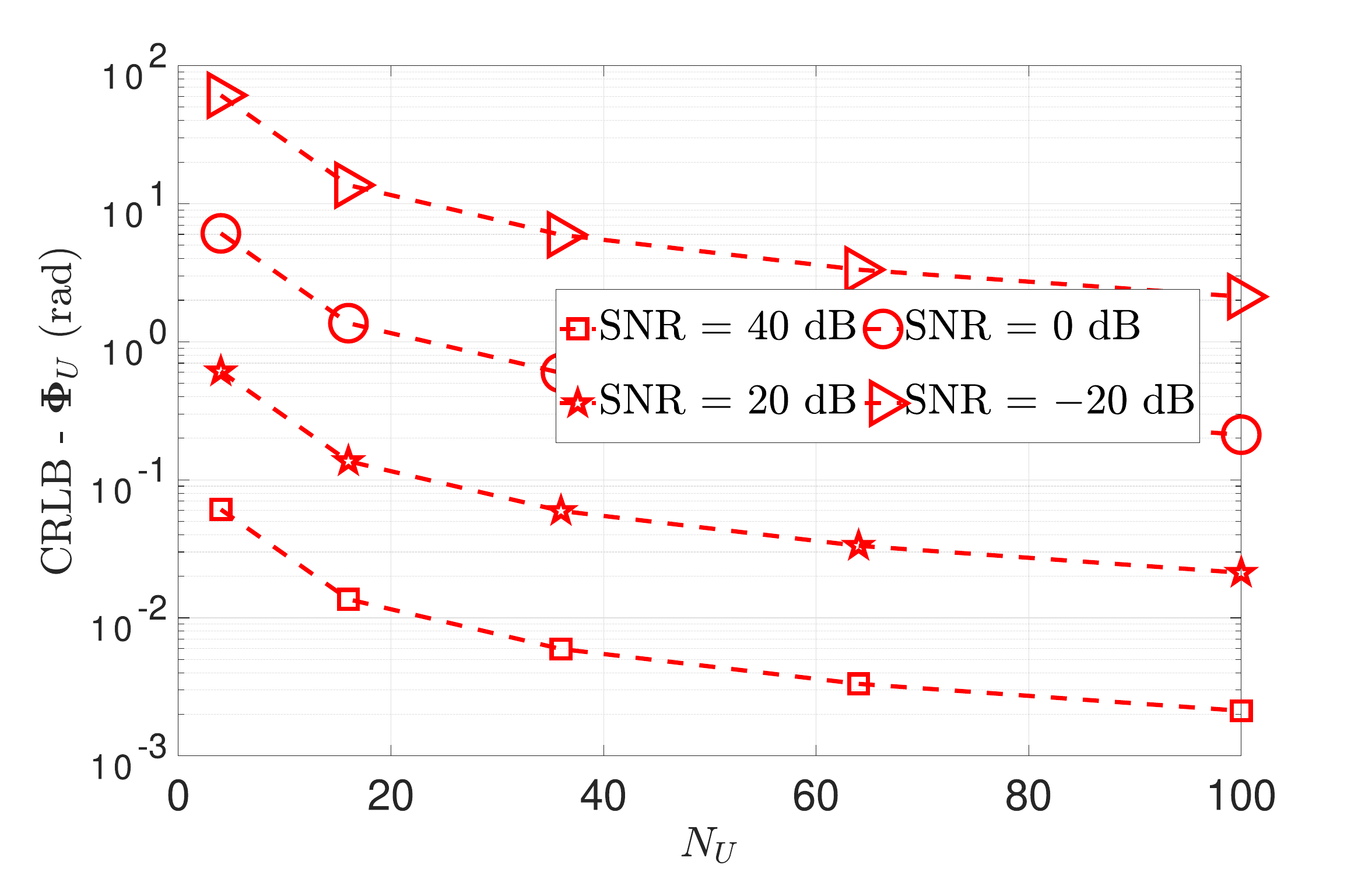}
\label{Results/9D_PHIU_NU/_NB_3_N_K_3_N_U___delta_t_index_8_fcIndex_1_SNRIndex__}}
\caption{CRLB for $\bm{\Phi}_{U}$ in the $9$D localization scenario with $f_c = 1 \text{ GHz}$: (a) $\Delta_t = 10 \text{ s}$ and (b) $\Delta_t = 80 \text{ s}$.}
\label{fig:Results/9D_PHIU_NU/_NB_3_N_K_3_N_U___delta_t_index_5_8_fcIndex_1_SNRIndex__}
\end{figure}

\begin{figure}[htb!]
\centering
\subfloat[]{\includegraphics[ width= 3.2in]{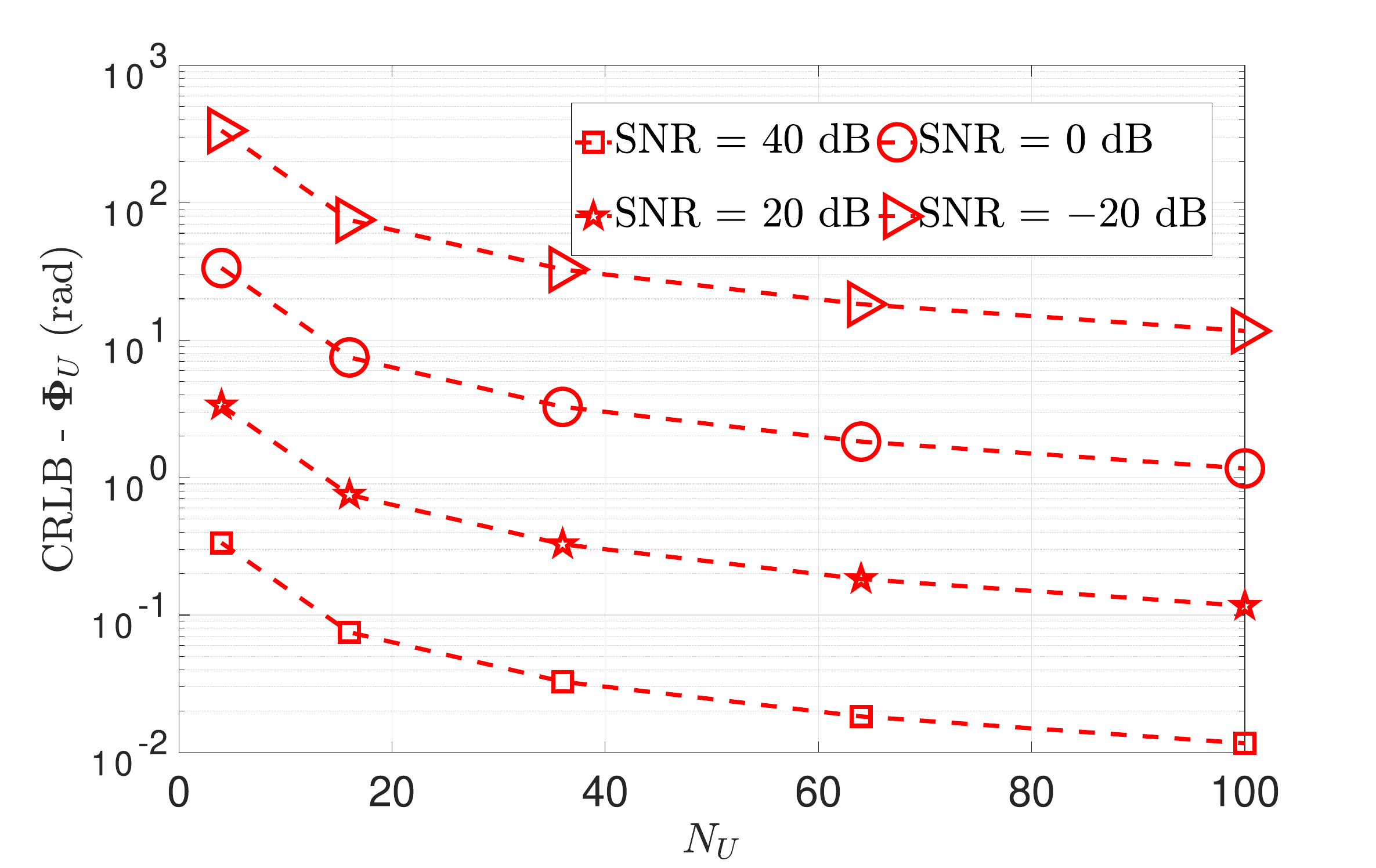}
\label{Results/9D_PHIU_NU/_NB_3_N_K_3_N_U___delta_t_index_5_fcIndex_4_SNRIndex__}}
\hfil
\subfloat[]{\includegraphics[ width= 3.2in]{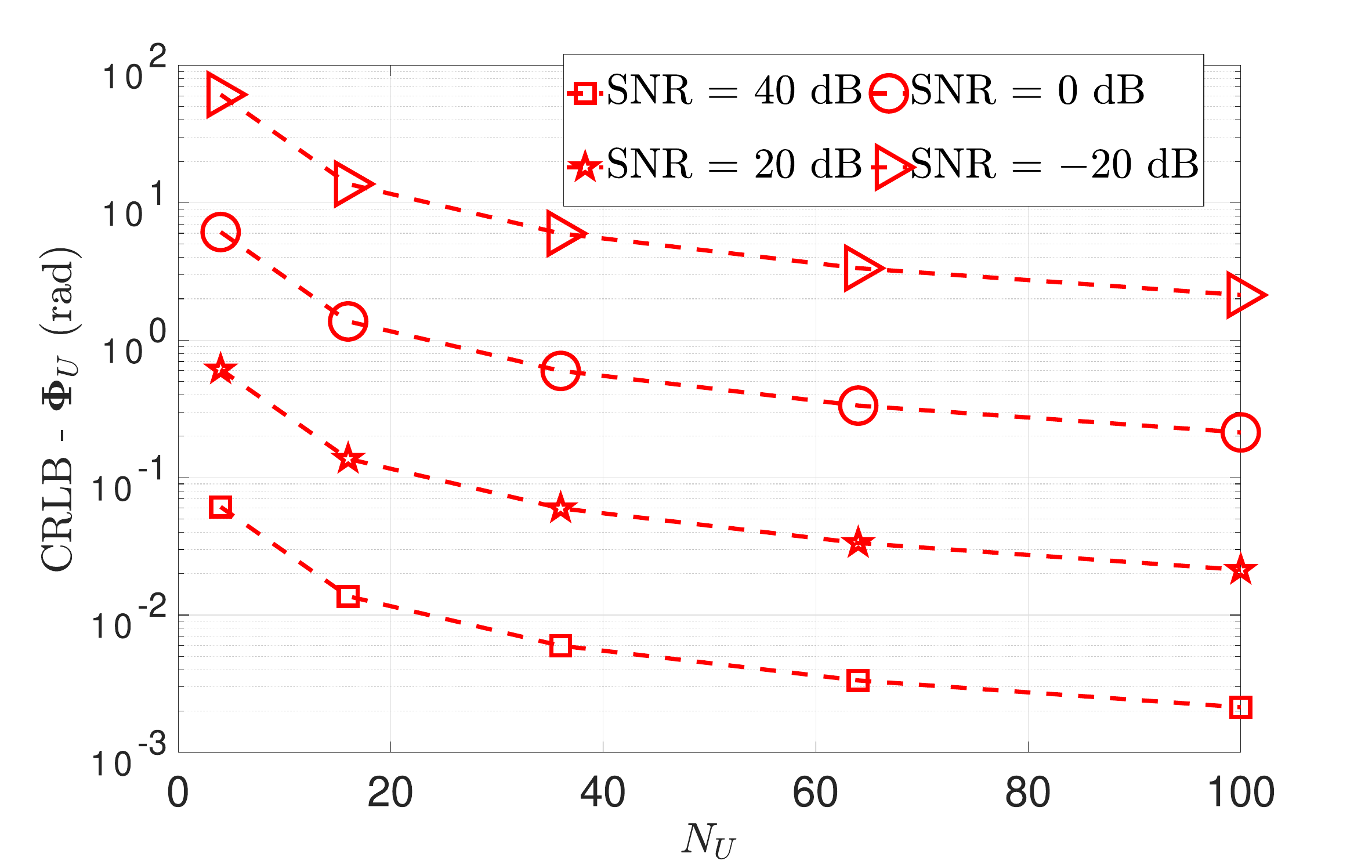}
\label{Results/9D_PHIU_NU/_NB_3_N_K_3_N_U___delta_t_index_8_fcIndex_4_SNRIndex__}}
\caption{CRLB for $\bm{\Phi}_{U}$ in the $9$D localization scenario with $f_c = 1 \text{ GHz}$: (a) $\Delta_t = 10 \text{ s}$ and (b) $\Delta_t = 80 \text{ s}$.}
\label{fig:Results/9D_PHIU_NU/_NB_3_N_K_3_N_U___delta_t_index_5_8_fcIndex_4_SNRIndex__}
\end{figure}

\begin{figure}[htb!]
\centering
\subfloat[]{\includegraphics[ width= 3.2in]{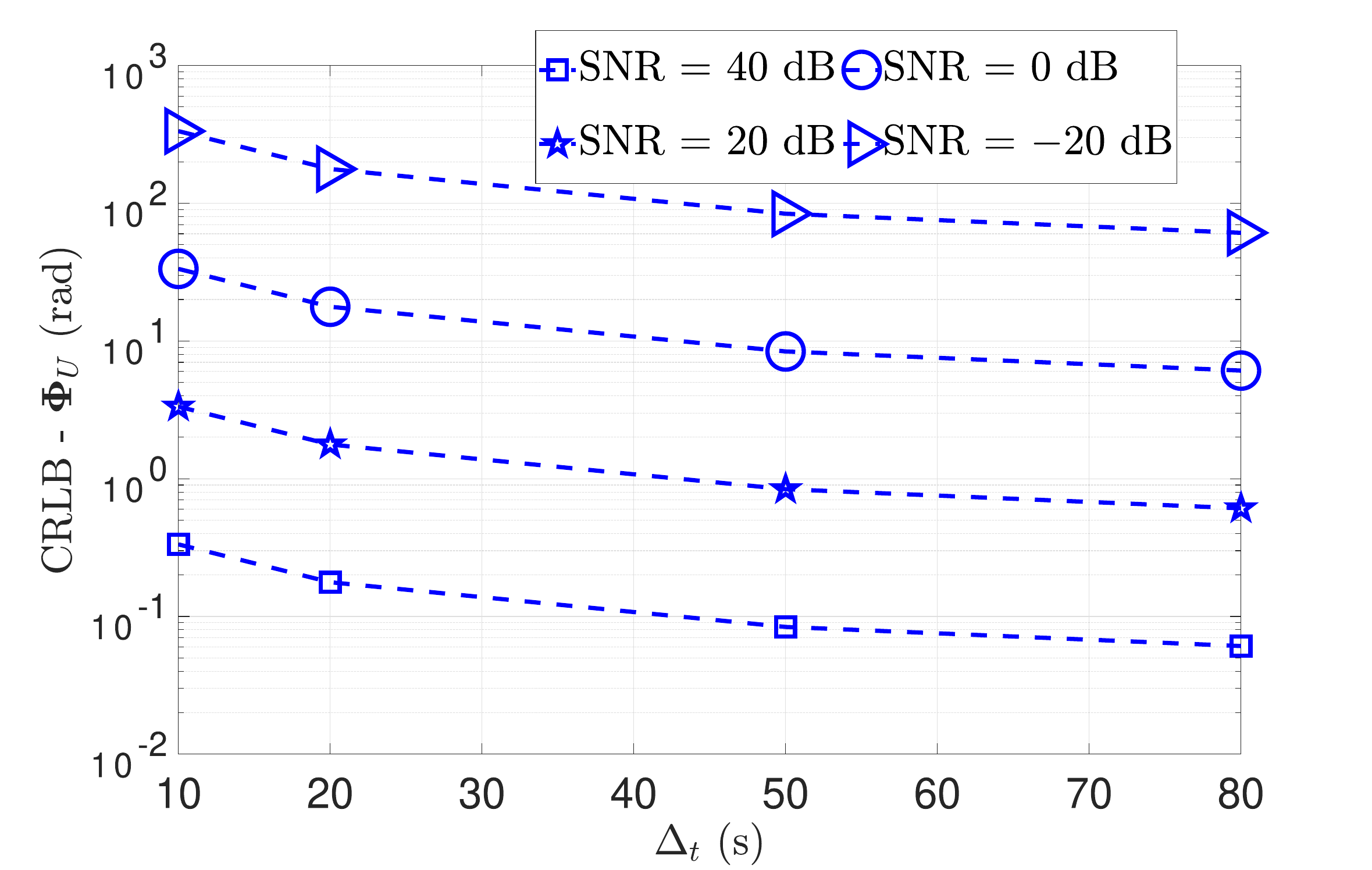}
\label{fig:Results/9D_PHIU_Time/_NB_3_N_K_3_N_U_2_delta_t_index___fcIndex_1_SNRIndex__}}
\hfil
\subfloat[]{\includegraphics[ width= 3.2in]{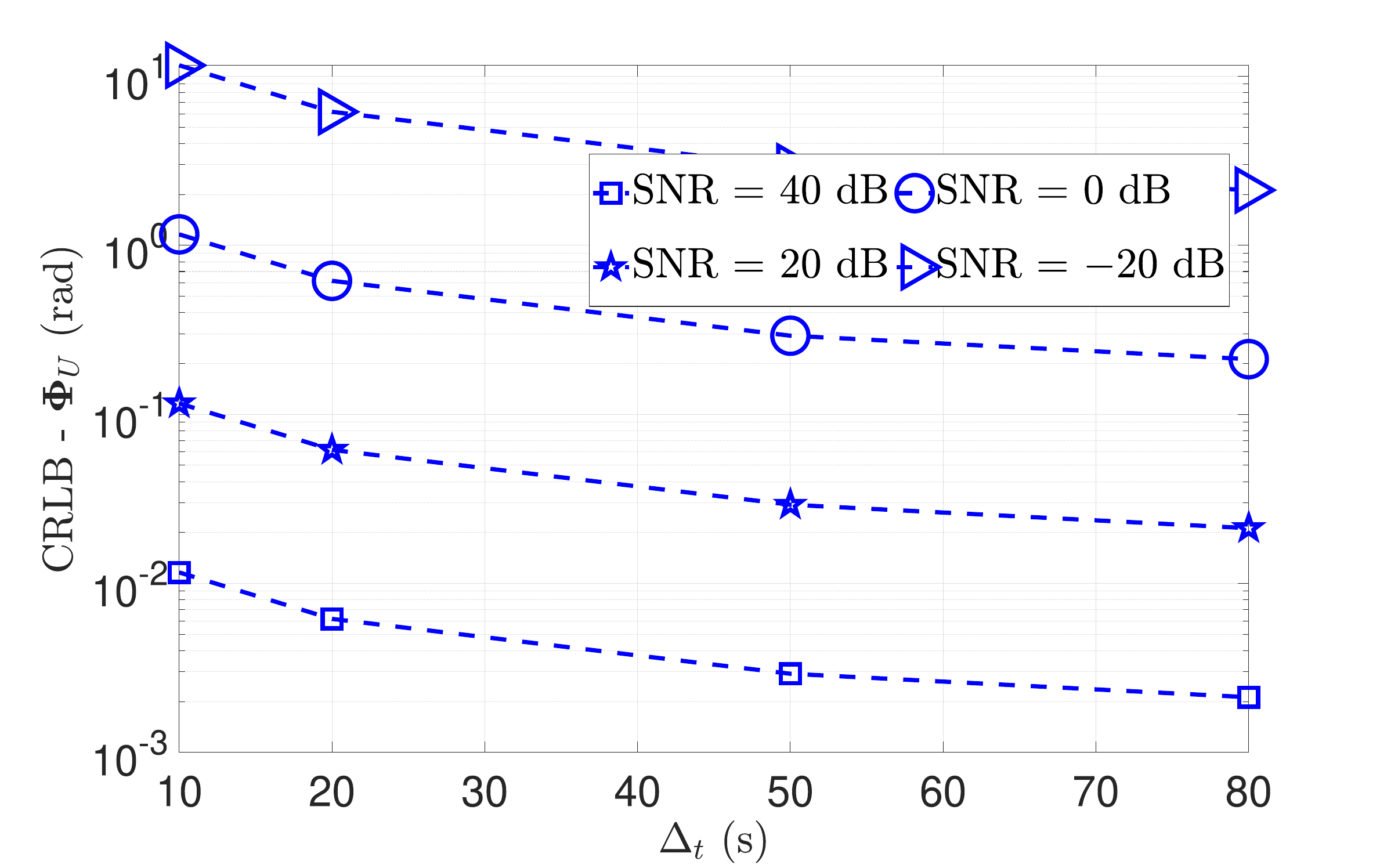}
\label{fig:Results/9D_PHIU_Time/_NB_3_N_K_3_N_U_10_delta_t_index___fcIndex_1_SNRIndex__}}
\caption{CRLB for $\bm{\Phi}_{U}$ in the $9$D localization scenario with $f_c = 1 \text{ GHz}$: (a) $N_U = 4$  and (b) $N_U = 100$.}
\label{fig:Results/9D_PHIU_Time/_NB_3_N_K_3_N_U_2_10_delta_t_index___fcIndex_1_SNRIndex__}
\end{figure}

\begin{figure}[htb!]
\centering
\subfloat[]{\includegraphics[ width= 3.2in]{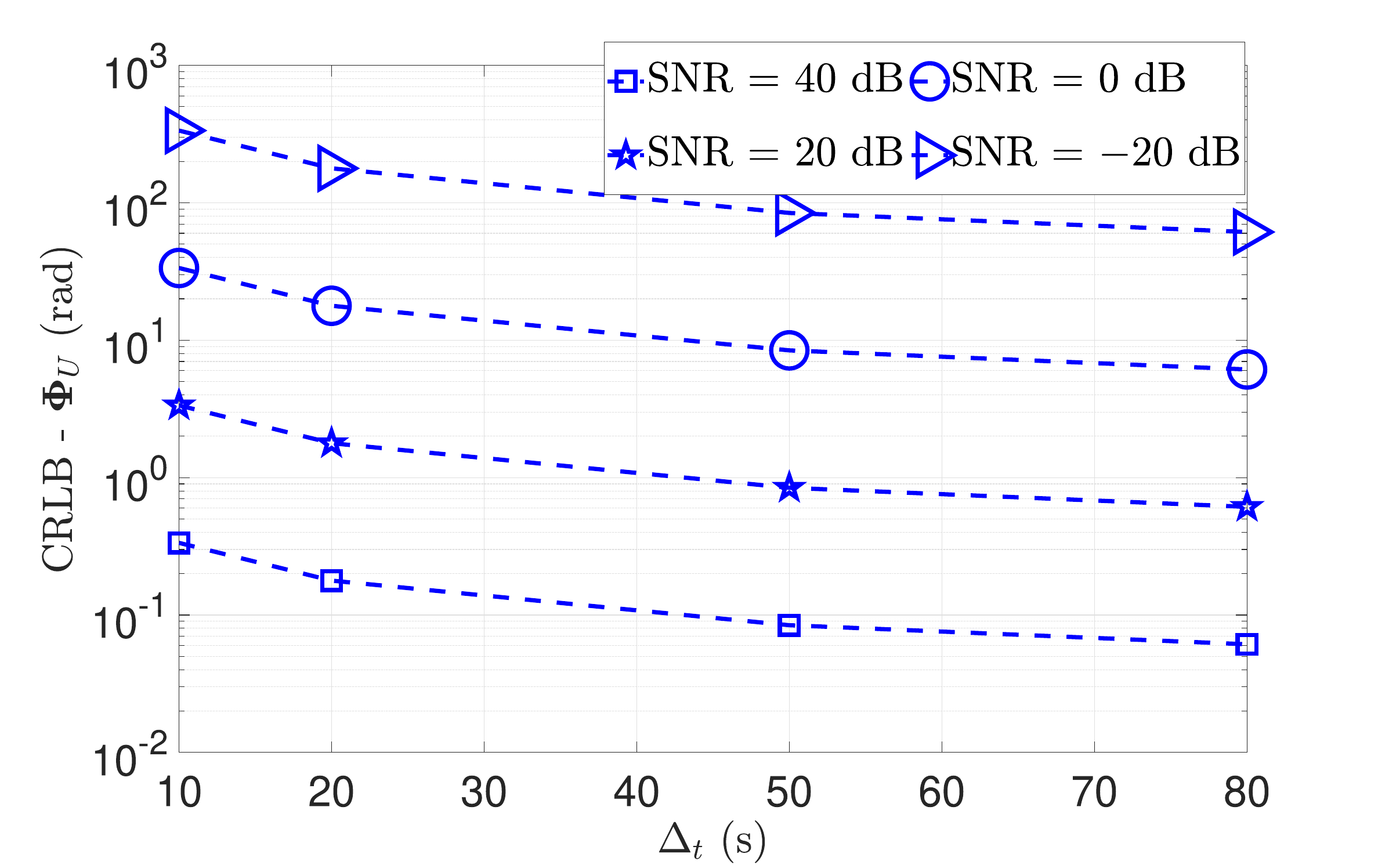}
\label{fig:Results/9D_PHIU_Time/_NB_3_N_K_3_N_U_2_delta_t_index___fcIndex_4_SNRIndex__}}
\hfil
\subfloat[]{\includegraphics[ width= 3.2in]{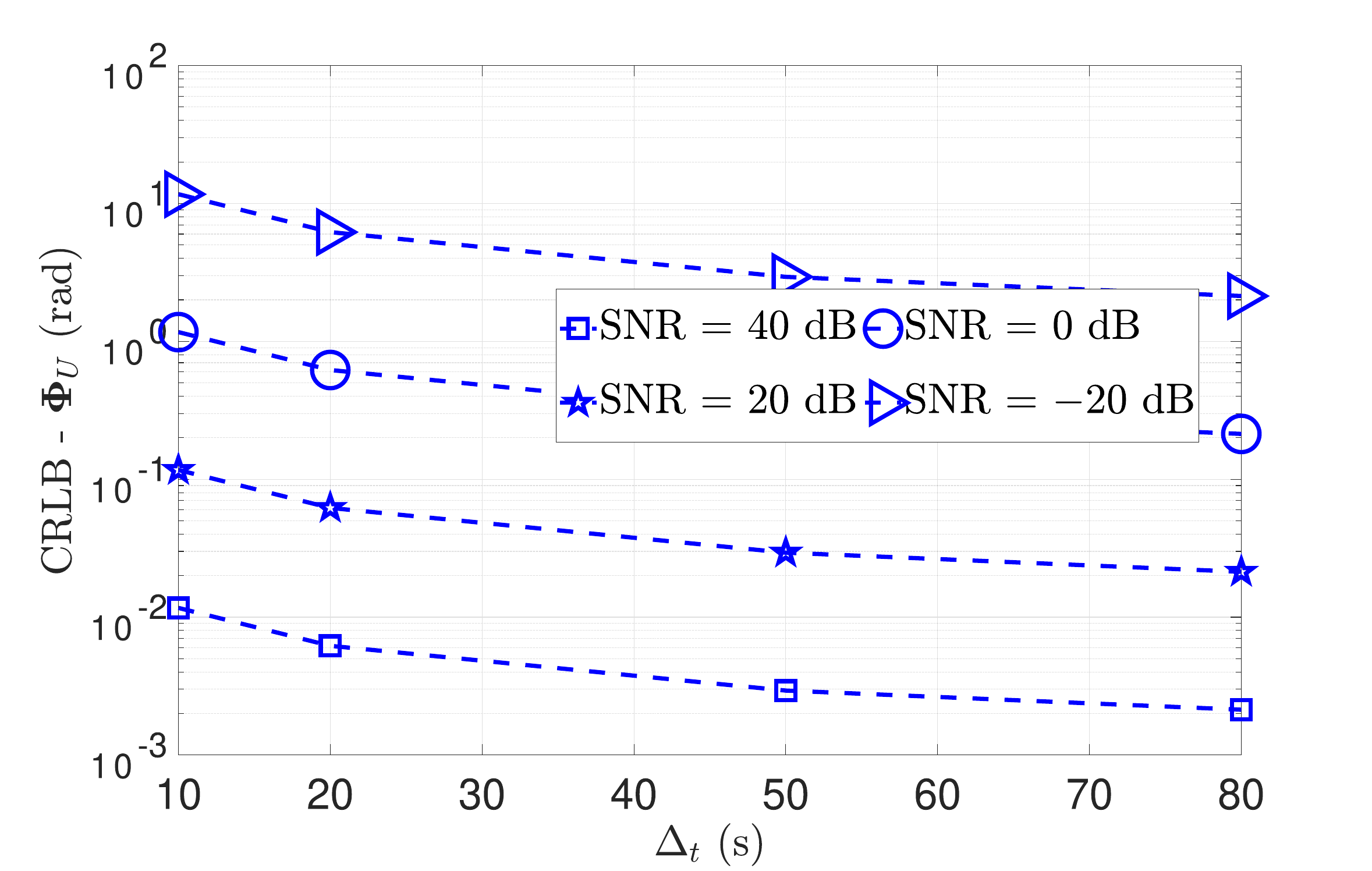}
\label{fig:Results/9D_PHIU_Time/_NB_3_N_K_3_N_U_10_delta_t_index___fcIndex_4_SNRIndex__}}
\caption{CRLB for $\bm{\Phi}_{U}$ in the $9$D localization scenario with $f_c = 60 \text{ GHz}$: (a) $N_U = 4$  and (b) $N_U = 100$.}
\label{fig:Results/9D_PHIU_Time/_NB_3_N_K_3_N_U_2_10_delta_t_index___fcIndex_4_SNRIndex__}
\end{figure}

\begin{figure}[htb!]
\centering
\subfloat[]{\includegraphics[ width= 3.2in]{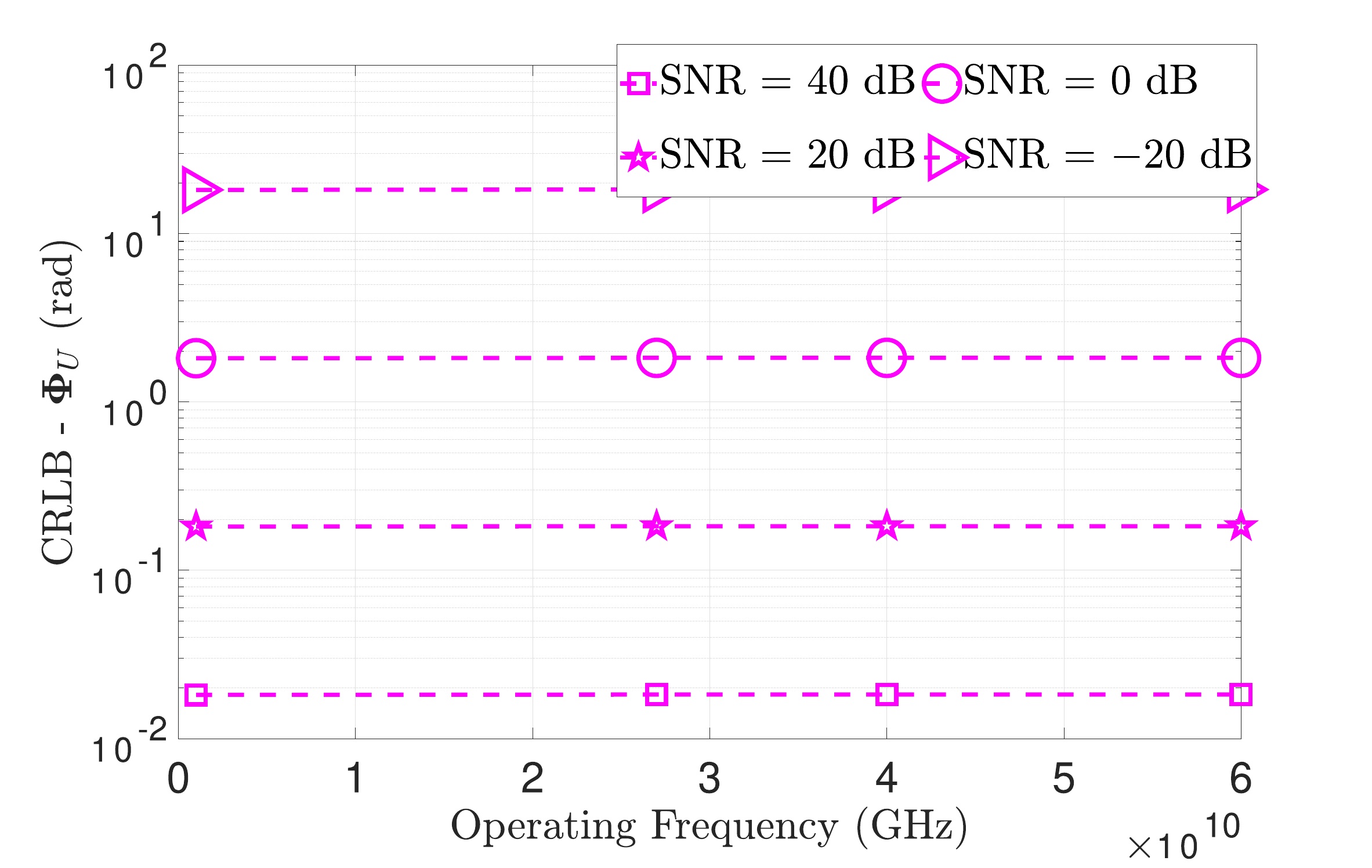}
\label{fig:Results/9D_PHIU_fc/_NB_3_N_K_3_N_U_8_delta_t_index_5_fcIndex___SNRIndex__}}
\hfil
\subfloat[]{\includegraphics[ width= 3.2in]{Results/9D_PHIU_fc/_NB_3_N_K_3_N_U_8_delta_t_index_5_fcIndex___SNRIndex__.pdf}
\label{fig:Results/9D_PHIU_fc/_NB_3_N_K_3_N_U_8_delta_t_index_8_fcIndex___SNRIndex__}}
\caption{CRLB for $\bm{\Phi}_{U}$ in the $9$D localization scenario with $N_U = 64$: (a) $\Delta_t = 10 \text{ s}$ and (b)  $\Delta_t = 80 \text{ s}$.}
\label{fig:Results/9D_PHIU_fc/_NB_3_N_K_3_N_U_8_delta_t_index_5_8_fcIndex___SNRIndex__}
\end{figure}

 \subsubsection{$N_K = 4$, $N_B = 1$, \textit{and} $N_U = 1$}
The necessary conditions in Corollary (\ref{corollary:FIM_6D_3D_position_3D_position_3D_velocity_1}) and Corollary (\ref{corollary:FIM_6D_3D_position_3D_position_3D_velocity_2}) are met. Hence, the possibility of estimating only the $3$D position exists. Simulation results indicate that the FIM in this case of $3$D position estimation, $\mathbf{J}_{ \bm{\bm{y}}; \bm{p}_{U}}^{\mathrm{ee}}$ is positive definite even when there is both a time and frequency offset. From Theorem \ref{theorem:FIM_6D_3D_joint_3D_position_3D_velocity}, $\mathbf{J}_{ \bm{\bm{y}}; \bm{v}_{U}}^{\mathrm{ee}}$  also has to be positive definite and $\bm{v}_{U}$ can be estimated when there is both a time and a frequency offset. Also, $\bm{p}_{U}$ and $\bm{v}_{U}$ can be jointly estimated with both offsets.

 \subsubsection{$N_K = 3$, $N_B = 2$, \textit{and} $N_U = 1$ or $N_K = 2$, $N_B = 3$, \textit{and} $N_U = 1$}
The necessary conditions in Corollary (\ref{corollary:FIM_6D_3D_position_3D_position_3D_velocity_1}) and Corollary (\ref{corollary:FIM_6D_3D_position_3D_position_3D_velocity_2}) are met even in the presence of a time or frequency offset, or both. Hence,  the possibility of estimating only the $3$D position exists even in the presence of a time or frequency offset or both.  Simulation results indicate that the FIM in this case of $3$D position estimation, $\mathbf{J}_{ \bm{\bm{y}}; \bm{p}_{U}}^{\mathrm{ee}}$ is positive definite in the presence of both offsets. From Theorem \ref{theorem:FIM_6D_3D_joint_3D_position_3D_velocity}, $\mathbf{J}_{ \bm{\bm{y}}; \bm{v}_{U}}^{\mathrm{ee}}$  also has to be positive definite and $\bm{v}_{U}$ can be estimated when there is both a time and a frequency offset. Also, $\bm{p}_{U}$ and $\bm{v}_{U}$ can be jointly estimated with both offsets.

It is possible to use only the delays for estimating $\bm{p}_{U}$, $\bm{v}_{U}$, or both when there is no time offset. Also, it is possible to use only the Dopplers for estimating $\bm{p}_{U}$, $\bm{v}_{U}$, or both when there is no frequency offset. Using combinations of delays and Dopplers are useful for estimating $\bm{p}_{U}$, $\bm{v}_{U}$, or both when there is both a time and a frequency offset.

 \subsubsection{$N_K = 3$, $N_B = 3$, \textit{and} $N_U = 1$}
 In this case, the delays are useful for estimating $\bm{p}_{U}$, $\bm{v}_{U}$, or both, even when there is both a time and a frequency offset. Similarly, in this case, the Dopplers are useful for estimating $\bm{p}_{U}$, $\bm{v}_{U}$, or both, even when there is both a time and a frequency offset. Finally, combinations of delays and Dopplers are useful for estimating $\bm{p}_{U}$, $\bm{v}_{U}$, or both when there is both a time and a frequency offset.

\begin{table*}[!t]
\renewcommand{\arraystretch}{1.3}
\caption{Information available for $6$D localization when the $3$D velocity is known - $\tau$ indicates that delay can be used to estimate the corresponding location parameters, and b indicates that combinations of the delay and Doppler can be used to estimate the corresponding location parameters.}
\label{Information_available_for_6_D_localization_when_the_3_D_velocity_is_known}
\centering
\begin{tabular}{|c|c|c|c|c|}
\hline
 & No offset & Time offset & Frequency offset & Both \\
\hline
$N_K = 1$, $N_B = 2$, \textit{and} $N_U > 1$ & $\bm{p}_{U}$, $\Phi_{U}$ - $\tau$, B  & - & $\bm{p}_{U}$, $\Phi_{U}$ - $\tau$, B & - \\
\hline
$N_K = 2$, $N_B = 1$, \textit{and} $N_U > 1$ & $\bm{p}_{U}$, $\Phi_{U}$ - $\tau$, B  & - & $\bm{p}_{U}$, $\Phi_{U}$ - $\tau$, B & - \\
\hline
$N_K = 2$, $N_B = 2$, \textit{and} $N_U > 1$ & $\bm{p}_{U}$, $\bm{\Phi}_{U}$ - $\tau$, B & $\bm{p}_{U}$, $\bm{\Phi}_{U}$ - $\tau$, B & $\bm{p}_{U}$, $\bm{\Phi}_{U}$ - $\tau$, B & $\bm{p}_{U}$, $\bm{\Phi}_{U}$ - $\tau$, B \\
\hline
$N_K = 3$, $N_B = 2$, \textit{and} $N_U > 1$ & $\bm{p}_{U}$, $\bm{\Phi}_{U}$ - $\tau$, B & $\bm{p}_{U}$, $\bm{\Phi}_{U}$ - $\tau$, B & $\bm{p}_{U}$, $\bm{\Phi}_{U}$ - $\tau$, B & $\bm{p}_{U}$, $\bm{\Phi}_{U}$ - $\tau$, B \\
\hline
$N_K = 2$, $N_B = 3$, \textit{and} $N_U > 1$ & $\bm{p}_{U}$, $\bm{\Phi}_{U}$ - $\tau$, B & $\bm{p}_{U}$, $\bm{\Phi}_{U}$ - $\tau$, B & $\bm{p}_{U}$, $\bm{\Phi}_{U}$ - $\tau$, B & $\bm{p}_{U}$, $\bm{\Phi}_{U}$ - $\tau$, B \\
\hline
\end{tabular}
\end{table*}

\begin{table*}[!t]
\renewcommand{\arraystretch}{1.3}
\caption{Information available for $6$D localization when the $3$D position is known - $\tau$ indicates that delay can be used to estimate the corresponding location parameters, and b indicates that combinations of the delay and Doppler can be used to estimate the corresponding location parameters.}
\label{Information_available_for_6_D_localization_when_the_3_D_position_is_known}
\centering
\begin{tabular}{|c|c|c|c|c|}
\hline
 & No offset & Time offset & Frequency offset & Both \\
\hline
$N_K = 3$, $N_B = 2$, \textit{and} $N_U > 1$ & $\bm{v}_{U}$, $\bm{\Phi}_{U}$ - $\tau$, B & $\bm{v}_{U}$, $\bm{\Phi}_{U}$ - $\tau$, B & $\bm{v}_{U}$, $\bm{\Phi}_{U}$ - $\tau$, B & $\bm{v}_{U}$, $\bm{\Phi}_{U}$ - $\tau$, B \\
\hline
$N_K = 2$, $N_B = 3$, \textit{and} $N_U > 1$ & $\bm{v}_{U}$, $\bm{\Phi}_{U}$ - $\tau$, B & $\bm{v}_{U}$, $\bm{\Phi}_{U}$ - $\tau$, B & $\bm{v}_{U}$, $\bm{\Phi}_{U}$ - $\tau$, B & $\bm{v}_{U}$, $\bm{\Phi}_{U}$ - $\tau$, B \\
\hline
\end{tabular}
\end{table*}

\subsection{Information available to find the $3$D position and $3$D orientation of the receiver when the $3$D velocity is known}

 \subsubsection{$N_K = 1$, $N_B = 2$, \textit{and} $N_U > 1$ or $N_K = 2$, $N_B = 1$, \textit{and} $N_U > 1$}
The presence of a time or frequency offset hinders the estimation of either the $3$D position, $3$D orientation, or both. When there is neither a time nor frequency offset, the necessary conditions in Corollary (\ref{corollary:FIM_6D_3D_orientation_3D_position_3D_orientation_1}) and Corollary (\ref{corollary:FIM_6D_3D_orientation_3D_position_3D_orientation_2}) are met. Hence, the possibility of estimating only the $3$D orientation exists. Simulation results indicate that the FIM in this case of $3$D orientation estimation, $\mathbf{J}_{ \bm{\bm{y}}; \bm{\Phi}_{U}}^{\mathrm{ee}}$ is positive definite. Hence, the $3$D orientation can be estimated. From Theorem \ref{theorem:FIM_6D_3D_joint_3D_position_3D_orientation}, $\mathbf{J}_{ \bm{\bm{y}}; \bm{p}_{U}}^{\mathrm{ee}}$  also has to be positive definite and $\bm{p}_{U}$ can be estimated when there is no time or frequency offset. Also, $\bm{p}_{U}$ and $\bm{\Phi}_{U}$ can be jointly estimated in this case with no time or frequency offset.

\begin{table*}[!t]
\renewcommand{\arraystretch}{1.3}
\caption{Information available for $9$D localization -  b indicates that combinations of the delay and Doppler can be used to estimate the corresponding location parameters.}
\label{Information_available_for_9_D_localization}
\centering
\begin{tabular}{|c|c|c|c|c|}
\hline
 & No offset & Time offset & Frequency offset & Both \\
\hline
$N_K = 3$, $N_B = 3$, \textit{and} $N_U > 1$ & $\bm{p}_{U}$, $\bm{v}_{U}$, $\bm{\Phi}_{U}$ -  B & $\bm{p}_{U}$, $\bm{v}_{U}$, $\bm{\Phi}_{U}$ -  B & $\bm{p}_{U}$, $\bm{v}_{U}$, $\bm{\Phi}_{U}$ -  B  & $\bm{p}_{U}$, $\bm{v}_{U}$, $\bm{\Phi}_{U}$ -  B  \\
\hline
\end{tabular}
\end{table*}

 \subsubsection{$N_K = 2$, $N_B = 2$, \textit{and} $N_U > 1$}
 The presence of a time offset hinders the estimation of either the $3$D position, $3$D orientation, or both. When there is no time offset (even in the presence of frequency offsets), the necessary conditions in Corollary (\ref{corollary:FIM_6D_3D_orientation_3D_position_3D_orientation_1}) and Corollary (\ref{corollary:FIM_6D_3D_orientation_3D_position_3D_orientation_2}) are met. Hence, the possibility of estimating only the $3$D orientation exists. Simulation results indicate that the FIM in this case of $3$D orientation estimation, $\mathbf{J}_{ \bm{\bm{y}}; \bm{\Phi}_{U}}^{\mathrm{ee}}$ is positive definite. Hence, the $3$D orientation can be estimated. From Theorem \ref{theorem:FIM_6D_3D_joint_3D_position_3D_orientation}, $\mathbf{J}_{ \bm{\bm{y}}; \bm{p}_{U}}^{\mathrm{ee}}$  also has to be positive definite and $\bm{p}_{U}$ can be estimated when there is no time offset. Also, $\bm{p}_{U}$ and $\bm{\Phi}_{U}$ can be jointly estimated in this case with no time offset. Note that it is possible to use only the delay measurements.

 \subsubsection{$N_K = 2$, $N_B = 3$, \textit{and} $N_U > 1$ or $N_K = 3$, $N_B = 2$, \textit{and} $N_U > 1$}

The necessary conditions in Corollary (\ref{corollary:FIM_6D_3D_orientation_3D_position_3D_orientation_1}) and Corollary (\ref{corollary:FIM_6D_3D_orientation_3D_position_3D_orientation_2}) are met even in the presence of time or frequency offset, or both. Hence,  the possibility of estimating only the $3$D orientation exists even in the presence of time or frequency offset, or both.  Simulation results indicate that the FIM in this case of $3$D orientation estimation, $\mathbf{J}_{ \bm{\bm{y}}; \bm{\Phi}_{U}}^{\mathrm{ee}}$ is positive definite in the presence of both offsets. From Theorem \ref{theorem:FIM_6D_3D_joint_3D_position_3D_orientation}, $\mathbf{J}_{ \bm{\bm{y}}; \bm{p}_{U}}^{\mathrm{ee}}$  also has to be positive definite and $\bm{p}_{U}$ can be estimated when there is both a time and a frequency offset. Also, $\bm{p}_{U}$ and $\bm{\Phi}_{U}$ can be jointly estimated with both offsets. It is possible to use only the delays for estimating $\bm{p}_{U}$, $\bm{\Phi}_{U}$, or both when there is both a time and a frequency offset.

 \subsection{Information available to find the $3$D velocity and $3$D orientation of the receiver when the $3$D position is known }

  \subsubsection{$N_K = 2$, $N_B = 3$, \textit{and} $N_U > 1$ or $N_K = 3$, $N_B = 2$, \textit{and} $N_U > 1$}

The necessary conditions in Corollary (\ref{corollary:FIM_6D_3D_velocity_3D_velocity_3D_orientation_1}) and Corollary (\ref{corollary:FIM_6D_3D_velocity_3D_velocity_3D_orientation_2}) are met even in the presence of a time or frequency offset, or both. Hence,  the possibility of estimating only the $3$D orientation exists even in the presence of a time or frequency offset, or both.  Simulation results indicate that the FIM in this case of $3$D orientation estimation, $\mathbf{J}_{ \bm{\bm{y}}; \bm{\Phi}_{U}}^{\mathrm{ee}}$ is positive definite in the presence of both offsets. From Theorem \ref{theorem:FIM_6D_3D_joint_3D_position_3D_orientation}, $\mathbf{J}_{ \bm{\bm{y}}; \bm{v}_{U}}^{\mathrm{ee}}$  also has to be positive definite and $\bm{v}_{U}$ can be estimated when there is both a time and a frequency offset. Also, $\bm{v}_{U}$ and $\bm{\Phi}_{U}$ can be jointly estimated with both offsets. It is possible to use only the delays for estimating $\bm{v}_{U}$, $\bm{\Phi}_{U}$, or both when there is both a time and a frequency offset.

 \subsection{Information available to find the $3$D position of the receiver when the $3$D velocity and $3$D orientation are unknown}

  \subsubsection{$N_K = 3$, $N_B = 3$, \textit{and} $N_U > 1$}
Even in the presence of both offsets, there is enough information to estimate: i) $\mathbf{J}_{ \bm{\bm{y}}; \bm{\Phi}_{U}}^{\mathrm{e}}$, ii) $\mathbf{J}_{ \bm{\bm{y}}; \bm{v}_{U}}^{\mathrm{e}}$, and iii) $\mathbf{J}_{ \bm{\bm{y}}; \bm{v}_{U}}^{\mathrm{ee}}$. Hence, the loss in information about $\bm{p}_{U}$ due to the unknown $\bm{\Phi}_{U}$ and $\bm{v}_{U}$ which is specified by  $\mathbf{J}_{ \bm{\bm{y}}; \bm{p}_{U}}^{nu}$ exists. Further, since there is enough information to estimate $\mathbf{J}_{ \bm{\bm{y}}; \bm{p}_{U}}^{\mathrm{e}}$ then  $\mathbf{J}_{ \bm{\bm{y}}; \bm{p}_{U}}^{\mathrm{eee}}$ can also be estimated by Theorem \ref{theorem:FIM_9D_position}. Simulation results verify that there is enough information for $\mathbf{J}_{ \bm{\bm{y}}; \bm{p}_{U}}^{\mathrm{eee}}$ to be positive definite. Hence, $\bm{p}_{U}$ can be estimated in the presence of both offsets.

\subsubsection{Simulation results}
Here, we present simulation results for the CRLB when estimating $\bm{p}_{U}$ with $N_K = 3$, $N_B = 3$, \textit{and} $N_U > 1$. We consider the following SNR values $\{40 \text{ dB}, 20 \text{ dB}, 0 \text{ dB}, -20 \text{ dB}\}$ and different number of receive antennas. In Fig. \ref{Results:9D_PU_NB_3_N_K_3_N_U___delta_t_index_5_8_fcIndex_1_SNRIndex_1_NU}, the operating frequency is $1 \text{ GHz}$ with $\Delta_t = 10 \text{ s} $ in Fig. \ref{fig:Results/9D_PU_NU/_NB_3_N_K_3_N_U___delta_t_index_5_fcIndex_1_SNRIndex_1} and  $\Delta_t = 80 \text{ s} $ in Fig. \ref{fig:Results/9D_PU_NU/_NB_3_N_K_3_N_U___delta_t_index_8_fcIndex_1_SNRIndex_1}. In Fig. \ref{fig:Results/9D_PU_NU/_NB_3_N_K_3_N_U___delta_t_index_5_fcIndex_1_SNRIndex_1}, at an SNR of $-20 \text{ dB}$, the CRLB starts out at about $10 \text{ km}$ $(N_U  = 4)$ and decreases to $8 \text{ km}$ $(N_U  = 100).$ At an SNR of $0 \text{ dB}$, the CRLB starts out at about $2 \text{ km}$ $(N_U  = 4)$ and decreases to $0.5 \text{ km}$ $(N_U  = 100).$ At an SNR of $40 \text{ dB}$, we can achieve m-level accuracy irrespective of the number of receive antennas. Now in Fig. \ref{fig:Results/9D_PU_NU/_NB_3_N_K_3_N_U___delta_t_index_8_fcIndex_1_SNRIndex_1}, the achievable errors are smaller. We can achieve dm-level accuracy irrespective of the number of receive antennas. Also, the errors at the lower SNR values are reduced. For instance, the error at an SNR of $-20 \text{ dB}$ with $N_U = 4$ is only $800 \text{ m}$ as opposed to $10 \text{ km}$ in Fig. \ref{fig:Results/9D_PU_NU/_NB_3_N_K_3_N_U___delta_t_index_5_fcIndex_1_SNRIndex_1}.  This improvement is because of the improvement in the GDOP attained due to the longer duration between transmission time slots. 

In Fig. \ref{Results:9D_PU_NB_3_N_K_3_N_U___delta_t_index_5_8_fcIndex_4_SNRIndex_1_NU} at $f_c = 60 \text{ GHz}$, the error values are much lower. In fact, for similar parameters, the error values range from $0.5 \text{ km}$ to $10^{-3}$ of a m.  In Fig. \ref{fig:Results/9D_PU_NU/_NB_3_N_K_3_N_U___delta_t_index_5_fcIndex_4_SNRIndex_1}, at an SNR of $-20 \text{ dB}$, the CRLB starts out at about $500 \text{ m}$ $(N_U  = 4)$ and decreases to $95 \text{ m}$ $(N_U  = 100).$ At an SNR of $0 \text{ dB}$, the CRLB starts out at about $30 \text{ m}$ $(N_U  = 4)$ and decreases to $9 \text{ m}$ $(N_U  = 100).$ At an SNR of $40 \text{ dB}$, we can achieve m-level accuracy irrespective of the number of receive antennas. Now in Fig. \ref{fig:Results/9D_PU_NU/_NB_3_N_K_3_N_U___delta_t_index_8_fcIndex_4_SNRIndex_1}, the achievable errors are smaller, we can achieve better than cm-level accuracy. Also, the errors at the lower SNR values are reduced. For instance, the error at an SNR of $-20 \text{ dB}$ with $N_U = 4$ is only $10 \text{ m}$ as opposed to $500 \text{ m}$ in Fig. \ref{fig:Results/9D_PU_NU/_NB_3_N_K_3_N_U___delta_t_index_5_fcIndex_4_SNRIndex_1}.  This improvement is because of the improvement in the GDOP attained due to the longer duration between transmission time slots. 

In Figs. \ref{Results:_NB_3_N_K_3_N_U_2_10_delta_t_index___fcIndex_1_SNRIndex_1deltat} and \ref{Results:_NB_3_N_K_3_N_U_2_10_delta_t_index___fcIndex_4_SNRIndex_1deltat}, the improvement in positioning error due to an increase in the length of the time interval between the transmission time slots is more clearly seen. This improvement in positioning error is due to the speed of the LEO satellites. The speed of satellites means that the same satellite can act as multiple anchors in different time slots while still achieving good GDOP. From Fig. \ref{Results:9D_PU_fc_NB_3_N_K_3_N_U_2_delta_t_index_1_3_fcIndex_1_SNRIndex_1_NU}, we notice that increasing the operating frequency causes a substantial reduction in the achievable error in estimating the $\bm{p}_{U}$.

 \subsection{Information available to find the $3$D velocity of the receiver when the $3$D position and $3$D orientation are unknown}
   \subsubsection{$N_K = 3$, $N_B = 3$, \textit{and} $N_U > 1$}

Even in the presence of both offsets, there is enough information to estimate: i) $\mathbf{J}_{ \bm{\bm{y}}; \bm{p}_{U}}^{\mathrm{e}}$, ii) $\mathbf{J}_{ \bm{\bm{y}}; \bm{\Phi}_{U}}^{\mathrm{e}}$, and iii) $\mathbf{J}_{ \bm{\bm{y}}; \bm{\Phi}_{U}}^{\mathrm{ee}}$. Hence, the loss in information about $\bm{v}_{U}$ due to the unknown $\bm{p}_{U}$ and $\bm{\Phi}_{U}$ which is specified by  $\mathbf{J}_{ \bm{\bm{y}}; \bm{v}_{U}}^{nu}$ exists. Further, since there is enough information to estimate $\mathbf{J}_{ \bm{\bm{y}}; \bm{v}_{U}}^{\mathrm{e}}$ then  $\mathbf{J}_{ \bm{\bm{y}}; \bm{v}_{U}}^{\mathrm{eee}}$ can also be estimated by Theorem \ref{theorem:FIM_9D_velocity}. Simulation results verify that there is enough information for $\mathbf{J}_{ \bm{\bm{y}}; \bm{v}_{U}}^{\mathrm{eee}}$ to be positive definite. Hence, $\bm{v}_{U}$ can be estimated in the presence of both offsets.

\subsubsection{Simulation results}
Here, we present simulation results for the CRLB when estimating $\bm{v}_{U}$ with $N_K = 3$, $N_B = 3$, \textit{and} $N_U > 1$. We consider the following SNR values $\{40 \text{ dB}, 20 \text{ dB}, 0 \text{ dB}, -20 \text{ dB}\}$ and different number of receive antennas. In Fig. \ref{fig:Results/9D_VU_NU/_NB_3_N_K_3_N_U___delta_t_index_5_8_fcIndex_1_SNRIndex__}, the operating frequency is $1 \text{ GHz}$ with $\Delta_t = 10 \text{ s} $ in Fig. \ref{fig:Results/9D_VU_NU/_NB_3_N_K_3_N_U___delta_t_index_5_fcIndex_1_SNRIndex__} and  $\Delta_t = 80 \text{ s} $ in Fig. \ref{fig:Results/9D_VU_NU/_NB_3_N_K_3_N_U___delta_t_index_8_fcIndex_1_SNRIndex__}. In Fig. \ref{fig:Results/9D_VU_NU/_NB_3_N_K_3_N_U___delta_t_index_5_fcIndex_1_SNRIndex__}, at an SNR of $-20 \text{ dB}$, the CRLB starts out at about $20 \text{ km/s}$ $(N_U  = 4)$ and decreases to $8 \text{ km/s}$ $(N_U  = 100).$ At an SNR of $0 \text{ dB}$, the CRLB starts out at about $2 \text{ km/s}$ $(N_U  = 4)$ and decreases to $0.6 \text{ km/s}$ $(N_U  = 100).$ It is important to note that at an SNR of $40 \text{ dB}$, we can achieve m/s level accuracy irrespective of the number of receive antennas. Now in Fig. \ref{fig:Results/9D_VU_NU/_NB_3_N_K_3_N_U___delta_t_index_8_fcIndex_1_SNRIndex__}, the achievable errors are smaller. We can also achieve m/s-level accuracy irrespective of the number of receive antennas at high SNRs. However, the errors at the lower SNR values are reduced. For instance, the error at an SNR of $-20 \text{ dB}$ with $N_U = 4$ is only $1 \text{ km/s}$ as opposed to $20 \text{ km/s}$ in Fig. \ref{fig:Results/9D_VU_NU/_NB_3_N_K_3_N_U___delta_t_index_5_fcIndex_1_SNRIndex__}.  This improvement is because of the improvement in the GDOP attained due to the longer duration between transmission time slots. 

In Fig. \ref{fig:Results/9D_VU_NU/_NB_3_N_K_3_N_U___delta_t_index_5_8_fcIndex_4_SNRIndex__} at $f_c = 60 \text{ GHz}$, the error values are much lower. In fact, for similar parameters, the error values range from $1 \text{ km/s}$ to $10^{-1} \text{ m/s}$ in Fig. \ref{fig:Results/9D_VU_NU/_NB_3_N_K_3_N_U___delta_t_index_5_fcIndex_4_SNRIndex__}. This is in contrast to Fig. \ref{fig:Results/9D_VU_NU/_NB_3_N_K_3_N_U___delta_t_index_5_fcIndex_1_SNRIndex__} where the error values range from $20 \text{ km/s}$ to $5 \text{ m/s}$.  In Fig. \ref{fig:Results/9D_VU_NU/_NB_3_N_K_3_N_U___delta_t_index_5_fcIndex_4_SNRIndex__}, at an SNR of $-20 \text{ dB}$, the CRLB starts out at about $0.5 \text{ km/s}$ $(N_U  = 4)$ and decreases to $90 \text{ km/s}$ $(N_U  = 100).$ At an SNR of $0 \text{ dB}$, the CRLB starts out at about $50 \text{ m/s}$ $(N_U  = 4)$ and decreases to $9 \text{ m/s}$ $(N_U  = 100).$ At an SNR of $40 \text{ dB}$, we can achieve m-level accuracy irrespective of the number of receive antennas. Now in Fig. \ref{fig:Results/9D_VU_NU/_NB_3_N_K_3_N_U___delta_t_index_8_fcIndex_4_SNRIndex__}, the achievable errors are smaller. We can also achieve better than mm/s-level accuracy. Also, the errors at the lower SNR values are reduced. For instance, the error at an SNR of $-20 \text{ dB}$ with $N_U = 4$ is  $0.5 \text{ km/s}$ as opposed to $20 \text{ km/s}$ in Fig. \ref{fig:Results/9D_VU_NU/_NB_3_N_K_3_N_U___delta_t_index_5_fcIndex_1_SNRIndex__}.  This improvement is because of the improvement in the GDOP attained due to the longer duration between transmission time slots. 

In Figs. \ref{Results/9D_VU_Time/_NB_3_N_K_3_N_U_2_10_delta_t_index___fcIndex_1_SNRIndex__} and \ref{Results/9D_VU_Time/_NB_3_N_K_3_N_U_2_10_delta_t_index___fcIndex_4_SNRIndex__}, the improvement in the velocity error due to an increase in the length of the time interval between the transmission time slots is more clearly seen. This improvement in the velocity error is due to the speed of the LEO satellites. The speed of satellites means that the same satellite can act as multiple anchors in different time slots while still achieving good GDOP. From Fig. \ref{fig:Results/9D_VU_fc/_NB_3_N_K_3_N_U_8_delta_t_index_5_8_fcIndex___SNRIndex__}, we notice that increasing the operating frequency causes a substantial reduction in the achievable error in estimating $\bm{v}_{U}$.

 \subsection{Information available to find the $3$D orientation of the receiver when the $3$D position and $3$D velocity are unknown}
  \subsubsection{$N_K = 3$, $N_B = 3$, \textit{and} $N_U > 1$}

Even in the presence of both offsets, there is enough information to estimate: i) $\mathbf{J}_{ \bm{\bm{y}}; \bm{p}_{U}}^{\mathrm{e}}$, ii) $\mathbf{J}_{ \bm{\bm{y}}; \bm{v}_{U}}^{\mathrm{e}}$, and iii) $\mathbf{J}_{ \bm{\bm{y}}; \bm{v}_{U}}^{\mathrm{ee}}$. Hence, the loss in information about $\bm{\Phi}_{U}$ due to the unknown $\bm{p}_{U}$ and $\bm{v}_{U}$ which is specified by  $\mathbf{J}_{ \bm{\bm{y}}; \bm{\Phi}_{U}}^{nu}$ exists. Further, since there is enough information to estimate $\mathbf{J}_{ \bm{\bm{y}}; \bm{\Phi}_{U}}^{\mathrm{e}}$ then  $\mathbf{J}_{ \bm{\bm{y}}; \bm{\Phi}_{U}}^{\mathrm{eee}}$ can also be estimated by Theorem \ref{theorem:FIM_9D_orientation}. Simulation results verify that there is enough information for $\mathbf{J}_{ \bm{\bm{y}}; \bm{\Phi}_{U}}^{\mathrm{eee}}$ to be positive definite. Hence, $\bm{\Phi}_{U}$ can be estimated in the presence of both offsets.

\subsubsection{Simulation results}
Here, we present simulation results for the CRLB when estimating $\bm{\Phi}_{U}$ with $N_K = 3$, $N_B = 3$, \textit{and} $N_U > 1$. We consider the following SNR values $\{40 \text{ dB}, 20 \text{ dB}, 0 \text{ dB}, -20 \text{ dB}\}$ and different number of receive antennas. In Fig. \ref{fig:Results/9D_PHIU_NU/_NB_3_N_K_3_N_U___delta_t_index_5_8_fcIndex_1_SNRIndex__}, the operating frequency is $1 \text{ GHz}$ with $\Delta_t = 10 \text{ s} $ in Fig. \ref{Results/9D_PHIU_NU/_NB_3_N_K_3_N_U___delta_t_index_5_fcIndex_1_SNRIndex__} and  $\Delta_t = 80 \text{ s} $ in Fig. \ref{Results/9D_PHIU_NU/_NB_3_N_K_3_N_U___delta_t_index_8_fcIndex_1_SNRIndex__}. In Fig. \ref{Results/9D_PHIU_NU/_NB_3_N_K_3_N_U___delta_t_index_5_fcIndex_1_SNRIndex__}, at an SNR of $-20 \text{ dB}$, the CRLB starts out at about $500 \text{ rad}$ $(N_U  = 4)$ and decreases to $10 \text{ rad}$ $(N_U  = 100).$ At an SNR of $0 \text{ dB}$, the CRLB starts out at about $20 \text{ rad}$ $(N_U  = 4)$ and decreases to $1 \text{ rad}$ $(N_U  = 100).$ At an SNR of $40 \text{ dB}$, we can achieve  accuracy on the order of $10^{-1}$ irrespective of the number of receive antennas. Now, in Fig. \ref{fig:Results/9D_PHIU_NU/_NB_3_N_K_3_N_U___delta_t_index_5_8_fcIndex_4_SNRIndex__}, the achievable errors are smaller. At high SNRs, we can also achieve rad errors on the order of $10^{-2}$ irrespective of the number of receive antennas. Also, the errors at the lower SNR values are reduced.   This improvement is because of the improvement in the GDOP attained due to the longer duration between transmission time slots. The errors in Fig. \ref{fig:Results/9D_PHIU_NU/_NB_3_N_K_3_N_U___delta_t_index_5_8_fcIndex_1_SNRIndex__} and Fig. \ref{fig:Results/9D_PHIU_NU/_NB_3_N_K_3_N_U___delta_t_index_5_8_fcIndex_4_SNRIndex__} are identical indicating that the operating frequency does not affect the achievable orientation error.

In Figs. \ref{fig:Results/9D_PHIU_Time/_NB_3_N_K_3_N_U_2_10_delta_t_index___fcIndex_1_SNRIndex__} and \ref{fig:Results/9D_PHIU_Time/_NB_3_N_K_3_N_U_2_10_delta_t_index___fcIndex_4_SNRIndex__}, the improvement in orientation error due to an increase in the length of the time interval between the transmission time slots is more clearly seen. This improvement in orientation error is due to the speed of the LEO satellites. The speed of satellites means that the same satellite can act as multiple anchors in different time slots while still achieving good GDOP. From Fig. \ref{fig:Results/9D_PHIU_fc/_NB_3_N_K_3_N_U_8_delta_t_index_5_8_fcIndex___SNRIndex__}, we notice that increasing the operating frequency does not change the achievable orientation error.

\section{Conclusion}
This paper has presented an analytical investigation of the available information that is obtainable in the signals from multiple LEOs during different transmission time slots received on a multiple antenna receiver and demonstrated the utility of these signals for $9$D localization ($3$D position, $3$D orientation, and $3$D velocity estimation). The first stage in our analysis involved deriving the FIM for the channel parameters present in the signals received from LEOs in the same or multiple constellations during multiple transmission time slots. This first stage was enabled by defining a system model that captures i) the possibility of a time offset between LEOs caused by having cheap synchronization clocks, ii) the possibility of a frequency offset between LEOs, iii) the unknown Doppler rate caused by the short coherence time in high mobility LEO based satellite systems, and iv) multiple transmission time slots from a particular LEO. The FIM for the channel parameters are transformed into the FIM for the location parameters. To enable this transformation, we started with the $3$D localization cases: i) $3$D positioning with known $3$D velocity and known $3$D orientation, ii) $3$D orientation estimation with known $3$D position and known $3$D velocity, and iii) $3$D velocity estimation with known $3$D position and known $3$D orientation. Subsequently, we derived the FIM for the $9$D localization ($3$D position, $3$D orientation, and $3$D velocity estimation) in terms of the FIM for the $3$D localization. With these derivations, we presented the number of LEOs, the operating frequency, the number of transmission time slots, and the number of receive antennas that allow for different levels of location estimation. One key result is that in the presence of time and frequency offsets and Doppler rate, it is possible to perform $9$D localization ($3$D position, $3$D velocity, and $3$D orientation estimation) of a receiver by utilizing the signals from three LEO satellites observed during three transmission time slots received through multiple receive antennas.

\appendix 
	 \label{appendix_channel} 

\subsection{Entries in transformation matrix}
\label{Appendix_Entries_in_transformation_matrix}
The derivative of the delay from the $u^{\text{th}}$ receive antenna to the $b^{\text{th}}$ LEO satellite during the $k^{\text{th}}$ time slot with respect to the $\bm{p}_{U}$  is simply the unit vector from the $u^{\text{th}}$ receive antenna to the $b^{\text{th}}$ LEO satellite during the $k^{\text{th}}$ time slot normalized by the speed of light. This is represented below
$$
\nabla_{\bm{p}_{U}}\tau_{bu,k} \triangleq  \nabla_{\bm{p}_{U}}\frac{\left\|\mathbf{p}_{u,k}-\mathbf{p}_{b,k}\right\|}{c} = \nabla_{\bm{p}_{U}}\frac{\bm{d}_{bu,k}}{c} = \frac{\bm{\Delta}_{bu,k}}{c}.
$$
The derivative of the Doppler observed at the receiver measured from the $b^{\text{th}}$ LEO satellite during the $k^{\text{th}}$ time slot with respect to  $\bm{p}_{U}$  is 
$$
\nabla_{\bm{p}_{U}}\nu_{b,k} \triangleq \frac{(\bm{v}_{B} - \bm{v}_{U}) - \bm{\Delta}_{bU,k}^{\mathrm{T}}(\bm{v}_{B} - \bm{v}_{U})\bm{\Delta}_{bU,k}}{c\bm{d}_{bU,k}^{-1}}.
$$
The derivative of the delay from the $u^{\text{th}}$ receive antenna to the $b^{\text{th}}$ LEO satellite during the $k^{\text{th}}$ time slot with respect to $\alpha_{U}$  is
$$
\nabla_{{\alpha}_{U}}\tau_{bu,k} \triangleq   \frac{\bm{\Delta}_{bu,k}^{\mathrm{T}} \nabla_{\alpha_{U}}\bm{Q}_{U}\Tilde{\bm{s}}_{u}}{c}.
$$
The derivative of the delay from the $u^{\text{th}}$ receive antenna to the $b^{\text{th}}$ LEO satellite during the $k^{\text{th}}$ time slot with respect to  $\psi_{U}$  is
$$
\nabla_{{\psi}_{U}}\tau_{bu,k} \triangleq   \frac{\bm{\Delta}_{bu,k}^{\mathrm{T}} \nabla_{\psi_{U}}\bm{Q}_{U}\Tilde{\bm{s}}_{u}}{c}.
$$
The derivative of the delay from the $u^{\text{th}}$ receive antenna to the $b^{\text{th}}$ LEO satellite during the $k^{\text{th}}$ time slot with respect to  $\varphi_{U}$  is
$$
\nabla_{{\varphi}_{U}}\tau_{bu,k} \triangleq   \frac{\bm{\Delta}_{bu,k}^{\mathrm{T}} \nabla_{\varphi_{U}}\bm{Q}_{U}\Tilde{\bm{s}}_{u}}{c}.
$$
The derivative of the delay from the $u^{\text{th}}$ receive antenna to the $b^{\text{th}}$ LEO satellite during the $k^{\text{th}}$ time slot with respect to  $\bm{\Phi}_{U}$  is
$$
\nabla_{\bm{\Phi}_{U}}\tau_{bu,k} \triangleq   
\left[\begin{array}{c}
{\bm{\Delta}_{bu,k}^{\mathrm{T}} \nabla_{\alpha_{U}}\bm{Q}_{U}\Tilde{\bm{s}}_{u}} \\
{\bm{\Delta}_{bu,k}^{\mathrm{T}} \nabla_{\psi_{U}}\bm{Q}_{U}\Tilde{\bm{s}}_{u}} \\
{\bm{\Delta}_{bu,k}^{\mathrm{T}} \nabla_{\varphi_{U}}\bm{Q}_{U}\Tilde{\bm{s}}_{u}}
\end{array}\right]
$$
The derivative of the delay from the $u^{\text{th}}$ receive antenna to the $b^{\text{th}}$ LEO satellite during the $k^{\text{th}}$ time slot with respect to $\bm{v}_{U}$  is
$$
\nabla_{\bm{v}_{U}}\tau_{bu,k} \triangleq   (k - 1) \Delta_{t}\frac{\bm{\Delta}_{bu,k}}{c}.
$$
The derivative of the Doppler observed at the receiver measured from the $b^{\text{th}}$ LEO satellite during the $k^{\text{th}}$ time slot with respect to  $\bm{v}_{U}$  is 
$$
\nabla_{\bm{v}_{U}}\nu_{b,k} \triangleq   -\frac{\bm{\Delta}_{bU,k}}{c}.
$$

\subsection{Elements in $\mathbf{J}_{ \bm{\bm{y}}; \bm{\kappa}_1}^{}$}
\label{Appendix_Entries_in_the_EFIM_Location_1}
The proof of the elements in $\mathbf{J}_{ \bm{\bm{y}}; \bm{\kappa}_1}^{}$ are presented in this section. We start with the elements that are related to the $3D$ position.
\subsubsection{Proof of the FIM related to the $3$D Position of the receiver}
\label{Appendix_lemma_FIM_3D_position}
The FIM of the $3$D position of the receiver can be written as
$$
\begin{aligned}
\bm{F}_{{{y} }}(\bm{y}_{}| \bm{\eta} ;\bm{p}_{U},\bm{p}_{U}) 
=& \sum_{b,k^{'},u^{'},k^{},u^{}}\bm{\Delta}_{bu,k} \bm{F}_{{{y} }}({y}_{bu,k}| {\eta} ;{\tau}_{bu,k},{\tau}_{bu^{'},k^{'}}) \Delta_{bu^{'},k^{'}}^{\mathrm{T}} +  \bm{\Delta}_{bu^{},k^{}} \bm{F}_{{{y} }}({y}_{bu,k}| {\eta} ;{\tau}_{bu^{},k^{}},{\nu}_{b,k^{'}}) \nabla_{\bm{p}_{U}}^{\mathrm{T}} \nu_{b,k^{'}} \\ &+ \nabla_{\bm{p}_{U}} \nu_{b,k} \bm{F}_{{{y} }}({y}_{bu,k}| {\eta} ;{\nu}_{b,k},{\tau}_{bu^{'},k^{'}}) \Delta_{bu^{'},k^{'}}^{\mathrm{T}}
+ \nabla_{\bm{p}_{U}} \nu_{b,k} \bm{F}_{{{y} }}({y}_{bu,k}| {\eta} ;{\nu}_{b,k},{\nu}_{b^{'},k^{'}}) \nabla_{\bm{p}_{U}}^{\mathrm{T}} \nu_{b^{},k^{'}}.
\end{aligned}
$$
Now due to independent measurements across receive antennas and time slots, and the fact that $ \bm{F}_{{{y} }}({y}_{bu,k}| {\eta} ;{\nu}_{b,k},{\tau}_{bu^{},k^{}}) = 0, \; \; \forall b, \forall u, \forall k,$ we have
$$
\begin{aligned}
&\bm{F}_{{{y} }}(\bm{y}_{}| \bm{\eta} ;\bm{p}_{U},\bm{p}_{U}) = \sum_{b,k^{},u^{}}\bm{\Delta}_{bu,k} \bm{F}_{{{y} }}({y}_{bu,k}| {\eta} ;{\tau}_{bu,k},{\tau}_{bu^{},k^{}}) \bm{\Delta}_{bu^{},k^{}}^{\mathrm{T}}  + \nabla_{\bm{p}_{U}} \nu_{b,k} \bm{F}_{{{y} }}({y}_{bu,k}| {\eta} ;{\nu}_{b,k},{\nu}_{b^{},k^{}}) \nabla_{\bm{p}_{U}}^{\mathrm{T}} \nu_{b,k},
\end{aligned}
$$
with appropriate substitutions, we have
$$
\begin{aligned}
\bm{F}_{{{y} }}(\bm{y}_{}| \bm{\eta} ;\bm{p}_{U},\bm{p}_{U}) = \sum_{b,k^{},u^{}} \underset{bu,k}{\operatorname{SNR}}  \Bigg[ \frac{\omega_{b,k}}{c^2}   \bm{\Delta}_{bu,k} 
 \bm{\Delta}_{bu^{},k^{}}^{\mathrm{T}}  + \frac{f_{c}^2 t_{obu,k}^2}{2} \nabla_{\bm{p}_{U}} \nu_{b,k} \nabla_{\bm{p}_{U}}^{\mathrm{T}} \nu_{b,k}\Bigg].
\end{aligned}
$$

\subsubsection{Proof of the FIM relating to the $3$D Position and $3$D orientation of the receiver}
\label{Appendix_lemma_FIM_3D_position_3D_orientation}
The FIM relating the $3$D position and $3$D orientation of the receiver is
$$
\begin{aligned}
\bm{F}_{{{y} }}(\bm{y}_{}| \bm{\eta} ;\bm{p}_{U},\bm{\Phi}_{U}) 
= \sum_{b,k^{'},u^{'},k^{},u^{}}\bm{\Delta}_{bu,k} \bm{F}_{{{y} }}({y}_{bu,k}| {\eta} ;{\tau}_{bu,k},{\tau}_{bu^{'},k^{'}}) \nabla_{\bm{\Phi}_{U}}^{\mathrm{T}}{\tau}_{bu^{'},k^{'}} + \nabla_{\bm{p}_{U}} \nu_{b,k} \bm{F}_{{{y} }}({y}_{bu,k}| {\eta} ;{\nu}_{b,k},{\tau}_{bu^{'},k^{'}}) \nabla_{\bm{\Phi}_{U}}^{\mathrm{T}}{\tau}_{bu^{'},k^{'}}.
\end{aligned}
$$
Now due to independent measurements across receive antennas and time slots, and the fact that $ \bm{F}_{{{y} }}({y}_{bu,k}| {\eta} ;{\nu}_{b,k},{\tau}_{bu^{},k^{}}) = 0, \; \; \forall b, \forall u, \forall k,$ we have
$$
\begin{aligned}
\bm{F}_{{{y} }}(\bm{y}_{}| \bm{\eta} ;\bm{p}_{U},\bm{\Phi}_{U}) 
= \sum_{b,k^{'},u^{'},k^{},u^{}}\bm{\Delta}_{bu,k} \bm{F}_{{{y} }}({y}_{bu,k}| {\eta} ;{\tau}_{bu,k},{\tau}_{bu^{'},k^{'}}) \nabla_{\bm{\Phi}_{U}}^{\mathrm{T}}{\tau}_{bu^{'},k^{'}} 
\end{aligned}
$$
with appropriate substitutions, we have
$$
\begin{aligned}
&{\bm{F}_{{{y} }}(\bm{y}_{}| \bm{\eta} ;\bm{p}_{U},\bm{\Phi}_{U}) = }  {\sum_{b,k^{},u^{}} \underset{bu,k}{\operatorname{SNR}}  \Bigg[ \frac{\omega_{b,k}}{c}   \bm{\Delta}_{bu,k} 
 \nabla_{\bm{\Phi}_{U}}^{\mathrm{T}} \tau_{bu^{},k^{}}}    \Bigg].
\end{aligned}
$$

\subsubsection{Proof of the FIM relating to the $3$D Position and $3$D velocity of the receiver}
\label{Appendix_lemma_FIM_3D_position_3D_velocity}
The FIM relating the $3$D position and $3$D velocity of the receiver is
$$
\begin{aligned}
\bm{F}_{{{y} }}(\bm{y}_{}| \bm{\eta} ;\bm{p}_{U},\bm{v}_{U}) 
&= \sum_{b,k^{'},u^{'},k^{},u^{}}\bm{\Delta}_{bu,k}  \bm{F}_{{{y} }}({y}_{bu,k}| {\eta} ;{\tau}_{bu,k},{\tau}_{bu^{'},k^{'}}) \nabla_{\bm{v}_{U}}^{\mathrm{T}} \tau_{bu^{'},k^{'}} + \bm{\Delta}_{bu,k}  \bm{F}_{{{y} }}({y}_{bu,k}| {\eta} ;{\tau}_{bu^{},k^{}},{\nu}_{b,k^{'}}) \nabla_{\bm{v}_{U}}^{\mathrm{T}} \nu_{b,k^{'}} \\ &+ \nabla_{\bm{p}_{U}} \nu_{b,k} \bm{F}_{{{y} }}({y}_{bu,k}| {\eta} ;{\nu}_{b,k},{\tau}_{bu^{'},k^{'}}) \nabla_{\bm{v}_{U}}^{\mathrm{T}} \tau_{bu^{'},k^{'}}
+ \nabla_{\bm{p}_{U}} \nu_{b,k} \bm{F}_{{{y} }}({y}_{bu,k}| {\eta} ;{\nu}_{b,k},{\nu}_{b^{'},k^{'}}) \nabla_{\bm{v}_{U}}^{\mathrm{T}} \nu_{b^{},k^{'}}.
\end{aligned}
$$
Now due to independent measurements across receive antennas and time slots, and the fact that $ \bm{F}_{{{y} }}({y}_{bu,k}| {\eta} ;{\nu}_{b,k},{\tau}_{bu^{},k^{}}) = 0, \; \; \forall b, \forall u, \forall k,$ we have
$$
\begin{aligned}
\bm{F}_{{{y} }}(\bm{y}_{}| \bm{\eta} ;\bm{p}_{U},\bm{v}_{U}) \\ &= \sum_{b,k^{},u^{}}\bm{\Delta}_{bu,k} \bm{F}_{{{y} }}({y}_{bu,k}| {\eta} ;{\tau}_{bu,k},{\tau}_{bu^{},k^{}}) \nabla_{\bm{v}_{U}}^{\mathrm{T}} \tau_{bu^{},k^{}} + \nabla_{\bm{p}_{U}} \nu_{b,k} \bm{F}_{{{y} }}({y}_{bu,k}| {\eta} ;{\nu}_{b,k},{\nu}_{b^{},k^{}}) \nabla_{\bm{v}_{U}}^{\mathrm{T}} \nu_{b,k},
\end{aligned}
$$
with appropriate substitutions, we have
$$
\begin{aligned}
{\bm{F}_{{{y} }}(\bm{y}_{}| \bm{\eta} ;\bm{p}_{U},\bm{v}_{U}) = } 
\medmath{{\sum_{b,k^{},u^{}} \underset{bu,k}{\operatorname{SNR}}  \Bigg[ \frac{(k - 1) \omega_{b,k}\Delta_{t}}{c^2}    \bm{\Delta}_{bu,k} 
\bm{\Delta}_{bu,k}^{\mathrm{T}} }    - {\frac{f_{c}^2 t_{obu,k}^2\nabla_{\bm{p}_{U}} \nu_{b,k} \bm{\Delta}_{bU,k}^{\mathrm{T}}}{2 c}  \Bigg]}}.
\end{aligned}
$$

\subsubsection{Proof of the FIM related to the $3$D orientation of the receiver}
\label{Appendix_lemma_FIM_3D_orientation}
The FIM of the $3$D orientation of the receiver can be written as
$$
\begin{aligned}
\bm{F}_{{{y} }}(\bm{y}_{}| \bm{\eta} ;\bm{\Phi}_{U},\bm{\Phi}_{U}) 
= \sum_{b,k^{'},u^{'},k^{},u^{}} \nabla_{\bm{\Phi}_{U}}{\tau}_{bu^{},k^{}}  \bm{F}_{{{y} }}({y}_{bu,k}| {\eta} ;{\tau}_{bu,k},{\tau}_{bu^{'},k^{'}}) \nabla_{\bm{\Phi}_{U}}^{\mathrm{T}}{\tau}_{bu^{'},k^{'}}.
\end{aligned}
$$
Due to independent measurements across receive antennas and time slots, and with appropriate substitutions, we have
$$
{\bm{F}_{{{y} }}(\bm{y}_{}| \bm{\eta} ;\bm{\Phi}_{U},\bm{\Phi}_{U}) = } {\sum_{b,k^{},u^{}} \underset{bu,k}{\operatorname{SNR}}  \Bigg[ \omega_{b,k}     \nabla_{\bm{\Phi}_{U}} \tau_{bu^{},k^{}}
 \nabla_{\bm{\Phi}_{U}}^{\mathrm{T}} \tau_{bu^{},k^{}}}    \Bigg].
$$
\subsubsection{Proof of the FIM relating to the $3$D orientation and $3$D velocity of the receiver}
\label{Appendix_lemma_FIM_3D_orientation_3D_velocity}
The FIM relating the $3$D orientation and $3$D velocity of the receiver is
$$
\begin{aligned}
\bm{F}_{{{y} }}(\bm{y}_{}| \bm{\eta} ;\bm{\Phi}_{U},\bm{v}_{U}) 
&= \sum_{b,k^{'},u^{'},k^{},u^{}} \nabla_{\bm{\Phi}_{U}}{\tau}_{bu,k}  \bm{F}_{{{y} }}({y}_{bu,k}| {\eta} ;{\tau}_{bu,k},{\tau}_{bu^{'},k^{'}}) \nabla_{\bm{v}_{U}}^{\mathrm{T}} \tau_{bu^{'},k^{'}} +   \nabla_{\bm{\Phi}_{U}}{\tau}_{bu,k}  \bm{F}_{{{y} }}({y}_{bu,k}| {\eta} ;{\tau}_{bu^{},k^{}},{\nu}_{b,k^{'}}) \nabla_{\bm{v}_{U}}^{\mathrm{T}} \nu_{b,k^{'}} \\ &+ \nabla_{\bm{\Phi}_{U}} \nu_{b,k} \bm{F}_{{{y} }}({y}_{bu,k}| {\eta} ;{\nu}_{b,k},{\tau}_{bu^{'},k^{'}}) \nabla_{\bm{v}_{U}}^{\mathrm{T}} \tau_{bu^{'},k^{'}}
+ \nabla_{\bm{\Phi}_{U}} \nu_{b,k} \bm{F}_{{{y} }}({y}_{bu,k}| {\eta} ;{\nu}_{b,k},{\nu}_{b^{'},k^{'}}) \nabla_{\bm{v}_{U}}^{\mathrm{T}} \nu_{b^{},k^{'}}.
\end{aligned}
$$
Now due to independent measurements across receive antennas and time slots, and the fact that $ \bm{F}_{{{y} }}({y}_{bu,k}| {\eta} ;{\nu}_{b,k},{\tau}_{bu^{},k^{}}) = 0, \; \; \forall b, \forall u, \forall k,$ we have
$$
\begin{aligned}
\bm{F}_{{{y} }}(\bm{y}_{}| \bm{\eta} ;\bm{\Phi}_{U},\bm{v}_{U}) = \sum_{b,k^{},u^{}}  \nabla_{\bm{\Phi}_{U}}{\tau}_{bu,k
}\bm{F}_{{{y} }}({y}_{bu,k}| {\eta} ;{\tau}_{bu,k},{\tau}_{bu^{},k^{}}) \nabla_{\bm{v}_{U}}^{\mathrm{T}} \tau_{bu^{},k^{}},
\end{aligned}
$$
with appropriate substitutions, we have
$$
\begin{aligned}
{\bm{F}_{{{y} }}(\bm{y}_{}| \bm{\eta} ;\bm{\Phi}_{U},\bm{v}_{U})  } = {\sum_{b,k^{},u^{}} \underset{bu,k}{\operatorname{SNR}}  \Bigg[ \frac{(k - 1) \Delta_{t}^{}\omega_{b,k}}{c}     \nabla_{\bm{\Phi}_{U}} \tau_{bu^{},k^{}}
 \bm{\Delta}_{bu,k}^{\mathrm{T}} }    \Bigg].
\end{aligned}
$$
\subsubsection{Proof of the FIM related to the $3$D velocity of the receiver}
\label{Appendix_lemma_FIM_3D_velocity}
The FIM of the $3$D velocity of the receiver can be written as
$$
\begin{aligned}
\bm{F}_{{{y} }}(\bm{y}_{}| \bm{\eta} ;\bm{v}_{U},\bm{v}_{U}) 
\\&= \sum_{b,k^{'},u^{'},k^{},u^{}} \nabla_{\bm{v}_{U}} \tau_{bu,k^{}} \bm{F}_{{{y} }}({y}_{bu,k}| {\eta} ;{\tau}_{bu,k},{\tau}_{bu^{'},k^{'}}) \nabla_{\bm{v}_{U}}^{\mathrm{T}} \tau_{bu^{'},k^{'}} +  \nabla_{\bm{v}_{U}} \tau_{bu,k^{}}  \bm{F}_{{{y} }}({y}_{bu,k}| {\eta} ;{\tau}_{bu^{},k^{}},{\nu}_{b,k^{'}}) \nabla_{\bm{v}_{U}}^{\mathrm{T}} \nu_{b,k^{'}} \\ &+ \nabla_{\bm{v}_{U}} \nu_{b,k} \bm{F}_{{{y} }}({y}_{bu,k}| {\eta} ;{\nu}_{b,k},{\tau}_{bu^{'},k^{'}}) \nabla_{\bm{v}_{U}}^{\mathrm{T}} \tau_{bu^{'},k^{'}}
+ \nabla_{\bm{v}_{U}} \nu_{b,k} \bm{F}_{{{y} }}({y}_{bu,k}| {\eta} ;{\nu}_{b,k},{\nu}_{b^{'},k^{'}}) \nabla_{\bm{v}_{U}}^{\mathrm{T}} \nu_{b^{},k^{'}}.
\end{aligned}
$$
Now due to independent measurements across receive antennas and time slots, and the fact that $ \bm{F}_{{{y} }}({y}_{bu,k}| {\eta} ;{\nu}_{b,k},{\tau}_{bu^{},k^{}}) = 0, \; \; \forall b, \forall u, \forall k,$ we have
$$
\begin{aligned}
\bm{F}_{{{y} }}(\bm{y}_{}| \bm{\eta} ;\bm{v}_{U},\bm{v}_{U}) = \sum_{b,k^{},u^{}} \nabla_{\bm{v}_{U}} \tau_{bu,k^{}} \bm{F}_{{{y} }}({y}_{bu,k}| {\eta} ;{\tau}_{bu,k},{\tau}_{bu^{},k^{}}) \nabla_{\bm{v}_{U}}^{\mathrm{T}} \tau_{bu,k^{}}  + \nabla_{\bm{v}_{U}} \nu_{b,k} \bm{F}_{{{y} }}({y}_{bu,k}| {\eta} ;{\nu}_{b,k},{\nu}_{b^{},k^{}}) \nabla_{\bm{v}_{U}}^{\mathrm{T}} \nu_{b,k},
\end{aligned}
$$
with appropriate substitutions, we have
$$
\begin{aligned}
{\bm{F}_{{{y} }}(\bm{y}_{}| \bm{\eta} ;\bm{v}_{U},\bm{v}_{U}) = } \medmath{\sum_{b,k^{},u^{}} \underset{bu,k}{\operatorname{SNR}}  \Bigg[ \frac{(k - 1)^2 \Delta_{t}^{2}\omega_{b,k}}{c^2}     \bm{\Delta}_{bu,k} 
\bm{\Delta}_{bu,k}^{\mathrm{T}} }  + \medmath{\frac{f_{c}^2 t_{obu,k}^2\bm{\Delta}_{bU,k}\bm{\Delta}_{bU,k}^{\mathrm{T}}}{2 c^2}  \Bigg]}.
\end{aligned}
$$

\subsection{Elements in $\mathbf{J}_{ \bm{\bm{y}}; \bm{\kappa}_1}^{nu}$}
\label{Appendix_subsection:information_loss_FIM_channel_parameters}
The proof of the expressions for the elements in $\mathbf{J}_{ \bm{\bm{y}}; \bm{\kappa}_1}^{nu}$ are presented in this section. The definition of the information loss provided in Definition \ref{definition_EFIM} is $$\mathbf{J}_{ \bm{\bm{y}}; \bm{\kappa}_1}^{nu}  = \mathbf{J}_{ \bm{\bm{y}}; \bm{\kappa}_1, \bm{\kappa}_2}^{} \mathbf{J}_{ \bm{\bm{y}}; \bm{\kappa}_2}^{-1} \mathbf{J}_{ \bm{\bm{y}}; \bm{\kappa}_1, \bm{\kappa}_2}^{\mathrm{T}}.$$
With this definition, we start with the loss of information about the $3$D position of the receiver due to uncertainty in the nuisance parameters ${\bm{\kappa}_{2}}$. 

\subsubsection{Proof of Lemma \ref{lemma:information_loss_FIM_3D_position}}
\label{Appendix_subsection:information_loss_3D_position_3D_position}
By expanding the definition of the loss in information and taking the appropriate submatrix, we get the first equality in (\ref{appendix_equ:lemma_nuisance_3D_position_3D_position}). With appropriate substitutions, we get the second equality. The second equality is equivalent to the expression in Lemma \ref{lemma:information_loss_FIM_3D_position}.

\subsubsection{Proof of Lemma \ref{lemma:information_loss_FIM_3D_position_3D_orientation}}
\label{Appendix_subsection:information_loss_3D_position_3D_orientation}
By expanding the definition of the loss in information and taking the appropriate submatrix, we get the first equality in (\ref{appendix_equ:lemma_nuisance_3D_position_3D_orientation}). With appropriate substitutions, we get the second equality. The second equality is equivalent to the expression in Lemma \ref{lemma:information_loss_FIM_3D_position_3D_orientation}.

\subsubsection{Proof of Lemma \ref{lemma:information_loss_FIM_3D_position_3D_veloctiy}}
\label{Appendix_subsection:information_loss_3D_position_3D_velocity}
Again, by expanding the definition of the loss in information and taking the appropriate submatrix, we get the first equality in (\ref{appendix_equ:lemma_nuisance_3D_position_3D_velocity}). With appropriate substitutions, we get the second equality. The second equality is equivalent to the expression in Lemma \ref{lemma:information_loss_FIM_3D_position_3D_veloctiy}.

\subsubsection{Proof of Lemma \ref{lemma:information_loss_FIM_3D_orientation_3D_orientation}}
\label{Appendix_subsection:information_loss_3D_orientation_3D_orientation}
Expanding the definition of the loss in information and taking the appropriate submatrix, we get the first equality in (\ref{appendix_equ:lemma_nuisance_3D_orientation_3D_orientation}). With appropriate substitutions, we get the second equality. The second equality is equivalent to the expression in Lemma \ref{lemma:information_loss_FIM_3D_orientation_3D_orientation}.

\subsubsection{Proof of Lemma \ref{lemma:information_loss_FIM_3D_orientation_3D_veloctiy}}
\label{Appendix_subsection:information_loss_3D_orientation_3D_velocity}
Expanding the definition of the loss in information and taking the appropriate submatrix, we get the first equality in (\ref{appendix_equ:lemma_nuisance_3D_orientation_3D_velocity}). With appropriate substitutions, we get the second equality. The second equality is equivalent to the expression in Lemma \ref{lemma:information_loss_FIM_3D_orientation_3D_veloctiy}.

\subsubsection{Proof of Lemma \ref{lemma:information_loss_FIM_3D_velocity_3D_velocity}}
\label{Appendix_subsection:information_loss_3D_velocity_3D_velocity}
Expanding the definition of the loss in information and taking the appropriate submatrix, we get the first equality in (\ref{appendix_equ:lemma_nuisance_3D_velocity_3D_velocity}). With appropriate substitutions, we get the second equality. The second equality is equivalent to the expression in Lemma \ref{lemma:information_loss_FIM_3D_velocity_3D_velocity}.

\begin{figure*}
\begin{align}
\begin{split}
\label{appendix_equ:lemma_nuisance_3D_position_3D_position}
&[\mathbf{J}_{ \bm{\bm{y}}; \bm{\kappa}_1}^{nu}]_{[1:3,1:3]} = \\ &\sum_{b}\Bigg[\Bigg[\sum_{k^{},u^{}} \Bigg[ \nabla_{\bm{p}_{U}} \tau_{bu,k^{}} \bm{F}_{{{y} }}({y}_{bu,k}| {\eta} ;{\tau}_{bu,k},{\beta}_{bu^{'},k^{'}}) \nabla_{{\beta}_{bu^{'},k^{'}}}{\beta}_{bu^{'},k^{'}} + \nabla_{\bm{p}_{U}} \nu_{b,k^{}} \bm{F}_{{{y} }}({y}_{bu,k}| {\eta} ;{\nu}_{b,k},{\beta}_{bu^{'},k^{'}}) \nabla_{{\beta}_{bu^{'},k^{'}}}{\beta}_{bu^{'},k^{'}}\Bigg]\Bigg]\\
&\times \Bigg[ \nabla_{{\beta}_{bu^{'},k^{'}}}{\beta}_{bu^{'},k^{'}} \bm{F}_{{{y} }}({y}_{bu,k}| {\eta} ;{\beta}_{bu^{'},k^{'}},{\beta}_{bu^{'},k^{'}}) \nabla_{{\beta}_{bu^{'},k^{'}}}{\beta}_{bu^{'},k^{'}}\Bigg]^{-1} \\
&\times \Bigg[\sum_{k^{},u^{}} \Bigg[ \nabla_{{\beta}_{bu^{'},k^{'}}}{\beta}_{bu^{'},k^{'}}  \bm{F}_{{{y} }}({y}_{bu,k}| {\eta} ;{\beta}_{bu^{'},k^{'}},{\tau}_{bu,k}) \nabla_{\bm{p}_{U}}^{\mathrm{T}} \tau_{bu,k^{}}  + \nabla_{{\beta}_{bu^{'},k^{'}}}{\beta}_{bu^{'},k^{'}}  \bm{F}_{{{y} }}({y}_{bu,k}| {\eta} ;{\beta}_{bu^{'},k^{'}},{\nu}_{b,k}) \nabla_{\bm{p}_{U}}^{\mathrm{T}} \nu_{b,k^{}}\Bigg]\Bigg] \Bigg] \\ 
&+
\sum_{b}\Bigg[\Bigg[\sum_{k^{},u^{}} \Bigg[ \nabla_{\bm{p}_{U}} \tau_{bu,k^{}} \bm{F}_{{{y} }}({y}_{bu,k}| {\eta} ;{\tau}_{bu,k},{\delta}_{b}) \nabla_{{\delta}_{b}}{\delta}_{b} + \nabla_{\bm{p}_{U}} \nu_{b,k^{}} \bm{F}_{{{y} }}({y}_{bu,k}| {\eta} ;{\nu}_{b,k},{\delta}_{b}) \nabla_{{\delta}_{b}}{\delta}_{b}\Bigg]\Bigg]\\
&\times \Bigg[\sum_{k^{},u^{}} \Bigg[ \nabla_{{\delta}_{b}}{\delta}_{b} \bm{F}_{{{y} }}({y}_{bu,k}| {\eta} ;{\delta}_{b},{\delta}_{b}) \nabla_{{\delta}_{b}}{\delta}_{b}\Bigg]\Bigg]^{-1} \\
&\times \Bigg[\sum_{k^{},u^{}} \Bigg[ \nabla_{{\delta}_{b}}{\delta}_{b}  \bm{F}_{{{y} }}({y}_{bu,k}| {\eta} ;{\delta}_{b},{\tau}_{bu,k}) \nabla_{\bm{p}_{U}}^{\mathrm{T}} \tau_{bu,k^{}}  + \nabla_{{\delta}_{b}}{\delta}_{b}  \bm{F}_{{{y} }}({y}_{bu,k}| {\eta} ;{\delta}_{b},{\nu}_{b,k}) \nabla_{\bm{p}_{U}}^{\mathrm{T}} \nu_{b,k^{}}\Bigg]\Bigg] \Bigg] \\ 
&+
\sum_{b}\Bigg[\Bigg[\sum_{k^{},u^{}} \Bigg[ \nabla_{\bm{p}_{U}} \tau_{bu,k^{}} \bm{F}_{{{y} }}({y}_{bu,k}| {\eta} ;{\tau}_{bu,k},{\epsilon}_{b}) \nabla_{{\epsilon}_{b}}{\epsilon}_{b} + \nabla_{\bm{p}_{U}} \nu_{b,k^{}} \bm{F}_{{{y} }}({y}_{bu,k}| {\eta} ;{\nu}_{b,k},{\epsilon}_{b}) \nabla_{{\epsilon}_{b}}{\epsilon}_{b}\Bigg]\Bigg]\\
&\times \Bigg[ \sum_{k,u}\Bigg[ \nabla_{{\epsilon}_{b}}{\epsilon}_{b} \bm{F}_{{{y} }}({y}_{bu,k}| {\eta} ;{\epsilon}_{b},{\epsilon}_{b}) \nabla_{{\epsilon}_{b}}{\epsilon}_{b}\Bigg]\Bigg]^{-1} \\
&\times \Bigg[\sum_{k^{},u^{}} \Bigg[ \nabla_{{\epsilon}_{b}}{\epsilon}_{b}  \bm{F}_{{{y} }}({y}_{bu,k}| {\eta} ;{\epsilon}_{b},{\tau}_{bu,k}) \nabla_{\bm{p}_{U}}^{\mathrm{T}} \tau_{bu,k^{}}  + \nabla_{{\epsilon}_{b}}{\epsilon}_{b}  \bm{F}_{{{y} }}({y}_{bu,k}| {\eta} ;{\epsilon}_{b},{\nu}_{b,k}) \nabla_{\bm{p}_{U}}^{\mathrm{T}} \nu_{b,k^{}}\Bigg]\Bigg] \Bigg] \\ 
&= \sum_{b}\Bigg[
 \Bigg[\sum_{k^{},u^{}} \bm{F}_{{{y} }}({y}_{bu,k}| {\eta} ;{\delta}_{b},{\delta}_{b}) \Bigg]^{-1} \norm{ \sum_{k^{},u^{}}   \bm{F}_{{{y} }}({y}_{bu,k}| {\eta} ;{\delta}_{b},{\tau}_{bu,k}) \nabla_{\bm{p}_{U}}^{\mathrm{T}} \tau_{bu,k^{}}  }^{2} \\ &+  \Bigg[ \sum_{k,u}  \bm{F}_{{{y} }}({y}_{bu,k}| {\eta} ;{\epsilon}_{b},{\epsilon}_{b})\Bigg]^{-1}\norm{\sum_{k^{},u^{}}   \bm{F}_{{{y} }}({y}_{bu,k}| {\eta} ;{\epsilon}_{b},{\nu}_{b,k}) \nabla_{\bm{p}_{U}}^{\mathrm{T}} \nu_{b,k^{}} }^2
\end{split}
\end{align}
\end{figure*}

\begin{figure*}
\begin{align}
\begin{split}
\label{appendix_equ:lemma_nuisance_3D_position_3D_orientation}
&[\mathbf{J}_{ \bm{\bm{y}}; \bm{\kappa}_1}^{nu}]_{[1:3,4:6]} =
\\ &\sum_{b}\Bigg[\Bigg[\sum_{k^{},u^{}} \Bigg[ \nabla_{\bm{p}_{U}} \tau_{bu,k^{}} \bm{F}_{{{y} }}({y}_{bu,k}| {\eta} ;{\tau}_{bu,k},{\beta}_{bu^{'},k^{'}}) \nabla_{{\beta}_{bu^{'},k^{'}}}{\beta}_{bu^{'},k^{'}} + \nabla_{\bm{p}_{U}} \nu_{b,k^{}} \bm{F}_{{{y} }}({y}_{bu,k}| {\eta} ;{\nu}_{b,k},{\beta}_{bu^{'},k^{'}}) \nabla_{{\beta}_{bu^{'},k^{'}}}{\beta}_{bu^{'},k^{'}}\Bigg]\Bigg]\\
&\times \Bigg[ \nabla_{{\beta}_{bu^{'},k^{'}}}{\beta}_{bu^{'},k^{'}} \bm{F}_{{{y} }}({y}_{bu,k}| {\eta} ;{\beta}_{bu^{'},k^{'}},{\beta}_{bu^{'},k^{'}}) \nabla_{{\beta}_{bu^{'},k^{'}}}{\beta}_{bu^{'},k^{'}}\Bigg]^{-1} \\
&\times \Bigg[\sum_{k^{},u^{}} \Bigg[ \nabla_{{\beta}_{bu^{'},k^{'}}}{\beta}_{bu^{'},k^{'}}  \bm{F}_{{{y} }}({y}_{bu,k}| {\eta} ;{\beta}_{bu^{'},k^{'}},{\tau}_{bu,k}) \nabla_{\bm{\Phi}_{U}}^{\mathrm{T}} \tau_{bu,k^{}} \Bigg]\Bigg] \Bigg] \\ 
&+
\sum_{b}\Bigg[\Bigg[\sum_{k^{},u^{}} \Bigg[ \nabla_{\bm{p}_{U}} \tau_{bu,k^{}} \bm{F}_{{{y} }}({y}_{bu,k}| {\eta} ;{\tau}_{bu,k},{\delta}_{b}) \nabla_{{\delta}_{b}}{\delta}_{b} + \nabla_{\bm{p}_{U}} \nu_{b,k^{}} \bm{F}_{{{y} }}({y}_{bu,k}| {\eta} ;{\nu}_{b,k},{\delta}_{b}) \nabla_{{\delta}_{b}}{\delta}_{b}\Bigg]\Bigg]\\
&\times \Bigg[\sum_{k^{},u^{}} \Bigg[ \nabla_{{\delta}_{b}}{\delta}_{b} \bm{F}_{{{y} }}({y}_{bu,k}| {\eta} ;{\delta}_{b},{\delta}_{b}) \nabla_{{\delta}_{b}}{\delta}_{b}\Bigg]\Bigg]^{-1}  \Bigg[\sum_{k^{},u^{}} \Bigg[ \nabla_{{\delta}_{b}}{\delta}_{b}  \bm{F}_{{{y} }}({y}_{bu,k}| {\eta} ;{\delta}_{b},{\tau}_{bu,k}) \nabla_{\bm{\Phi}_{U}}^{\mathrm{T}} \tau_{bu,k^{}}\Bigg]\Bigg] \Bigg] \\ 
&+
\sum_{b}\Bigg[\Bigg[\sum_{k^{},u^{}} \Bigg[ \nabla_{\bm{p}_{U}} \tau_{bu,k^{}} \bm{F}_{{{y} }}({y}_{bu,k}| {\eta} ;{\tau}_{bu,k},{\epsilon}_{b}) \nabla_{{\epsilon}_{b}}{\epsilon}_{b} + \nabla_{\bm{p}_{U}} \nu_{b,k^{}} \bm{F}_{{{y} }}({y}_{bu,k}| {\eta} ;{\nu}_{b,k},{\epsilon}_{b}) \nabla_{{\epsilon}_{b}}{\epsilon}_{b}\Bigg]\Bigg]\\
&\times \Bigg[ \sum_{k,u}\Bigg[ \nabla_{{\epsilon}_{b}}{\epsilon}_{b} \bm{F}_{{{y} }}({y}_{bu,k}| {\eta} ;{\epsilon}_{b},{\epsilon}_{b}) \nabla_{{\epsilon}_{b}}{\epsilon}_{b}\Bigg]\Bigg]^{-1}  \Bigg[\sum_{k^{},u^{}} \Bigg[ \nabla_{{\epsilon}_{b}}{\epsilon}_{b}  \bm{F}_{{{y} }}({y}_{bu,k}| {\eta} ;{\epsilon}_{b},{\tau}_{bu,k}) \nabla_{\bm{\Phi}_{U}}^{\mathrm{T}} \tau_{bu,k^{}} \Bigg]\Bigg] \Bigg] \\ 
&=\sum_{b}
\Bigg[\Bigg[\sum_{k^{},u^{}}  \nabla_{\bm{p}_{U}} \tau_{bu,k^{}} \bm{F}_{{{y} }}({y}_{bu,k}| {\eta} ;{\tau}_{bu,k},{\delta}_{b})  \Bigg] \Bigg[\sum_{k^{},u^{}}  \bm{F}_{{{y} }}({y}_{bu,k}| {\eta} ;{\delta}_{b},{\delta}_{b}) \Bigg]^{-1}  \Bigg[\sum_{k^{},u^{}}   \bm{F}_{{{y} }}({y}_{bu,k}| {\eta} ;{\delta}_{b},{\tau}_{bu,k}) \nabla_{\bm{\Phi}_{U}}^{\mathrm{T}} \tau_{bu,k^{}}\Bigg] \Bigg]  \\ 
\end{split}
\end{align}
\end{figure*}

\begin{figure*}
\begin{align}
\begin{split}
\label{appendix_equ:lemma_nuisance_3D_position_3D_velocity}
&[\mathbf{J}_{ \bm{\bm{y}}; \bm{\kappa}_1}^{nu}]_{[1:3,7:9]} =  \\ & \sum_{b}\Bigg[\Bigg[\sum_{k^{},u^{}} \Bigg[ \nabla_{\bm{p}_{U}} \tau_{bu,k^{}} \bm{F}_{{{y} }}({y}_{bu,k}| {\eta} ;{\tau}_{bu,k},{\beta}_{bu^{'},k^{'}}) \nabla_{{\beta}_{bu^{'},k^{'}}}{\beta}_{bu^{'},k^{'}} + \nabla_{\bm{p}_{U}} \nu_{b,k^{}} \bm{F}_{{{y} }}({y}_{bu,k}| {\eta} ;{\nu}_{b,k},{\beta}_{bu^{'},k^{'}}) \nabla_{{\beta}_{bu^{'},k^{'}}}{\beta}_{bu^{'},k^{'}}\Bigg]\Bigg]\\
&\times \Bigg[ \nabla_{{\beta}_{bu^{'},k^{'}}}{\beta}_{bu^{'},k^{'}} \bm{F}_{{{y} }}({y}_{bu,k}| {\eta} ;{\beta}_{bu^{'},k^{'}},{\beta}_{bu^{'},k^{'}}) \nabla_{{\beta}_{bu^{'},k^{'}}}{\beta}_{bu^{'},k^{'}}\Bigg]^{-1} \\
&\times \Bigg[\sum_{k^{},u^{}} \Bigg[ \nabla_{{\beta}_{bu^{'},k^{'}}}{\beta}_{bu^{'},k^{'}}  \bm{F}_{{{y} }}({y}_{bu,k}| {\eta} ;{\beta}_{bu^{'},k^{'}},{\tau}_{bu,k}) \nabla_{\bm{v}_{U}}^{\mathrm{T}} \tau_{bu,k^{}}  + \nabla_{{\beta}_{bu^{'},k^{'}}}{\beta}_{bu^{'},k^{'}}  \bm{F}_{{{y} }}({y}_{bu,k}| {\eta} ;{\beta}_{bu^{'},k^{'}},{\nu}_{b,k}) \nabla_{\bm{v}_{U}}^{\mathrm{T}} \nu_{b,k^{}}\Bigg]\Bigg] \Bigg] \\ 
&+
\sum_{b}\Bigg[\Bigg[\sum_{k^{},u^{}} \Bigg[ \nabla_{\bm{p}_{U}} \tau_{bu,k^{}} \bm{F}_{{{y} }}({y}_{bu,k}| {\eta} ;{\tau}_{bu,k},{\delta}_{b}) \nabla_{{\delta}_{b}}{\delta}_{b} + \nabla_{\bm{p}_{U}} \nu_{b,k^{}} \bm{F}_{{{y} }}({y}_{bu,k}| {\eta} ;{\nu}_{b,k},{\delta}_{b}) \nabla_{{\delta}_{b}}{\delta}_{b}\Bigg]\Bigg]\\
&\times \Bigg[\sum_{k^{},u^{}} \Bigg[ \nabla_{{\delta}_{b}}{\delta}_{b} \bm{F}_{{{y} }}({y}_{bu,k}| {\eta} ;{\delta}_{b},{\delta}_{b}) \nabla_{{\delta}_{b}}{\delta}_{b}\Bigg]\Bigg]^{-1} \\
&\times \Bigg[\sum_{k^{},u^{}} \Bigg[ \nabla_{{\delta}_{b}}{\delta}_{b}  \bm{F}_{{{y} }}({y}_{bu,k}| {\eta} ;{\delta}_{b},{\tau}_{bu,k}) \nabla_{\bm{v}_{U}}^{\mathrm{T}} \tau_{bu,k^{}}  + \nabla_{{\delta}_{b}}{\delta}_{b}  \bm{F}_{{{y} }}({y}_{bu,k}| {\eta} ;{\delta}_{b},{\nu}_{b,k}) \nabla_{\bm{v}_{U}}^{\mathrm{T}} \nu_{b,k^{}}\Bigg]\Bigg] \Bigg] \\ 
&+
\sum_{b}\Bigg[\Bigg[\sum_{k^{},u^{}} \Bigg[ \nabla_{\bm{p}_{U}} \tau_{bu,k^{}} \bm{F}_{{{y} }}({y}_{bu,k}| {\eta} ;{\tau}_{bu,k},{\epsilon}_{b}) \nabla_{{\epsilon}_{b}}{\epsilon}_{b} + \nabla_{\bm{p}_{U}} \nu_{b,k^{}} \bm{F}_{{{y} }}({y}_{bu,k}| {\eta} ;{\nu}_{b,k},{\epsilon}_{b}) \nabla_{{\epsilon}_{b}}{\epsilon}_{b}\Bigg]\Bigg]\\
&\times \Bigg[ \sum_{k,u}\Bigg[ \nabla_{{\epsilon}_{b}}{\epsilon}_{b} \bm{F}_{{{y} }}({y}_{bu,k}| {\eta} ;{\epsilon}_{b},{\epsilon}_{b}) \nabla_{{\epsilon}_{b}}{\epsilon}_{b}\Bigg]\Bigg]^{-1} \\
&\times \Bigg[\sum_{k^{},u^{}} \Bigg[ \nabla_{{\epsilon}_{b}}{\epsilon}_{b}  \bm{F}_{{{y} }}({y}_{bu,k}| {\eta} ;{\epsilon}_{b},{\tau}_{bu,k}) \nabla_{\bm{v}_{U}}^{\mathrm{T}} \tau_{bu,k^{}}  + \nabla_{{\epsilon}_{b}}{\epsilon}_{b}  \bm{F}_{{{y} }}({y}_{bu,k}| {\eta} ;{\epsilon}_{b},{\nu}_{b,k}) \nabla_{\bm{v}_{U}}^{\mathrm{T}} \nu_{b,k^{}}\Bigg]\Bigg] \Bigg] \\ 
&=
\sum_{b}\Bigg[\Bigg[\sum_{k^{},u^{}}  \nabla_{\bm{p}_{U}} \tau_{bu,k^{}} \bm{F}_{{{y} }}({y}_{bu,k}| {\eta} ;{\tau}_{bu,k},{\delta}_{b})  \Bigg] \Bigg[\sum_{k^{},u^{}} \bm{F}_{{{y} }}({y}_{bu,k}| {\eta} ;{\delta}_{b},{\delta}_{b}) \Bigg]^{-1}  \Bigg[\sum_{k^{},u^{}}   \bm{F}_{{{y} }}({y}_{bu,k}| {\eta} ;{\delta}_{b},{\tau}_{bu,k}) \nabla_{\bm{v}_{U}}^{\mathrm{T}} \tau_{bu,k^{}}  \Bigg] \Bigg] \\ 
&+
\sum_{b}\Bigg[\Bigg[\sum_{k^{},u^{}}  \nabla_{\bm{p}_{U}} \nu_{b,k^{}} \bm{F}_{{{y} }}({y}_{bu,k}| {\eta} ;{\nu}_{b,k},{\epsilon}_{b})\Bigg] \Bigg[ \sum_{k,u}  \bm{F}_{{{y} }}({y}_{bu,k}| {\eta} ;{\epsilon}_{b},{\epsilon}_{b}) \Bigg]^{-1} 
 \Bigg[\sum_{k^{},u^{}}   \bm{F}_{{{y} }}({y}_{bu,k}| {\eta} ;{\epsilon}_{b},{\nu}_{b,k}) \nabla_{\bm{v}_{U}}^{\mathrm{T}} \nu_{b,k^{}}\Bigg] \Bigg]
\end{split}
\end{align}
\end{figure*}

\begin{figure*}
\begin{align}
\begin{split}
\label{appendix_equ:lemma_nuisance_3D_orientation_3D_orientation}
&[\mathbf{J}_{ \bm{\bm{y}}; \bm{\kappa}_1}^{nu}]_{[4:6,4:6]} = \\ &\sum_{b}\Bigg[\Bigg[\sum_{k^{},u^{}}  \nabla_{\bm{\Phi
}_{U}} \tau_{bu,k^{}} \bm{F}_{{{y} }}({y}_{bu,k}| {\eta} ;{\tau}_{bu,k},{\beta}_{bu^{'},k^{'}}) \nabla_{{\beta}_{bu^{'},k^{'}}}{\beta}_{bu^{'},k^{'}}\Bigg] \Bigg[  \bm{F}_{{{y} }}({y}_{bu,k}| {\eta} ;{\beta}_{bu^{'},k^{'}},{\beta}_{bu^{'},k^{'}}) \Bigg]^{-1} \\
&\times \Bigg[\sum_{k^{},u^{}} \nabla_{{\beta}_{bu^{'},k^{'}}}{\beta}_{bu^{'},k^{'}}  \bm{F}_{{{y} }}({y}_{bu,k}| {\eta} ;{\beta}_{bu^{'},k^{'}},{\tau}_{bu,k}) \nabla_{\bm{\Phi}_{U}}^{\mathrm{T}} \tau_{bu,k^{}}\Bigg] \Bigg] \\ 
&+
\sum_{b}\Bigg[\sum_{k^{},u^{}}  \nabla_{\bm{\Phi}_{U}} \tau_{bu,k^{}} \bm{F}_{{{y} }}({y}_{bu,k}| {\eta} ;{\tau}_{bu,k},{\delta}_{b}) \nabla_{{\delta}_{b}}{\delta}_{b} \Bigg] \Bigg[\sum_{k^{},u^{}}  \nabla_{{\delta}_{b}}{\delta}_{b} \bm{F}_{{{y} }}({y}_{bu,k}| {\eta} ;{\delta}_{b},{\delta}_{b}) \nabla_{{\delta}_{b}}{\delta}_{b}\Bigg]^{-1} \\  &\times\Bigg[\sum_{k^{},u^{}}  \nabla_{{\delta}_{b}}{\delta}_{b}  \bm{F}_{{{y} }}({y}_{bu,k}| {\eta} ;{\delta}_{b},{\tau}_{bu,k}) \nabla_{\bm{\Phi}_{U}}^{\mathrm{T}} \tau_{bu,k^{}}\Bigg] \Bigg] \\ 
&+
\sum_{b}\Bigg[\sum_{k^{},u^{}}   \nabla_{\bm{\Phi}_{U}} \tau_{b,k^{}} \bm{F}_{{{y} }}({y}_{bu,k}| {\eta} ;{\tau}_{b,k},{\epsilon}_{b}) \nabla_{{\epsilon}_{b}}{\epsilon}_{b} \Bigg[ \sum_{k,u} \nabla_{{\epsilon}_{b}}{\epsilon}_{b} \bm{F}_{{{y} }}({y}_{bu,k}| {\eta} ;{\epsilon}_{b},{\epsilon}_{b}) \nabla_{{\epsilon}_{b}}{\epsilon}_{b}\Bigg]^{-1} \\ &\times \Bigg[\sum_{k^{},u^{}}  \nabla_{{\epsilon}_{b}}{\epsilon}_{b}  \bm{F}_{{{y} }}({y}_{bu,k}| {\eta} ;{\epsilon}_{b},{\tau}_{bu,k}) \nabla_{\bm{\Phi}_{U}}^{\mathrm{T}} \tau_{bu,k^{}} \Bigg] \Bigg] \\ 
&=\sum_{b}
\Bigg[\Bigg[\sum_{k^{},u^{}}  \nabla_{\bm{\phi}_{U}} \tau_{bu,k^{}} \bm{F}_{{{y} }}({y}_{bu,k}| {\eta} ;{\tau}_{bu,k},{\delta}_{b})  \Bigg] \Bigg[\sum_{k^{},u^{}}  \bm{F}_{{{y} }}({y}_{bu,k}| {\eta} ;{\delta}_{b},{\delta}_{b}) \Bigg]^{-1}  \Bigg[\sum_{k^{},u^{}}   \bm{F}_{{{y} }}({y}_{bu,k}| {\eta} ;{\delta}_{b},{\tau}_{bu,k}) \nabla_{\bm{\Phi}_{U}}^{\mathrm{T}} \tau_{bu,k^{}}\Bigg] \Bigg]  \\ 
\end{split}
\end{align}
\end{figure*}

\begin{figure*}
\begin{align}
\begin{split}
\label{appendix_equ:lemma_nuisance_3D_orientation_3D_velocity}
&[\mathbf{J}_{ \bm{\bm{y}}; \bm{\kappa}_1}^{nu}]_{[4:6,7:9]} = \\ & \sum_{b}\Bigg[\Bigg[\sum_{k^{},u^{}}  \nabla_{\bm{\Phi}_{U}} \tau_{bu,k^{}} \bm{F}_{{{y} }}({y}_{bu,k}| {\eta} ;{\tau}_{bu,k},{\beta}_{bu^{'},k^{'}}) \nabla_{{\beta}_{bu^{'},k^{'}}}{\beta}_{bu^{'},k^{'}} \Bigg]\\
&\times \Bigg[ \nabla_{{\beta}_{bu^{'},k^{'}}}{\beta}_{bu^{'},k^{'}} \bm{F}_{{{y} }}({y}_{bu,k}| {\eta} ;{\beta}_{bu^{'},k^{'}},{\beta}_{bu^{'},k^{'}}) \nabla_{{\beta}_{bu^{'},k^{'}}}{\beta}_{bu^{'},k^{'}}\Bigg]^{-1} \\
&\times \Bigg[\sum_{k^{},u^{}} \Bigg[ \nabla_{{\beta}_{bu^{'},k^{'}}}{\beta}_{bu^{'},k^{'}}  \bm{F}_{{{y} }}({y}_{bu,k}| {\eta} ;{\beta}_{bu^{'},k^{'}},{\tau}_{bu,k}) \nabla_{\bm{v}_{U}}^{\mathrm{T}} \tau_{bu,k^{}}  + \nabla_{{\beta}_{bu^{'},k^{'}}}{\beta}_{bu^{'},k^{'}}  \bm{F}_{{{y} }}({y}_{bu,k}| {\eta} ;{\beta}_{bu^{'},k^{'}},{\nu}_{b,k}) \nabla_{\bm{v}_{U}}^{\mathrm{T}} \nu_{b,k^{}}\Bigg]\Bigg] \Bigg] \\ 
&+
\sum_{b}\Bigg[\Bigg[\sum_{k^{},u^{}}  \nabla_{\bm{\Phi}_{U}} \tau_{bu,k^{}} \bm{F}_{{{y} }}({y}_{bu,k}| {\eta} ;{\tau}_{bu,k},{\delta}_{b}) \nabla_{{\delta}_{b}}{\delta}_{b}\Bigg] \Bigg[\sum_{k^{},u^{}}  \nabla_{{\delta}_{b}}{\delta}_{b} \bm{F}_{{{y} }}({y}_{bu,k}| {\eta} ;{\delta}_{b},{\delta}_{b}) \nabla_{{\delta}_{b}}{\delta}_{b}\Bigg]^{-1} \\
&\times \Bigg[\sum_{k^{},u^{}} \Bigg[ \nabla_{{\delta}_{b}}{\delta}_{b}  \bm{F}_{{{y} }}({y}_{bu,k}| {\eta} ;{\delta}_{b},{\tau}_{bu,k}) \nabla_{\bm{v}_{U}}^{\mathrm{T}} \tau_{bu,k^{}}  + \nabla_{{\delta}_{b}}{\delta}_{b}  \bm{F}_{{{y} }}({y}_{bu,k}| {\eta} ;{\delta}_{b},{\nu}_{b,k}) \nabla_{\bm{v}_{U}}^{\mathrm{T}} \nu_{b,k^{}}\Bigg]\Bigg] \Bigg] \\ 
&+
\sum_{b}\Bigg[\Bigg[\sum_{k^{},u^{}}  \nabla_{\bm{\Phi}_{U}} \tau_{bu,k^{}} \bm{F}_{{{y} }}({y}_{bu,k}| {\eta} ;{\tau}_{bu,k},{\epsilon}_{b}) \nabla_{{\epsilon}_{b}}{\epsilon}_{b}\Bigg] \Bigg[ \sum_{k,u}\Bigg[ \nabla_{{\epsilon}_{b}}{\epsilon}_{b} \bm{F}_{{{y} }}({y}_{bu,k}| {\eta} ;{\epsilon}_{b},{\epsilon}_{b}) \nabla_{{\epsilon}_{b}}{\epsilon}_{b}\Bigg]\Bigg]^{-1} \\ &\times
 \Bigg[\sum_{k^{},u^{}} \nabla_{{\epsilon}_{b}}{\epsilon}_{b}  \bm{F}_{{{y} }}({y}_{bu,k}| {\eta} ;{\epsilon}_{b},{\tau}_{bu,k}) \nabla_{\bm{v}_{U}}^{\mathrm{T}} \tau_{bu,k^{}}  + \nabla_{{\epsilon}_{b}}{\epsilon}_{b}  \bm{F}_{{{y} }}({y}_{bu,k}| {\eta} ;{\epsilon}_{b},{\nu}_{b,k}) \nabla_{\bm{v}_{U}}^{\mathrm{T}} \nu_{b,k^{}}\Bigg] \Bigg] \\ 
&=
\sum_{b}\Bigg[\Bigg[\sum_{k^{},u^{}}  \nabla_{\bm{\Phi}_{U}} \tau_{bu,k^{}} \bm{F}_{{{y} }}({y}_{bu,k}| {\eta} ;{\tau}_{bu,k},{\delta}_{b})  \Bigg] \Bigg[\sum_{k^{},u^{}} \bm{F}_{{{y} }}({y}_{bu,k}| {\eta} ;{\delta}_{b},{\delta}_{b}) \Bigg]^{-1}  \Bigg[\sum_{k^{},u^{}}   \bm{F}_{{{y} }}({y}_{bu,k}| {\eta} ;{\delta}_{b},{\tau}_{bu,k}) \nabla_{\bm{v}_{U}}^{\mathrm{T}} \tau_{bu,k^{}}  \Bigg] \Bigg] \\ 
&+
\sum_{b}\Bigg[\Bigg[\sum_{k^{},u^{}}  \nabla_{\bm{\Phi}_{U}} \nu_{b,k^{}} \bm{F}_{{{y} }}({y}_{bu,k}| {\eta} ;{\nu}_{b,k},{\epsilon}_{b})\Bigg] \Bigg[ \sum_{k,u}  \bm{F}_{{{y} }}({y}_{bu,k}| {\eta} ;{\epsilon}_{b},{\epsilon}_{b}) \Bigg]^{-1} 
 \Bigg[\sum_{k^{},u^{}}   \bm{F}_{{{y} }}({y}_{bu,k}| {\eta} ;{\epsilon}_{b},{\nu}_{b,k}) \nabla_{\bm{v}_{U}}^{\mathrm{T}} \nu_{b,k^{}}\Bigg] \Bigg]
\end{split}
\end{align}
\end{figure*}

\begin{figure*}
\begin{align}
\begin{split}
\label{appendix_equ:lemma_nuisance_3D_velocity_3D_velocity}
&[\mathbf{J}_{ \bm{\bm{y}}; \bm{\kappa}_1}^{nu}]_{[7:9,7:9]} = \\ &\sum_{b}\Bigg[\Bigg[\sum_{k^{},u^{}} \Bigg[ \nabla_{\bm{v}_{U}} \tau_{bu,k^{}} \bm{F}_{{{y} }}({y}_{bu,k}| {\eta} ;{\tau}_{bu,k},{\beta}_{bu^{'},k^{'}}) \nabla_{{\beta}_{bu^{'},k^{'}}}{\beta}_{bu^{'},k^{'}} + \nabla_{\bm{v}_{U}} \nu_{b,k^{}} \bm{F}_{{{y} }}({y}_{bu,k}| {\eta} ;{\nu}_{b,k},{\beta}_{bu^{'},k^{'}}) \nabla_{{\beta}_{bu^{'},k^{'}}}{\beta}_{bu^{'},k^{'}}\Bigg]\Bigg]\\
&\times \Bigg[ \nabla_{{\beta}_{bu^{'},k^{'}}}{\beta}_{bu^{'},k^{'}} \bm{F}_{{{y} }}({y}_{bu,k}| {\eta} ;{\beta}_{bu^{'},k^{'}},{\beta}_{bu^{'},k^{'}}) \nabla_{{\beta}_{bu^{'},k^{'}}}{\beta}_{bu^{'},k^{'}}\Bigg]^{-1} \\
&\times \Bigg[\sum_{k^{},u^{}} \Bigg[ \nabla_{{\beta}_{bu^{'},k^{'}}}{\beta}_{bu^{'},k^{'}}  \bm{F}_{{{y} }}({y}_{bu,k}| {\eta} ;{\beta}_{bu^{'},k^{'}},{\tau}_{bu,k}) \nabla_{\bm{v}_{U}}^{\mathrm{T}} \tau_{bu,k^{}}  + \nabla_{{\beta}_{bu^{'},k^{'}}}{\beta}_{bu^{'},k^{'}}  \bm{F}_{{{y} }}({y}_{bu,k}| {\eta} ;{\beta}_{bu^{'},k^{'}},{\nu}_{b,k}) \nabla_{\bm{v}_{U}}^{\mathrm{T}} \nu_{b,k^{}}\Bigg]\Bigg] \Bigg] \\ 
&+
\sum_{b}\Bigg[\Bigg[\sum_{k^{},u^{}} \Bigg[ \nabla_{\bm{v}_{U}} \tau_{bu,k^{}} \bm{F}_{{{y} }}({y}_{bu,k}| {\eta} ;{\tau}_{bu,k},{\delta}_{b}) \nabla_{{\delta}_{b}}{\delta}_{b} + \nabla_{\bm{v}_{U}} \nu_{b,k^{}} \bm{F}_{{{y} }}({y}_{bu,k}| {\eta} ;{\nu}_{b,k},{\delta}_{b}) \nabla_{{\delta}_{b}}{\delta}_{b}\Bigg]\Bigg]\\
&\times \Bigg[\sum_{k^{},u^{}} \Bigg[ \nabla_{{\delta}_{b}}{\delta}_{b} \bm{F}_{{{y} }}({y}_{bu,k}| {\eta} ;{\delta}_{b},{\delta}_{b}) \nabla_{{\delta}_{b}}{\delta}_{b}\Bigg]\Bigg]^{-1} \\
&\times \Bigg[\sum_{k^{},u^{}} \Bigg[ \nabla_{{\delta}_{b}}{\delta}_{b}  \bm{F}_{{{y} }}({y}_{bu,k}| {\eta} ;{\delta}_{b},{\tau}_{bu,k}) \nabla_{\bm{v}_{U}}^{\mathrm{T}} \tau_{bu,k^{}}  + \nabla_{{\delta}_{b}}{\delta}_{b}  \bm{F}_{{{y} }}({y}_{bu,k}| {\eta} ;{\delta}_{b},{\nu}_{b,k}) \nabla_{\bm{v}_{U}}^{\mathrm{T}} \nu_{b,k^{}}\Bigg]\Bigg] \Bigg] \\ 
&+
\sum_{b}\Bigg[\Bigg[\sum_{k^{},u^{}} \Bigg[ \nabla_{\bm{v}_{U}} \tau_{bu,k^{}} \bm{F}_{{{y} }}({y}_{bu,k}| {\eta} ;{\tau}_{bu,k},{\epsilon}_{b}) \nabla_{{\epsilon}_{b}}{\epsilon}_{b} + \nabla_{\bm{v}_{U}} \nu_{b,k^{}} \bm{F}_{{{y} }}({y}_{bu,k}| {\eta} ;{\nu}_{b,k},{\epsilon}_{b}) \nabla_{{\epsilon}_{b}}{\epsilon}_{b}\Bigg]\Bigg]\\
&\times \Bigg[ \sum_{k,u}\Bigg[ \nabla_{{\epsilon}_{b}}{\epsilon}_{b} \bm{F}_{{{y} }}({y}_{bu,k}| {\eta} ;{\epsilon}_{b},{\epsilon}_{b}) \nabla_{{\epsilon}_{b}}{\epsilon}_{b}\Bigg]\Bigg]^{-1} \\
&\times \Bigg[\sum_{k^{},u^{}} \Bigg[ \nabla_{{\epsilon}_{b}}{\epsilon}_{b}  \bm{F}_{{{y} }}({y}_{bu,k}| {\eta} ;{\epsilon}_{b},{\tau}_{bu,k}) \nabla_{\bm{v}_{U}}^{\mathrm{T}} \tau_{bu,k^{}}  + \nabla_{{\epsilon}_{b}}{\epsilon}_{b}  \bm{F}_{{{y} }}({y}_{bu,k}| {\eta} ;{\epsilon}_{b},{\nu}_{b,k}) \nabla_{\bm{v}_{U}}^{\mathrm{T}} \nu_{b,k^{}}\Bigg]\Bigg] \Bigg] \\ 
&= \sum_{b}\Bigg[
 \Bigg[\sum_{k^{},u^{}} \bm{F}_{{{y} }}({y}_{bu,k}| {\eta} ;{\delta}_{b},{\delta}_{b}) \Bigg]^{-1} \norm{ \sum_{k^{},u^{}}   \bm{F}_{{{y} }}({y}_{bu,k}| {\eta} ;{\delta}_{b},{\tau}_{bu,k}) \nabla_{\bm{v}_{U}}^{\mathrm{T}} \tau_{bu,k^{}}  }^{2} \\ &+  \Bigg[ \sum_{k,u}  \bm{F}_{{{y} }}({y}_{bu,k}| {\eta} ;{\epsilon}_{b},{\epsilon}_{b})\Bigg]^{-1}\norm{\sum_{k^{},u^{}}   \bm{F}_{{{y} }}({y}_{bu,k}| {\eta} ;{\epsilon}_{b},{\nu}_{b,k}) \nabla_{\bm{v}_{U}}^{\mathrm{T}} \nu_{b,k^{}} }^2
\end{split}
\end{align}
\end{figure*}

\subsection{$9$D Localization}
In this section, we present proofs for the EFIMs of the $3$D position, $3$D orientation, and $3$D velocity.

\subsubsection{Proof for Theorem \ref{theorem:FIM_9D_position}}
\label{Appendix_subsection:FIM_9D_position}
To highlight the available information for the $3$D position estimation, we present the appropriate EFIM below.

\begin{equation}
\label{equ:FIM_9D_3D_position_3D_orientation_3D_velocity}
\begin{aligned}
\left[\begin{array}{ccc}
\mathbf{J}_{ \bm{\bm{y}}; \bm{p}_{U}}^{\mathrm{e}} & \mathbf{J}_{ \bm{\bm{y}};[ \bm{p}_{U}, \bm{\Phi}_{U}]}^{\mathrm{e}} & \mathbf{J}_{ \bm{\bm{y}};[ \bm{p}_{U}, \bm{v}_{U}]}^{\mathrm{e}} \\
(\mathbf{J}_{ \bm{\bm{y}};[ \bm{p}_{U}, \bm{\Phi}_{U}]}^{\mathrm{e}})^{\mathrm{T}} & \mathbf{J}_{ \bm{\bm{y}}; \bm{\Phi}_{U}}^{\mathrm{e}} & \mathbf{J}_{ \bm{\bm{y}};[ \bm{\Phi}_{U}, \bm{v}_{U}]}^{\mathrm{e}} \\
(\mathbf{J}_{ \bm{\bm{y}};[ \bm{p}_{U}, \bm{v}_{U}]}^{\mathrm{e}})^{\mathrm{T}} & (\mathbf{J}_{ \bm{\bm{y}}; [\bm{\Phi}_{U} \bm{v}_{U}]}^{\mathrm{e}})^{\mathrm{T}} & \mathbf{J}_{ \bm{\bm{y}}; \bm{v}_{U}}^{\mathrm{e}} \\
\end{array}\right].
\end{aligned}
\end{equation}
The EFIM matrix in (\ref{equ:FIM_9D_3D_position_3D_orientation_3D_velocity}) has the following structure
\begin{equation}
\label{equ:FIM_9D_3D_position_3D_orientation_3D_velocity_structure}
\begin{aligned}
\left[\begin{array}{cc}
\mathbf{J}_{ \bm{\bm{y}}; \bm{p}_{U}}^{\mathrm{e}} & \bm{B} \\
\bm{B}^{\mathrm{T}} & \bm{C} \\
\end{array}\right].
\end{aligned}
\end{equation}
Here, the loss in information about $\bm{p}_{U}$ due to the unknown $\bm{\Phi}_{U}$ and $\bm{v}_{U}$ which is specified by  $\mathbf{J}_{ \bm{\bm{y}}; \bm{p}_{U}}^{nu}$ can be defined as
$
\mathbf{J}_{ \bm{\bm{y}}; \bm{p}_{U}}^{nu} = \bm{B} \bm{C}^{-1} \bm{B}^{\mathrm{T}}.
$ Using the definition of the inverse of a block $2$D matrix \cite{horn2012matrix}, we can expand  $\bm{B} \bm{C}^{-1} \bm{B}^{\mathrm{T}}$ and write $
\mathbf{J}_{ \bm{\bm{y}}; \bm{p}_{U}}^{nu}$ as (\ref{equ_theorem:FIM_3D_nuisance_position}). Finally, the location EFIM for the $3$D position is (\ref{equ:FIM_9D_3D_position_3D_position_3D_orientation_3D_velocity}).

\subsubsection{Proof for Theorem \ref{theorem:FIM_9D_velocity}}
\label{Appendix_subsection:FIM_9D_velocity}
To highlight the available information for the $3$D velocity estimation, we present the appropriate EFIM below.

\begin{equation}
\label{equ:FIM_9D_3D_position_3D_orientation_3D_velocity_2}
\begin{aligned}
\left[\begin{array}{ccc}
\mathbf{J}_{ \bm{\bm{y}}; \bm{v}_{U}}^{\mathrm{e}} & \mathbf{J}_{ \bm{\bm{y}};[ \bm{v}_{U}, \bm{p}_{U}]}^{\mathrm{e}} & \mathbf{J}_{ \bm{\bm{y}};[  \bm{v}_{U},\bm{\Phi}_{U}]}^{\mathrm{e}} \\
(\mathbf{J}_{ \bm{\bm{y}};[  \bm{v}_{U}, \bm{p}_{U}]}^{\mathrm{e}})^{\mathrm{T}} & \mathbf{J}_{ \bm{\bm{y}}; \bm{p}_{U}}^{\mathrm{e}} & \mathbf{J}_{ \bm{\bm{y}};[ \bm{p}_{U}, \bm{\Phi}_{U}]}^{\mathrm{e}} \\
(\mathbf{J}_{ \bm{\bm{y}};[ \bm{v}_{U}, \bm{\Phi}_{U}]}^{\mathrm{e}})^{\mathrm{T}} & (\mathbf{J}_{ \bm{\bm{y}}; [\bm{p}_{U}, \bm{\Phi}_{U}]}^{\mathrm{e}})^{\mathrm{T}} & \mathbf{J}_{ \bm{\bm{y}}; \bm{\Phi}_{U}}^{\mathrm{e}} \\
\end{array}\right].
\end{aligned}
\end{equation}
The EFIM matrix in (\ref{equ:FIM_9D_3D_position_3D_orientation_3D_velocity_2}) has the following structure
\begin{equation}
\label{equ:FIM_9D_3D_position_3D_orientation_3D_velocity_structure_2}
\begin{aligned}
\left[\begin{array}{cc}
\mathbf{J}_{ \bm{\bm{y}}; \bm{v}_{U}}^{\mathrm{e}} & \bm{B} \\
\bm{B}^{\mathrm{T}} & \bm{C} \\
\end{array}\right].
\end{aligned}
\end{equation}
The rest of the proof is identical to the proof in Appendix \ref{Appendix_subsection:FIM_9D_position}.

\subsubsection{Proof for Theorem \ref{theorem:FIM_9D_orientation}}
\label{Appendix_subsection:FIM_9D_orientation}
To highlight the available information for the $3$D orientation estimation, we present the appropriate EFIM below.

\begin{equation}
\label{equ:FIM_9D_3D_position_3D_orientation_3D_velocity_1}
\begin{aligned}
\left[\begin{array}{ccc}
\mathbf{J}_{ \bm{\bm{y}}; \bm{\Phi}_{U}}^{\mathrm{e}} & \mathbf{J}_{ \bm{\bm{y}};[ \bm{\Phi}_{U}, \bm{p}_{U}]}^{\mathrm{e}} & \mathbf{J}_{ \bm{\bm{y}};[ \bm{\Phi}_{U}, \bm{v}_{U}]}^{\mathrm{e}} \\
(\mathbf{J}_{ \bm{\bm{y}};[  \bm{\Phi}_{U}, \bm{p}_{U}]}^{\mathrm{e}})^{\mathrm{T}} & \mathbf{J}_{ \bm{\bm{y}}; \bm{p}_{U}}^{\mathrm{e}} & \mathbf{J}_{ \bm{\bm{y}};[ \bm{p}_{U}, \bm{v}_{U}]}^{\mathrm{e}} \\
(\mathbf{J}_{ \bm{\bm{y}};[ \bm{\Phi}_{U}, \bm{v}_{U}]}^{\mathrm{e}})^{\mathrm{T}} & (\mathbf{J}_{ \bm{\bm{y}}; [\bm{p}_{U}, \bm{v}_{U}]}^{\mathrm{e}})^{\mathrm{T}} & \mathbf{J}_{ \bm{\bm{y}}; \bm{v}_{U}}^{\mathrm{e}} \\
\end{array}\right].
\end{aligned}
\end{equation}
The EFIM matrix in (\ref{equ:FIM_9D_3D_position_3D_orientation_3D_velocity_1}) has the following structure
\begin{equation}
\label{equ:FIM_9D_3D_position_3D_orientation_3D_velocity_structure_1}
\begin{aligned}
\left[\begin{array}{cc}
\mathbf{J}_{ \bm{\bm{y}}; \bm{\Phi}_{U}}^{\mathrm{e}} & \bm{B} \\
\bm{B}^{\mathrm{T}} & \bm{C} \\
\end{array}\right].
\end{aligned}
\end{equation}
The rest of the proof is identical to the proof in Appendix \ref{Appendix_subsection:FIM_9D_position}.

\clearpage

{
\bibliographystyle{IEEEtran}
\bibliography{refs}

\begin{thebibliography}{10}
\providecommand{\url}[1]{#1}
\csname url@samestyle\endcsname
\providecommand{\newblock}{\relax}
\providecommand{\bibinfo}[2]{#2}
\providecommand{\BIBentrySTDinterwordspacing}{\spaceskip=0pt\relax}
\providecommand{\BIBentryALTinterwordstretchfactor}{4}
\providecommand{\BIBentryALTinterwordspacing}{\spaceskip=\fontdimen2\font plus
\BIBentryALTinterwordstretchfactor\fontdimen3\font minus \fontdimen4\font\relax}
\providecommand{\BIBforeignlanguage}[2]{{%
\expandafter\ifx\csname l@#1\endcsname\relax
\typeout{** WARNING: IEEEtran.bst: No hyphenation pattern has been}%
\typeout{** loaded for the language `#1'. Using the pattern for}%
\typeout{** the default language instead.}%
\else
\language=\csname l@#1\endcsname
\fi
#2}}
\providecommand{\BIBdecl}{\relax}
\BIBdecl

\bibitem{emenonye2023_VTC_conf_Minimal}
D.-R. Emenonye, H.~S. Dhillon, and R.~M. Buehrer, ``Minimal configurations to achieve {3D} positioning with unsynchronized {LEO} satellites,'' \emph{accepted to {{IEEE 100th Vehicular Technology Conference (VTC2024-Fall)}}}, 2024.

\bibitem{emenonye2023_VTC_conf_unsyn}
------, ``Can unsynchronized {LEOs} provide 3d orientation for a ground receiver?'' \emph{accepted to {{IEEE 100th Vehicular Technology Conference (VTC2024-Fall)}}}, 2024.

\bibitem{Fundamental_Performance_Bounds_for_Carrier_Phase_Positioning_LEO_PNT}
J.~Kang, P.~E. N, J.~Lee, H.~Wymeersch, and S.~Kim, ``Fundamental performance bounds for carrier phase positioning in {LEO}-{PNT} systems,'' in \emph{Proc., IEEE Intl. Conf. on Acoustics, Speech, and Sig. Proc. (ICASSP)}, 2024, pp. 13\,496--13\,500.

\bibitem{Broadband_LEO_Constellations_for_Navigation}
T.~Reid, A.~Neish, T.~Walter, and P.~Enge, ``Broadband {LEO} constellations for navigation,'' \emph{Navigation}, vol.~65, Jun. 2018.

\bibitem{Economical_Fused_LEO_GNSS}
P.~A. Iannucci and T.~E. Humphreys, ``Economical fused {LEO} {GNSS},'' in \emph{{Proc., IEEE/ION Position, Location and Navigation Symposium (PLANS)}}, 2020, pp. 426--443.

\bibitem{Empowering_the_Tracking_Performance_of_LEOBased_Positioning_by_Means_of_Meta_Signals}
A.~Nardin, F.~Dovis, and J.~A. Fraire, ``Empowering the tracking performance of {LEO}-based positioning by means of meta-signals,'' \emph{IEEE J. of Radio Frequency Identification}, vol.~5, no.~3, pp. 244--253, Sep. 2021.

\bibitem{Performance_Analysis_of_a_Multi_Slope_Chirp_Spread_Spectrum}
D.~Egea-Roca, J.~López-Salcedo, G.~Seco-Granados, and E.~Falletti, ``Performance analysis of a multi-slope chirp spread spectrum signal for pnt in a {LEO} constellation,'' in \emph{proc., Satellite Navigation Technology Workshop (NAVITEC)}, 2022, pp. 1--9.

\bibitem{Integrated_Communications_and_Localization_for_Massive_MIMO_LEO_Satellite}
L.~You, X.~Qiang, Y.~Zhu, F.~Jiang, C.~G. Tsinos, W.~Wang, H.~Wymeersch, X.~Gao, and B.~Ottersten, ``Integrated communications and localization for massive {MIMO} {LEO} satellite systems,'' \emph{IEEE Trans. on Wireless Commun.}, to appear.

\bibitem{Psiaki2020NavigationUC}
\BIBentryALTinterwordspacing
M.~L. Psiaki, ``Navigation using carrier doppler shift from a {LEO} constellation: {TRANSIT} on steroids,'' \emph{NAVIGATION}, vol.~68, no.~3, pp. 621--641, 2021. [Online]. Available: \url{https://onlinelibrary.wiley.com/doi/abs/10.1002/navi.438}
\BIBentrySTDinterwordspacing

\bibitem{Kassas2019NewAgeSN}
\BIBentryALTinterwordspacing
Z.~M. Kassas, J.~Morales, and J.~J. Khalife, ``New-age satellite-based navigation -- {STAN}: Simultaneous tracking and navigation with {LEO} satellite signals,'' 2019. [Online]. Available: \url{https://api.semanticscholar.org/CorpusID:214300844}
\BIBentrySTDinterwordspacing

\bibitem{Navigation_With_Differential_Carrier_Phase_Measurements_From_Megaconstellation_LEO_Satellites}
J.~Khalife, M.~Neinavaie, and Z.~M. Kassas, ``Navigation with differential carrier phase measurements from megaconstellation {LEO} satellites,'' in \emph{{Proc., IEEE/ION Position, Location and Navigation Symposium (PLANS)}}, 2020, pp. 1393--1404.

\bibitem{A_Hybrid_Analytical_Machine_Learning_Approach_for_LEO_Satellite_Orbit_Prediction}
J.~Haidar-Ahmad, N.~Khairallah, and Z.~M. Kassas, ``A hybrid analytical-machine learning approach for {LEO} satellite orbit prediction,'' in \emph{proc., International Conference on Information Fusion (FUSION)}, 2022, pp. 1--7.

\bibitem{Doppler_effect_Downlink_Receivers_OFDM_Low_earth_orbit_satellites_Bandwidth_Synchronization_Doppler_positioning_low_Earth_orbit}
M.~Neinavaie, J.~Khalife, and Z.~M. Kassas, ``Acquisition, doppler tracking, and positioning with starlink {LEO} satellites: First results,'' \emph{IEEE Trans. on Aerospace and Electronic Systems}, vol.~58, no.~3, pp. 2606--2610, Apr. 2022.

\bibitem{Ad_Astra_STAN_With_Megaconstellation_LEO_Satellites}
Z.~M. Kassas, N.~Khairallah, and S.~Kozhaya, ``Ad astra: Simultaneous tracking and navigation with megaconstellation {LEO} satellites,'' \emph{IEEE Aerospace and Electronic Systems Magazine}, to appear.

\bibitem{A_Hybrid_Analytical_Machine_Learning_Approach_for_LEO_Satellite_Orbit_Prediction_1}
J.~Haidar-Ahmad, N.~Khairallah, and Z.~M. Kassas, ``A hybrid analytical-machine learning approach for {LEO} satellite orbit prediction,'' in \emph{25th International Conference on Information Fusion (FUSION)}, 2022, pp. 1--7.

\bibitem{Assessing_Machine_Learning_for_LEO_Satellite_Orbit_Determination_in_Simultaneous_Tracking_and_Navigation}
T.~Mortlock and Z.~M. Kassas, ``Assessing machine learning for {LEO} satellite orbit determination in simultaneous tracking and navigation,'' in \emph{IEEE Aerospace Conference (50100)}, 2021, pp. 1--8.

\bibitem{Cognitive_Navigation_With_Unknown_OFDM_signals_With_Application_Terrestrial_5G_Starlink}
M.~Neinavaie and Z.~M. Kassas, ``Cognitive sensing and navigation with unknown {OFDM} signals with application to terrestrial {5G} and starlink {LEO} satellites,'' \emph{IEEE J on Selected Areas in Commun.}, vol.~42, no.~1, pp. 146--160, Jan. 2024.

\bibitem{Observability_Analysis_of_Receiver_Localization_Pseudorange}
R.~Sabbagh and Z.~M. Kassas, ``Observability analysis of receiver localization via pseudorange measurements from a single {LEO} satellite,'' \emph{IEEE Control Systems Lett.}, vol.~7, pp. 571--576, Jun. 2023.

\bibitem{Positioning_with_Starlink_LEO_Satellites_A_Blind_Doppler_Spectral_Approach}
S.~E. Kozhaya and Z.~M. Kassas, ``Positioning with starlink {LEO} satellites: A blind {Doppler} spectral approach,'' in \emph{IEEE 97th Vehicular Technology Conference (VTC2023-Spring)}, 2023, pp. 1--5.

\bibitem{Receiver_Design_for_Doppler_Positioning_with_Leo_Satellites}
J.~J. Khalife and Z.~M. Kassas, ``Receiver design for {Doppler} positioning with {LEO} satellites,'' in \emph{IEEE International Conference on Acoustics, Speech and Signal Processing (ICASSP)}, 2019, pp. 5506--5510.

\bibitem{Unveiling_Starlink_LEO_Satellite_OFDM_Like_Signal_Structure_Enabling_Precise_Positioning}
M.~Neinavaie and Z.~M. Kassas, ``Unveiling starlink {LEO} satellite {OFDM}-like signal structure enabling precise positioning,'' \emph{IEEE Trans. on Aerospace and Electronic Systems}, vol.~60, no.~2, pp. 2486--2489, Apr. 2024.

\bibitem{garcia2017direct}
N.~Garcia, H.~Wymeersch, E.~G. Larsson, A.~M. Haimovich, and M.~Coulon, ``Direct localization for massive {MIMO},'' \emph{IEEE Trans. on Signal Processing}, vol.~65, no.~10, pp. 2475--2487, May 2017.

\bibitem{8240645}
A.~Shahmansoori, G.~E. Garcia, G.~Destino, G.~Seco-Granados, and H.~Wymeersch, ``Position and orientation estimation through millimeter-wave {MIMO in {5G}} systems,'' \emph{IEEE Trans. on Wireless Commun.}, vol.~17, no.~3, pp. 1822--1835, Mar. 2018.

\bibitem{8515231}
R.~Mendrzik, H.~Wymeersch, G.~Bauch, and Z.~Abu-Shaban, ``Harnessing {NLOS} components for position and orientation estimation in {5G} millimeter wave {MIMO},'' \emph{IEEE Trans. on Wireless Commun.}, vol.~18, no.~1, pp. 93--107, Jan. 2019.

\bibitem{8356190}
Z.~Abu-Shaban, X.~Zhou, T.~Abhayapala, G.~Seco-Granados, and H.~Wymeersch, ``Error bounds for uplink and downlink {3D} localization in {5G} millimeter wave systems,'' \emph{IEEE Trans. on Wireless Commun.}, vol.~17, no.~8, pp. 4939--4954, Aug. 2018.

\bibitem{guerra2018single}
A.~Guerra, F.~Guidi, and D.~Dardari, ``Single-anchor localization and orientation performance limits using massive arrays: {MIMO vs. beamforming},'' \emph{IEEE Trans. on Wireless Commun.}, vol.~17, no.~8, pp. 5241--5255, Aug. 2018.

\bibitem{emenonye2023limits}
D.-R. Emenonye, H.~S. Dhillon, and R.~M. Buehrer, ``On the limits of single anchor localization: Near-field vs. far-field,'' \emph{IEEE Wireless Commun. Lett.}, Feb. 2023.

\bibitem{fascista2021downlink}
A.~Fascista, A.~Coluccia, H.~Wymeersch, and G.~Seco-Granados, ``Downlink single-snapshot localization and mapping with a single-antenna receiver,'' \emph{IEEE Trans. on Wireless Commun.}, vol.~20, no.~7, pp. 4672--4684, July 2021.

\bibitem{8755880}
------, ``Millimeter-wave downlink positioning with a single-antenna receiver,'' \emph{IEEE Trans. on Wireless Commun.}, vol.~18, no.~9, pp. 4479--4490, Sep. 2019.

\bibitem{li2019massive}
X.~Li, E.~Leitinger, M.~Oskarsson, K.~{\AA}str{\"o}m, and F.~Tufvesson, ``Massive {MIMO}-based localization and mapping exploiting phase information of multipath components,'' \emph{IEEE Trans. on Wireless Commun.}, vol.~18, no.~9, pp. 4254--4267, Sep. 2019.

\bibitem{9082200}
F.~Ghaseminajm, Z.~Abu-Shaban, S.~S. Ikki, H.~Wymeersch, and C.~R. Benson, ``Localization error bounds for {5G mmWave} systems under {I/Q} imbalance,'' \emph{IEEE Trans. on Veh. Technol.}, vol.~69, no.~7, pp. 7971--7975, July 2020.

\bibitem{5571889}
Y.~Shen, H.~Wymeersch, and M.~Z. Win, ``Fundamental limits of wideband localization — part {II}: Cooperative networks,'' \emph{IEEE Trans. on Info. Theory}, vol.~56, no.~10, pp. 4981--5000, Oct. 2010.

\bibitem{9606768}
Y.~Xiong, N.~Wu, Y.~Shen, and M.~Z. Win, ``Cooperative localization in massive networks,'' \emph{IEEE Trans. on Info. Theory}, vol.~68, no.~2, pp. 1237--1258, Feb. 2022.

\bibitem{7364259}
Y.~Han, Y.~Shen, X.-P. Zhang, M.~Z. Win, and H.~Meng, ``Performance limits and geometric properties of array localization,'' \emph{IEEE Trans. on Info. Theory}, vol.~62, no.~2, pp. 1054--1075, Feb. 2016.

\bibitem{5571900}
Y.~Shen and M.~Z. Win, ``Fundamental limits of wideband localization — part {I}: A general framework,'' \emph{IEEE Trans. on Info. Theory}, vol.~56, no.~10, pp. 4956--4980, Oct. 2010.

\bibitem{8264743}
S.~Hu, F.~Rusek, and O.~Edfors, ``Beyond massive {MIMO}: The potential of positioning with large intelligent surfaces,'' \emph{IEEE Trans. on Signal Processing}, vol.~66, no.~7, pp. 1761--1774, Apr. 2018.

\bibitem{9729782}
Z.~Wang, Z.~Liu, Y.~Shen, A.~Conti, and M.~Z. Win, ``Location awareness in beyond {5G} networks via reconfigurable intelligent surfaces,'' \emph{IEEE J. Sel. Areas Commun.}, vol.~40, no.~7, pp. 2011--2025, July 2022.

\bibitem{9781656}
M.~Z. Win, Z.~Wang, Z.~Liu, Y.~Shen, and A.~Conti, ``Location awareness via intelligent surfaces: A path toward holographic {NLN},'' \emph{IEEE Veh. Technology Magazine}, vol.~17, no.~2, pp. 37--45, June 2022.

\bibitem{9508872}
A.~Elzanaty, A.~Guerra, F.~Guidi, and M.-S. Alouini, ``Reconfigurable intelligent surfaces for localization: Position and orientation error bounds,'' \emph{IEEE Trans. on Signal Processing}, vol.~69, pp. 5386--5402, Aug. 2021.

\bibitem{9625826}
D.~Dardari, N.~Decarli, A.~Guerra, and F.~Guidi, ``{LOS/NLOS} near-field localization with a large reconfigurable intelligent surface,'' \emph{IEEE Trans. on Wireless Commun.}, vol.~21, no.~6, pp. 4282--4294, June 2022.

\bibitem{9500663}
Z.~Abu-Shaban, K.~Keykhosravi, M.~F. Keskin, G.~C. Alexandropoulos, G.~Seco-Granados, and H.~Wymeersch, ``Near-field localization with a reconfigurable intelligent surface acting as lens,'' in \emph{Proc., IEEE Intl. Conf. on Commun. (ICC)}, 2021.

\bibitem{9782100}
A.~Fascista, M.~F. Keskin, A.~Coluccia, H.~Wymeersch, and G.~Seco-Granados, ``{RIS}-aided joint localization and synchronization with a single-antenna receiver: Beamforming design and low-complexity estimation,'' \emph{IEEE J. of Sel. Topics in Signal Processing}, to appear.

\bibitem{9528041}
K.~Keykhosravi, M.~F. Keskin, S.~Dwivedi, G.~Seco-Granados, and H.~Wymeersch, ``Semi-passive {3D} positioning of multiple {RIS}-enabled users,'' \emph{IEEE Trans. on Veh. Technol.}, vol.~70, no.~10, pp. 11\,073--11\,077, Oct. 2021.

\bibitem{emenonye2022fundamentals}
D.-R. Emenonye, H.~S. Dhillon, and R.~M. Buehrer, ``Fundamentals of {RIS}-aided localization in the far-field,'' \emph{IEEE Trans. on Wireless Commun.}, vol.~23, no.~4, pp. 3408--3424, Apr. 2024.

\bibitem{emenonye2023_ICC_conf_workshop}
------, ``Limitations of {RIS}-aided localization: Inspecting the relationships between channel parameters,'' in \emph{{Proc., IEEE Intl. Conf. on Commun. (ICC) Workshop}}, 2023, pp. 1890--1895.

\bibitem{emenonye2022ris}
------, ``{RIS}-aided localization under position and orientation offsets in the near and far field,'' \emph{IEEE Trans. on Wireless Commun.}, vol.~22, no.~12, pp. 9327--9345, Dec. 2023.

\bibitem{emenonye2023_ICC_conf}
------, ``Estimation of {RIS} misorientation in both near and far field regimes,'' in \emph{{Proc., IEEE Intl. Conf. on Commun. (ICC)}}, 2023, pp. 2013--2018.

\bibitem{RIS_Aided_Kinematic}
D.-R. Emenonye, A.~Sarker, A.~T. Asbeck, H.~S. Dhillon, and R.~M. Buehrer, ``{RIS}-aided kinematic analysis for remote rehabilitation,'' \emph{IEEE Sensors Journal}, vol.~23, no.~19, pp. 22\,679--22\,692, Oct. 2023.

\bibitem{OTFS_Enabled_RIS}
D.-R. Emenonye, A.~Pradhan, H.~S. Dhillon, and R.~M. Buehrer, ``{OTFS} enabled {RIS}-aided localization: Fundamental limits and potential drawbacks,'' in \emph{{Proc., IEEE/ION Position, Location and Navigation Symposium (PLANS)}}, 2023, pp. 344--353.

\bibitem{9774917}
K.~Keykhosravi, M.~F. Keskin, G.~Seco-Granados, P.~Popovski, and H.~Wymeersch, ``{RIS}-enabled {SISO} localization under user mobility and spatial-wideband effects,'' \emph{IEEE J. of Sel. Topics in Signal Processing}, pp. 1--1, to appear.

\bibitem{9573365}
M.~Neinavaie, J.~Khalife, and Z.~M. Kassas, ``Cognitive opportunistic navigation in private networks with {5G} signals and beyond,'' \emph{IEEE J. of Selected Topics in Signal Processing}, vol.~16, no.~1, pp. 129--143, Jan. 2022.

\bibitem{9369049}
K.~Shamaei and Z.~M. Kassas, ``Receiver design and time of arrival estimation for opportunistic localization with {5G} signals,'' \emph{IEEE Trans. on Wireless Commun.}, vol.~20, no.~7, pp. 4716--4731, July 2021.

\bibitem{kassas2021carpe}
Z.~Kassas, A.~Abdallah, and M.~Orabi, ``{Carpe signum}: seize the signal--opportunistic navigation with {5G},'' 2021.

\bibitem{lavalle2006planning}
S.~M. LaValle, \emph{{Planning Algorithms}}.\hskip 1em plus 0.5em minus 0.4em\relax Cambridge University Press, 2006.

\bibitem{kay1993fundamentals}
S.~M. Kay, \emph{{Fundamentals of Statistical Signal Processing: Estimation Theory}}.\hskip 1em plus 0.5em minus 0.4em\relax Prentice-Hall, Inc., 1993.

\bibitem{horn2012matrix}
R.~A. Horn and C.~R. Johnson, \emph{{Matrix Analysis}}.\hskip 1em plus 0.5em minus 0.4em\relax Cambridge University Press, 2012.

\end{thebibliography}
}
\end{document}